\documentclass[12pt]{PisaPhdThesis_LaTeX}

\usepackage{amsfonts}
\usepackage{amsmath}
\usepackage{amsthm}
\usepackage{amssymb}
\usepackage{graphicx}
\usepackage{subfigure}
\usepackage[hyphens]{url}
\usepackage{hyperref}
\usepackage{xspace}
\usepackage{setspace}

\newtheoremstyle{mystyle}{}{}{\itshape}{}{\bfseries}{.}{6 pt}{\thmname{#1}\thmnumber{ #2}\thmnote{ {\bfseries(#3)}}}
\theoremstyle{mystyle}
\newtheorem{theorem}{Theorem}[chapter]
\newtheorem{remark}[theorem]{Remark}
\newtheorem{observation}[theorem]{Observation}
\newtheorem{corollary}[theorem]{Corollary}
\newtheorem{proposition}[theorem]{Proposition}
\newtheorem{lemma}[theorem]{Lemma}
\newtheorem{conjecture}[theorem]{Conjecture}
\newtheorem{definition}[theorem]{Definition}
\numberwithin{equation}{chapter}

\newcommand{\R}{\mathbb{R}}
\newcommand{\qq}[1]{\textquotedblleft #1\textquotedblright\xspace}
\newcommand{\complexityclass}[1]{\textbf{#1}}
\newcommand{\computproblem}[1]{\textsc{#1}}
\renewcommand{\P}{\complexityclass{P}\xspace}
\newcommand{\NP}{\complexityclass{NP}\xspace}
\newcommand{\APX}{\complexityclass{APX}\xspace}
\newcommand{\PTAS}{\complexityclass{PTAS}\xspace}
\newcommand{\LOGSPACE}{\complexityclass{L}\xspace}

\newcommand{\PSPACE}{\complexityclass{PSPACE}\xspace}
\newcommand{\NPSPACE}{\complexityclass{NPSPACE}\xspace}

\newcommand{\TSAT}{\computproblem{3SAT}\xspace}
\newcommand{\SSP}{\computproblem{SSP}\xspace}
\newcommand{\SSPext}{\computproblem{Searchlight Scheduling Problem}\xspace}
\newcommand{\PSSPext}{\computproblem{Partial Searchlight Scheduling Problem}\xspace}
\newcommand{\PSSP}{\computproblem{PSSP}\xspace}

\newcommand{\TPSSP}{\computproblem{3PSSP}\xspace}
\newcommand{\TSSP}{\computproblem{3SSP}\xspace}
\newcommand{\TSSPext}{\computproblem{3-dimensional Searchlight Scheduling Problem}\xspace}
\newcommand{\TRSSP}{\computproblem{3PSSP}\xspace}
\newcommand{\TTSSPext}{\computproblem{3-dimensional Timed Searchlight Scheduling Problem}\xspace}
\newcommand{\TTSSP}{\computproblem{3TSSP}\xspace}
\newcommand{\ART}{\computproblem{Art Gallery Problem}\xspace}
\newcommand{\EENCL}{\computproblem{EE-NCL}\xspace}
\newcommand{\EEANCL}{\computproblem{EE-ANCL}\xspace} 

\newcommand{\open}[1]{\widetilde{#1}}

\makeatletter
\newcommand{\Rmnum}[1]{\expandafter\@slowromancap\romannumeral #1@}
\makeatother

\begin{document}

\title{Guarding and Searching Polyhedra}
\author{\textsc{Giovanni Viglietta}}
\supervisor{Prof.\ Linda Pagli}
\referee{Prof.\ Peter Widmayer\and Prof.\ Masafumi Yamashita}
\chair{Prof.\ Fabrizio Broglia}
\phdnumber{\Rmnum{24}}

\maketitle
\onehalfspacing
\pagestyle{empty}

\begin{abstract}
Guarding and searching problems have been of fundamental interest since the early years of Computational Geometry. Both are well-developed areas of research and have been thoroughly studied in planar polygonal settings.

In this thesis we tackle the \ART and the \SSPext in 3-dimensional polyhedral environments, putting special emphasis on edge guards and orthogonal polyhedra.

We solve the \ART with reflex edge guards in orthogonal polyhedra having reflex edges in just two directions: generalizing a classic theorem by O'Rourke, we prove that $\left\lfloor r/2 \right\rfloor+1$ reflex edge guards are sufficient and occasionally necessary, where $r$ is the number of reflex edges. We also show how to compute guard locations in $O(n\log n)$ time.

Then we investigate the \ART with mutually parallel edge guards in orthogonal polyhedra with $e$ edges, showing that $\left\lfloor 11e/72 \right\rfloor$ edge guards are always sufficient and can be found in linear time, improving upon the previous state of the art, which was $\left\lfloor e/6 \right\rfloor$. We also give tight inequalities relating $e$ with the number of reflex edges $r$, obtaining an upper bound on the guard number of $\left\lfloor 7r/12 \right\rfloor+1$.

We further study the \ART with edge guards in polyhedra having faces oriented in just four directions, obtaining a lower bound of $\left\lfloor e/6 \right\rfloor-1$ edge guards and an upper bound of $\left\lfloor (e+r)/6 \right\rfloor$ edge guards.

All the previously mentioned results hold for polyhedra of any genus. Additionally, several guard types and guarding modes are discussed, namely open and closed edge guards, and orthogonal and non-orthogonal guarding.

Next, we model the \TSSPext,  the problem of searching a given polyhedron by suitably turning some half-planes
around their axes, in order to catch an evasive intruder. After discussing several generalizations of classic theorems, we study the problem of efficiently placing guards in a given polyhedron, in order to make it searchable. For general polyhedra, we give an upper bound of $r^2$ on the number of guards, which reduces to $r$ for orthogonal polyhedra.

Then we prove that it is strongly \NP-hard to decide if a given polyhedron is entirely searchable by a given set of guards. We further prove that, even under the assumption that an orthogonal polyhedron is searchable, approximating the minimum search time within a small-enough constant factor to the optimum is still strongly \NP-hard.

Finally, we show that deciding if a specific region of an orthogonal polyhedron is searchable is strongly \PSPACE-hard. By further improving our construction, we show that the same problem is strongly \PSPACE-complete even for planar orthogonal polygons. Our last results are especially meaningful because no similar hardness theorems for 2-dimensional scenarios were previously known.
\end{abstract}

\begin{dedication}
To Gianni
\end{dedication}

\begin{acknowledgments}
I would like to express my gratitude to my supervisor Linda Pagli, who has been an immense source of encouragement and advice throughout the last three years, along with Roberto Grossi and Giuseppe Prencipe.

I am also thankful to Joseph O'Rourke for offering me constant support and for hosting me at Smith College for three months. During my stay, he has been a teacher, a colleague and a friend, providing me with everything I could possibly need.

I am grateful to the Computational Geometry research group at MIT, and particularly to Erik and Martin Demaine, for giving me a chance to take part in their lively and inspiring meetings.

I deeply appreciated the precious advice of Herbert Edelsbrunner, James King, Jorge Urrutia, and Masafumi Yamashita.

Finally, I would like to thank all my coauthors, who contributed to raising my research skills and lowering my Erd\H{o}s number. I wish them the best of luck in their future endeavors. In alphabetical order, I thank Zachary Abel, Greg Aloupis, Nadia M.\ Benbernou, Erik D.\ Demaine, Martin L.\ Demaine, Sarah Eisenstat, Alan Guo, Anastasia Kurdia, Anna Lubiw, Maurizio Monge, Joseph O'Rourke, Linda Pagli, Giuseppe Prencipe, Andr\'e Schulz, Diane L.\ Souvaine, Godfried T.\ Toussaint, Jorge Urrutia, and Andrew Winslow.
\\

\hfill
{\it Giovanni Viglietta}
\end{acknowledgments}

\frontmatter
\pagestyle{underheadings}
\tableofcontents
\markthissection{CONTENTS}
\markthischapter{CONTENTS}

\begin{introduction}
\markthischapter{INTRODUCTION}\markthissection{INTRODUCTION}
Recent advances in robotics and distributed computing have made it a realistic goal to employ robots in surveillance and tracking, and also motivated countless theoretical studies in several areas of Computational Geometry, particularly concerning vision and motion planning. In the last decades, in order to capture the key aspects of these applications in simple, yet meaningful ways, several theoretical models have been proposed and studied. 

Arguably the oldest and most fundamental problem of this kind is the \ART. In its original formulation, the problem consists in finding a (minimum) set of points in a given polygonal room from which the entire room is visible. The sides of the polygon model the room's walls, and the chosen points would be the locations of static guards, who collectively see the whole room and protect it from intruders.

The \ART has been generalized and extended in several directions, encompassing wider classes of geometric shapes, such as polygons with holes. At the same time, stronger results were obtained for special shapes, such as orthogonal polygons.

On the other hand, the problem may also be viewed as a pursuit-evasion game, in which some guards attempt to capture a moving intruder. This scenario can be modulated not only geometrically, but also based on the sensorial and motional capabilities of the agents involved. Here, the problem shifts from statically \emph{guarding} an environment to dynamically \emph{searching} it.

In the \SSPext, guards are static, but carry a laser beam that can be turned around, in search of a moving intruder. The goal of the guards is to move their lasers in concert, according to a fixed schedule, and detect the intruder with a laser in a bounded amount of time, no matter how fast he moves and which path he takes. The goal of the intruder is to escape the lasers and remain forever unseen.

Despite their seeming simplicity, both problems have proven quite challenging. While the \ART, as an optimization problem, is known to be \NP-hard, the exact complexity of the \SSPext is still unknown, and the capabilities of various types of searchers are not well-understood either.

Very little is known about guarding and searching higher-dimensional environments. None of the most fruitful properties of polygons, such as triangulability and the existence of an inherent ordering of edges and vertices (e.g., clockwise order), can be meaningfully generalized to higher dimensions. Even for 3-dimensional polyhedral environments, a small set of negative results is known and, to the best of our knowledge, major non-trivial positive theorems, algorithms or techniques have yet to be found.

In this thesis, we set out to investigate the \ART and the \SSPext in 3-dimensional polyhedra. The purpose of our research is twofold: on the one hand, we seek to generalize some well-established theorems, as well as to provide entirely new results; on the other hand, we aim at gaining better insights on some long-standing open problems in 2-dimensional scenarios, by abstracting them and expanding our perspective.

\subsection*{Thesis structure and contributions}

\paragraph{Part~\ref{part1}.}
This part is devoted to background information and some preliminary results.

Chapter~\ref{chapter1} surveys some well-known results on the \ART and the \SSPext in the plane. We do not attempt to give a thorough account of the state of the art of these problems, but we focus on fundamental results and techniques, and a few additional aspects that will be relevant in later chapters. 

In Chapter~\ref{chapter2}, we define our notion of polyhedron, pointing out the main features of our model. 
Historically, correct definitions of polyhedra have been quite elusive, and this is an especially delicate task in our thesis, because different polyhedral models yield somewhat different theorems.

In Chapter~\ref{chapter3}, we define several types of guards and guarding modes in polyhedra. In particular, the notions of \emph{open segment guard} and \emph{orthogonal guarding} are novel contributions. We also prove some basic results concerning face guards and point guards, concluding with some upper and lower bounds on edge guard numbers, and the computational hardness of the minimum edge guard problem.

\paragraph{Part~\ref{part2}.}
We dedicate this part to the \ART with edge guards in polyhedra, which we believe to be the most natural extension of the original 2-dimensional \ART with vertex guards in polygons, and more likely to yield similar results. We are mainly concerned with orthogonal polyhedra, which enjoy several deep structural properties that we thoroughly investigate and put to use.

In Chapter~\ref{chapter4}, we focus on orthogonal polyhedra with reflex edges in only two directions (as opposed to three). We generalize a classic theorem by O'Rourke, providing a tight bound on the number of required reflex edge guards, in terms of the number of reflex edges in the polyhedron. We also obtain upper bounds in terms of the total number of edges, and we formulate some conjectures. To conclude the chapter, we present an $O(n\log n)$ time algorithm to compute a guard set matching our upper bounds.

In Chapter~\ref{chapter5}, we consider the problem of guarding orthogonal polyhedra with mutually parallel edge guards. We give an upper bound on the number of guards that improves on the previous state of the art, due to Urrutia. En route to this result, we also give tight inequalities involving the number of reflex edges in an orthogonal polyhedron and the total number of edges.

Chapter~\ref{chapter6} extends the techniques used in Chapter~\ref{chapter5} to polyhedra with faces oriented in four directions (as opposed to three). We obtain new upper and lower bounds, and formulate some conjectures.

\paragraph{Part~\ref{part3}.}
Here we investigate an extension of the \SSPext to polyhedral environments.

In Chapter~\ref{chapter7}, we carefully detail our 3-dimensional guard model. There are several meaningful ways to extend the classic 2-dimensional model, and we choose to employ segment guards who rotate half-planes of light with one degree of freedom. We thoroughly motivate our choice with observations and examples.

In Chapter~\ref{chapter8}, we study one special type of guard, which we call \emph{filling} guard. We show how filling guards somewhat act as boundary guards in 2-dimensional \SSPext, in that they enjoy the same positive properties, and enable an extension of the classic ``one-way sweep strategy'' by Sugihara, Suzuki and Yamashita. As a further application of the concept of filling guard, we characterize the searchable problem instances containing only one guard. Finally, we give a polynomial time algorithm to decide if a guard is filling.

In Chapter~\ref{chapter9}, we consider the problem of placing guards in a given polyhedron, in order to make it searchable. We give a quadratic upper bound on the number of guards that holds for general polyhedra, and we further show that placing one guard on each reflex edge suffices for orthogonal polyhedra.

Chapter~\ref{chapter10} is devoted to computational complexity issues. We first show that deciding if a given problem instance is searchable is strongly \NP-hard. Then we consider the problem of minimizing search time, with the promise that the input instance is searchable. We prove that approximating the minimum search time is \NP-hard, even for orthogonal polyhedra.

Finally, in Chapter~\ref{chapter11}, we consider a generalized \SSPext, in which we do not necessarily have to search the entire environment, but just a subregion given as input. We prove that the 3-dimensional version of this \PSSPext is strongly \PSPACE-hard, and that its 2-dimensional version is \PSPACE-complete. Our last result is especially meaningful because it is the first hardness theorem related to the 2-dimensional \SSPext, and it comes as a more sophisticated version of its 3-dimensional counterpart.

\paragraph{Publications.}
Most of the material in this thesis has already been published. Publications include the upper bounds on the number of parallel edge guards in orthogonal polyhedra (joint work with Benbernou et al.,~\cite{viglietta4}), most of Part~\ref{part3} (sole author,~\cite{viglietta3}), the \NP-hardness of minimizing search time in orthogonal polyhedra (joint work with Monge,~\cite{viglietta}), and the \PSPACE-completeness of partially searching orthogonal polygons (sole author,~\cite{viglietta2}).
\end{introduction}

\mainmatter

\part{Preliminaries}\label{part1}
\chapter{Background}\label{chapter1}
\begin{chapterabstract}
We survey some classic results on the \ART and the \SSPext for planar polygons. Rather than providing a complete account of the state of the art for these two problems, we illustrate some fundamental theorems and techniques that we are going to generalize or employ in later chapters.

Specifically, we delineate the role of partitioning polygons as the first step in computing guard locations for the \ART, and we point out that minimizing the number of guards is computationally hard.

Then we survey the main results on the \SSPext, such as the one-way sweep strategy and its consequences, and the sequentializability of search schedules.
\end{chapterabstract}

\section{\ART}
Consider a polygonal region $P$, representing an empty room. We say that two points in the plane can \emph{see} each other if the line segment connecting them lies entirely inside $P$ or on its boundary. The classic \ART consists in finding the minimum number $g(n)$ such that, in any polygon of $n$ vertices, there exists a set of $g(n)$ points, called \emph{guards}, that collectively see the whole interior of $P$.

Several results related to the \ART are surveyed in \cite{art,shermer,urrutia2000,zylinski}. In this section, we describe those that are most significant to our thesis.

\subsection*{General polygons}
\paragraph{Simple polygons.}
The \ART was originally posed for simple polygons (i.e., polygons with no holes and whose boundary does not self-intersect) by Klee at a conference in Stanford, in 1973. It was later established by Chv\'atal that $$\left\lfloor \frac n 3 \right\rfloor$$
guards are necessary and sufficient to guard simple polygons with $n$ vertices (see~\cite{chvatal}), and Chv\'atal's proof was simplified by Fisk shortly after (see~\cite{fisk}).

We sketch Fisk's proof. Let $P$ be a polygon with $n$ vertices. It is well-known that $P$'s interior is partitionable into $n-2$ disjoint triangles by drawing $n-3$ diagonals. Such a decomposition is called \emph{triangulation}, and can be computed in linear time (see~\cite{chazelletriang,art}). Also, a triangulation of a polygon is a plane graph on its vertex set, whose dual is a tree. By an inductive argument on the tree structure, one can easily prove that the triangulation graph can be 3-colored, as shown in Figure~\ref{f1:art1a}. Since some color is used on at most $\lfloor n/3 \rfloor$ vertices, we place a guard at every such vertex to guard the entire polygon: indeed, each triangle in the triangulation has vertices of all three colors, and is therefore guarded. The whole process can be performed in linear time.

\begin{figure}[h]
\centering{\includegraphics[width=0.55\linewidth]{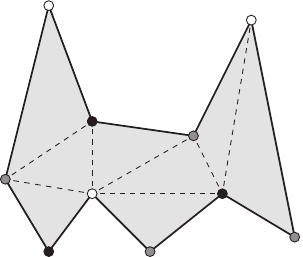}}
\caption{Illustration of Fisk's proof.}
\label{f1:art1a}
\end{figure}

This upper bound on the guard number is also tight, because occasionally $\lfloor n/3 \rfloor$ guards are necessary: consider for example the \qq{comb polygon} depicted in Figure~\ref{f1:art1b}. Each ``tooth'' of the comb requires one guard, as no point sees the tip of two distinct tooths.

\begin{figure}[h]
\centering{\includegraphics[width=0.75\linewidth]{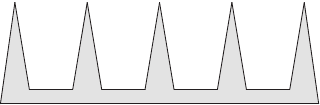}}
\caption{Comb polygon. It has $3k$ vertices and requires $k$ guards.}
\label{f1:art1b}
\end{figure}

\paragraph{Polygons with holes.}
If $P$ has holes, the previous argument does not hold, because the triangulation's dual is no longer a tree. However, a more sophisticated argument, given in~\cite{holes2}, reveals that
$$\left\lfloor \frac {n+h} 3 \right\rfloor$$
guards are always sufficient and occasionally necessary, where $h$ is the number of holes.

\paragraph{Reflex vertices.}
Fisk's method does not exploit $P$'s shape at all, and this yields particularly poor results when $P$ is convex: a single guard is sufficient in this case, but Fisk's  algorithm still places a linear amount of guards.

Let $r$ be the number of reflex vertices in $P$ (i.e., the vertices whose corresponding internal angle exceeds $180^\circ$). Choose a concave internal angle $\alpha$, and cut $P$ along $\alpha$'s bisector, obtaining two polygons with at most $r-1$ reflex vertices in total. By inductively repeating the process, $P$ gets partitioned into at most $r+1$ convex parts, and each part is such that one of its vertices is a reflex vertex of $P$. Therefore, if $r\geqslant 1$, then $r$ guards are always sufficient to guard $P$.
To see that $r$ guards are sometimes necessary, consider the ``shutter polygons'' in Figure~\ref{f1:art2}.

\begin{figure}[h]
\centering{\includegraphics[width=0.8\linewidth]{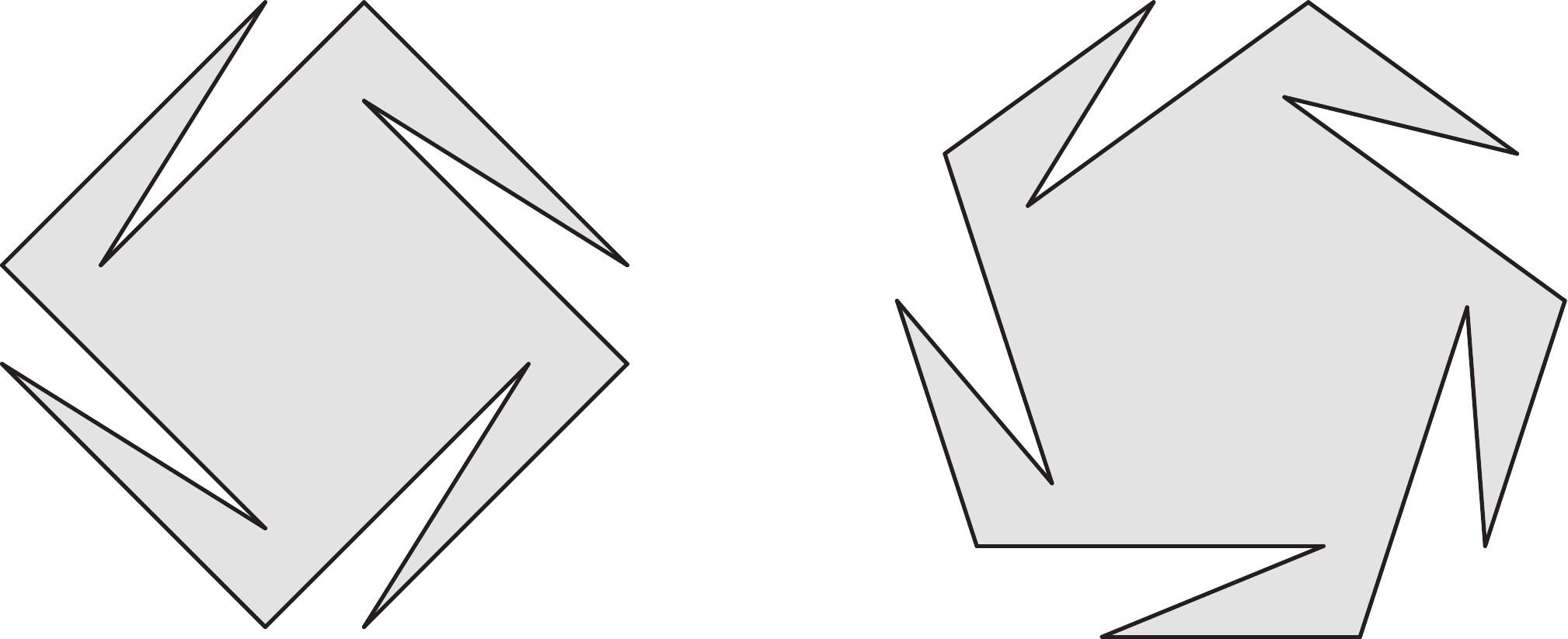}}
\caption{Shutter polygons. At least one guard per reflex vertex is needed.}
\label{f1:art2}
\end{figure}

Notice that $r$ can take any value between $0$ and $n-3$, so this algorithm is occationally better, occasionally worse than the previous one, depending on $P$'s shape.

\subsection*{Orthogonal polygons}
\paragraph{Simple orthogonal polygons.}
Better bounds can be obtained by restricting the \ART to \emph{orthogonal polygons}, i.e., polygons whose edges meet only at right angles. Their simple structure allows for more ingenious partitions, such as \emph{convex quadrilateralizations} (i.e., partitions into convex quadrilaterals whose vertices are also vertices of $P$) and partitions into \emph{L-shaped} pieces (i.e., orthogonal hexagons). Both partitions can be employed to design a linear time algorithm that places
$$\left\lfloor \frac n 4 \right\rfloor = \left\lfloor \frac r 2 \right\rfloor +1$$
guards that collectively see $P$ (refer to~\cite{art,urrutia2000}). A partition into L-shaped pieces is shown in Figure~\ref{f1:orthoa}. Such a partition, originally due to O'Rourke, will be generalized in Chapter~\ref{chapter4} to a wide class of orthogonal polyhedra, yielding a similar upper bound on guard numbers.

\begin{figure}[h]
\centering{\includegraphics[width=0.7\linewidth]{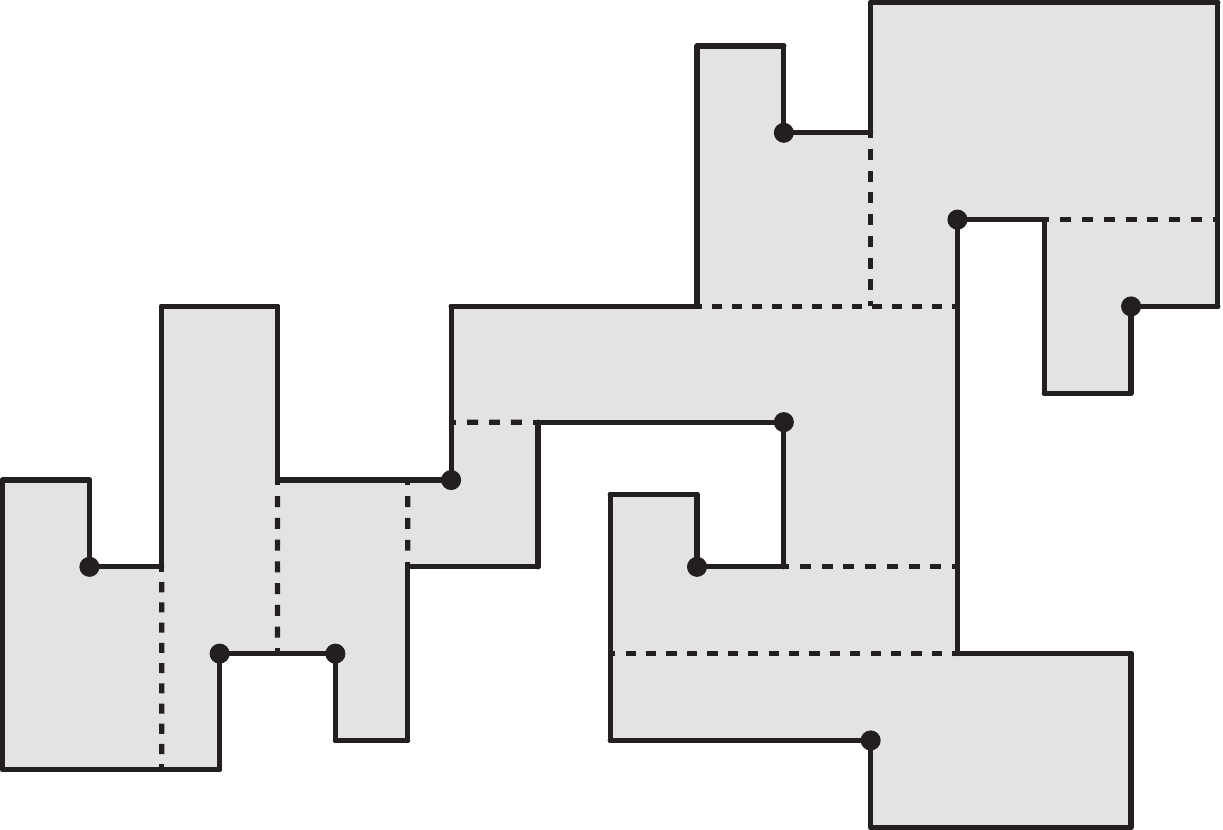}}
\caption{Orthogonal polygon partitioned into L-shaped pieces, and a guarding set.}
\label{f1:orthoa}
\end{figure}

This many guards are occasionally necessary, for example to guard the ``orthogonal combs'' shown in Figure~\ref{f1:orthob}.

\paragraph{Orthogonal polygons with holes.}
If holes are allowed, then the same upper bound of $\lfloor n/4\rfloor$ guards holds, as established by Hoffmann in~\cite{holes3}. However, if guards are required to lie on vertices, no tight bound is known. O'Rourke proved that
$$\left\lfloor \frac{n+2h}{4}\right\rfloor$$
vertex guards are sufficient, and Shermer conjectured that
$$\left\lfloor \frac{n+h}{4}\right\rfloor$$
are necessary and sufficient (see~\cite{art}).

\begin{figure}[h]
\centering{\includegraphics[width=0.65\linewidth]{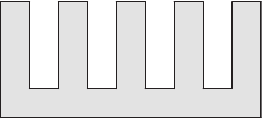}}
\caption{Orthogonal comb. It has $4k$ vertices and requires $k$ guards.}
\label{f1:orthob}
\end{figure}

\subsection*{Patroling guards}
A different guard model was proposed by Avis and Toussaint in~\cite{edge2}. Imagine that a static intruder has to be found by a squad of point guards, each of which ``patrols'' an area. The easiest way to model this problem is to let each guard move back and forth on a line segment completely contained in the polygon. The intruder is caught the moment a guard sees him during the patrol.

The example in Figure~\ref{f1:edgega} reveals that
$$\left\lfloor \frac n4\right\rfloor$$
patroling guards may be necessary.

\begin{figure}[h]
\centering
\subfigure[]{\label{f1:edgega}\includegraphics[scale=2]{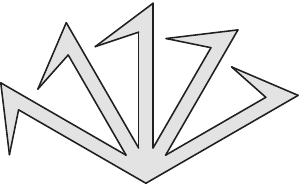}}\qquad
\subfigure[]{\label{f1:edgegb}\includegraphics[scale=2]{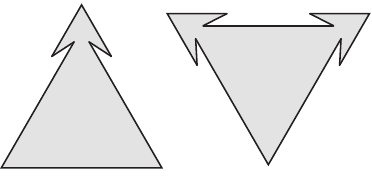}}
\caption{Lower bound constructions for patroling guard numbers.}
\label{f1:edgeg}
\end{figure}

As established by O'Rourke in~\cite{edge1}, $\lfloor n/4\rfloor$ patroling guards are also sufficient to guard any simple polygon, and their paths may be chosen among the polygon's diagonals.

However, Paige and Shermer noted that, if we insist on confining each guard to move on an edge of the polygon, then $\lfloor n/4\rfloor$ guards may not be sufficient, as the two examples in Figure~\ref{f1:edgegb} show (refer to~\cite{shermer}).

Since no substantially different counterexamples are known, Toussaint conjectured that
$$\left\lfloor \frac n4\right\rfloor+o(1)$$
guards patroling edges are sufficient to guard any simple polygon. The best known upper bound on edge guards is
$$\left\lfloor \frac{3n}{10}\right\rfloor+o(1),$$
as proved by Shermer in~\cite{edge4}.

For simple orthogonal polygons, Aggarwal proved in~\cite{edge3} that
$$\left\lfloor \frac{3n+4}{16}\right\rfloor$$
patroling guards are necessary and sufficient, and may even be placed on diagonals.

It follows that, not surprisingly, patroling guards are indeed more ``powerful'' than stationary guards, in all the aforementioned settings.

\subsection*{Minimizing guards}
It was shown by Lee and Lin in~\cite{lee} that minimizing the number of static point guards required to guard a given simple polygon is strongly \NP-hard. A similar result for orthogonal polygons is also known. In Chapter~\ref{chapter3}, we will propose a variation that can be also extended to the problem of edge-guarding orthogonal polyhedra.

Recently, in~\cite{king1}, King and Kirkpatrick gave an $O(\log \log \rm{OPT})$-approximation algorithm to guard a simple polygon with guards lying on the perimeter.

On the other hand, the problem of minimizing guards in simple polygons was shown to be \APX-hard by Eidenbenz in~\cite{artapprox3}, but whether it is in \APX remains an open problem. With no restrictions on the number of holes, the problem is as hard to approximate as \computproblem{SET COVER}, as shown by Eidenbenz, Stamm and Widmayer in~\cite{artapprox2}. As a consequence, it is \NP-hard to find a guarding set of size $o(\log n)$.

\section{\SSPext}
\subsection*{Basics}
The \SSPext (\SSP) was first studied in~\cite{search1} by Sugihara, Suzuki and Yamashita as a more ``dynamic'' version of the \ART. It is a search problem in simple polygons, where some stationary guards are tasked to locate an evasive, moving intruder by illuminating him with \emph{searchlights}. Each guard carries a searchlight, modeled as a 1-dimensional ray that can be continuously rotated about the guard itself, whereas the intruder follows an ``unpredictable'' continuous path, running at unbounded speed and trying to avoid the searchlights. Since the guards cannot know the position of the intruder until they catch him in their lights, the movements of the searchlights must follow a fixed \emph{schedule}, which should guarantee that the intruder is caught in finite time, regardless of the path he decides to take. The search takes place in a polygonal region, whose sides act as obstacles both for the intruder's movements and for the guards' searchlights. In a way, the polygonal boundary benefits the intruder, who can hide behind corners and avoid scanning searchlights. But it can also turn into a {\it cul-de-sac}, if the guards manage to force the intruder into an enclosed area from which he cannot escape.

Thus \SSP\ is the problem of deciding if  there exists a successful search schedule for a given finite set of guards in a given simple polygon. Figure~\ref{fig:1} shows an instance of \SSP\ with a search schedule.

\begin{figure}[h]
\centering
\subfigure[]{\includegraphics[width=0.35\linewidth]{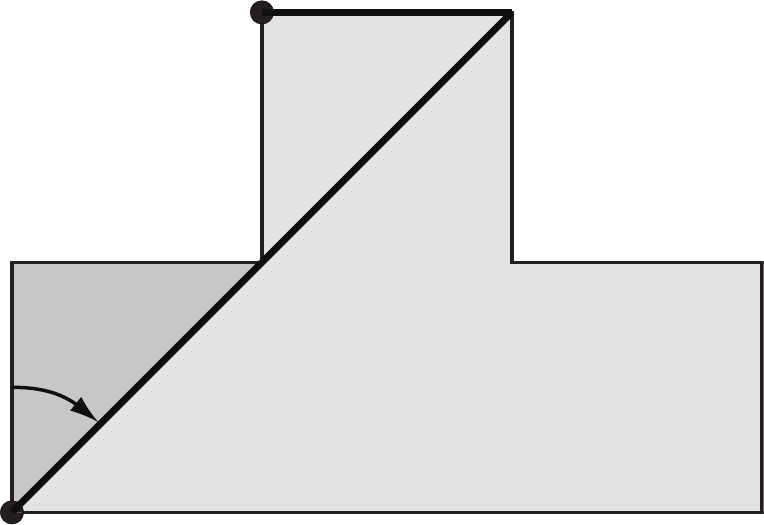}}\qquad \qquad
\subfigure[]{\includegraphics[width=0.35\linewidth]{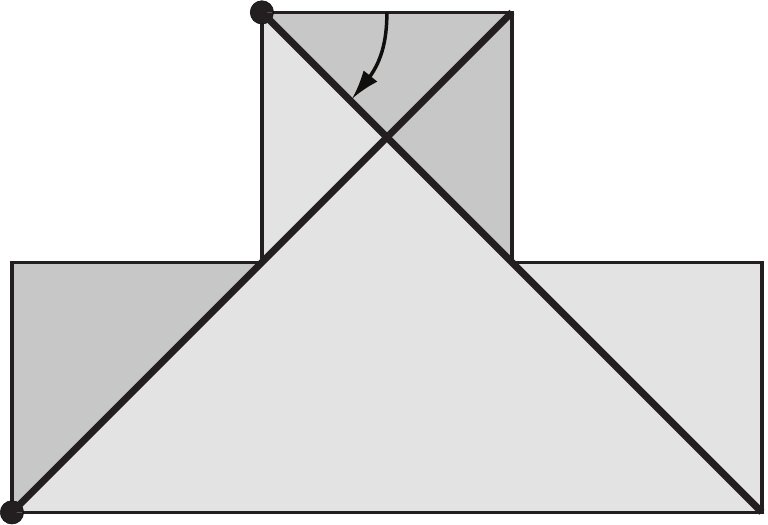}}\\ \vspace{0.5cm}
\subfigure[]{\includegraphics[width=0.35\linewidth]{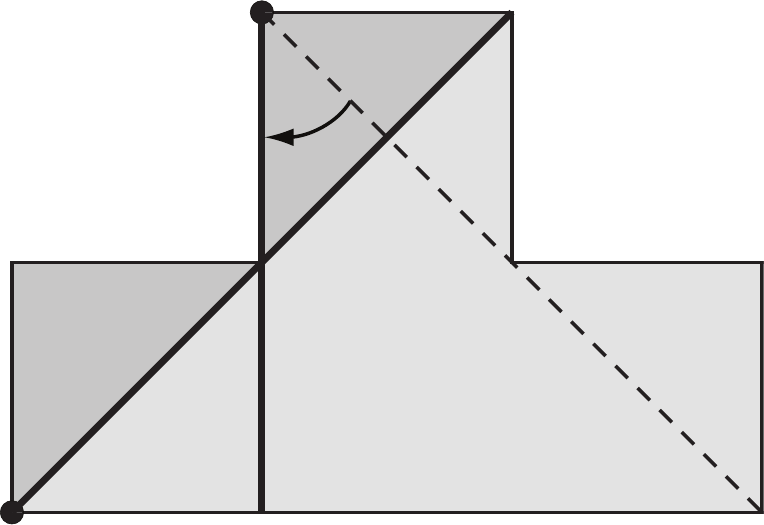}}\qquad \qquad
\subfigure[]{\includegraphics[width=0.35\linewidth]{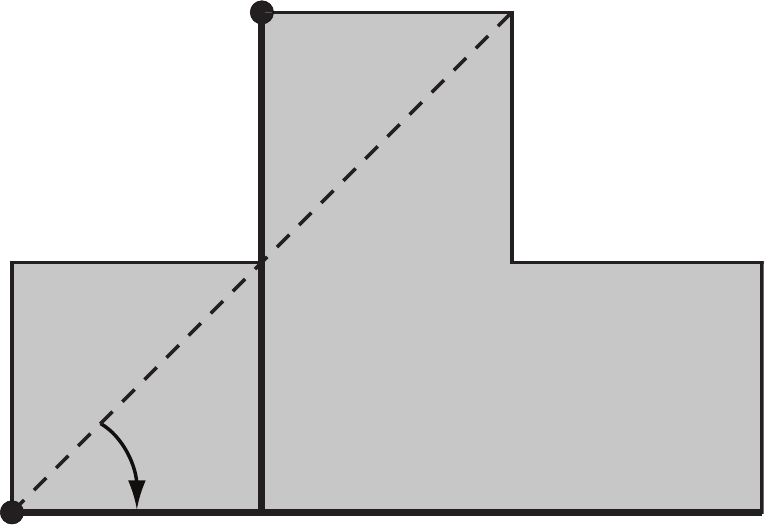}}
\caption{Search schedule for two guards in a polygon. At each stage, the light area is still \textquotedblleft contaminated\textquotedblright, whereas the darker areas have been cleared.}
\label{fig:1}
\end{figure}

\begin{sloppypar}A trivial necessary condition for searchability is that the guard positions should guarantee that no point in the polygon is invisible to all guards. In other words, the guard set should at least be a solution to the \ART\ in the given polygon, otherwise the intruder could sit at an uncovered point and never be discovered.\end{sloppypar}

Another simple necessary condition is that every guard lying in the interior of the polygon (thus not on the boundary) should be visible to at least one other guard. Without this, the intruder could remain in a neighborhood of a guard and just avoid its rotating searchlight.

Concerning the problem of minimizing the number of guards to search a given polygon, \cite{search1} also contains a characterization of the simple polygons that are searchable by one or two suitably placed guards. A similar characterization for three guards was also found by the same authors, but never published.\footnote{Source: private communication with M.\ Yamashita, 2010.}

On the other hand, Yamashita, Suzuki and Kameda, in~\cite{search2}, gave some upper bounds on the minimum number of guards required to search a polygon (possibly with holes) as a function of the number of guards needed to solve the \ART in that polygon. In the same paper, the original \SSP is further extended by introducing guards carrying $k\geqslant 1$ searchlights, called $k$-guards.

\subsection*{One-way sweep strategy}

Several sufficient conditions for searchability of simple polygons are detailed in~\cite{search1}, all employing a general search heuristic called the \emph{one-way sweep strategy} (OWSS), illustrated in Figure~\ref{f1:owss}.

\begin{figure}[h]
\centering
\subfigure[]{\includegraphics[width=0.29\linewidth]{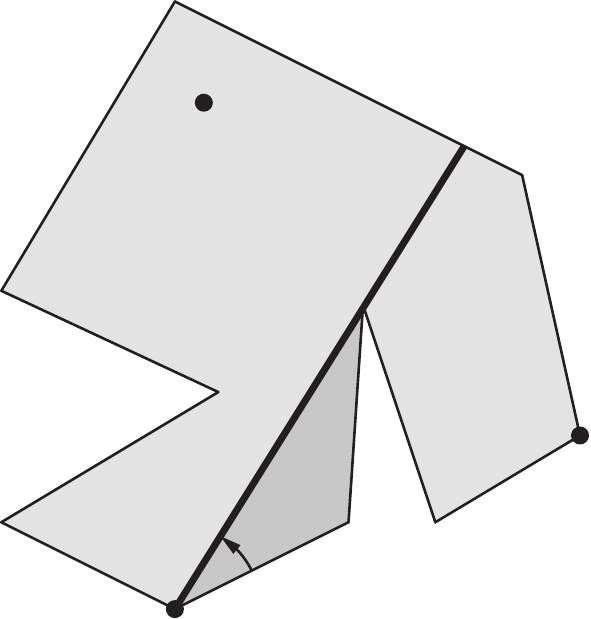}}\qquad
\subfigure[]{\includegraphics[width=0.29\linewidth]{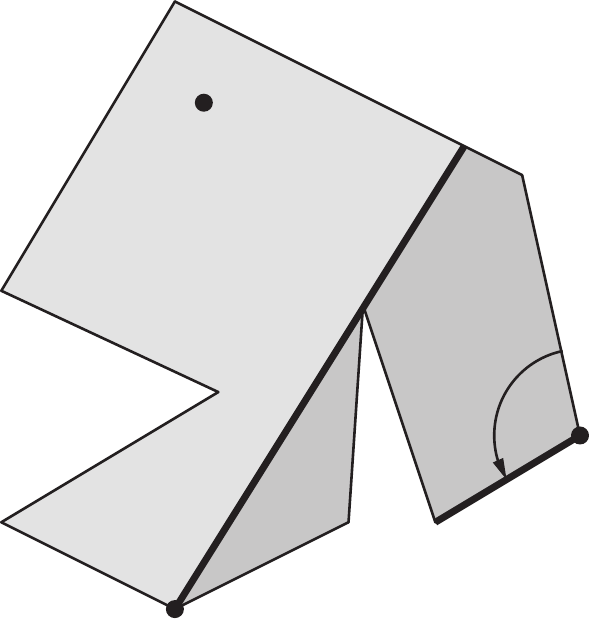}}\qquad
\subfigure[]{\includegraphics[width=0.29\linewidth]{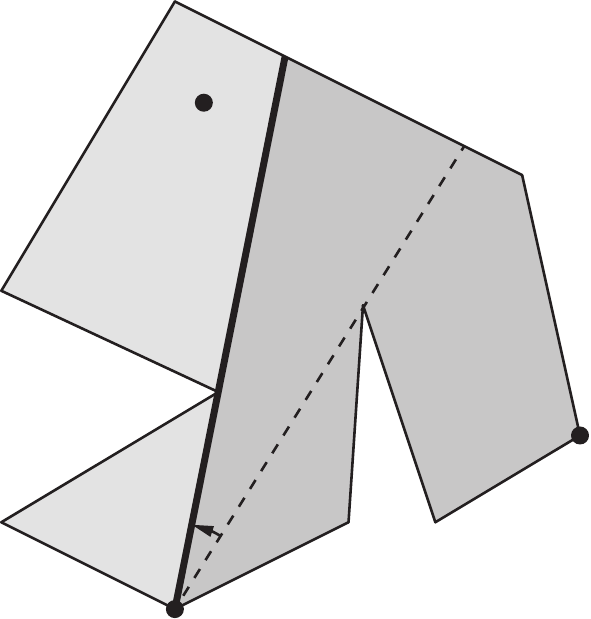}}\\ \vspace{0.25cm}
\subfigure[]{\includegraphics[width=0.29\linewidth]{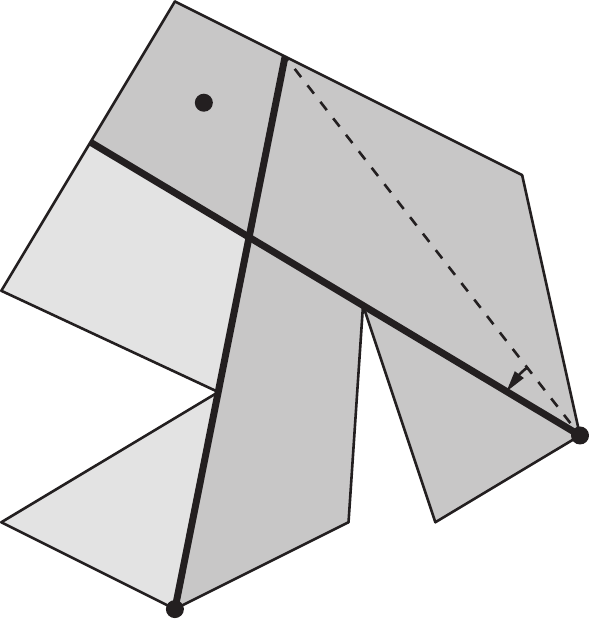}}\qquad
\subfigure[]{\includegraphics[width=0.29\linewidth]{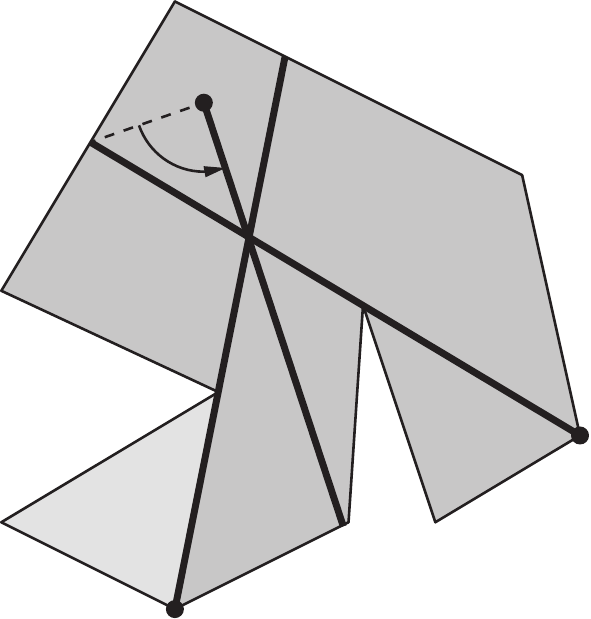}}\qquad
\subfigure[]{\includegraphics[width=0.29\linewidth]{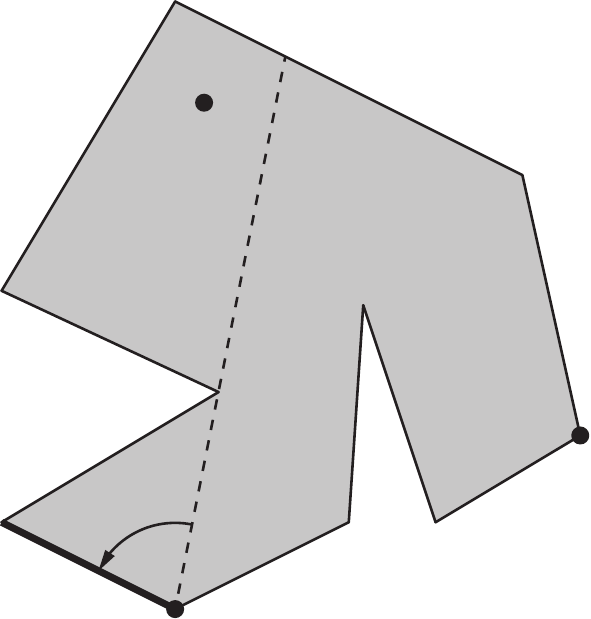}}
\caption{Illustration of the one-way sweep strategy.}
\label{f1:owss}
\end{figure}

First of all, two assumptions are made on the input data: every point in the input polygon $P$ is visible to at least one guard, and every guard is either on the boundary of $P$ or visible to another guard. As previously noted, if either condition does not hold, $P$ is trivially unsearchable. The OWSS algorithm takes a \emph{semiconvex subpolygon} $R$, ``supported'' by the searchlights of a set of guards $K$, and an additional guard $\ell$ not in $K$. $R$ is partly bounded by the searchlights of $K$, and partly bounded by $P$'s boundary. Indeed, the part bounded by the searchlights contains no points of non-convexity, hence the name \textquotedblleft semiconvex\textquotedblright. $\ell$ is assumed to be outside $R$ or on its boundary, in such a way that it can see at least one point inside $R$. Now, the heuristic attempts to clear $R$ as follows: the searchlights in $K$ are never moved, and $\ell$ starts sweeping $R$ counterclockwise, until it reaches the boundary of an invisible subregion $R'$. If possible, OWSS is applied recursively on $R'$, by adding $\ell$ to $K$ and picking a new searchlight $\ell '$ that meets the above requirements. Otherwise (i.e., if such $\ell'$ does not exist), either $R'$ can be cleared by the guards that lie inside it (any method can be applied here), or the heuristic fails.

The above algorithm shows that \SSP may be reduced to instances with no guards on the polygon's boundary. The OWSS can also be used as a tool to prove many remarkable sufficient conditions for the existence of a search schedule, based on the structure of the \emph{visibility graph} of the input guards (i.e., the graph whose edges connect the couples of guards that can see each other).

For example, there exists a search schedule if every connected component of the visibility graph contains at least one guard lying on the boundary of the polygon. Most notably, if all the guards lie on the boundary (and collectively see the whole polygon), then they also have a search schedule.

Another sufficient condition for the existence of a search schedule is that the visibility graph is connected, and there exists a guard whose removal does not disconnect it. Yet another sufficient condition is that the visibility graph is connected, and there exist two guards such that no point in the polygon is visible to both of them. More details can be found in~\cite{search1}.

In~\cite{bullo2}, Obermeyer, Ganguli and Bullo described an \emph{asynchronous distributed} version of the OWSS, in which the guards are unaware of what lies outside their field of view, and some degree of parallelism is also achieved.

We will propose an extension of the OWSS to polyhedra in Chapter~\ref{chapter8}.

\subsection*{Complexity and discretization}
The problem of determining the computational complexity of \SSP\ was not directly addressed in the seminal papers~\cite{search1,search2}, but has acquired more interest over time, and remains open. We will obtain some related complexity results in Chapters~\ref{chapter10} and~\ref{chapter11}.

Obermeyer, Ganguli and Bullo showed in~\cite{bullo} that \SSP is solvable in double exponential time, by means of an \emph{exact cell decomposition} technique that reduces the space of all possible schedules to a finite graph, which is then searched systematically. We will show in Lemma~\ref{lemma1} that such a technique actually implies that \SSP is in \PSPACE.

The discretization method consists in identifying a finite set of \emph{critical angles} for each guard. Guard $\ell$ is aiming its searchlight at a critical angle if and only if: 
\begin{enumerate}
\item[(i)] $\ell$ is located on some edge of the polygon and aims its searchlight along that edge;
\item[(ii)] $\ell$'s ray is tangent to a subpolygon that $\ell$ cannot see;
\item[(iii)] $\ell$ aims at another visible guard;
\item[(iv)] $\ell$ aims directly away from another visible guard.
\end{enumerate}
It can be proved that, if a search schedule exists, then it can be transformed into one in which only one guard is active at a time, and guards only stop or change direction at critical angles. Such a schedule is said to be \emph{sequential} and \emph{critical}. In Chapters~\ref{chapter7} and~\ref{chapter8}, we will briefly discuss some conditions allowing sequentiability to be extended to 3-dimensional scenarios.

In~\cite{bullo}, the problem of minimizing search time, in a scenario where searchlights have bounded angular speeds, is also mentioned as an interesting and yet unexplored variation of \SSP. In Chapters~\ref{chapter9} and~\ref{chapter10}, we will marginally consider this problem, as well.
\chapter{Polyhedra}\label{chapter2}
\begin{chapterabstract}
We define general polyhedra and several subclasses, such as orthogonal polyhedra and $c$-oriented polyhedra. We review some well-known theorems, such as Euler's formula and the polyhedral Gauss--Bonnet theorem, and we explore some structural properties of orthogonal polyhedra.

Furthermore, we consider the problem of partitioning polyhedra into convex parts, first showing why planar triangulations do not generalize, and then presenting some alternative partitions, mainly due to Chazelle.
\end{chapterabstract}

\section{Closed and open polyhedra}

\paragraph{Basic definitions.}
Throughout the scientific literature, there is little agreement around a precise definition of the notion of \emph{polyhedron}. Unfortunately, different definitions, arising in different contexts, are not always equivalent. Here follow the definitions that we will use for the rest of the thesis.

\begin{definition}[(closed) polyhedron]\label{polydef}
A \emph{(closed) polyhedron} is the union of a finite set of closed tetrahedra (with mutually disjoint interiors) embedded in $\R^3$, whose boundary is a connected 2-manifold.
\end{definition}

\begin{remark}
We stress that a polyhedron, by itself, is merely a subset of $\R^3$, with no additional structure.
\end{remark}

As a consequence of our definition, a (closed) polyhedron is a compact topological space and its boundary is homeomorphic to a sphere or a $g$-torus. $g$ is also called the \emph{genus} of the polyhedron, and it is zero if the polyhedron is homeomorphic to a ball. Moreover, the complement of a polyhedron with respect to $\R^3$ is connected. (For a review on basic topological terminology and facts, refer to~\cite{alexandroff,kosniowski}.)

Sometimes, when studying visibility problems, we may be interested in a slightly different notion of polyhedron.

\begin{definition}[open polyhedron]
An \emph{open polyhedron} is the interior of a closed polyhedron (i.e., a closed polyhedron minus its boundary).
\end{definition}

\paragraph{Faces, vertices, edges.}
Since a polyhedron's boundary is piecewise linear, the notion of \emph{face} of a polyhedron is well-defined as a maximal planar subset of its boundary with connected and non-empty relative interior. Thus a face is a plane polygon, possibly with holes, and possibly with some degeneracies, such as hole boundaries touching each other at a single vertex. Any vertex of a face is also considered a \emph{vertex} of the polyhedron.

\emph{Edges} are defined as minimal non-degenerate straight line segments shared by two distinct faces and connecting two vertices of the polyhedron. Since a polyhedron's boundary is an orientable 2-manifold, the relative interior of an edge lies on the boundary of exactly two faces, thus determining an internal dihedral angle (with respect to the polyhedron). An edge is \emph{reflex} if its internal dihedral angle is reflex, i.e., strictly greater than $180^\circ$. Hence, convex polyhedra have no reflex edges.

\begin{observation}\label{o2:fve}
The intersection of any two distinct faces of a polyhedron is either empty, or a vertex, or an edge. Each edge of a polyhedron contains exactly two vertices as its endpoints, and lies on the relative boundary of exactly two non-coplanar faces. Each vertex of a polyhedron belongs to at least three edges and to at least three faces.
\end{observation}

A discussion on possible degeneracies arising from an analogous definition of polyhedron, along with a collection of notable examples, is contained in~\cite{chazellethesis}.

\paragraph{$c$-oriented polyhedra.}
A useful parameter to express the complexity of a polyhedron is the number of different orientations its faces have (see for instance~\cite{coriented}).

\begin{definition}[$c$-oriented polyhedron]
A polyhedron is \emph{$c$-oriented} if $c$ vectors exist, any three of which form a basis for $\mathbb R^3$, such that each face is orthogonal to one of the vectors.
\end{definition}

It is understood that the orientation of a face is merely the slope of the plane in which the face is embedded. This should not be confused with the orientation the face inherits as being part of an orientable 2-manifold, i.e., a polyhedron's boundary. For instance, a cube is 3-oriented even if there are six distinct normal vectors pointing out of its faces. Similarly, a tetrahedron and a regular octahedron are both 4-oriented, a regular dodecahedron is 6-oriented, and a regular icosahedron is 10-oriented.

\paragraph{Representing polyhedra.}
There is no obvious choice of data structure for representing polyhedra. Although discussing this topic goes beyond the scope of our thesis, occasionally we will have to give algorithms and express their complexity.

We will generally assume that a polyhedron is stored as the array of its vertices, together with the array of its faces. Each face is a sequence of indices into the vertex array. The outer boundary of each face will be
given in counterclockwise order with respect to a normal vector poiniting outside, while its holes will be given in clockwise order.

There are some issues arising from the finiteness of the representation of coordinates, since the vertices of the same face are supposed to be coplanar. Coherency of data may still be guaranteed (with obvious losses in the representability of some polyhedral shapes), for example by triangulating the faces, but we will not delve into this topic here. For a discussion on some more sophisticated ways to represent polyhedra, refer to~\cite{chazellethesis,cgc}

When expressing asymptotic bounds with respect to the size of a polyhedron, we will generically refer to $n$ as the size of its representation. The robustness of this choice is also testified by the results presented in the next section. For instance, due to Euler's formula and Observation~\ref{o2:fve}, in a simply connected polyhedron the number of faces is bounded by a linear function of the number of edges, which is bounded by a linear function of the number of vertices, etc.

\section{Notable identities}

\paragraph{Euler's formula.}
One of the reasons why we chose our polyhedral model is that the Euler's formula holds, along with its several consequences.

\begin{theorem}[Euler--Poincar\'e]
For any polyhedron of genus $g$ with $f$ faces, $e$ edges, $v$ vertices, and a total of $h$ holes on its faces, the identity
$$f + v + 2g = e + h + 2$$
holds.
\hfill\qed
\end{theorem}

A proof of the formula for polyhedra of arbitrary genus was first given by Poincar\'e in~\cite{poincare}, whereas Euler only gave the formula for simply connected polyhedra (also refer to~\cite{manifold} for a comprehensive treatment of related subjects).

\begin{remark}
Euler's formula is often presented without the $h$ term, in that it is usually assumed that the polyhedron's faces are triangulated. It is straightforward to obtain our version of the formula, by induction on $h$. Indeed, a face with $n$ vertices and $h$ holes can be partitioned into $n+2h-2$ triangles by drawing $n+3h-3$ extra edges. Such triangles can then be treated as ``degenerate'' (i.e., coplanar) faces.
\end{remark}

\paragraph{Polyhedral Gauss--Bonnet theorem.}
For each vertex $v$ of a polyhedron, we denote by $k_v$ its \emph{curvature}, defined as $2\pi$ minus the sum of the face angles incident to $v$.

The following theorem, relating the total curvature of the vertices of a polyhedron with its genus, was originally discovered by Descartes, and later extended in various directions by several authors, most notably Gauss and Bonnet.

\begin{theorem}[polyhedral Gauss--Bonnet]\label{t2:gauss}
In any polyhedron of genus $g$,
$$\sum_v k_v = 2\pi (2-2g),$$
where $v$ ranges through all vertices.
\hfill\qed
\end{theorem}

In the case of polyhedra, the proof follows elementarily from Euler's formula, and can be found in~\cite{do-dcg-11}.

The polyhedral Gauss--Bonnet theorem will be used in Chapter~\ref{chapter5} to obtain tight inequalities relating the number of edges in an orthogonal polyhedron with the number of reflex edges.

\paragraph{Genus and reflex edges.}
An elementary consequence of the polyhedral Gauss--Bonnet theorem is that higher genuses imply proportional amounts of reflex edges.

\begin{theorem}[Chazelle--Shouraboura]\label{t2:chaz}
In any polyhedron, the genus $g$ and the number of reflex edges $r$ are related by the inequality
$$g\leqslant r+1.$$\hfill\qed
\end{theorem}

A simple proof of this inequality was given by Chazelle and Shouraboura in~\cite{polypart3}.

\section{Orthogonal polyhedra}

\paragraph{Terminology.}
A polyhedron is said to be \emph{orthogonal} if each of its edges is parallel to some axis. It follows that orthogonal polyhedra are 3-oriented. A \emph{cuboid} is a convex orthogonal polyhedron (some authors refer to cuboids as \emph{boxes}).

Each face of an orthogonal polyhedron is an orthogonal polygon, perhaps with holes, perhaps with degeneracies such as hole boundaries touching each other at a single vertex.

Whenever considering orthogonal polyhedra, we implicitly define the \emph{vertical} direction (or \emph{up-down} direction) to be the direction parallel to the $z$ axis. Accordingly, any direction parallel to the $xy$ plane is called \emph{horizontal}, and the direction parallel to the $x$ axis (resp.~$y$ axis) is called \emph{left-right} (resp.~\emph{front-back}) direction. Thus, the edges and faces of any orthogonal polyhedron are either vertical or horizontal, and the $x$-orthogonal (resp.~$y$-orthogonal) faces are called \emph{lateral faces} (resp.~\emph{front faces}).

\paragraph{$k$-reflex orthogonal polyhedra.}
Furthermore, we say that an orthogonal polyhedron is $k$-reflex if it has reflex edges parallel to at most $k$ distinct axes. Thus, a cuboid is 0-reflex, an orthogonal prism is 1-reflex, and the polyhedra in Figures~\ref{f2:types} and~\ref{f2:octoplex} are 3-reflex. In Chapter~\ref{chapter4}, we will consider a variation of the \ART restricted to 2-reflex orthogonal polyhedra.

\paragraph{Vertex classification.}
In any orthogonal polyhedron, all convex dihedral angles are $90^\circ$ wide, and all reflex dihedral angles are $270^\circ$ wide.
Based on the number of
incident reflex and convex edges, the vertices of orthogonal polyhedra form six distinct classes, denoted here by A, B, C, D, E and F, defined as follows.

Consider the eight octants determined by the coordinate axes intersecting at
a given vertex, and place a sufficiently small regular octahedron around the vertex, such that each of its faces lies in
a distinct octant. By Definition~\ref{polydef}, the set of the octahedron's faces that fall inside (resp.\ outside) our orthogonal polyhedron is connected: recall that the boundary of a polyhedron is a 2-manifold.

Consider all possible ways of partitioning the faces of the octahedron into two non-empty connected sets, up to isometry (refer to Figure~\ref{f2:vtypes}):

\begin{figure}[h]
\centering{\includegraphics[width=0.66\linewidth]{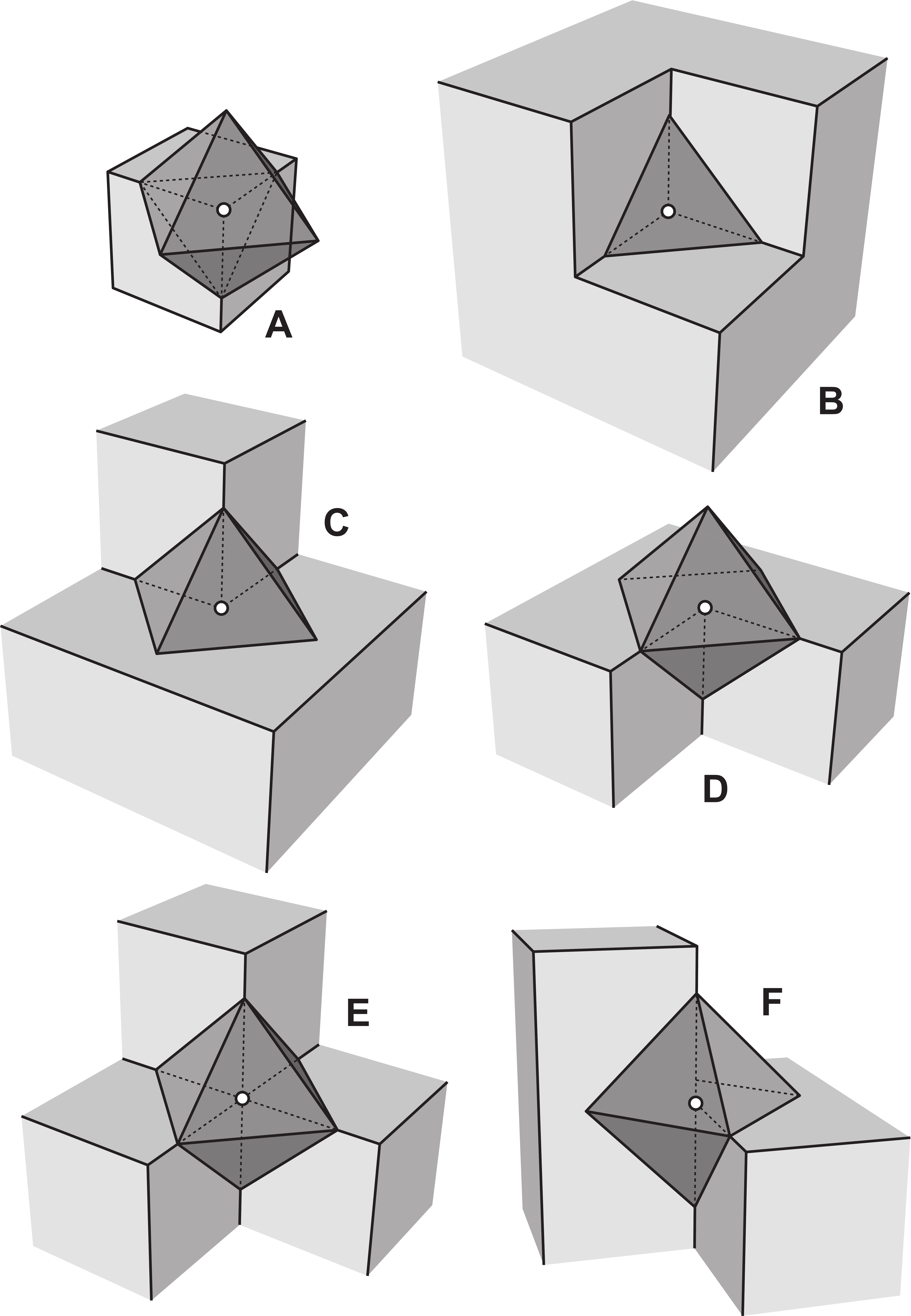}}
\caption{Vertex types for orthogonal polyhedra.}
\label{f2:vtypes}
\end{figure}

\begin{itemize}
\item There is essentially a single way to select one face (resp.\ seven faces). This corresponds to an A-vertex (resp.\ a B-vertex).

\item There is a single way to select two faces (resp.\ six faces). This case does not correspond to a vertex of the orthogonal polyhedron: it implies that the considered point is not a vertex of any face on which it lies.

\item There is a single way to select three faces (resp.\ five faces). This corresponds to a D-vertex (resp.\ a C-vertex).

\item There are three ways to select four faces. One of them implies that the point lies in the middle of a face,
hence it does not correspond to a vertex.
The other two choices correspond to an E-vertex and an F-vertex, respectively.
\end{itemize}

Figure~\ref{f2:types} shows a polyhedron exhibiting vertices of all types.

\begin{figure}[h]
\centering{\includegraphics[width=0.55\linewidth]{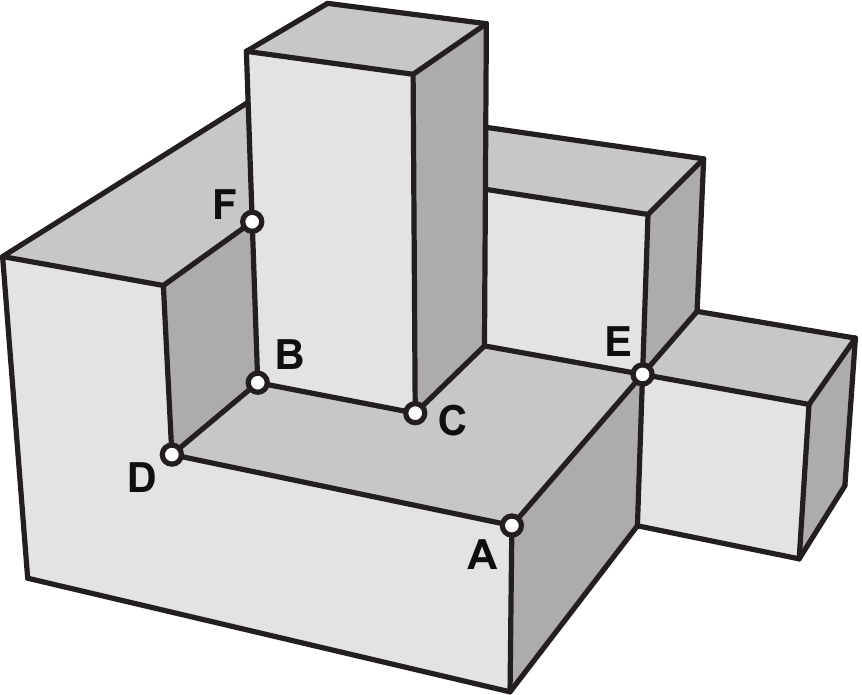}}
\caption{Orthogonal polyhedron displaying vertices of all types.}
\label{f2:types}
\end{figure}

\paragraph{Bounds on the number of reflex edges.}
For orthogonal polygons with $h$ holes, there is a simple formula relating the number of vertices $n$ with the number of reflex vertices $r$:
$$n=2r-4h+4.$$
This can be easily proved by induction on $h$.

For edges in orthogonal polyhedra there is no such identity, but the total number of edges still bounds the number of reflex edges, both from above and from below. In Chapter~\ref{chapter5}, we will prove that the following tight inequalities hold for every orthogonal polyhedron of genus $g$, with $e$ edges, of which $r>0$ are reflex:
$$\frac{e}{6} + 2g - 2\ \leqslant\ r\ \leqslant\ \frac{5e}{6} - 2g - 12.$$
The proof is yet another application of the polyhedral Gauss--Bonnet theorem.

\section{Convex partitions}

\subsection*{Tetrahedralizations}

The main difficulty that arises when attempting to extend the results surveyed in Chapter~\ref{chapter1} to polyhedra is that triangulation, the most basic and versatile tool used for polygons, does not generalize.
Indeed, there are polyhedra whose interior cannot be partitioned into tetrahedra without the addition of extra vertices (also called \emph{Steiner points}). 

\paragraph{The Sch\"onhardt polyhedron.}
A simple example of an untetrahedralizable polyhedron can be obtained by ``twisting'' a triangular prism by $30^\circ$, while bending inside the lateral faces, as
shown in Figures~\ref{f2:notri} and~\ref{f2:notri0}. This is also known as the \emph{Sch\"onhardt polyhedron}, see~\cite{herbert,art}.

\begin{figure}[h]
\centering{\includegraphics[width=0.85\linewidth]{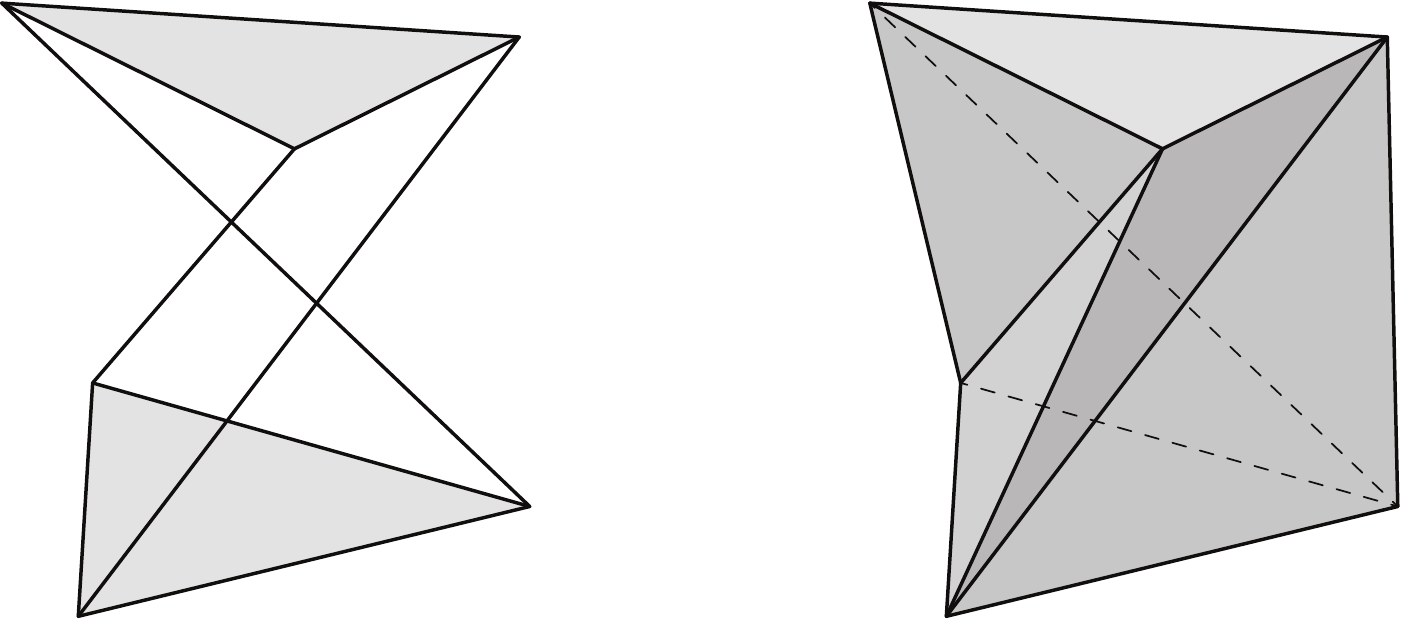}}
\caption{Construction of the Sch\"onhardt polyhedron.}
\label{f2:notri}
\end{figure}

\begin{figure}[h!]
\centering{\includegraphics[width=0.45\linewidth]{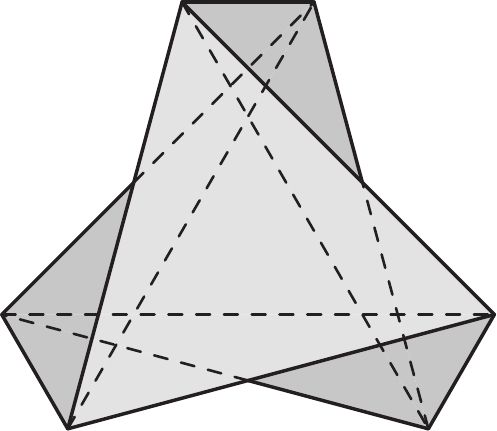}}
\caption{Sch\"onhardt polyhedron, top view.}
\label{f2:notri0}
\end{figure}

No matter how we pick four of the six vertices, the segment connecting two of them lies outside the polyhedron, which is then untetrahedralizable. Nonetheless, it can be decomposed into eight tetrahedra by adding a Steiner point in its center.

\paragraph{Deciding tetraheralizability.}
The Sch\"onhardt polyhedron was used by Ruppert and Seidel in~\cite{polypart4} to prove that deciding if a given polyhedron is tetrahedralizable without the addition of Steiner points is \NP-complete.

\paragraph{The octoplex.}
There are untetrahedralizable polyhedra also among orthogonal polyhedra. Figure~\ref{f2:octoplex} shows an example, called \emph{octoplex} in~\cite{adventures}.

To prove that the octoplex is not tetrahedralizable, observe that there are points $p$ around its center, such that the segment connecting $p$ with any vertex necessarily crosses the boundary. It follows that such points $p$ cannot belong to any tetrahedron in a tetrahedralization, hence the octoplex is untetrahedralizable.

\begin{figure}[h!]
\centering{\includegraphics[width=0.5\linewidth]{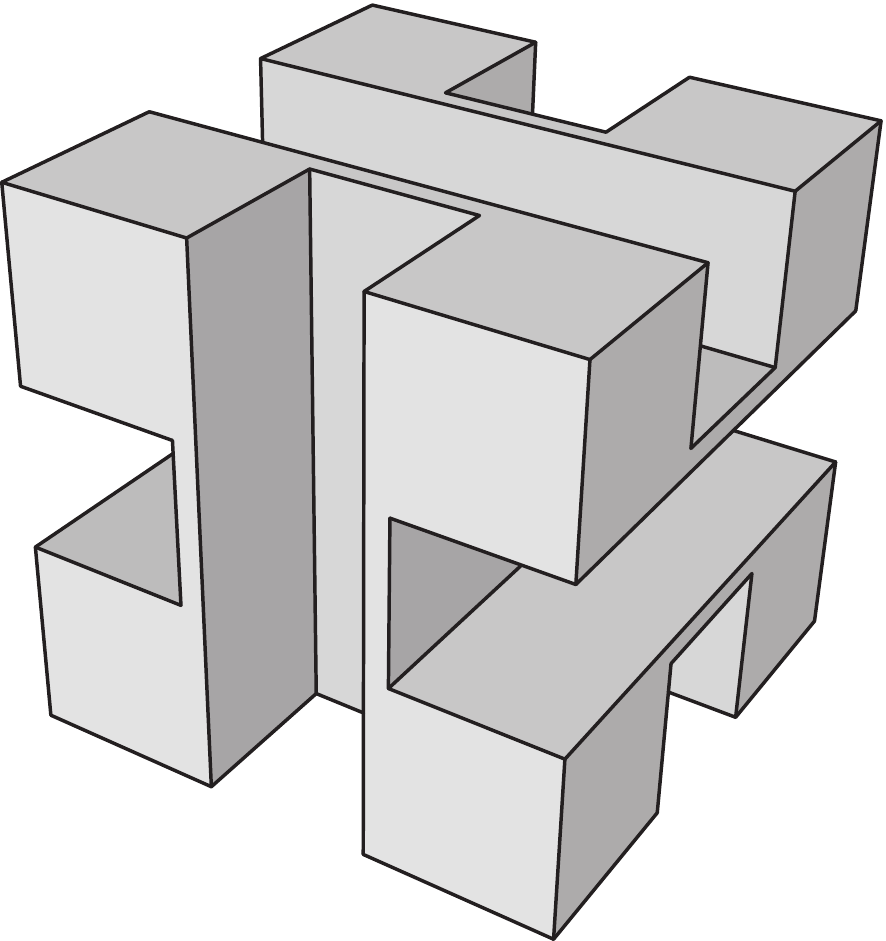}}
\caption{Octoplex.}
\label{f2:octoplex}
\end{figure}

\subsection*{Bounds on partitions into tetrahedra}

Chazelle and Palios proved in~\cite{polypart2} that any simply connected polyhedron with $n$ vertices and $r$ reflex edges can be partitioned, with the addition of Steiner points, into $O(n+r^2)$ tetrahedra. The partition can be computed in $O(n+r^2)$ space and $O((n + r^2) \log r)$ time. Since in
most applications $n$ greatly exceeds $r$, the method is viable in practice.

The algorithm starts by eliminating some low-degree vertices called ``cups'' in order to reduce the size of the polyhedron to $O(r)$, and then proceeds by erecting vertical ``fences'' that partition the resulting shape into $O(r^2)$ prisms, which are easily tetrahedralized.

Due to Theorem~\ref{t2:chaz}, the same upper bounds also hold for polyhedra of any genus, as observed by Chazelle and Shouraboura in~\cite{polypart3}.

The $O(n+r^2)$ bound on the number of tetrahedra is asymptotically tight in the worst case, as the class of \emph{Chazelle's polyhedra} shown in Figure~\ref{f2:chaz} implies (see~\cite{polypart1}).

\begin{figure}[h!]
\centering{\includegraphics[width=0.5\linewidth]{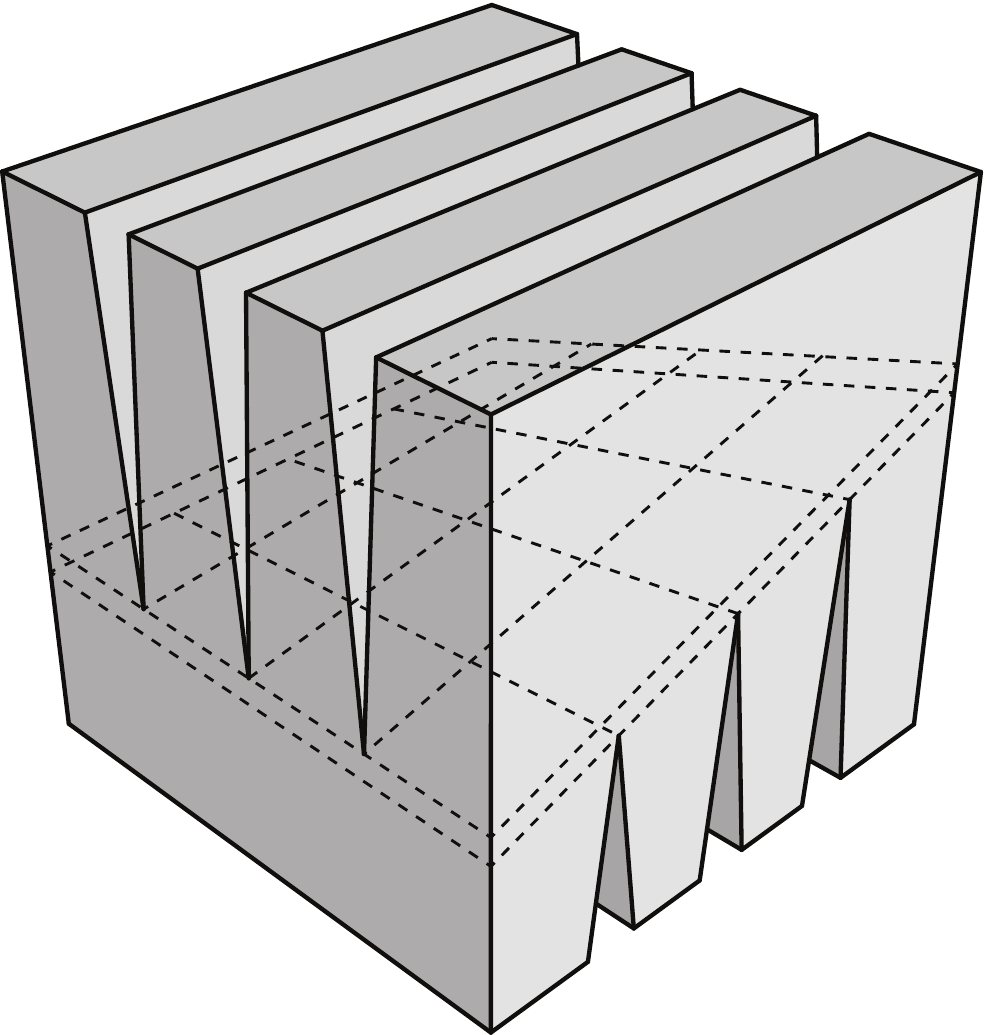}}
\caption{Chazelle's polyhedron.}
\label{f2:chaz}
\end{figure}

Indeed, the reflex edges on the bottom reflex edges lie on the hyperbolic paraboloid $z = xy$, while those on the top reflex edges lie on $z = xy +\varepsilon$. If $\varepsilon$ is small enough, the intersection of the region lying between the two hyperbolic paraboloids with any convex subset of the polyhedron can only have such a small volume that a quadratic number of convex pieces is necessary to cover the whole volume.

\subsection*{Decompositions into convex parts}

Note that Chazelle's polyhedra (Figure~\ref{f2:chaz}) also yield a lower bound of $\Omega(r^2)$ on the number of convex parts (not necessarily tetrahedral) required to decompose a polyhedron with $r$ reflex edges.

In~\cite{polypart1}, Chazelle showed how to partition any polyhedron into at most
$$\frac {r^2} 2 + \frac r 2 +1$$
convex pieces (with the addition of Steiner points) in $O(nr^3)$ time and $O(nr^2)$ space, where $n$ is the number of vertices.

The proposed algorithm, called ``revised naive decomposition'' in~\cite{polypart1}, picks a reflex edge and resolves it by cutting the polyhedron along a plane adjacent to that edge. Then the decomposition proceeds recursively on the resulting pieces, making sure that the reflex edge parts that originally belonged to the same reflex edge are all resolved with coplanar cuts. (A very similar decomposition will be thoroughly described in Chapter~\ref{chapter9} and applied to the \TSSPext.)

Again, the partition obtained is asymptotically optimal, and the algorithm is viable due to the low amount of reflex edges in polyhedra encountered in most applications.
\chapter{Guards}\label{chapter3}
\begin{chapterabstract}
We define several types of guards and guarding modes in polyhedra, also introducing some new concepts. Namely, we consider face guards and edge guards, each of which may be open or closed, and traditional point guards. Furthermore, we distinguish between orthogonal and non-orthogonal guarding.

We give matching upper and lower bounds quantifying the relationship between closed and open edge guards in orthogonal polyhedra, showing that closed edge guards are three times more ``powerful'' than open edge guards.

Next, we provide some upper and lower bounds on the number of point guards, edge guards, and face guards required to guard a given polyhedron. We review the state of the art on each problem, while also proving some new basic facts. We discuss both general and orthogonal polyhedra, and both closed and open guards.

Finally we focus on edge guards again, and we give some hardness proofs related to the computation and approximation of minimum edge guard numbers, employing almost direct generalizations of well-known planar constructions.
\end{chapterabstract}

\section{Visibility and guarding}

\subsection*{Visibility}
Given two points $x$ and $y$, we denote by $xy$ the (closed) straight line segment joining $x$ and $y$, and by $\open{xy}$ the corresponding \emph{open segment}, i.e., $\open{xy}=xy\setminus\{x,y\}$.

\emph{Visibility} with respect to a (closed or open) polyhedron $\mathcal{P}$ is a relation between points in $\mathbb{R}^3$: point $x$ \emph{sees} point
$y\neq x$ (equivalently, $y$ is \emph{visible} to $x$) if $xy \setminus \{x\}$ lies entirely in $\mathcal P$. $x$ sees itself if and only if it belongs to $\mathcal P$.

Note that, if $\mathcal P$ is a closed polyhedron, then visibility is a symmetric relation. In this case, a ``visibility segment'' $xy$ could touch $\mathcal P$'s boundary, or even lie on it.

On the other hand, if $\mathcal P$ is an open polyhedron, its boundary ``occludes'' visibility: if $x$ sees $y$, no portion of $xy$, except the endpoint $x$, can
lie on the boundary of $\mathcal P$. Hence, even if a boundary point cannot see itself, it can see some points inside $\mathcal P$.

When $\mathcal P$ is understood, we can safely omit any explicit reference to it, and just generically refer to visibility.

\begin{definition}[visibility region]\label{d3:visregion}
The \emph{visibility region} $\mathcal V(x)$ of a point $x\in \R^3$ (with respect to some polyhedron) is the set of points that are visible to $x$. The visibility region of a set $X \subseteq \mathbb R^3$, denoted by $\mathcal V(X)$, is the set of points that are visible to at least one point in $X$.
\end{definition}

With respect to open polyhedra, visibility regions are open sets. This property will be used in this chapter in a couple of occasions.

\begin{proposition}
\label{prop:1}
For every $X\subseteq \R^3$, the visibility region of $X$ with respect to an open polyhedron $\mathcal P$ is an open set.
\end{proposition}
\begin{proof}
We will prove our claim just for $X=\{x\}$. In general,
$$\mathcal V(X)= \bigcup_{x\in X} \mathcal V(x),$$
which is open because it is a union of open sets.

First of all, if $x$ does not lie in the topological closure of $\mathcal P$, then its visibility region is the empty set, which is open. Otherwise, let $f$ be a face of $\mathcal{P}$ not containing $x$. The region of space
occluded by $f$ is a closed set $\mathcal O(x,f)$, shaped like a truncated
unbounded pyramid with apex $x$ and  base $f$.
Taking the union of all $\mathcal O(x,f)$'s, for every face $f$ not containing $x$, we obtain
a closed set $\mathcal O(x)$, because $\mathcal P$ has finitely many faces.

The region occluded by the faces containing $x$ is the corresponding
unbounded solid angle, external with respect to $\mathcal P$, which is a
closed set.  Its union with $\mathcal O(x)$ is again a closed set, whose complement is therefore an open set. By definition, this is exactly $\mathcal V(x)$.
\end{proof}

\begin{observation}
For closed polyhedra, a weaker statement holds: the visibility region of any point, with respect to a closed polyhedron, is a closed set (the proof is similar to that of Proposition~\ref{prop:1}). However, the visibility region of an infinite point set, with respect to a closed polyhedron, may be neither a closed set nor an open set.
\end{observation}

\subsection*{Guard types and guarding modes}

Given a polyhedron $\mathcal P$, our variations of the \ART ask for a tight bound on the minimum $k$ such that there exists a \emph{guarding set} $G=\{g_1, \cdots, g_k\}$ that completely sees $\mathcal P$. In other terms, we require that
$$\mathcal P =  \bigcup_{i=1}^k \mathcal V(g_i).$$

The exact nature of the \emph{guards} $g_i$ depends on the particular variation of the problem we are considering. Our guards could be:

\begin{itemize}
\item \emph{Point guards}, i.e., points chosen anywhere.
\item \emph{Boundary point guards}, i.e., points chosen on $\mathcal P$'s boundary.
\item \emph{Vertex guards}, i.e., points chosen among $\mathcal P$'s vertices.
\item \emph{Segment guards}, i.e., line segments lying in the topological closure of $\mathcal P$.
\item \emph{Boundary segment guards}, i.e., line segments lying on $\mathcal P$'s boundary.
\item \emph{Edge guards}, i.e., line segments chosen among $\mathcal P$'s edges.
\item \emph{Face guards}, i.e., polygons chosen among $\mathcal P$'s faces.
\end{itemize}

\paragraph{Orthogonal guarding.}
We can further distinguish a more restrictive \emph{guarding mode}, which we will often use in conjunction with edge guards in orthogonal polyhedra. We say that an edge guard $e$ \emph{orthogonally sees}  a point $x$ (with respect to some polyhedron) if there is a point $y\in e$ that sees $x$, such that $xy$ is orthogonal to $e$. Similarly, we introduce the variant of the \ART in which an (orthogonal) polyhedron must be orthogonally guarded by edge guards.

\paragraph{Bound parameters.}
The parameters used to express bounds on guard numbers depend on the nature of the guards themselves. We usually bound point guards in terms of the number of vertices of the polyhedron, and we bound face guards in terms of the total number of faces. We may bound segment and edge guards either in terms of the total number of edges, or of the number of reflex edges.

\section{Open edge guards}

Although in the traditional variations of the \ART, guards are usually topologically closed sets, we introduce a different model, which we claim is more natural in certain circumstances.

Segment guards, edge guards and face guards may be \emph{closed} or \emph{open}, depending on whether they contain their relative boundary or not. For instance, an \emph{open edge guard} is an edge minus its endpoints.

\subsection*{Motivations}
In Part~\ref{part2}, we will study the \ART for edge guards. We will generally consider closed edge guards in conjunction with closed polyhedra, and open edge guards in conjunction with open polyhedra.

To better understand our motivations for studying open edge guards in open polyhedra, consider a polyhedron representing an empty room with solid walls. We are tasked to place guards in this room, who can detect unwelcome intruders. Because an intruder cannot hide ``within'' a wall, but rather must be located inside the room, there is no need to guard the walls of the room, i.e., the boundary of the polyhedron.

A guarding problem can alternatively be viewed as an illumination problem, with guards acting as light sources. Incandescent lights are modeled as point guards, and fluorescent lights are modeled as segment guards. In the latter case, it may be convenient to disregard the endpoints of the edge guards.
Indeed, the amount of light that a point interior to the polyhedron
receives is proportional to the total length of the segments illuminating it. Employing the open edge guard model in open polyhedra ensures that if a point is illuminated, it receives a strictly positive amount of light, and makes the model more realistic. 

Observe that these two notions of illuminated points (visible to an
open edge guard or receiving a strictly positive amount of light from
closed edge guards) cease to be equivalent in the case of closed polyhedra.

\subsection*{Comparing open and closed edge guards}

Although closed edge guards have only two more points than their open counterparts, they can be three times more ``powerful''. This is because the number of open guards needed to guard a polyhedron may be three times the corresponding number of closed edge guards. For open orthogonal polyhedra, this bound becomes tight.

\begin{theorem}
\label{th:1}
Any open orthogonal polyhedron guardable by $k$ closed edge guards is guardable by at most $3k$ open edge guards, and this bound is tight.
\end{theorem}

\begin{proof}
Given a set of $k$ closed edges that guard the entire polyhedron, we first construct a guarding set of
open edges of size at most $3k$ and then show that this set also guards the entire polyhedron. The construction is
as follows: for each closed edge $uv$ from the original guarding set, place the open edge $\open{uv}$ into the new guarding
set. To replace the endpoint $u$, add a reflex edge $\open{uw}$, with $w \neq v$, if one such edge exists. Otherwise, add any other edge incident to $u$. Similarly, an incident edge, preferably reflex, is selected for the other endpoint $v$. Hence, for each edge of the original guarding set, at most three open edges are placed in the new guarding set.

To prove the equivalence of the two guarding sets, we show that the region that was guarded
by endpoint $u$ of the closed edge $uv$ from the original guarding set is guarded by some point belonging to the
interior of $uv$ or the interior of $uw$, as chosen above, 
i.e., $\mathcal V (u) \subseteq \mathcal V (\open{uv}) \cup \mathcal V (\open{uw})$.

Let $x \in \mathcal V (u)$.
By Proposition~\ref{prop:1}, a small-enough ball $\mathcal{B}$ centered at $x$  
belongs to $\mathcal V(x)$. We consider a right circular cone $\mathcal{C}$ with apex $u$, whose base is centered at $x$ and is contained
in $\mathcal{B}$. Clearly, $\mathcal{C} \subset \mathcal V(u)$. Let $\mathcal{D}$ be a (small-enough) ball centered at $u$ that does not intersect any face of the polyhedron except those containing $u$ (refer to Figure~\ref{fig:coneball}).

\begin{figure}[h]
\centering{\includegraphics[width=.95\linewidth]{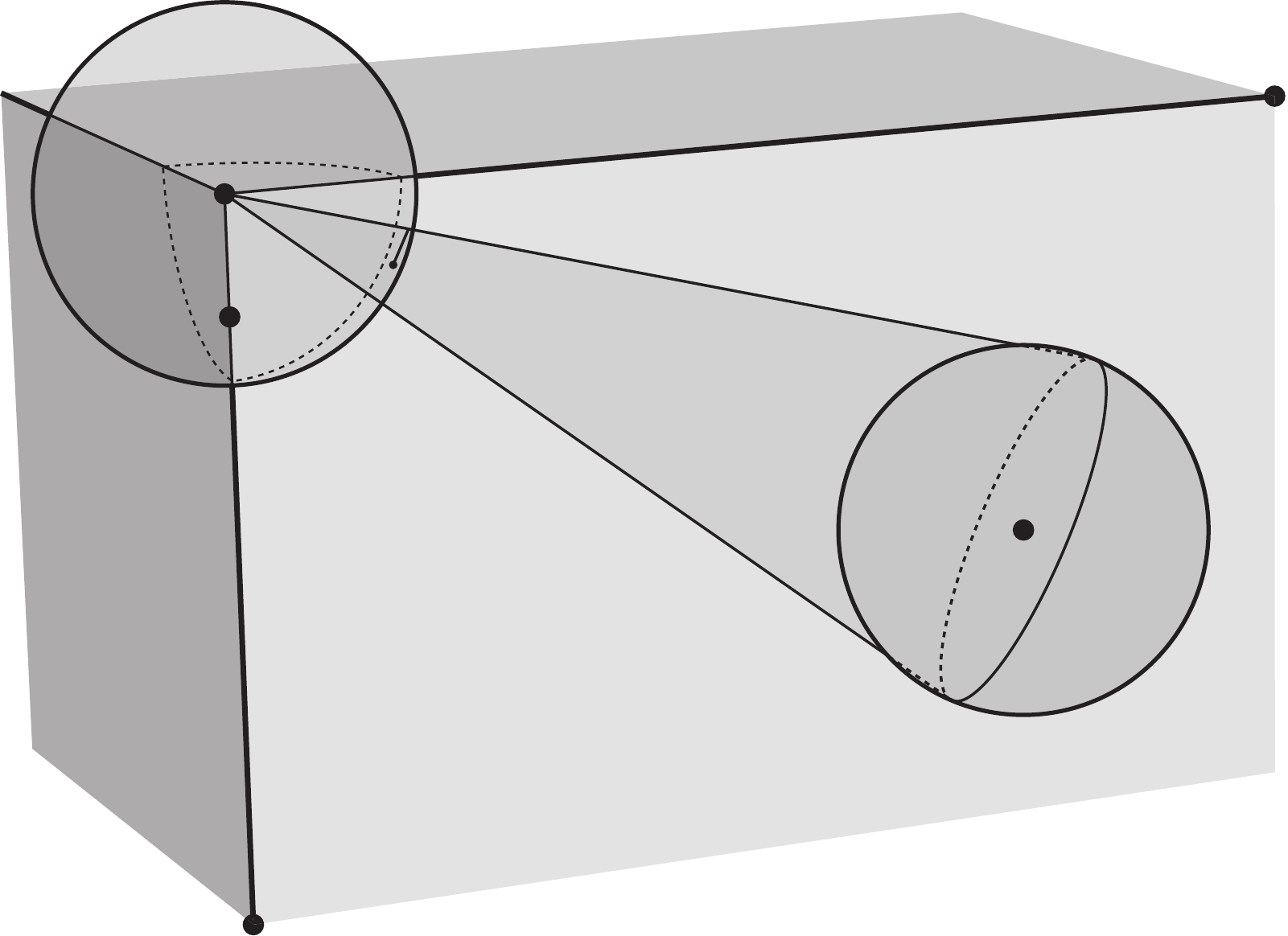}}
\put(-348,222){$u$}
\put(-320,12){$v$}
\put(-20,246){$w$}
\put(-293,209){$y$}
\put(-89,112){$x$}
\put(-327,183){$z$}
\put(-275,215){$\rho$}
\Large
\put(-344,275){$\mathcal D$}
\put(-200,170){$\mathcal C$}
\put(-50,104){$\mathcal B$}
\caption{Construction from the proof of Theorem~\ref{th:1}.}
\label{fig:coneball}
\end{figure}
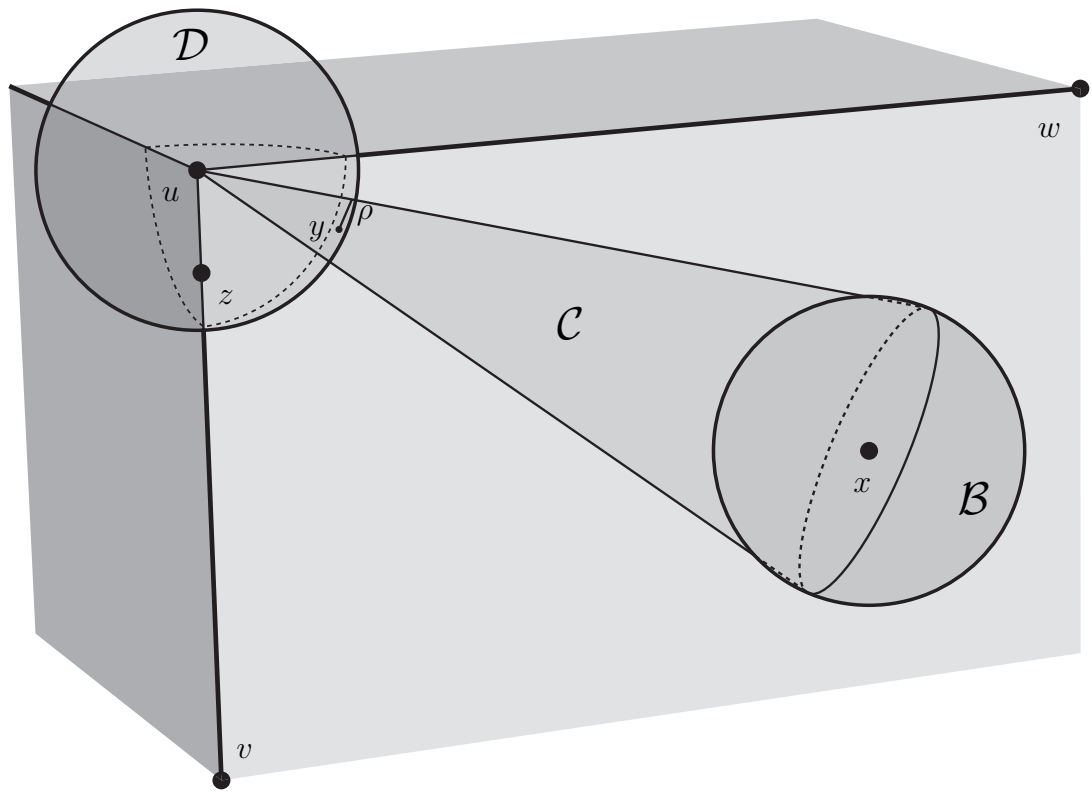

We prove that $\mathcal{D} \cap \mathcal{P} \subseteq \mathcal V(\open{uv} \cap \mathcal{D}) \cup \mathcal V(\open{uw} \cap \mathcal{D})$.

If $u$ is an A-vertex (recall the classification of vertices from Chapter~\ref{chapter2} and refer to Figure~\ref{f2:vtypes}), then $\mathcal{D} \cap \mathcal{P} \subseteq \mathcal V (\open{uv} \cap \mathcal{D})$. If
$u$ is a B-vertex (as illustrated in Figure~\ref{fig:coneball}), then of the eight octants in which $\mathcal{D}$ is partitioned by
orthogonal planes crossing at its center, one is external to $\mathcal{P}$.
Out of the seven octants that need be guarded, six are guarded
by $\open{uv} \cap \mathcal{D}$. The same holds for $\open{uw}$, and together they
guard all seven octants (two of the octants guarded by $\open{uv}$ are missing a face, but those two faces are guarded by $\open{uw}$).

In all other cases ($u$ is a $C$-, $D$-, $E$- or $F$-vertex), 
either $uv$ or $uw$ is a reflex edge.
Assume without loss of generality that $uv$ is reflex.
Then, $\open{uv} \cap \mathcal{D}$ sees all of
$\mathcal{D} \cap \mathcal{P}$ (refer again to Figure~\ref{f2:vtypes}). 

The boundaries of $\mathcal{D}$ and $\mathcal{C}$ intersect at a circle of
radius $\rho > 0$. Let $y$ be the center of that circle. By the above reasoning, there is
a point $z$ on $\open{uv} \cap \mathcal{D}$ or on $\open{uw} \cap \mathcal{D}$ that sees $y$, and hence
the entire open segment $\open{uz}$ sees $y$. Pick a point $t$ on $\open{uz}$
such that $||ut||< \rho$. Then $t$ sees $x$, hence $x$ is guarded.

A similar argument holds for the visibility region of the other endpoint, $v$, of
$uv$.

\begin{figure}[h]
\centering{\includegraphics[width=.55\linewidth]{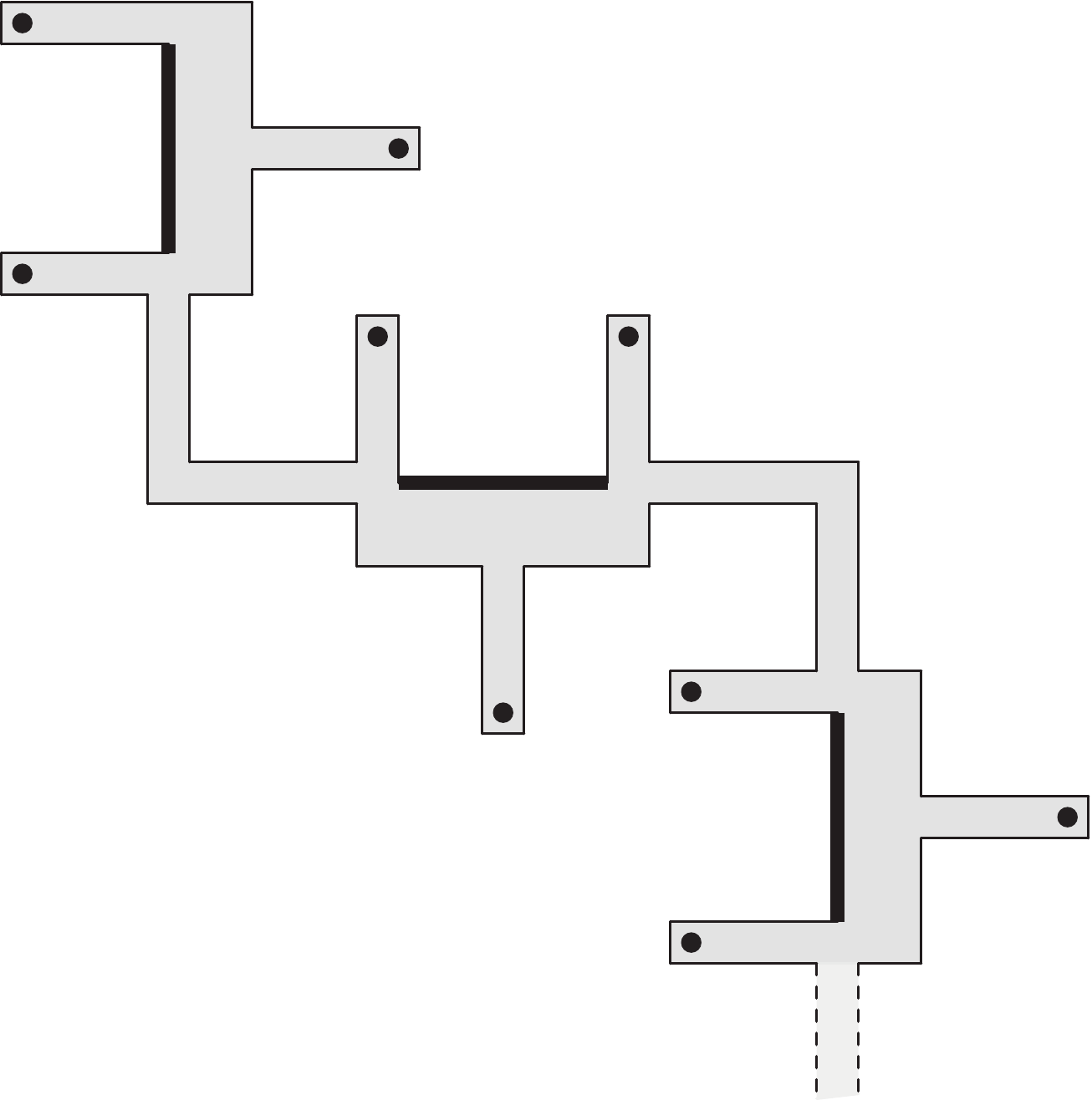}}
\caption{Matching ratio in Theorem~\ref{th:1}. Notice that the same example also solves the corresponding problem for 2-dimensional polygons.}
\label{fig:openclosed}
\end{figure}

To see that 3 is the best achievable ratio between the
number of open and closed edge guards, consider the
polygon in Figure~\ref{fig:openclosed} and extrude it to an orthogonal
prism. Each large dot in the figure represents a distinguished point located in the interior of the prism. The only (closed) edges that can see more than two selected points are the highlighted edges (located on the lower or upper base of the prism).
Picking those edges as guards yields the minimum guarding set.
On the other hand, the relative interior of any edge
can see at most one distinguished point. Therefore, at least as many
open edge guards as there are distinguished points are necessary.
\end{proof} 

Note that the above analysis does not hold in the case of closed polyhedra, since we can no longer argue that a single closed edge guard is locally dominated by three open edge guards.

\subsection*{Planar open edge guards}

The open guard model that we introduced was later studied, in the case of open edge guards in 2-dimensional polygons, by T\'oth, Toussaint and Winslow in~\cite{open1}, and subsequently by Cannon, Souvaine and Winslow in~\cite{open2}.

Concerning planar open edge guards, we present a result that will serve as a lemma for a theorem in Chapter~\ref{chapter9}. We want to guard a polygon $P$ with open edge guards and, for technical reasons (see the proof of Theorem~\ref{heur}), we also want to force the selection of a specific edge.

\begin{lemma}\label{edges}Any polygon with $r$ reflex vertices, $h$ holes and a distinguished edge $e$ is guardable by at most $r-h+1$ open edge guards, one of which lies on $e$.\end{lemma}

\begin{proof}Let $P$ be a polygon, select any reflex vertex $v$ and draw the bisector of the corresponding internal angle, until it again hits the boundary of $P$. If the ray hits another vertex, slightly rotate it about $v$, so that it instead hits the interior of an edge. Two situations can occur: either $P$ gets partitioned in two parts, with $r-1$ total reflex vertices and $h$ total holes, or two connected components of the boundary are joined, so that $P$ loses both a hole and a reflex vertex.

\begin{figure}[h]
\centering
\subfigure[]{\label{fig:7a}\includegraphics[scale=0.35]{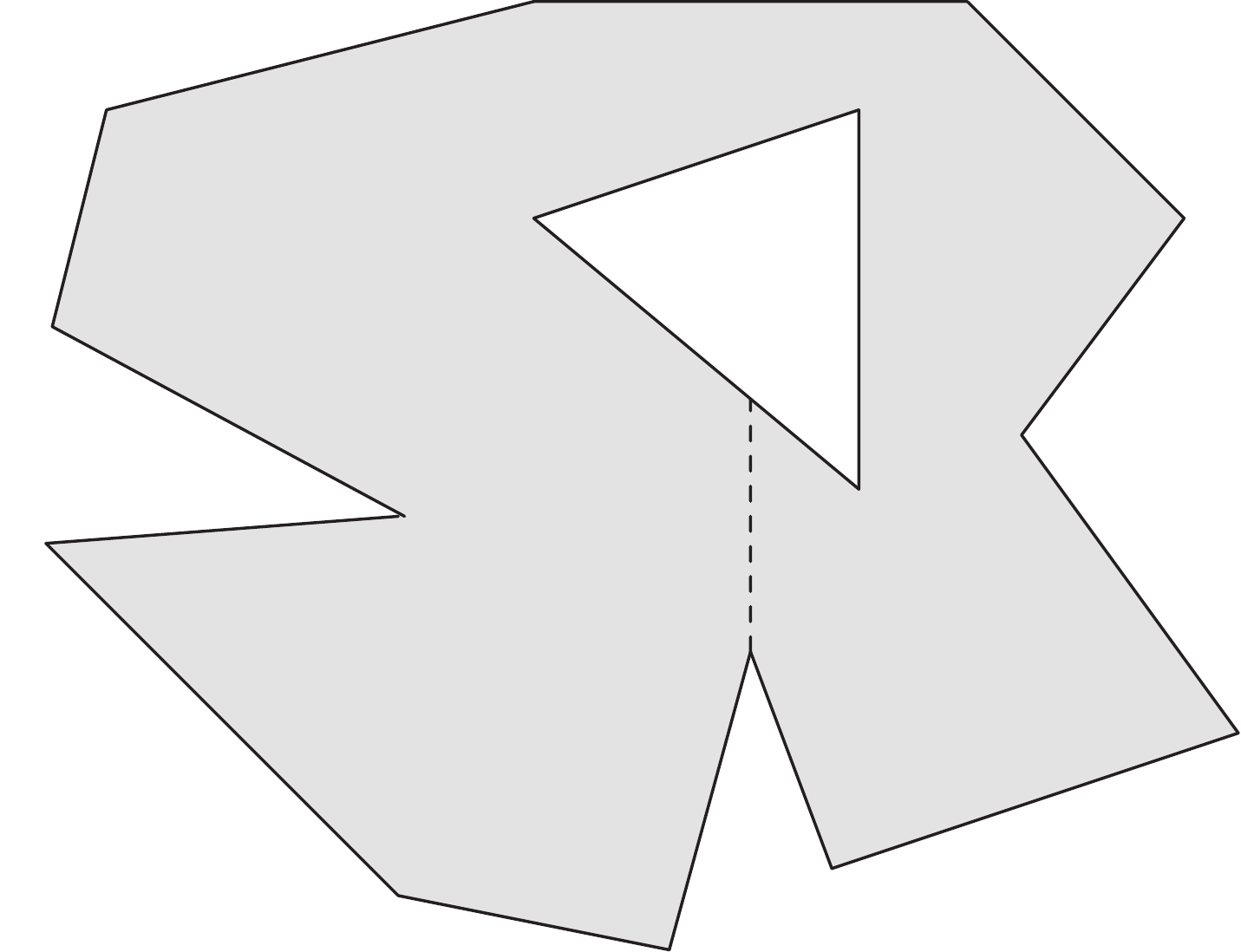}}\qquad \qquad
\subfigure[]{\label{fig:7b}\includegraphics[scale=0.35]{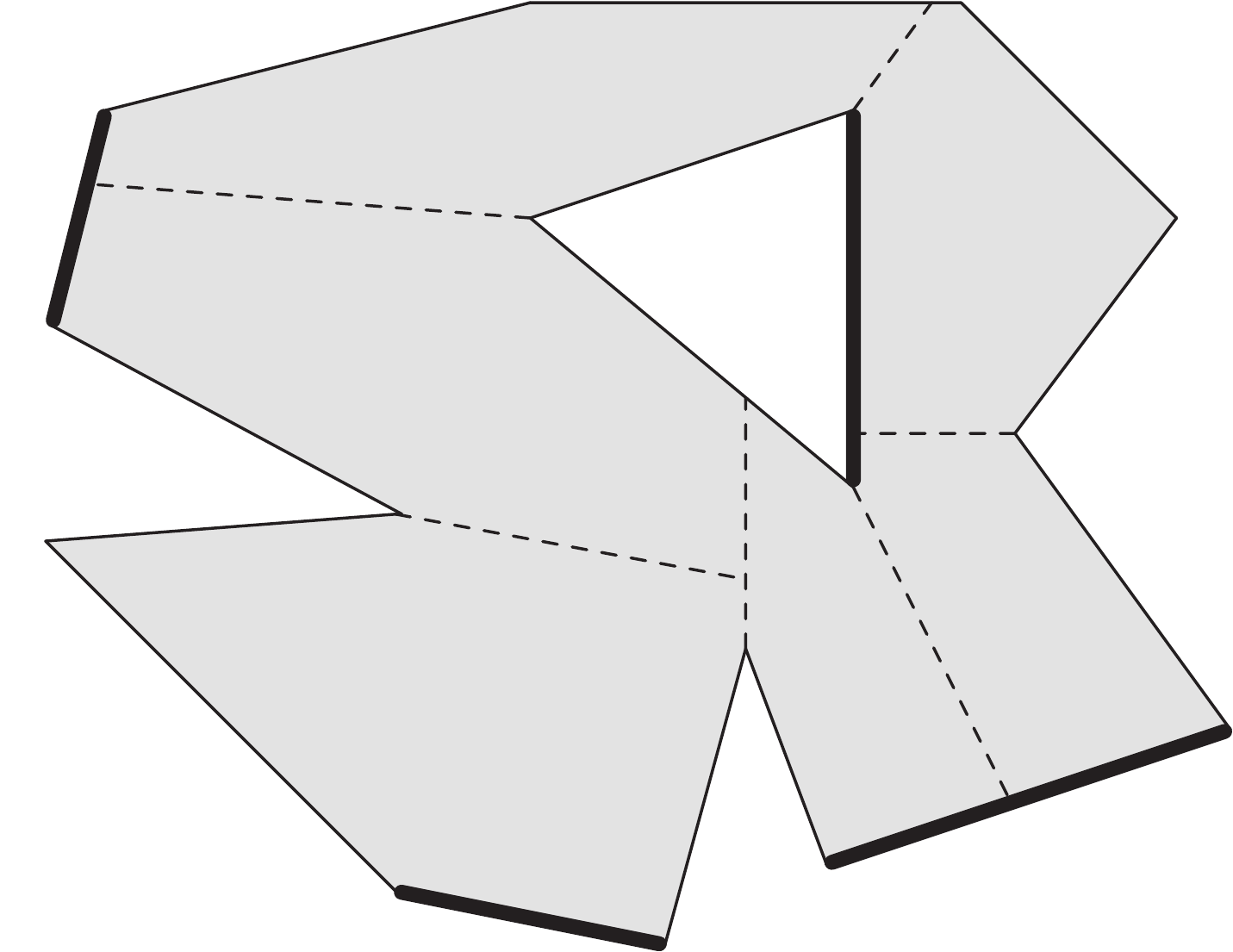}}
\put(-144,86){$e$}
\put(-336,86){$e$}
\caption{Example of the construction in Lemma~\ref{edges}. In \subref{fig:7a} the first step is shown, where the dashed line acts as a degenerate edge of the resulting polygon. In \subref{fig:7b} the final partition is shown, with the selected edges represented as thick lines.}
\label{fig:7}
\end{figure}

Repeat the process inductively on the resulting polygons, until no reflex vertex remains. Notice that polygonal boundaries may be \emph{degenerate} in the intermediate steps of this construction, meaning that a single segment should occasionally be regarded as two coincident segments (refer to Figure~\ref{fig:7a}). $P$ is now partitioned into convex pieces, which of course have no holes. Thus, during the process, the number of holes decreased $h$ times, while the number of pieces in the partition increased $r-h$ times, resulting in $r-h+1$ convex polygons. Additionally, each polygon has at least one edge lying on $P$'s boundary, and conversely every internal point of each edge of $P$ sees at least one complete region. Place a guard on the distinguished edge $e$, thus guarding at least one region of the partition. For each unguarded (or partially guarded) region, choose an edge of $P$ that completely sees it, and place a guard on it. As a result, $P$ is completely guarded and at most $r-h+1$ open edge guards have been placed.\end{proof}

\begin{observation}
The previous bound is asymptotically tight (in terms of $r$), because polygons with $r$ reflex vertices can be constructed that require $r$ open edge guards, for every $r$ (e.g., the shutter polygons in Figure~\ref{f1:art2}).
\end{observation}

\section{Bounds on point guard numbers}

Recall from Chapter~\ref{chapter1} that any 2-dimensional polygon can be guarded by vertex guards, due to the existence of triangulations. In Chapter~\ref{chapter2} we showed that triangulations do not generalize to 3-dimensional polyhedra. Actually, there are polyhedra that are not guardable by vertex guards, such as the octoplex, shown in Figure~\ref{f2:octoplex}. Indeed, even assigning a guard to each vertex of the octoplex fails to guard some points around its center.

\subsection*{Orthogonal polyhedra}
The class of polyhedra illustrated in Figure~\ref{f3:seidel} somewhat generalizes the octoplex. These simply connected orthogonal polyhedra are called \emph{multiplexes} in~\cite{adventures} and \emph{Seidel polyhedra} in~\cite{do-dcg-11,art}.

\begin{figure}[h]
\centering
\subfigure[outer view]{\includegraphics[width=0.4\linewidth]{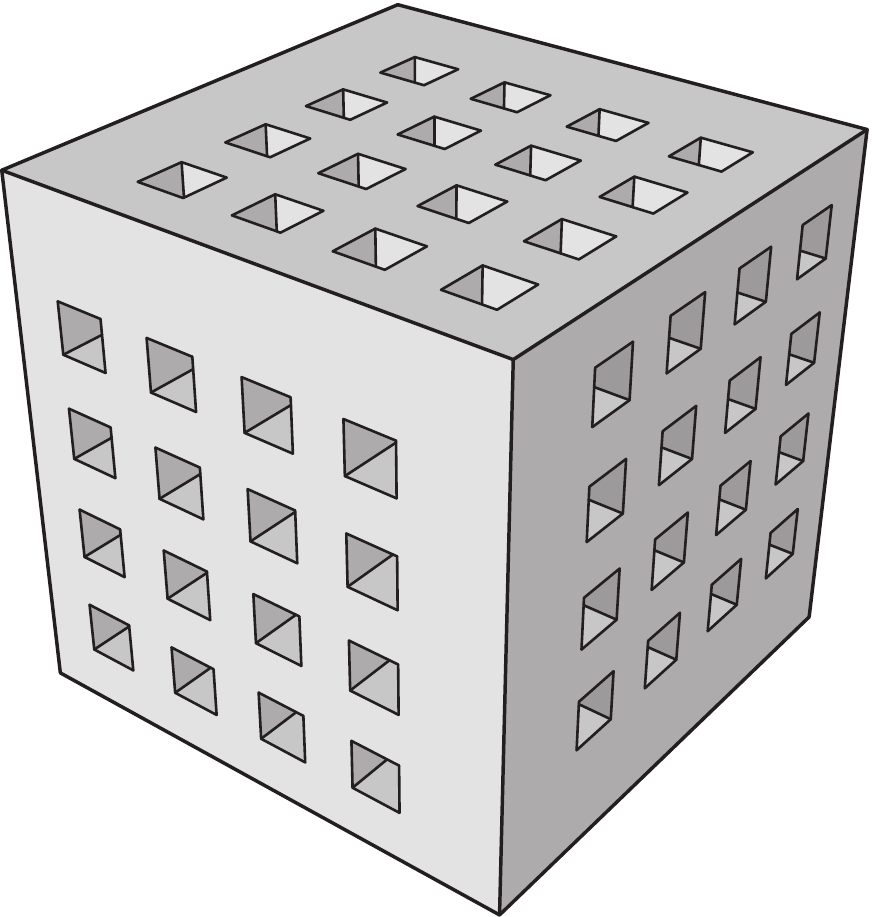}\label{f3:seidela}}\qquad \qquad
\subfigure[cross section]{\includegraphics[width=0.4\linewidth]{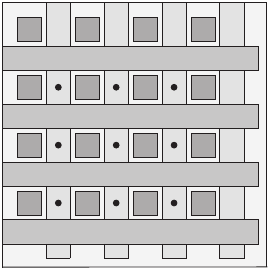}\label{f3:seidelb}}
\caption{Multiplex.}
\label{f3:seidel}
\end{figure}

To construct a multiplex, start from a cube and carve an array of deep cuboidal dents on three of
its faces. The dents do not intersect inside the cube;
instead, they subdivide its interior into smaller empty cubes, with narrow
cracks along every edge. Figure~\ref{f3:seidelb} shows a cross section of a multiplex, with a dot in the center of each small cube. Around each dot, there are points that are invisible to all vertices. Therefore, vertex guards are insufficient to guard multiplexes, and the number of necessary point guards is proportional to the number of small cubes. It easily follows that $\Omega(n^{3/2})$ point guards are necessary to guard multiplexes, where $n$ is the number of their vertices.

For orthogonal polyhedra, this bound is asymptotically tight, as proved by Paterson and Yao in~\cite{paterson}, via Binary Space Partitioning.

\begin{theorem}[Paterson--Yao]\label{t3:patyao}
$\Theta(n^{3/2})$ point guards are sufficient and occasionally necessary to guard an orthogonal polyhedron with $n$ vertices.\hfill\qed
\end{theorem}

\subsection*{General polyhedra}

For general polyhedra, a quadratic upper bound can be obtained with little effort:

\begin{observation}
For general polyhedra with $r$ reflex edges, an upper bound of $$\frac {r^2} 2 + \frac r 2 +1$$ guards trivially follows from Chazelle's ``revised naive decomposition'' (refer to Chapter~\ref{chapter2} and~\cite{polypart1}). Furthermore, such guards can be chosen on reflex edges.
\end{observation}

The general impression is that point guards are too ``weak'' to guard polyhedra, in that they do not allow generalizations of any of the results on 2-dimensional polygons mentioned in Chapter~\ref{chapter1}. In particular, no linear upper bound on the number of point guards exists, due to Theorem~\ref{t3:patyao}.

\section{Bounds on edge guard numbers}

\subsection*{Motivations}
We already discussed the analogy between guarding problems and illumination problems when introducing open edge guards. Recall from Chapter~\ref{chapter1} that another application of edge guards, and segment guards in general, is that of modeling patorling point guards moving back and forth on a line.

The question whether and how edge guards in polyhedra are truly more ``powerful'' than point guards immediately arises. In the next sections, we will show that the situation with edge guards is much more appealing, in terms of the number of guards required to guard a polyhedron.

Here, we merely point out a basic fact that holds even for edge guards in simple orthogonal polygons, which is another evidence of the intrinsic superiority of edge guards with respect to point guards.

\begin{figure}[h]
\centering{\includegraphics[width=.35\linewidth]{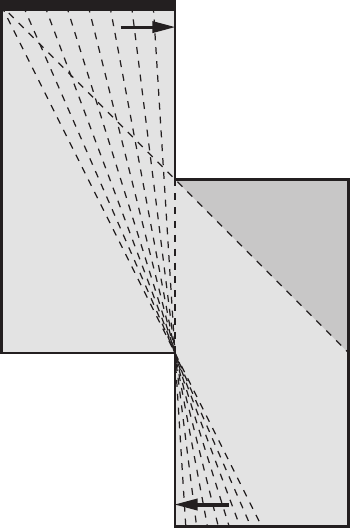}}
\caption{The upper edge guard cannot be replaced by a finite number of its points.}
\label{f3:edgeq}
\end{figure}

\begin{observation}\label{o3:edge}
There are simple orthogonal polygons in which an edge guard cannot be replaced by a finite number of its points. Figure~\ref{f3:edgeq} illustrates an example: if a subset $G$ of the upper edge $\ell$ is such that $\mathcal V(G)=\mathcal V(\ell)$, then the right endpoint of $\ell$ must be a limit point of $G$.
\end{observation}

Of course, Observation~\ref{o3:edge} holds for open and closed edge guards, and for open and closed polygons and polyhedra alike.

\subsection*{Bounds in terms of $e$}

\paragraph{Solvability.}
First of all, the \ART with edge guards is always ``solvable'', in that assigning guards to every edge is sufficient to guard any polyhedron $\mathcal P$. Indeed, given a point $x\in \mathcal P$, a cross section of $\mathcal P$ through $x$ is a polygon, whose vertices lie on $\mathcal P$'s edges. Because a polygon is guarded by its vertex set, $x$ is guarded. We may even select the cross section so that it contains no vertex of $\mathcal P$ (other than $x$ itself, if it coincides with a vertex).

\begin{observation}
Any polyhedron is guardable by its open edge set.
\end{observation}

\paragraph*{Urrutia's bounds.}

More bounds on edge guard numbers were given by Urrutia in his survey~\cite[Section~10]{urrutia2000}.

\begin{observation}[Urrutia]\label{o3:urr1}
There are polyhedra with $e$ edges that require at least
$$\left\lfloor \frac e 6\right\rfloor-1$$
(closed) edge guards to be guarded, for arbitrarily large $e$. Figure~\ref{fig3:lower1} shows an example.
\end{observation}

\begin{figure}[h]
\centering
\includegraphics[width=0.55\linewidth]{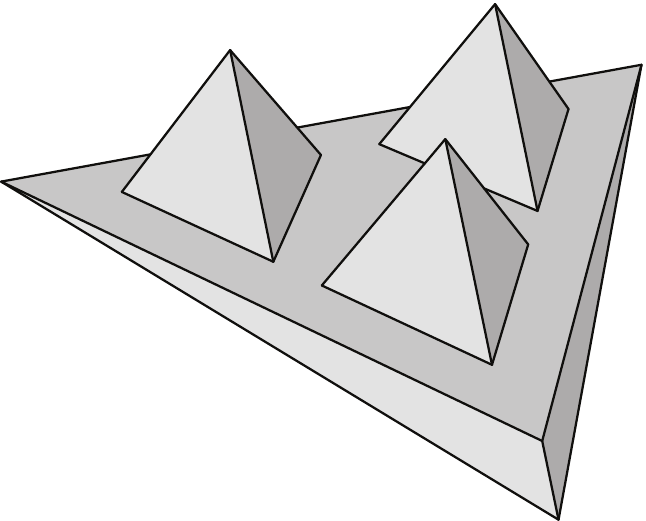}
\caption{Polyhedron with $6k+6$ edges that requires $k$ edge guards.}
\label{fig3:lower1}
\end{figure}

Urrutia conjectures this lower bound to be tight.

\begin{conjecture}[Urrutia]\label{con:urrutia}
Any genus-zero polyhedron with $e$ edges is guardable by at most
$$\frac e 6+O(1)$$
(closed) edge guards.
\end{conjecture}

For orthogonal polyhedra, the best known lower bound halves:

\begin{observation}[Urrutia]\label{obs3:urrutia}
There are orthogonal polyhedra with $e$ edges that require at least
$$\left\lfloor \frac e {12}\right\rfloor-1$$
(closed) edge guards to be guarded, for arbitrarily large $e$. Figure~\ref{fig3:lower2} shows an example.
\end{observation}

\begin{figure}[h]
\centering{\includegraphics[width=.65\linewidth]{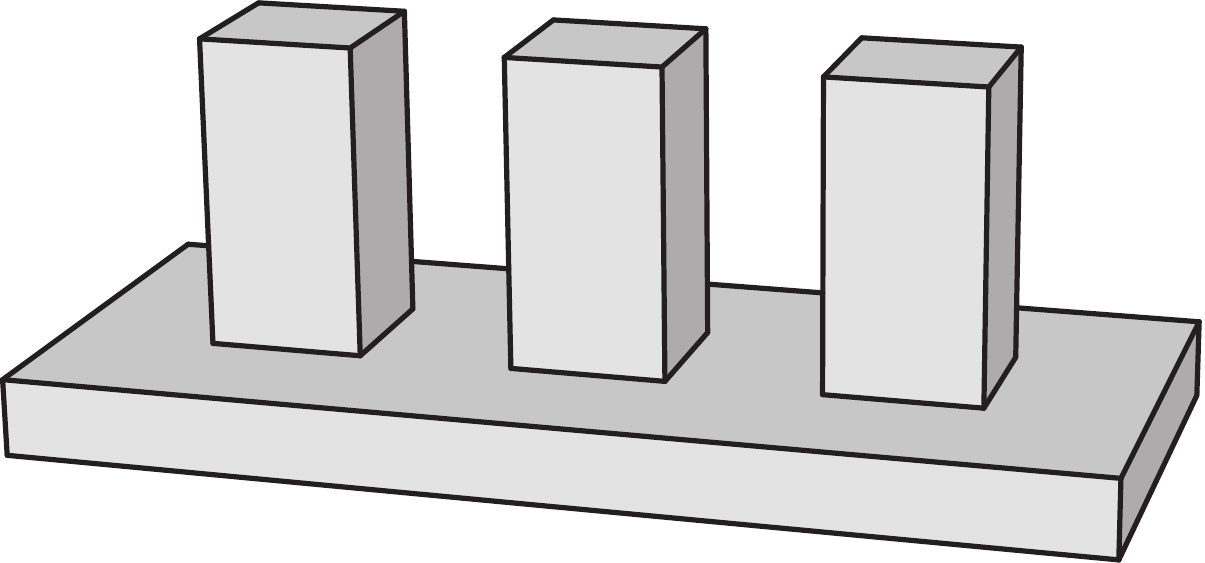}}
\caption{Orthogonal polyhedron with $12k+12$ edges that requires $k$ edge guards.}
\label{fig3:lower2}
\end{figure}

Once again, the lower bound is conjectured to be tight.

\begin{conjecture}[Urrutia]\label{con:urrutia}
Any orthogonal polyhedron with $e$ edges is guardable by at most
$$\frac e {12}+O(1)$$
(closed) edge guards.
\end{conjecture}

Urrutia also stated the following theorem, without proof. We will improve it in two directions, in Chapters~\ref{chapter4} and~\ref{chapter5}.

\begin{theorem}[Urrutia]\label{t:urrutia}
Any orthogonal polyhedron with $e$ edges is guardable by at most
$$\left\lfloor \frac e 6\right\rfloor$$
(closed) edge guards.\hfill\qed
\end{theorem}

\paragraph*{Upper bound for closed edge guards.}

A recent breakthrough by Cano, T\'oth and Urrutia, appeared in~\cite{edgenew}, slightly lowered the trivial upper bound of $e$ closed edge guards in general polyhedra.

\begin{theorem}[Cano--T\'oth--Urrutia]\label{thm:cano}
Any polyhedron with $e$ edges is guardable by at most
$$\left\lfloor \frac{27e}{32}\right\rfloor$$
(closed) edge guards.\hfill\qed
\end{theorem}

\subsection*{Bounds in terms of $r$}

\paragraph{Guarding with reflex edges.}
Remarkably, assigning guards only to reflex edges is sufficient to guard any polyhedron.

\begin{lemma}\label{l3:reflex}
Any open (resp.\ closed) non-convex polyhedron can be guarded by assigning an open (resp.\ closed) edge guard to each reflex edge.
\end{lemma}
\begin{proof}
Let $x$ be a point in a non-convex polyhedron $\mathcal P$, and let $y$ be an interior point of a reflex edge. If $y$ sees $x$, we are done. Otherwise, consider any shortest path from $x$ to $y$ in the topological closure of $\mathcal P$ (because this is a compact set, shortest paths within it exist, although they are not necessarily unique). Such path is a polygonal chain whose bending points lie on reflex edges (see~\cite{path}). The first bending point $z$ sees $x$. If $\mathcal P$ is closed and edge guards are closed, we are done. Otherwise, if $\mathcal P$ is an open polyhedron and $z$ lies on a (non-convex) vertex, then there is a reflex edge $\ell$ with an endpoint in $z$ such that a small-enough neighborhood of $z$ in $\ell$ sees $x$ (refer to Proposition~\ref{prop:1}).
\end{proof}

If $r>0$ is the number of reflex edges, this establishes an upper bound of $r$ edge guards.

\paragraph{Lower bound.}

For general (open) polyhedra, we have a matching lower bound of $r$ open edge guards.

\begin{theorem}\label{t3:reflex}
$r$ open edge guards are sufficient and occasionally necessary to guard an open polyhedron having $r$ reflex edges.
\end{theorem}
\begin{proof}
Sufficiency follows from Lemma~\ref{l3:reflex}. To see that $r$ open edge guards may be necessary, extrude the shutter polygons in Figure~\ref{f1:art2} to obtain ``shutter prisms'', which have $r$ reflex edges and require as many open edge guards.
\end{proof}

For orthogonal polyhedra and open edge guards, we can also give a lower bound in terms of $r$ that we believe to be tight (see our conjectures in Chapters~\ref{chapter4} and~\ref{chapter5}).

\begin{observation}\label{o3:rourke}
There are orthogonal polyhedra with $r$ reflex edges that require at least
$$\left\lfloor \frac r 2\right\rfloor+1$$
open edge guards, for any $r$. The ``staircase polyhedron'' in Figure~\ref{fig3:tight} is an example.
\end{observation}
	
\begin{figure}[h]
\centering{\includegraphics[width=.85\linewidth]{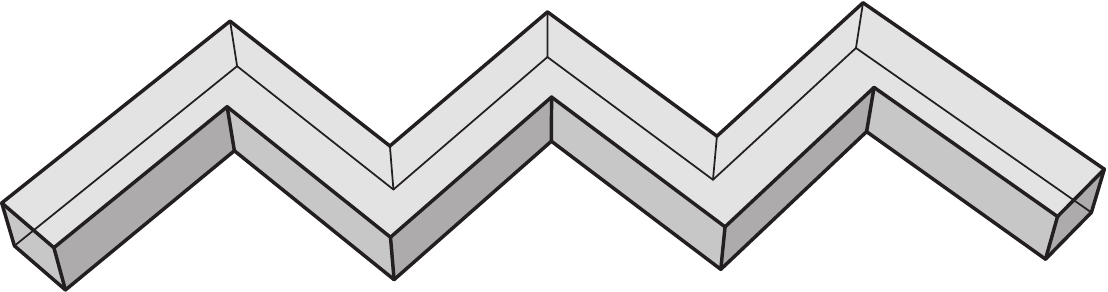}}
\caption{Orthogonal polyhedron with $r$ reflex edges that requires $\lfloor r/2\rfloor+1$ open edge guards.}
\label{fig3:tight}
\end{figure}

\section{Bounds on face guard numbers}

\subsection*{Motivations}
It is rather hard to find concrete applications of face guards. The obvious analogy with illumination that we mentioned for edge guards suggests that face guards may be luminous panels or screens.

On the other hand, imagining that a point guard could patrol a whole face of a polyhedron poses some problems. Recall that, due to Observation~\ref{o3:edge}, a point guard patroling an edge may not be ``locally'' replaced by finitely many static point guards. We can exploit this fact to construct the class of polyhedra sketched in Figure~\ref{f3:faceq}.

\begin{figure}[h]
\centering
\subfigure[3-dimensional view]{\label{f3:faceqa}\includegraphics[height=7cm]{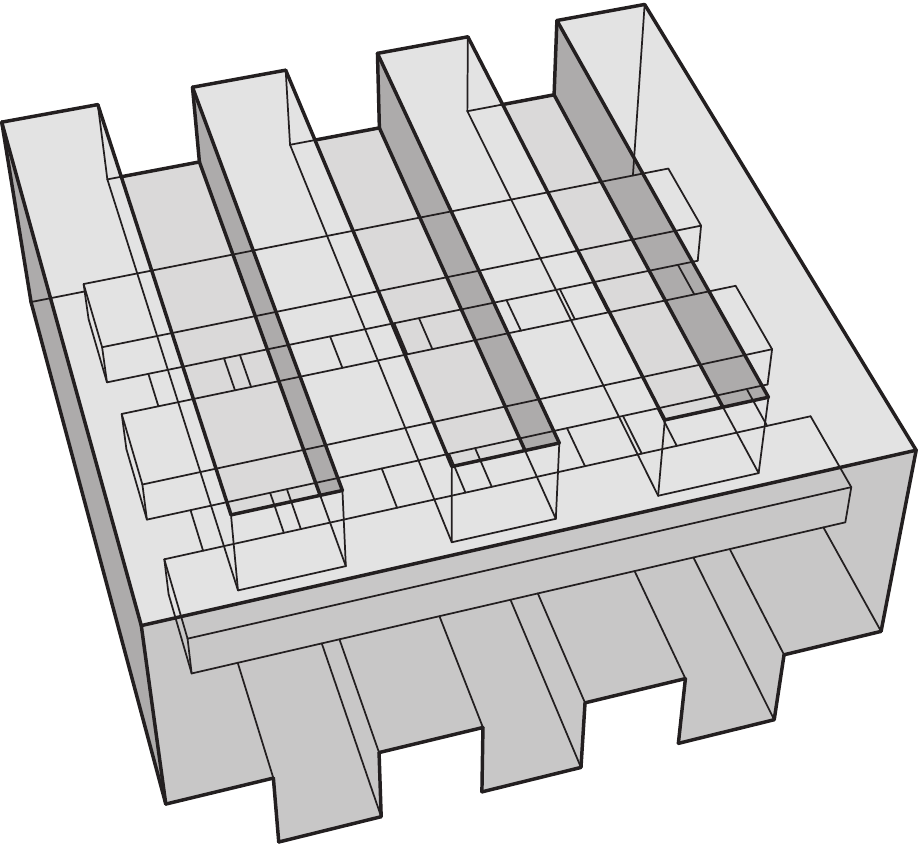}}\qquad \quad
\subfigure[top view]{\label{f3:faceqb}\includegraphics[height=6.5cm]{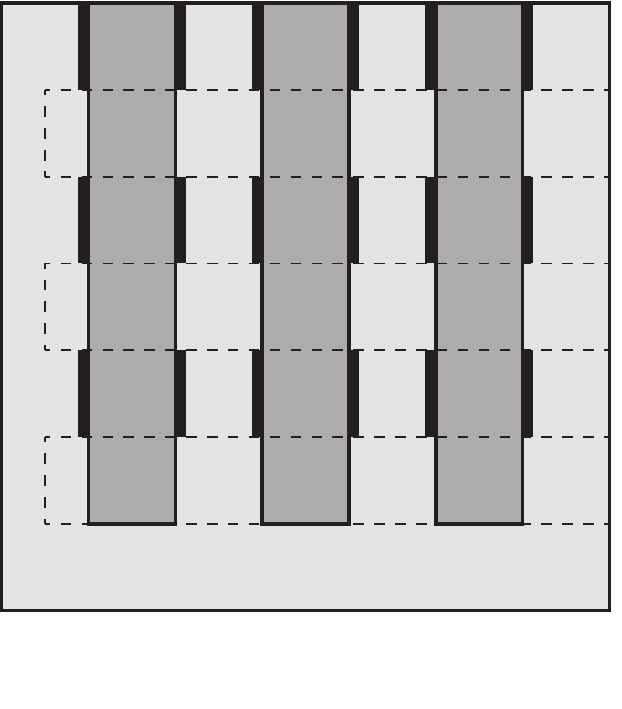}}
\caption{Orthogonal polyhedron whose top face cannot be ``dominated'' by a subquadratic number of segments lying on it.}
\label{f3:faceq}
\end{figure}

We start by cutting long parallel ``dents'' on opposite faces of a cuboid, in such a way that the resulting polyhedron looks like an extruded ``iteration'' of the polygon illustrated in Figure~\ref{f3:edgeq}. Then we stab this construction with a row of ``girders'' running orthogonally with respect to the dents, as Figure~\ref{f3:faceqa} suggests.

Suppose that a point guard has to patrol the top face of this construction, eventually seeing every point that is visible to that face. The situation is represented in Figure~\ref{f3:faceqb}, where the light-shaded region is the top face, and the dashed lines mark the underlying girders. By Observation~\ref{o3:edge}, and by the presence of the girders, each thick vertical segment must be approached by the patroling guard from the interior of the face.

Suppose that the polyhedron has $n$ dents and $n$ girders. Then, the number of its vertices, edges, or faces is~$\Theta(n)$. Now, if the guard moves along a polygonal chain lying on the top face, such a chain must have at least a vertex on each thick segment, which amounts to~$\Omega(n^2)$ vertices. Equivalently, if the face guard has to be substituted with segment guards lying on it, quadratically many guards are needed.

\begin{observation}
For arbitrarily large $n$, there are simply connected orthogonal polyhedra with $n$ edges having face guards that cannot be replaced by $o(n^2)$ segment guards lying on them.
\end{observation}

This suggests that face guards may not be the proper ``tool'' to model point guards patroling the surface of a polyhedron. Indeed, a face that counts as a single guard in this model could represent a point guard with a route of quadratic ``complexity'', or quadratically many patroling guards.

Even if we are allowed to replace a face guard with point guards patroling any segment in the polyhedron (i.e, not necessarily constrained to move on that face), a linear number of them may be needed. To see why, consider again the class of orthogonal polyhedra illustrated in Figure~\ref{fig3:lower2}, and arrange the ``chimneys'' in such a way that no straight line intersects more than two of them. If there are $n$ chimneys, then the complexity of the polyhedron is $\Theta(n)$, and a face guard lying on the bottom face must be replaced by $\Omega(n)$ segment guards. 

\begin{observation}
For arbitrarily large $n$, there are simply connected orthogonal polyhedra with $n$ edges having convex face guards that cannot be replaced by $o(n)$ segment guards (chosen anywhere).
\end{observation}

However, we know from the previous section that just a linear amount of edge guards is sufficient to guard any polyhedron, let alone ``dominate'' a face guard. Nonetheless, as we will show below, a linear amount of face guards may be needed to guard a polyhedron. Once again, face guards appear to be a needlessly ``powerful'' type of guard, and definitely a poor model for patroling guards.

For these reasons, in Parts~\ref{part2} and~\ref{part3} we will focus primarily on edge guards, which we think achieve the best tradeoff among pleasant theoretical properties, efficiency and applicability.

\subsection*{Upper bounds}

We provide a very simple upper bound on face guard numbers, which becomes tight for open face guards in orthogonal polyhedra.

\begin{theorem}\label{t3:face}
Any open (resp.\ closed) $c$-oriented polyhedron with $f$ faces is guardable by
$$\left\lfloor\frac f 2 - \frac f c\right\rfloor$$
open (resp.\ closed) face guards.
\end{theorem}
\begin{proof}
Let $\mathcal P$ be a polyhedron whose faces are orthogonal to $c\geqslant 3$ distinct vectors. Let $f_i$ be the number of faces orthogonal to the $i$-th vector $v_i$. Without loss of generality, $i<j$ implies $f_i \geqslant f_j$. Then,
$$f_1+f_2 \geqslant \left\lfloor \frac {2f} c \right\rfloor.$$
Assume the direction of the cross product $v_i \times v_j$ to be \emph{vertical}. Thus, there are at most
$$f-\left\lfloor \frac {2f} c \right\rfloor$$
non-vertical faces. Some of these are facing up, the others are facing down. Without loss of generality, at most half of them are facing down, and we assign a face guard to each of them. Therefore, at most
$$\left\lfloor\frac f 2 - \frac f c\right\rfloor$$
face guards have been assigned.

Let $x$ be a point in $\mathcal P$. Casting a vertical ray from $x$ directed upward, we eventually reach a point $y$ on a face $F$ that is assigned a guard. If $\mathcal P$ is closed and face guard are closed, then $x$ is guarded by $F$. 

Otherwise, if $\mathcal P$ and the face guards are open, $x$ may not be guarded by $y$, as $y$ may lie on an edge of $F$. However, because $\mathcal P$ is open and $x$ lies in its interior, it can see a neighborhood of $y$ belonging to $F$ (recall Proposition~\ref{prop:1}). Such a neighborhood contains points in the relative interior of $F$, which guard $x$.
\end{proof}

For orthogonal polyhedra, the upper bound given in Theorem~\ref{t3:face} becomes $\left\lfloor f/6\right\rfloor$. In this case, if face guards are closed, they also \emph{orthogonally} guard the polyhedron.

Our guarding strategy becomes less and less efficient as $c$ grows. If $c=f$, we get an upper bound of $\left \lfloor f/2\right\rfloor - 1$ face guards. The same construction works even for general polyhedra that are not $f$-oriented (i.e., even if three distinct face normal vectors do not form a basis for $\R^3$).

\begin{corollary}
Any open (resp.\ closed) polyhedron with $f$ faces is guardable by
$$\left \lfloor\frac f 2\right\rfloor - 1$$
open (resp.\ closed) face guards.\hfill\qed
\end{corollary}

\subsection*{Lower bounds}
\paragraph{Open face guards.}
We give a simple lower bound construction for open face guards in general polyhedra.

\begin{observation}
There are (open) polyhedra with $f$ faces that require at least
$$\left \lfloor\frac f 4\right\rfloor$$
open face guards, for arbitrarily large $f$. Figure~\ref{f3:facelower} shows an example.
\end{observation}

\begin{figure}[h]
\centering
\includegraphics[width=0.75\linewidth]{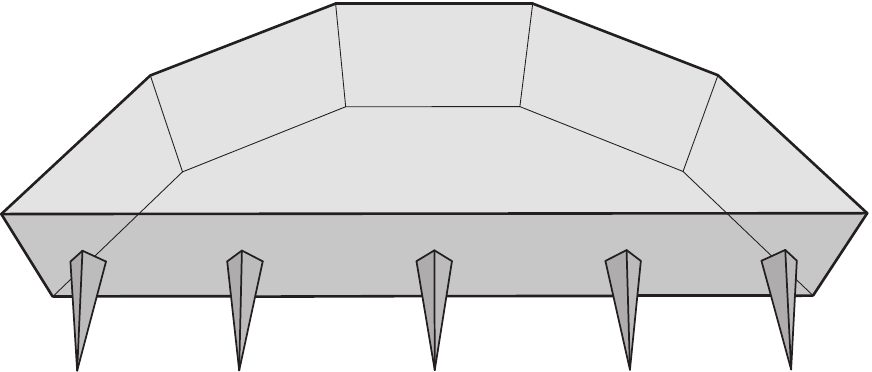}
\caption{Polyhedron with $4k+3$ faces that requires $k$ face guards. Each open face sees the tip of at most one of the $k$ tetrahedral ``spikes''.}
\label{f3:facelower}
\end{figure}

For open face guards in orthogonal polyhedra, we have a lower bound that is also tight.

\begin{theorem}
To guard an open orthogonal polyhedron with $f$ faces,
$$\left\lfloor \frac f 6\right\rfloor$$
open face guards are sufficient and occasionally necessary.
\end{theorem}
\begin{proof}
Sufficiency follows from Theorem~\ref{t3:face} with $c=3$. Necessity is implied by Figure~\ref{f3:face}. Indeed, any small L-shaped polyhedron that is attached to the big cuboid adds six faces to the construction, of which at least one must be selected. Moreover, no matter how these faces are selected, some portion of the big cuboid remains unguarded, and needs one more face guard.
\end{proof}

\begin{figure}[h]
\centering{\includegraphics[width=.65\linewidth]{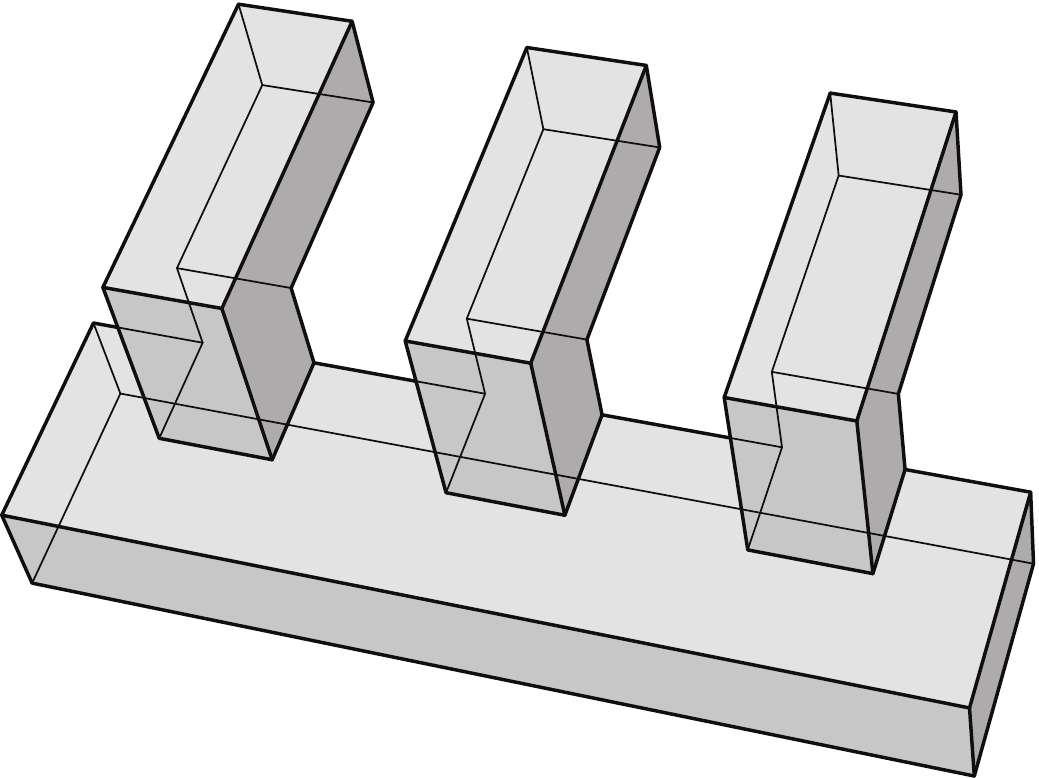}}
\caption{Orthogonal polyhedron that needs $f/6$ open face guards.}
\label{f3:face}
\end{figure}

\paragraph{Closed face guards.}
Some special cases of Theorem~\ref{t3:face} have been recently obtained by Souvaine, Veroy and Winslow, for closed face guards in closed polyhedra only (see~\cite{faceguards}). For this type of guards, they also gave two lower bounds.

\begin{observation}[Souvaine--Veroy--Winslow]
For arbitrarily large $f$, there are polyhedra with $f$ faces that require at least
$$\left\lfloor \frac f 5 \right\rfloor$$
closed face guards, and orthogonal polyhedra that require at least
$$\left\lfloor \frac f 7 \right\rfloor$$
closed face guards.
\end{observation}

\section{Hardness of edge-guarding}
Computing optimal edge guard numbers is hard even for simply connected orthogonal polyhedra, and this can be proved in almost the same way as for vertex or edge guards in polygons.

\paragraph{\NP-hardness.}
Several variations on the \ART with edge guards are \NP-hard, such as open or closed edge guarding, or reflex edge guarding, or guarding with mutually parallel edge guards. All reductions are obtained by operating small adjustments on a common pattern (inspired by the 2-dimensional one, originally given by Lee and Lin  in~\cite{lee}), which we roughly sketch here for reflex edge guards in simply connected orthogonal polyhedra.

\begin{theorem}\label{t3:hard}
Deciding whether a simply connected orthogonal polyhedron is guardable by $k$ (open or closed) reflex edge guards is strongly \NP-complete.
\end{theorem}
\begin{proof}
Membership in \NP is straightforward. To prove \NP-hardness, we reduce from \TSAT, so let $\varphi$ be a Boolean formula with $n$ variables and $m$ clauses. We will construct a genus-zero orthogonal prism that is guardable by $3m+3n+1$ reflex edge guards if and only if $\varphi$ is satisfiable.

One side of the construction contains an array of $m$ \emph{clause gadgets}, sketched in Figure~\ref{f3:np1}, and the opposite side contains an array of $n$ \emph{variable gadgets}, sketched in Figure~\ref{f3:np2}.

\begin{figure}[h]
\centering{\includegraphics[width=.85\linewidth]{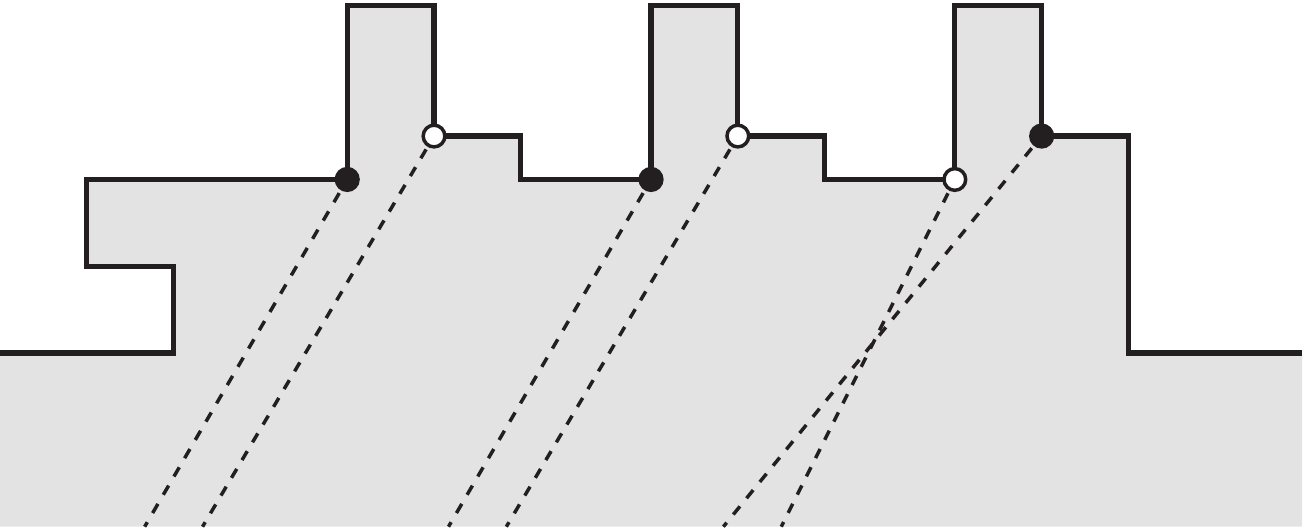}}
\caption{Clause gadget for $(x\vee y\vee \neg z)$.}
\label{f3:np1}
\end{figure}

Each clause gadget contains three \emph{alcoves}, corresponding to the three literals in the clause. Only two reflex edges can guard the bottom of each alcove; one is colored black, the other is colored white in the figure. These are called \emph{literal edges}. In each pair of literal edges, the lower one is colored black if and only if the literal is positive in $\varphi$. Thus, at least three literal edges per clause gadget must be selected. Three are sufficient if and only if at least one of them is the lower one in its pair (otherwise, if only the three topmost literal edges are selected, the small ``cleft'' on the left remains unguarded).

\begin{figure}[h]
\centering{\includegraphics[width=.7\linewidth]{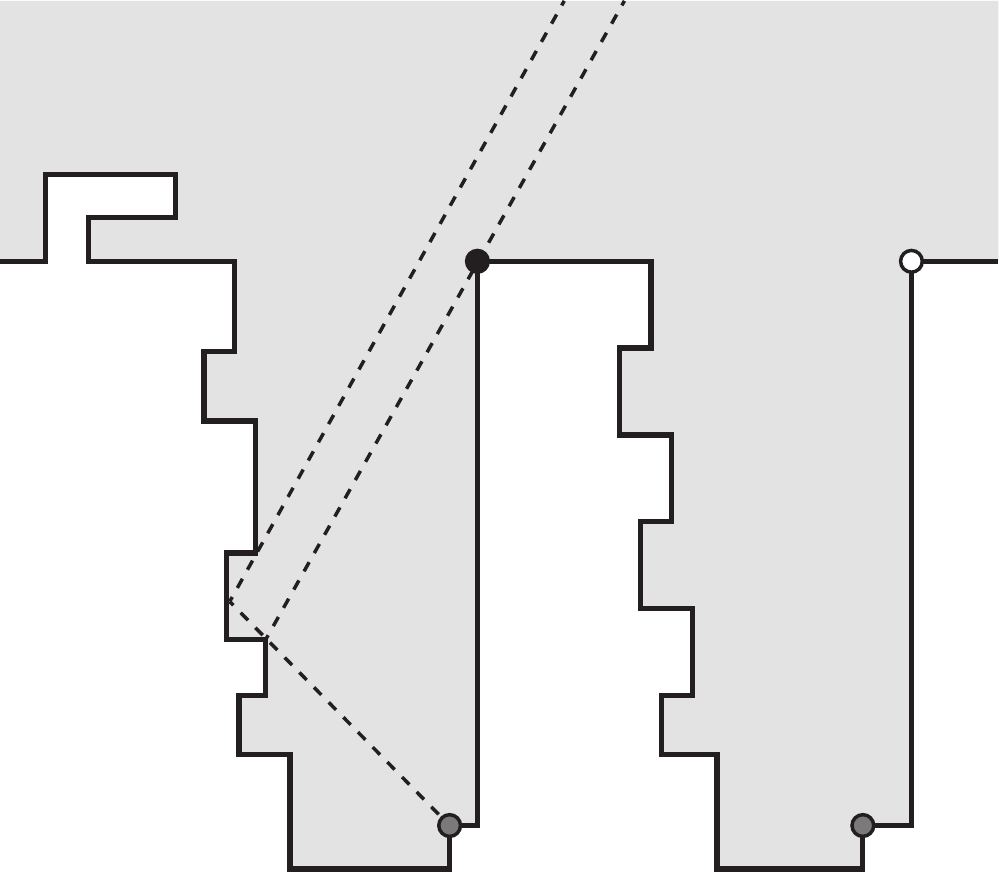}}
\caption{Variable gadget.}
\label{f3:np2}
\end{figure}

Each variable gadget is made of two \emph{wells} and a small enlosed area on the left, called \emph{nook}. Each well contains one \emph{dent} for each occurrence of the variable in $\varphi$. The two bottom reflex edges, marked with a gray dot, are always selected, and they guard the upper part of each dent. The lower parts of the dents in the left (resp.\ right) well are guardable by the edge marked with a black (resp.\ white) dot (these two are called \emph{variable edges}). Moreover, of the two dents corresponding to some literal $\ell$ (one in the left well and one in the right well), the leftmost (resp.\ rightmost) dent and the black (resp.\ white) variable edge are ``aligned'' with the black (resp.\ white) literal edge corresponding to $\ell$ (as the dashed lines in Figures~\ref{f3:np1} and~\ref{f3:np2} suggest: the two almost parallel lines coming out of the variable gadget should converge on one of the black dots in the clause gadget). Every single element of the construction is properly positioned and stretched, so that the bottom part of each dent of a variable gadget is completely visible to exactly one literal edge.

Now, if a variable is assigned the value true (resp.\ false), the white (resp.\ black) edge in the corresponding variable gadget is selected. Thus, the dents in the corresponding well are guarded, and the ones in the other dent are to be guarded by the respective literal edges, which are also selected. A clause is considered satisfied if and only if one of the three bottom literal edges has been selected. All the nooks are also guarded by variable edges, and we observe that the presence of nooks enforces the selection of at least one variable edge in each variable gadget. It follows that each variable gadget and each clause gadget can be guarded by exactly three reflex edge guards if and only if $\varphi$ is satisfiable.

If each gadget is guarded, then the whole construction is guarded, provided that one additional guard is assigned to the top-left reflex edge of the nook of the leftmost variable gadget (such edge must be selected anyway, as it is the only one that can see ``behind'' the leftmost nook).
\end{proof}

\paragraph{Hardness of approximation.}

Analogous hardness results also hold for the problem of approximating the minimum size of an edge-guarding set in a given polyhedron. As a general rule, if an \ART is hard for vertex guards in 2-dimensional polygons \emph{with holes}, then the corresponding problem for edge guards in \emph{simply connected} 3-dimensional polyhedra is also hard.

As an example, we show that the \emph{VC-dimension} of the \emph{range spaces} associated to edge-guarding problems in simply connected orthogonal polyhedra may grow unboundedly.

Observe that a guarding problem may be viewed as an instance of \computproblem{SET COVER}, in which some potential guard locations have to be selected, each of which covers a subset of the whole environment (refer to~\cite{artapproxfirst}).

The concept of VC-dimension, named after Vapnik and Chervonenkis (who defined it in~\cite{vcdimension}), is associated to the ``complexity'' of a \computproblem{SET COVER} instance, and is defined as the size of the largest \emph{shattered set} in that instance.

For guarding problems, a set of points $S$ is \emph{shattered} by the set of potential guard postions if, for each subset $S'\subseteq S$, there exists a (potential) guard $g$ such that $S \cap \mathcal V(g) = S'$.

It was shown in a series of papers by Blumer, Br\"onnimann, Kalai and others (refer to~\cite{blumer,bronnimann,kalai}) that any class of instances of \computproblem{SET COVER} whose VC-dimension is bounded by a constant has an $O(\log \rm{OPT})$-approximation algorithm.

Later, in~\cite{valtr}, Valtr proved that the VC-dimension associated to point guards in a simple polygon is not more than 23, while polygons with $h$ holes may yield VC-dimensions of $\Omega (\log h)$.

With a very similar construction to that found in~\cite{valtr} for polygons with holes, we show that simply connected orthogonal polyhedra with $e$ edges may have VC-dimensions of $\Omega (\log e)$, implying that an $O(\log \rm{OPT})$-approximation algorithm for the corresponding minimization \ART does \emph{not} follow from general theorems. This provides evidence that approximating minimum edge-guarding numbers is indeed hard.

\begin{figure}[h!]
\centering{\includegraphics[width=.65\linewidth]{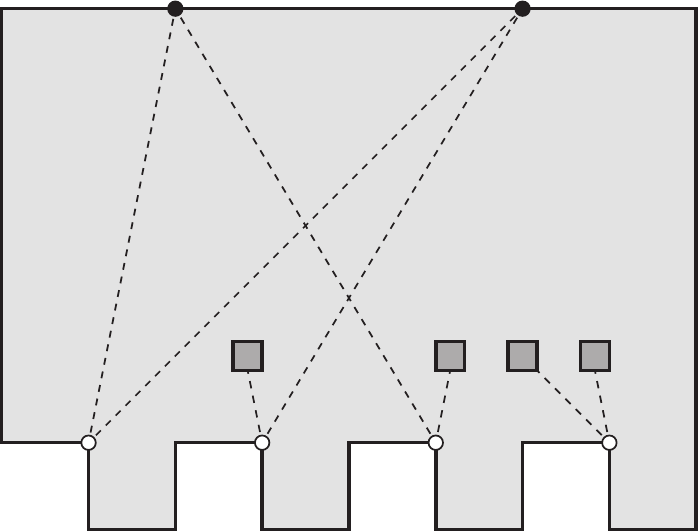}}
\caption{Set of two black points shattered by the white guards. Darker squares represent pillars.}
\label{f3:shatter}
\end{figure}

\begin{theorem}\label{t3:aaa}
There are simply connected orthogonal polyhedra with $e$ edges containing point sets of size $\Omega (\log e)$ that are shattered by the set of edges.
\end{theorem}

\begin{proof}
The construction is illustrated in Figure~\ref{f3:shatter} for a shattered set of two points. The outer boundary of the polyhedron is extruded to form an orthogonal prism, and the darker squares become high \emph{pillars} that start from the ``floor'' and almost touch the ``ceiling''.

The two black points lie on the floor, so that they are shattered by the white ones, which mark vertical reflex edges (dashed lines represent lines of sight).

It is easy to add more black points to the construction, and exponentially many white edges and pillars, so that the black points are shattered by the white edges.
\end{proof}

From the construction given in Theorem~\ref{t3:aaa}, it is clear how our pillars substitute holes in 2-dimensional polygons, thus enabling hardness proofs even for simply connected polyhedra.

\part{Guarding polyhedra}\label{part2}
\chapter{Reflex edge guards in 2-reflex orthogonal polyhedra}\label{chapter4}
\markthischapter{CHAPTER 4. \ EDGE GUARDS IN 2-REFLEX POLYHEDRA}

\begin{chapterabstract}
We consider the problem of guarding orthogonal polyhedra having reflex edges in just two directions (as opposed to three), by placing guards on reflex edges only.

We generalize a classic result by O'Rourke, showing that $$\left\lfloor \frac{r-g}{2} \right\rfloor +1$$ reflex edge guards are sufficient, where $r$ is the number of reflex edges in a given polyhedron and $g$ is its genus. This bound is tight for $g=0$.

Then we give a similar upper bound in terms of $e$, the total number of edges in the polyhedron. We prove that $$\left\lfloor \frac{e-4}{8} \right\rfloor +g$$ reflex edge guards are sufficient, whereas the previous best known bound, due to Urrutia, was $\lfloor e/6 \rfloor$ edge guards (not necessarily reflex).

Ultimately, we show that guard locations achieving the above bounds can be computed in $O(n \log n)$ time.

En route, we also discuss the setting in which guards and polyhedra are open, proving that the same results hold even in this more challenging case.
\end{chapterabstract}

\section{2-reflex orthogonal polyhedra}

\subsection*{Motivations}
This is the first of three chapters in which we study edge guards in polyhedra. Here we focus on 2-reflex orthogonal polyhedra, i.e., orthogonal polyhedra whose reflex edges lie in at most two different directions. This is a case of intermediate complexity, between the 1-reflex case (i.e., orthogonal prisms) and the full 3-reflex case (i.e., general orthogonal polyhedra).

Recall from Chapter~\ref{chapter1} that simple orthogonal polygons with $r$ reflex vertices can be guarded by $\lfloor r/2 \rfloor +1$ guards. This obviously extends to simply connected orthogonal prisms with $r$ reflex edges. Our main research question is whether the same bound in terms of $r$ extends to the whole class of orthogonal polyhedra. We are still unable to fully answer this question, although we have evidence that this may be the case.

However, we can prove (as we will do in this chapter) that the $\lfloor r/2 \rfloor +1$ upper bound holds at least for 2-reflex orthogonal polyhedra. We perceive this as a very important sub-case, and a necessary step toward a proof for general orthogonal polyhedra. Indeed, we believe that recursively ``cutting away'' 2-reflex orthogonal subpolyhedra from a given orthogonal polyhedron eventually yields a ``kernel'' that can be efficiently guarded, due to its structural properties.

Regardless of this theoretical aspect, 2-reflex orthogonal polyhedra have an interest by themselves, as they can already express a rich variety of shapes, which planar structures cannot attain.

\subsection*{Structure and terminology}

Without loss of generality, we stipulate that every 2-reflex orthogonal polyhedron encountered in this chapter has only horizontal reflex edges, and no vertical ones.

The intersection between a 2-reflex orthogonal polyhedron and a horizontal plane not containing any of its vertices is a collection of mutually disjoint rectangles. As the plane is moved upward or downward, the intersection does not change, until a vertex is reached.  Thus, a natural way to partition a 2-reflex orthogonal polyhedron is into maximal cuboids, any two of which are either disjoint or touch each other at their top or bottom faces. Each cuboid in the partition is called a \emph{brick}, and any non-empty intersection between two bricks is called their \emph{contact rectangle}.

Figure~\ref{fig4:1} shows all the possible ways two bricks can touch each other. They are viewed from above, and the light shaded brick lies on top of the darker one. Thick lines denote reflex edges, and the dashed ones are those covered by the top brick (hence hidden to the viewer).

\begin{figure}[h]
\centering
\subfigure[$e-e'=0$ $r-r'=4$]{\includegraphics[width=0.15\linewidth]{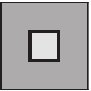}}\qquad
\subfigure[$e-e'=0$ $r-r'=3$]{\includegraphics[width=0.15\linewidth]{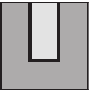}}\qquad
\subfigure[$e-e'=-3$ $r-r'=2$]{\includegraphics[width=0.15\linewidth]{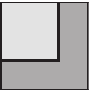}}\qquad
\subfigure[$e-e'=-6$ $r-r'=1$]{\includegraphics[width=0.15\linewidth]{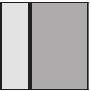}}\qquad
\subfigure[$e-e'=0$ $r-r'=2$]{\includegraphics[width=0.15\linewidth]{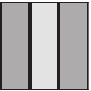}}\\
\subfigure[$e-e'=0$ $r-r'=4$]{\includegraphics[width=0.15\linewidth]{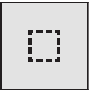}}\qquad
\subfigure[$e-e'=0$ $r-r'=3$]{\includegraphics[width=0.15\linewidth]{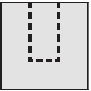}}\qquad
\subfigure[$e-e'=-3$ $r-r'=2$]{\includegraphics[width=0.15\linewidth]{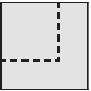}}\qquad
\subfigure[$e-e'=-6$ $r-r'=1$]{\includegraphics[width=0.15\linewidth]{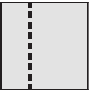}}\qquad
\subfigure[$e-e'=0$ $r-r'=2$]{\includegraphics[width=0.15\linewidth]{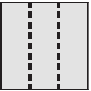}}\\
\subfigure[$e-e'=2$ $r-r'=3$]{\includegraphics[width=0.15\linewidth]{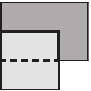}}\qquad
\subfigure[$e-e'=4$ $r-r'=4$]{\includegraphics[width=0.15\linewidth]{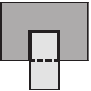}}\qquad
\subfigure[$e-e'=4$ $r-r'=3$]{\includegraphics[width=0.15\linewidth]{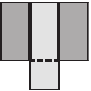}}\qquad
\subfigure[$e-e'=4$ $r-r'=4$]{\includegraphics[width=0.15\linewidth]{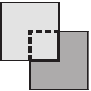}}\qquad
\subfigure[$e-e'=8$ $r-r'=4$]{\includegraphics[width=0.15\linewidth]{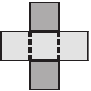}}\\
\subfigure[$e-e'=2$ $r-r'=3$]{\includegraphics[width=0.15\linewidth]{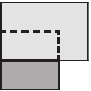}}\qquad
\subfigure[$e-e'=4$ $r-r'=4$]{\includegraphics[width=0.15\linewidth]{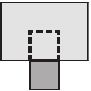}}\qquad
\subfigure[$e-e'=4$ $r-r'=3$]{\includegraphics[width=0.15\linewidth]{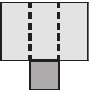}}\qquad
\subfigure[$e-e'=0$ $r-r'=2$]{\includegraphics[width=0.15\linewidth]{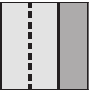}}\qquad
\subfigure[$e-e'=-1$ $r-r'=2$]{\label{f4:typet}\includegraphics[width=0.15\linewidth]{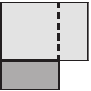}}
\caption{Possible contact rectangles of two adjacent bricks, viewed from above.}
\label{fig4:1}
\end{figure}

By ``cutting'' a 2-reflex orthogonal polyhedron along the contact rectangle of two adjacent bricks, all the reflex edges bordering the rectangle turn into convex edges, and several different things may happen to the edge set. A convex edge may split in two distinct edges, new edges may be created, and several edges may merge together.

Figure~\ref{fig4:1} also indicates the number of edges gained or lost during a split, for each different configuration. $e$ and $e'$ are the number of edges in the polyhedron before and after the split, respectively. Similarly, $r$ and $r'$ are the number of reflex edges before and after the split, respectively.

In Figure~\ref{fig4:2}, a type-(t) contact between two bricks is illustrated, before and after the cut. The polyhedron in the second picture has $e=23$ edges, of which $r=2$ are reflex. The polyhedra in the third picture have $e'=24$ edges in total, of which $r'=0$ are reflex (cf.\ Figure~\ref{f4:typet}).

\begin{figure}[h]
\centering
\includegraphics[width=\linewidth]{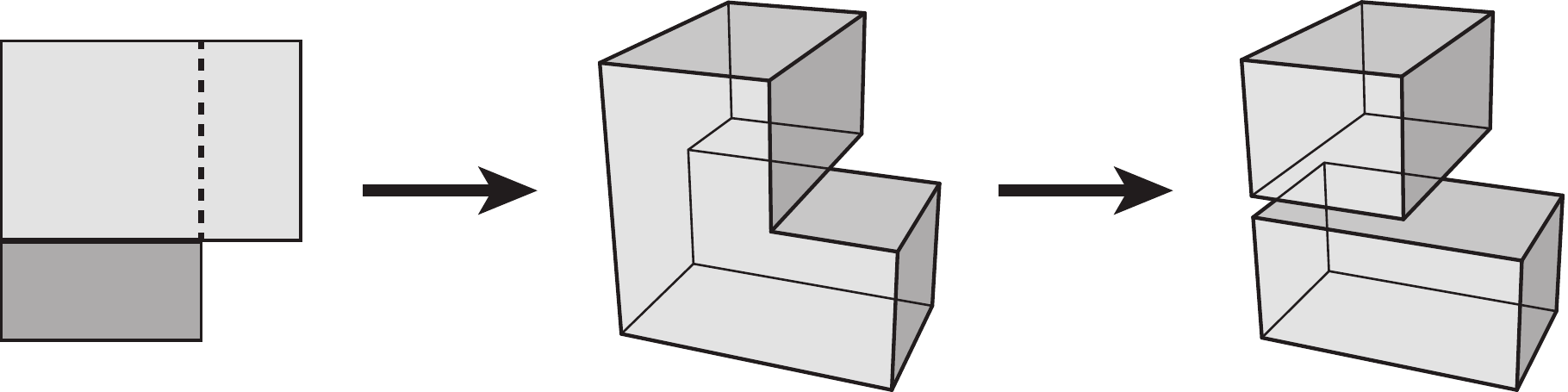}
\caption{Type-(t) contact between two bricks, before and after a cut.}
\label{fig4:2}
\end{figure}

\begin{remark}\label{r4:degen}
When cutting a polyhedron of positive genus along a contact rectangle, we may fail to disconnect it, but just lower its genus. The resulting polyhedron is \emph{degenerate}, in that its boundary is self-intersecting. We will occasionally encounter such degeneracies in intermediate steps of inductive proofs, and we will tolerate them whenever their presence will not invalidate our reasoning.
\end{remark}

Referring again to Figure~\ref{fig4:1}, we call each type-(a) or type-(f) contact rectangle a \emph{collar}, because its boundary is made of four reflex edges ``winding'' around a smaller brick. Singling out collars to treat them as separate cases will often be needed in our proofs. The (technical) reason is that collars minimize the ratio $$\frac{e-e'+12}{r-r'}.$$ This ratio is 3 for collars, whereas it is at least 4 for any other contact type.

We also want to single out contact types~(d) and~(i), because they produce just one reflex edge each, and this will turn out to be the hardest case to handle in our later constructions. These two contact types will be called \emph{primitive}, and a 2-reflex orthogonal polyhedron having only primitive contact rectangles will be called a \emph{stack}.

Observe that each brick of a stack may have up to two bricks attached to each of its horizontal faces. A stack in which each brick has either zero or two other bricks attached to its top face is called a \emph{castle} (Figure~\ref{fig4:castle} shows an example). The bottom brick of a castle is called its \emph{base brick}. It is straightforward to see that a castle has an even number of reflex edges, because they all come in parallel pairs.

\begin{figure}[h]
\centering
\includegraphics[width=0.65\linewidth]{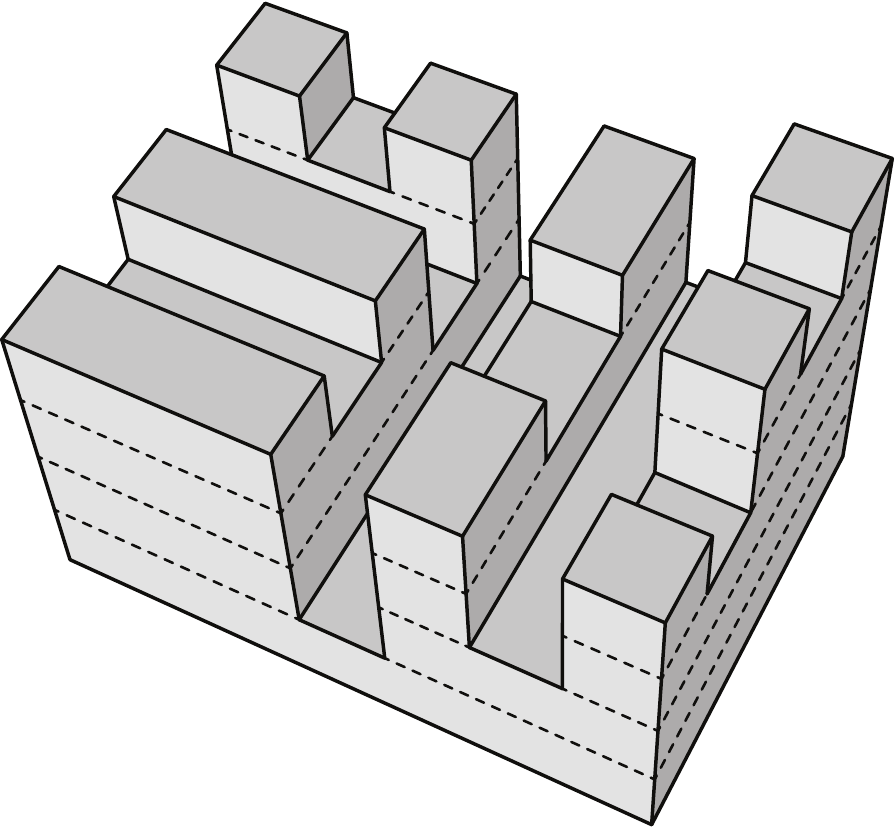}
\caption{Castle with dashed lines bordering contact rectangles between bricks.}
\label{fig4:castle}
\end{figure}

If a castle is turned upside down and its base is attached to another castle's base via a primitive contact rectangle, the resulting shape is a stack called \emph{double castle} (see Figure~\ref{fig4:dcastle}). It follows that a double castle has an odd number of reflex edges.

\begin{figure}[h]
\centering
\includegraphics[width=0.5\linewidth]{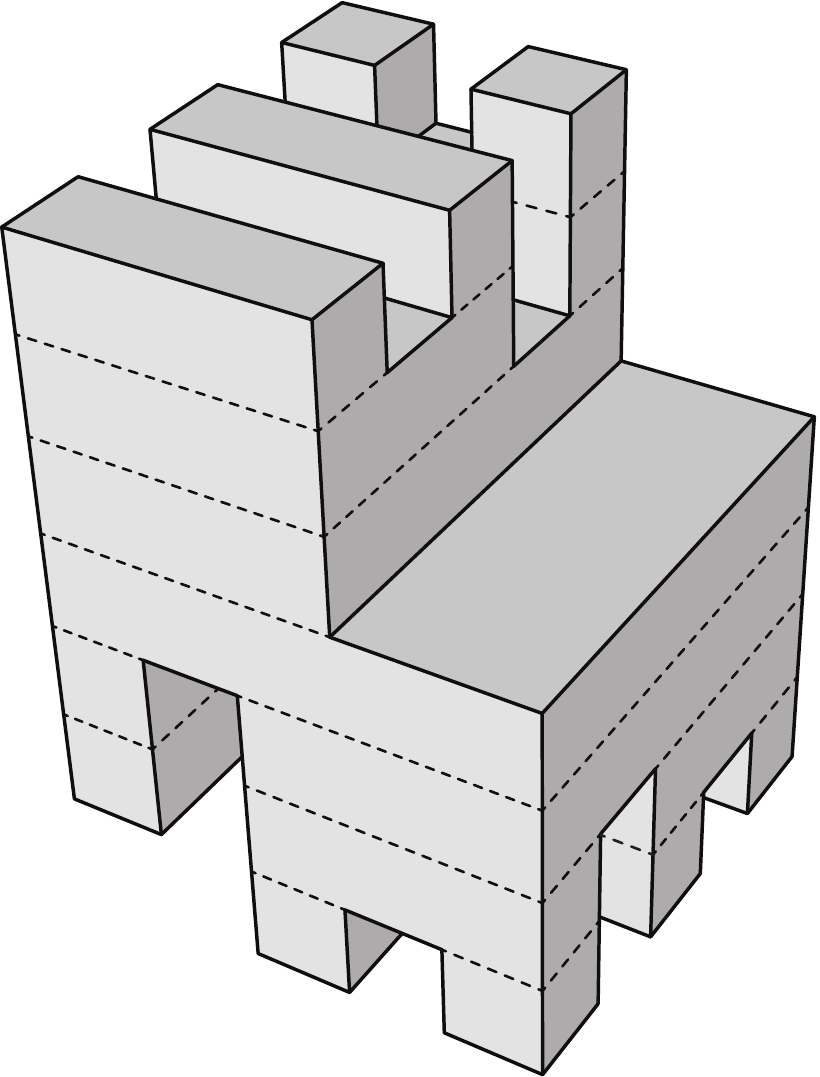}
\caption{Double castle.}
\label{fig4:dcastle}
\end{figure}

Castles and double castles will play a fundamental role in the next section.

\section{Bounds in terms of $r$}\label{sec4:2}
\markthissection{4.2. \ BOUNDS IN TERMS OF $r$}

Recall that orthogonal polygons with $n$ vertices and $h$ holes are guardable by $$\left\lfloor \frac{n+2h}{4} \right\rfloor$$ vertex guards, as established by O'Rourke (see~\cite[Theorem~5.1]{art}). Because $n=2r-4h+4$, where $r$ is the number of reflex vertices in the orthogonal polyhgon (by a straightforward induction on $h$), we can express the same upper bound in terms of $r$, as $$\left\lfloor \frac{r-h}{2} \right\rfloor+1.$$ Even though O'Rourke does not mention this aspect, a careful analysis of his method (a decomposition into L-shaped pieces, see~\cite[Sections~2.5,~2.6]{art}) reveals that, if $r>0$, all the guards can be chosen to lie on \emph{reflex} vertices.

Naturally, the above extends to orthogonal \emph{prisms} of arbitrary \emph{genus} (as opposed to orthogonal polygons with holes), and reflex \emph{edge} guards (as opposed to reflex vertex guards).

As prevoiusly stated, in this section we are going to further generalize this result to 2-reflex orthogonal polyhedra, showing that the exact same upper bound in terms of $r$ holds. It turns out that the central part of O'Rourke's main argument can be considerably simplified, and then rephrased and generalized in terms of polyhedra (\cite[Lemmas~2.13--2.15]{art} are condensed in our Lemma~\ref{l4:stack}), while the other parts of the proof require more sophisticated constructions and some novel ideas.

Eventually, everything boils down to guarding castles and double castles, so we will resolve these first, and then use them as building blocks to prove our full theorem.

We say that an orthogonal polyhedron is \emph{monotone} if it is a prism and if its intersection with any vertical line is either empty or a single line segment. In other terms, a monotone orthogonal polyhedron is constructed by extruding a monotone orthogonal polygon (refer to~\cite{art}).

\begin{lemma}\label{l4:mono}
Any open (resp.\ closed) monotone orthogonal polyhedron with $r>0$ reflex edges is orthogonally guardable by at most
\begin{equation}
\left\lfloor \frac{r}{2} \right\rfloor+1
\label{eq4:aaa}
\end{equation}
open (resp.\ closed) reflex edge guards.
\end{lemma}
\begin{proof}
Without loss of generality, suppose that the reflex edges are $y$-parallel (hence, all the points in each reflex edge have the same $x$ coordinate). Sort all reflex edges by increasing $x$ coordinate (breaking ties arbitrarily), and let $\{e_i\}_{1\leqslant i\leqslant r}$ be the sorted sequence. Now assign a guard to each edge whose index is an odd number, plus a guard to $e_r$ (if $r$ is even). Thus, the bound (\ref{eq4:aaa}) is matched.

To show that the polyhedron is guarded, let $x_i$ be the $x$ coordinate of $e_i$, for $1\leqslant i\leqslant r$. Additionally, let $x_0$ and $x_{r+1}$ be the lowest and the highest $x$ coordinate of a vertex of the polyhedron, respectively. Then, it is straightforward to see that a guard lying on $e_i$ guards at least the points in the polyhedron whose $x$ coordinate lies between $x_{i-1}$ and $x_{i+1}$, as Figure~\ref{fig4:mono} suggests.
\end{proof}

\begin{figure}[h]
\centering
\includegraphics[width=0.65\linewidth]{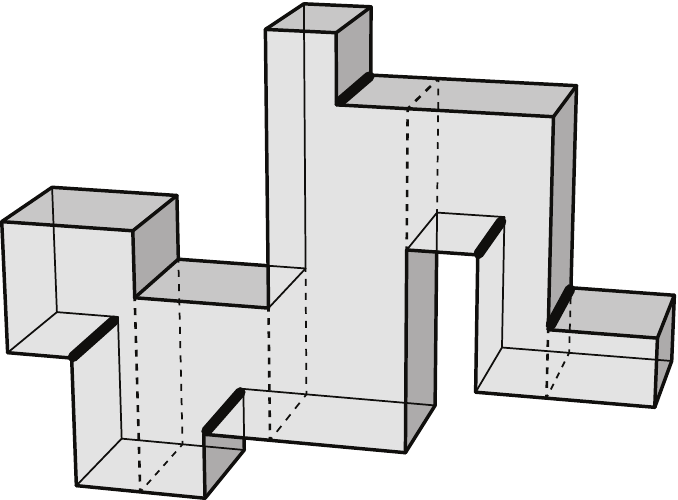}
\caption{Guarding a monotone orthogonal polyhedron. Thick reflex edges represent guards.}
\label{fig4:mono}
\end{figure}

Of course, among monotone orthogonal polyhedra there are all 1-reflex castles. In the next lemma we show that we can actually use one less guard for the remaining castles.

\begin{lemma}\label{l4:castle}
Any open (resp.\ closed) castle with $2r$ reflex edges that is not a prism is orthogonally guardable by at most $r$ open (resp.\ closed) reflex edge guards.
\end{lemma}
\begin{proof}
We prove our claim by well-founded induction on $r$. So, suppose the claim is true for all castles that are not prisms and have fewer than $2r$ reflex edges, and let a castle $\mathcal C$ be given, having exactly $2r$ reflex edges. Assuming that $\mathcal C$ is not a prism, it cannot be a cuboid, and hence $r$ is strictly positive and two castles $\mathcal C_1$ and $\mathcal C_2$ lie on top of $\mathcal C$. Let $e_1$ (resp.\ $e_2$) be the reflex edge bordering the contact rectangle between the base brick of $\mathcal C$ and the base brick of $\mathcal C_1$ (resp.\ $\mathcal C_2$). Let also $2r_1$ and $2r_2$ be the numbers of reflex edges of $\mathcal C_1$ and $\mathcal C_2$, respectively. It follows that $$r=r_1+r_2+1.$$

Three cases arise.
\begin{itemize}
\item If neither $\mathcal C_1$ nor $\mathcal C_2$ is a prism, they both satisfy the inductive hypothesis and can be guarded by $r_1$ and $r_2$ reflex edge guards, respectively. It is straightforward to see that, if guards are placed in this fashion, also the contact rectangles lying on the base brick of $\mathcal C$ are guarded, even if guards and polyhedra are open. Now we just assign a guard to $e_1$ in order to guard the base block of $\mathcal C$, and our upper bound of $r$ guards is matched.

\item If $\mathcal C_1$ is a prism and $\mathcal C_2$ is not (the symmetric case is analogous), then the induction hypothesis applies to $\mathcal C_2$, and we place $r_2$ reflex edge guards accordingly (again, guarding both $\mathcal C_2$ and the contact rectangle shared with the base brick of $\mathcal C$, even if guards and polyhedra are open). Now, because $\mathcal C_1$ is a prism, its reflex edges are all parallel, and two sub-cases arise.
\begin{itemize}
\item If the reflex edges of $\mathcal C_1$ are parallel to $e_1$ (or $\mathcal C_1$ has no reflex edges), then $\mathcal C_1$ and the base brick of $\mathcal C$ together form a monotone orthogonal polyhedron with $2r_1+1$ reflex edges in total. By Lemma~\ref{l4:mono}, such polyhedron can be guarded by $r_1+1$ guards. Along with the previously assigned $r_2$ guards, this yields $r$ guards, as desired.

\item If the reflex edges of $\mathcal C_1$ are not parallel to $e_1$, then they are orthogonal to $e_1$. As a consequence, all of $\mathcal C_1$ is visible to $e_1$, and so is the base brick of $\mathcal C$, as Figure~\ref{fig4:monoguard} illustrates. It follows that $r_2+1$ guards are sufficient in this case, which is less than $r$.
\end{itemize}

\begin{figure}[h]
\centering
\includegraphics[width=0.45\linewidth]{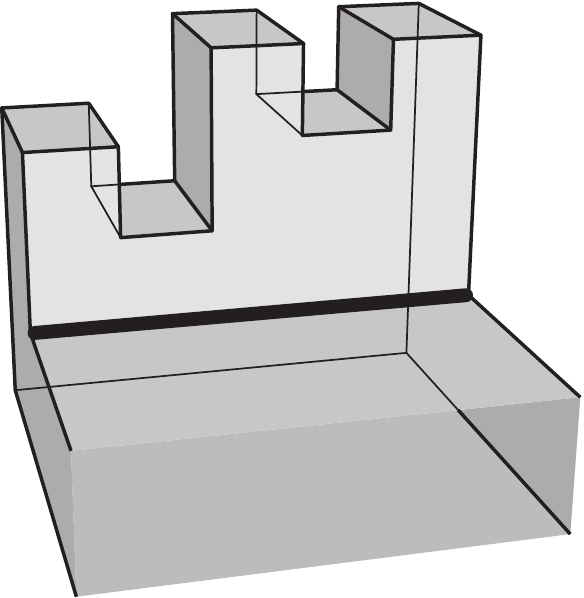}
\Large
\put(-30,100){$e_1$}
\put(-100,120){$\mathcal C_1$}
\caption{Edge $e_1$ orthogonally guards all of $\mathcal C_1$, plus the base brick.}
\label{fig4:monoguard}
\end{figure}

\item If both $\mathcal C_1$ and $\mathcal C_2$ are prisms, at least one of them (say, $\mathcal C_1$) must have reflex edges that are orthogonal to $e_1$, otherwise $\mathcal C$ would be a prism, too. Therefore $r_1\geqslant 1$, and $\mathcal C_1$ is guarded by assigning a guard to $e_1$. Such guard clearly sees also the contact rectangle to which it belongs. Again, two sub-cases arise.
\begin{itemize}
\item If the reflex edges of $\mathcal C_2$ are parallel to $e_2$, then $\mathcal C_2$ and the base brick of $\mathcal C$ form a monotone orthogonal polyhedron with $2r_2+1$ reflex edges, which is guardable by $r_2+1$ reflex edge guards, by Lemma~\ref{l4:mono}. Overall, we assigned $1+r_2+1 \leqslant r_1+r_2+1=r$ guards, as required.

\item If the reflex edges of $\mathcal C_2$ are orthogonal to $e_2$, then $e_2$ guards $\mathcal C_2$, along with the base brick of $\mathcal C$. We assigned only two guards, and $2\leqslant r_1+1\leqslant r$, concluding the proof.
\end{itemize}
\end{itemize}
\end{proof}

The two previous lemmas enable us to prove that our upper bound holds at least for double castles.

\begin{lemma}\label{l4:dcastle}
Any open (resp.\ closed) double castle with $r$ reflex edges is orthogonally guardable by at most
\begin{equation}
\left\lfloor \frac{r}{2} \right\rfloor +1
\label{eq4:ccc}
\end{equation}
open (resp.\ closed) reflex edge guards.
\end{lemma}
\begin{proof}
Let $\mathcal D$ be a double castle with $r$ reflex edges, made of two  castles $\mathcal C_1$ and $\mathcal C_2$ having $2r_1$ and $2r_2$ reflex edges, respectively. Let $e$ be the reflex edge lying on the contact rectangle between the two castles. Because $r=2r_1+2r_2+1$, (\ref{eq4:ccc}) can be rewritten as
\begin{equation}
r_1+r_2+1.
\label{eq4:ccc2}
\end{equation}

We distinguish three cases.
\begin{itemize}
\item If neither $\mathcal C_1$ nor $\mathcal C_2$ is a prism, by Lemma~\ref{l4:castle} they can be guarded by at most $r_1$ and $r_2$ reflex edge guards, respectively. Also the contact rectangle between $\mathcal C_1$ and $\mathcal C_2$ is clearly guarded, even if guards and polyhedra are open, and (\ref{eq4:ccc2}) is then satisfied.

\item If $\mathcal C_1$ is a prism and $\mathcal C_2$ is not (the symmetric case is analogous), by Lemma~\ref{l4:castle} $\mathcal C_2$ can be guarded by at most $r_2$ reflex edge guards. Two sub-cases arise.
\begin{itemize}
\item If $r_1=0$, $\mathcal C_1$ is a cuboid and can be guarded by $e$. In total, $r_2+1$ guards have been assigned, which matches (\ref{eq4:ccc2}).

\item If $r_1>0$, $\mathcal C_1$ can be guarded by $r_1+1$ guards, by Lemma~\ref{l4:mono}. Together with the previous $r_2$ guards, these match the upper bound (\ref{eq4:ccc2}).
\end{itemize}
Again, the contact rectangle between the two castles is clearly guarded.

\item If both $\mathcal C_1$ and $\mathcal C_2$ are prisms, we have three sub-cases.
\begin{itemize}
\item If the reflex edges of $\mathcal C_1$ and $\mathcal C_2$ are parallel to $e$, then $\mathcal D$ is a monotone orthogonal polyhedron with $2r_1+2r_2+1$ reflex edges, and according to Lemma~\ref{l4:mono} it is guardable by $r_1+r_2+1$ reflex edge guards, which agrees with (\ref{eq4:ccc2}). This holds also if $r_1=0$ or $r_2=0$.

\item If the reflex edges of $\mathcal C_1$ are parallel to $e$ and the reflex edges of $\mathcal C_2$ are orthogonal to $e$ (the symmetric case is analogous), then we may assume that $r_2\geqslant 1$. We assign one guard to $e$ in order to guard $\mathcal C_2$, and $r_1+1$ guards to reflex edges in $\mathcal C_1$, in accordance with Lemma~\ref{fig4:mono}. In total, we assigned $r_1+1+1 \leqslant r_1+r_2+1$ guards, thus satisfying (\ref{eq4:ccc2}).

\item Finally, if the reflex edges of both $\mathcal C_1$ and $\mathcal C_2$ are orthogonal to $e$, then a guard assigned to $e$ sees all of $\mathcal D$, which obviously satisfies (\ref{eq4:ccc2}).
\end{itemize}
\end{itemize}
\end{proof}

Now that we know how to guard double castles, we can move on to more general shapes, such as stacks. This is the last intermediate step before our main theorem, and it generalizes~\cite[Lemmas~2.13--2.15]{art}.

Recall that a stack has only primitive contact rectangles. If two neighboring bricks $\mathcal B$ and $\mathcal B'$ share a type-(d) (resp.\ type-(i)) contact rectangle (refer to Figure~\ref{fig4:1}), and $\mathcal B$ lies below (resp.\ above) $\mathcal B'$, then $\mathcal B$ is a \emph{parent} of $\mathcal B'$, and $\mathcal B'$ is a \emph{child} of $\mathcal B$. It follows that a brick in a stack can have at most two children above and two children below, and shares exactly one reflex edge with each child (see Figure~\ref{fig4:primitive}). Moreover, if a brick has one parent above (resp.\ below), then it has no other neighboring bricks above (resp.\ below). Finally, if a brick $\mathcal B$ has exactly one child $\mathcal B'$ on one side (regardless of the number of neighboring bricks on the other side), then the reflex edge shared by $\mathcal B$ and $\mathcal B'$ is said to be \emph{isolated}. It follows that the number of reflex edges in a stack and the number of isolated reflex edges have the same parity.

\begin{figure}[h]
\centering
\subfigure[]{\includegraphics[width=0.25\linewidth]{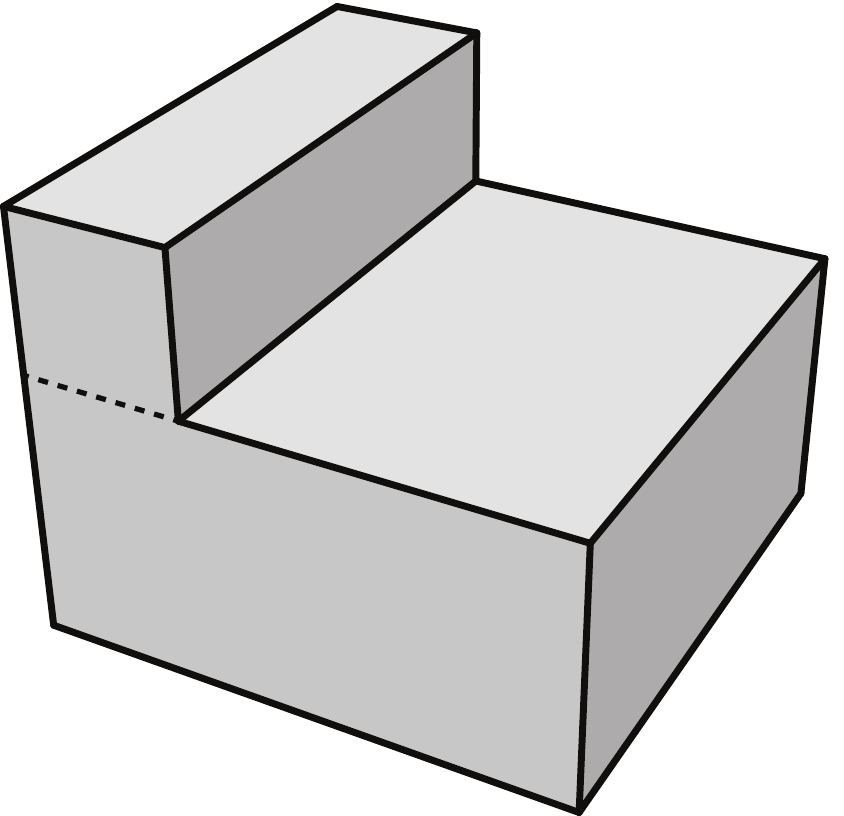}}\qquad
\subfigure[]{\includegraphics[width=0.25\linewidth]{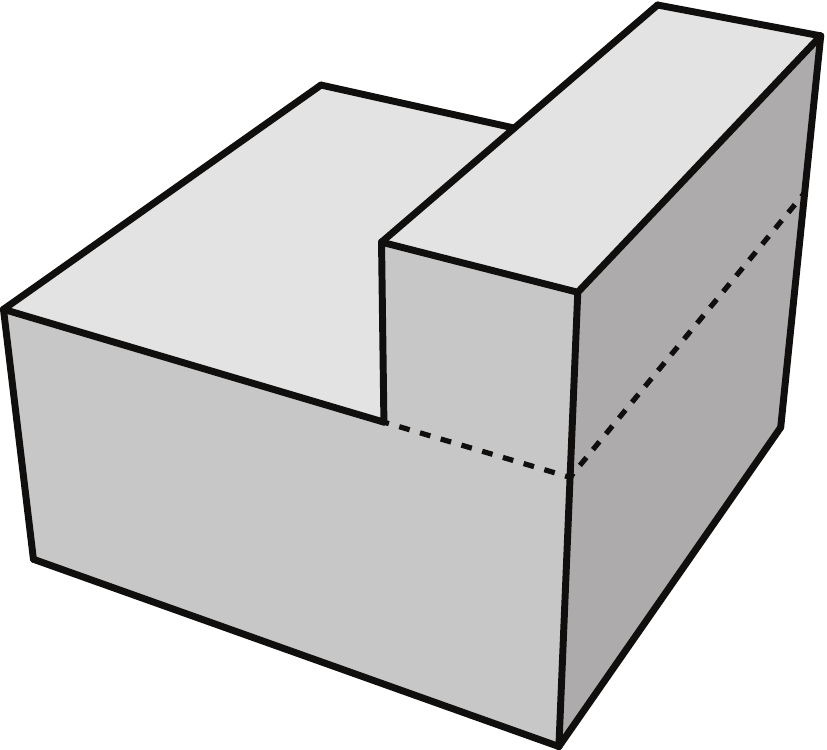}}\qquad
\subfigure[]{\includegraphics[width=0.25\linewidth]{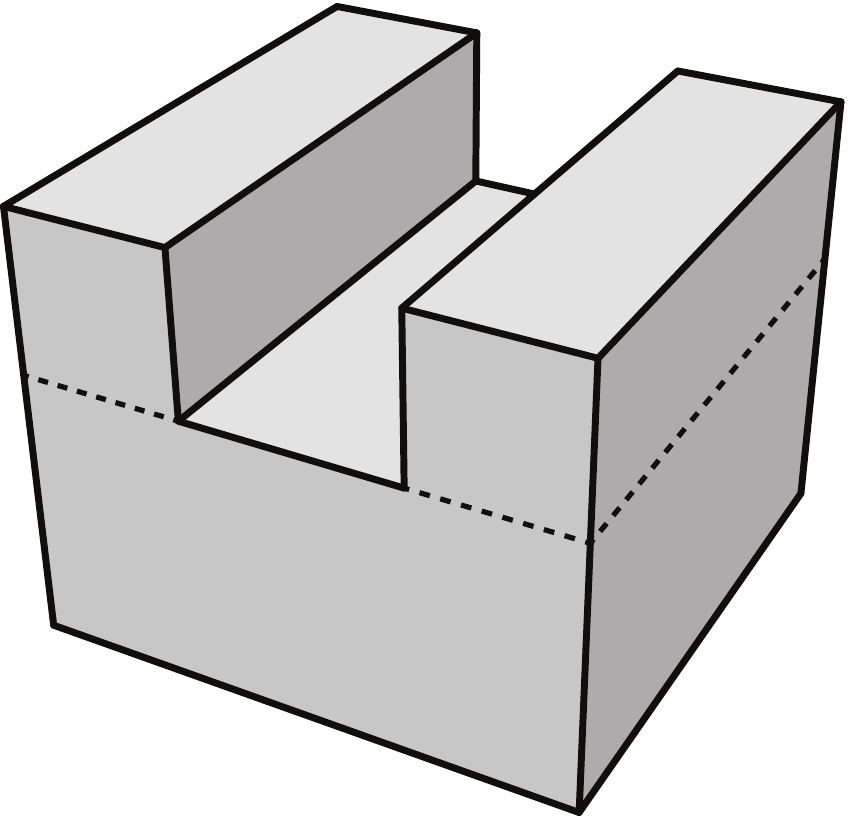}}\\ \vspace{0.25cm}
\subfigure[]{\includegraphics[width=0.25\linewidth]{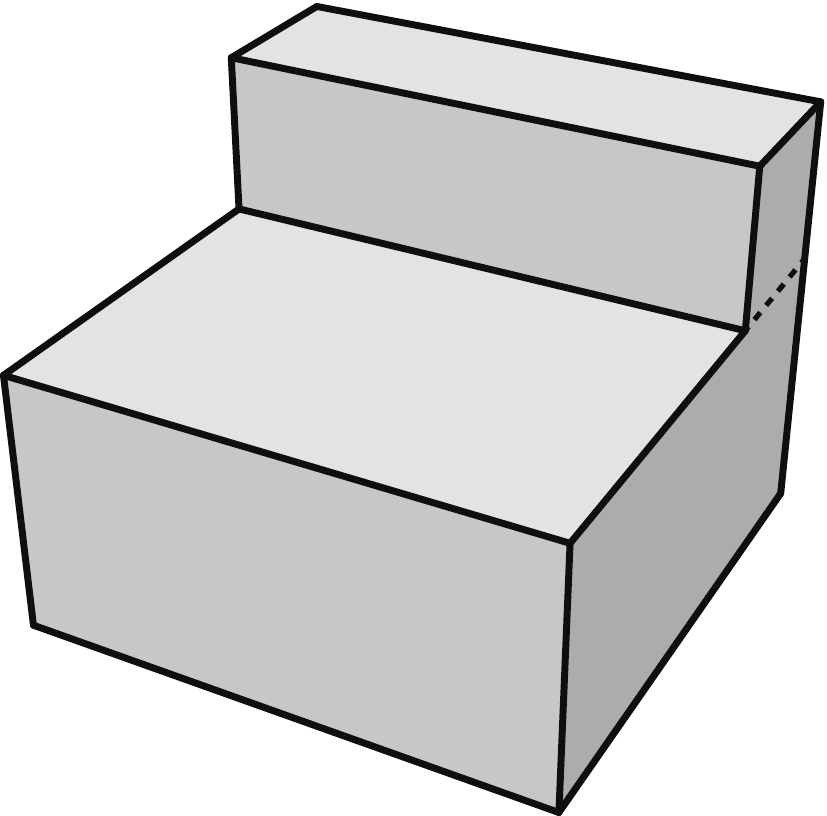}}\qquad
\subfigure[]{\includegraphics[width=0.25\linewidth]{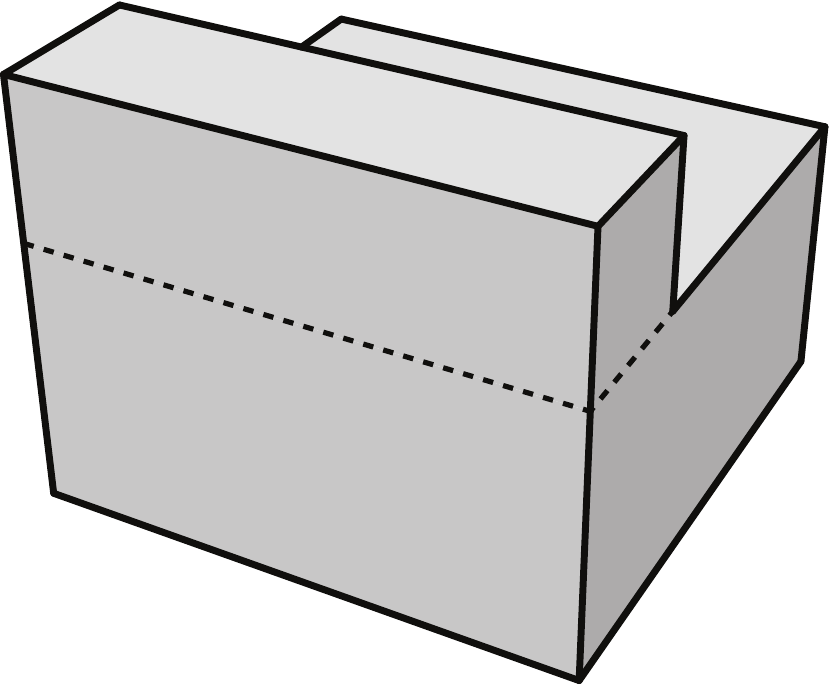}}\qquad
\subfigure[]{\includegraphics[width=0.25\linewidth]{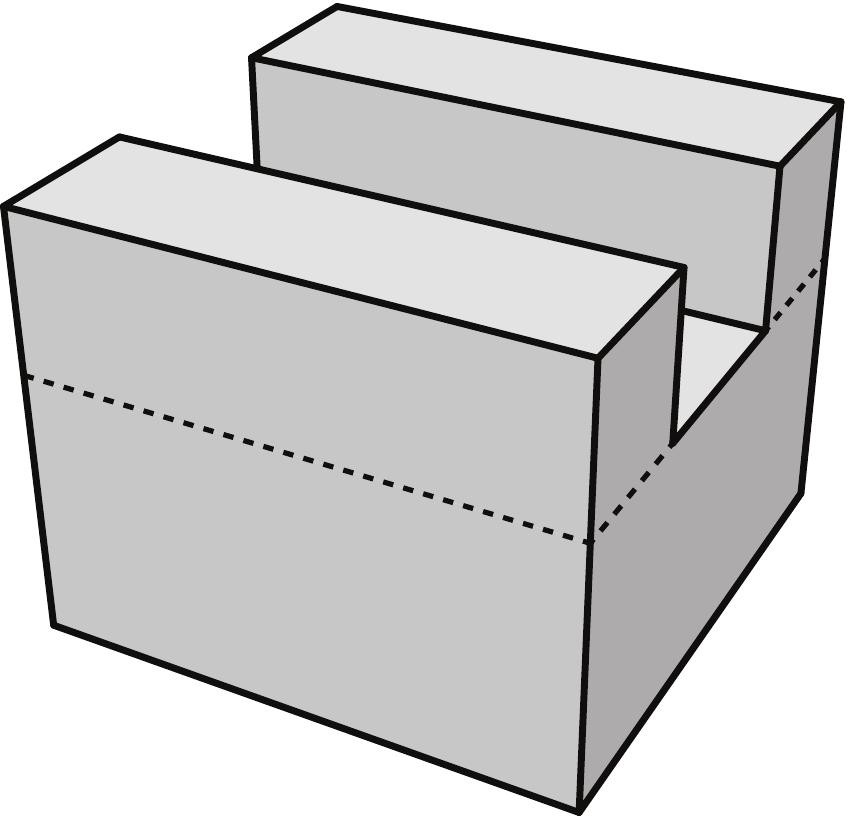}}
\caption{Different configurations for the upper children of a brick in a stack.}
\label{fig4:primitive}
\end{figure}

\begin{lemma}\label{l4:stack}
Any open (resp.\ closed) stack with $r>0$ reflex edges and genus $g$ is orthogonally guardable by at most
\begin{equation}
\left\lfloor \frac{r-g}{2} \right\rfloor+1
\label{eq4:ddd}
\end{equation}
open (resp.\ closed) reflex edge guards.
\end{lemma}
\begin{proof}
Our proof is by well-founded induction on $r$. Following O'Rourke (see~\cite[Section~2.5]{art}), we say that a contact rectangle $R$ in a simply connected stack $\mathcal S$ yields an \emph{odd cut} if $R$ partitions $\mathcal S$ into two stacks $\mathcal S_1$ and $\mathcal S_2$, one of which has an odd number of reflex edges. The presence of odd cuts in a simply connected stack is very desirable, in that it permits to successfully apply the inductive hypothesis.  Indeed, let $2r_1+1$ and $r_2$ be the number of reflex edges of $\mathcal S_1$ and $\mathcal S_2$ (the symmetric case is analogous). Then $r=2r_1+r_2+2$, because the cut resolves exactly one reflex edge. Recall that $\mathcal S$ is simply connected, and so are $\mathcal S_1$ and $\mathcal S_2$. Since $\mathcal S_1$ is non-convex, we can apply the inductive hypothesis on it, guarding it with at most
$$\left\lfloor \frac{2r_1+1}{2} \right\rfloor+1=r_1+1$$
reflex edge guards. Similarly, if $\mathcal S_2$ is non-convex, we can guard it with at most
\begin{equation}
\left\lfloor \frac{r_2}{2} \right\rfloor+1
\label{eq4:ddd2}
\end{equation}
reflex edge guards. On the other hand, if $\mathcal S_2$ is a cuboid, we can orthogonally guard it by assigning a guard to the reflex edge lying on $R$ (see Figure~\ref{fig4:primitive}). In this case, $r_2=0$, hence (\ref{eq4:ddd2}) still holds. As a result, we have guarded all of $\mathcal S$, except perhaps $R$ itself (if $\mathcal S$ is open), with at most
$$r_1+1+\left\lfloor \frac{r_2}{2} \right\rfloor+1 = \left\lfloor \frac{2r_1+r_2+2}{2} \right\rfloor+1=\left\lfloor \frac{r}{2} \right\rfloor+1$$
reflex edge guards. If $\mathcal S_2$ is a cuboid, then we have assigned a guard to the reflex edge lying on $R$, and so $R$ is obviously orthogonally guarded. Otherwise, let $\mathcal B_1$ and $\mathcal B_2$ be the bricks of $\mathcal S$ sharing $R$, with $\mathcal B_1$ child of $\mathcal B_2$. Then, $\mathcal B_1$'s interior is guarded by some guards not coplanar to $R$, which must also see $R$ itself. We stress that it makes sense to talk about odd cuts only for genus-zero stacks.

Now, let $\mathcal S$ be a stack with exactly $r>0$ reflex edges and genus $g$, and assume that the lemma's claim holds for all non-convex stacks with fewer than $r$ reflex edges. Four cases arise.

\begin{itemize}
\item Let $g>0$. Note that each cut along a contact rectangle either disconnects $\mathcal S$ or lowers its genus, and each cut resolves exactly one contact rectangle without creating or modifying other contact rectangles. Because $\mathcal S$ is partitioned by contact rectangles into cuboids, whose genus is zero, it follows that there exists at least one contact rectangle $R$ such that cutting $R$ yields a (degenerate, see Remark~\ref{r4:degen}) stack $\mathcal S'$ with $r'=r-1$ reflex edges and genus $g'=g-1$. Observe that $\mathcal S'$ is non-convex, because it is made of at least two bricks, so we can apply the inductive hypothesis on it and guard it with at most
$$\left\lfloor \frac{r'-g'}{2} \right\rfloor+1=\left\lfloor \frac{r-g}{2} \right\rfloor+1$$
reflex edge guards. $R$ may still be unguarded, if $\mathcal S$ is open. We show that this is not the case, as above: let $\mathcal B_1$ and $\mathcal B_2$ be the two bricks sharing $R$, with $\mathcal B_1$ child of $\mathcal B_2$. Then the interior of $\mathcal B_1$ is guarded in $\mathcal S'$ by some guards that are not coplanar with $R$, and hence $R$ is guarded in $\mathcal S$ by the same guards.

\item If $g=0$ and $r>0$ is even, then any contact rectangle yields an odd cut. Indeed, cutting $\mathcal S$ along a contact rectangle resolves exactly one reflex edge and leaves with two stacks (because $\mathcal S$ is simply connected). The reflex edge numbers in these two stacks must have opposite parity, because their sum must be odd. Hence, one of them is odd.

\item Let $g=0$, and let $\mathcal S$ have a brick $\mathcal B$ with exactly one neighboring brick above and exactly one below (Figure~\ref{fig4:oneone} sketches one possible configuration for $\mathcal B$). We show that one of the two contact rectangles bordering $\mathcal B$ yields an odd cut. Let $R_1$ be the upper contact rectangle and $R_2$ be the lower one. If $R_1$ does not yield an odd cut, the stack above $R_1$ has an even number of reflex edges. But then, the stack above $R_2$, which additionally includes $\mathcal B$ and the reflex edge belonging to $R_1$, has an odd number of reflex edges. It follows that either $R_1$ or $R_2$ yields an odd cut.

\begin{figure}[h]
\centering
\includegraphics[width=0.3\linewidth]{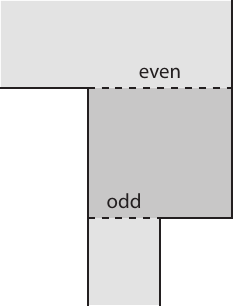}
\large
\put(-48,80){$\mathcal B$}
\put(-0,114){$R_1$}
\put(-100,42){$R_2$}
\caption{Sketch of a brick with one neighbor above and one below.}
\label{fig4:oneone}
\end{figure}

\item If none of the above is satisfied, we show that $\mathcal S$ must be a double castle, hence our claim holds by Lemma~\ref{l4:dcastle}. Let then $r$ be odd, so that $\mathcal S$ has at least one isolated reflex edge corresponding to a contact rectangle $R$. Additionally, let no brick of $\mathcal S$ have exactly one neighbor above and one neighbor below. We show that $\mathcal S'$, i.e., the stack above $R$, is a castle. (By a symmetric argument, the stack below $R$ will be an upside-down castle, and $\mathcal S$ will then be a double castle.) This is straightforward to see (build the castle from bottom to top, without putting bricks that have only one neighbor below and one above), however here follows a formal proof.

Let $d(\mathcal B)$ be the minimum number of bricks one has to traverse to reach brick $\mathcal B$ from $R$ (while always staying inside $\mathcal S'$), and let $\mathcal S'_h$ be the set of bricks $\mathcal B$ in $\mathcal S'$ such that $d(\mathcal B)\leqslant h$. We prove by induction on $h$ that $\mathcal S'_h$ is a castle whose base brick contains $R$.

The claim is true for $h=0$, because the brick just above $R$ (let it be $\widetilde{\mathcal B}$) is the only one in contact with $R$, and a brick is indeed a castle.

Observe that no brick in $\mathcal S'$ can be attached to the bottom face of $\widetilde{\mathcal B}$, because the reflex edge of $\mathcal S$ corresponding to $R$ is isolated. Hence, $\widetilde{\mathcal B}$ must have either zero or two neighbors in $\mathcal S'$, and both of them are above. It follows that $\mathcal S'_1$ is a castle, as well.

Let now $h\geqslant 1$ and let $\mathcal S'_h$ be a castle whose base brick $\widetilde{\mathcal B}$ contains $R$, and let us show that the same holds for $\mathcal S'_{h+1}$. Any brick $\mathcal B$ such that $d(\mathcal B)=h+1$ must be attached to some brick $\mathcal B'$ of $\mathcal S'_h$ such that $d(\mathcal B')=h$ (see Figure~\ref{fig4:stackcastle}). It is straightforward to see that any such $\mathcal B'$ has one parent brick below and no bricks above, in $\mathcal S'_h$. Hence, $\mathcal B$ can only be attached on top of $\mathcal B'$. Because $\mathcal B'$ cannot have only one top and one bottom neighbor in $\mathcal S'$, it follows that $\mathcal B$ cannot be a parent of $\mathcal B'$, but a child. Any other brick $\mathcal B''$ that is attached on top of $\mathcal B'$ in $\mathcal S'$ must also belong to $\mathcal S'_{h+1}$, by definition. As a consequence, one such brick $\mathcal B''$ must indeed be in $\mathcal S'_{h+1}$, otherwise $\mathcal B'$ would once again have only one top and one bottom neighbor in $\mathcal S'$. No new brick is attached to $\widetilde{\mathcal B}$, which stays the base of $\mathcal S'_{h+1}$. Thus the base brick of $\mathcal S'_{h+1}$ contains $R$, proving our claim and concluding the inductive proof.

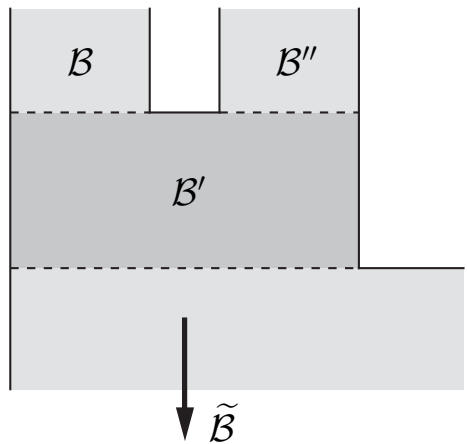
\begin{figure}[h]
\centering
\includegraphics[width=0.4\linewidth]{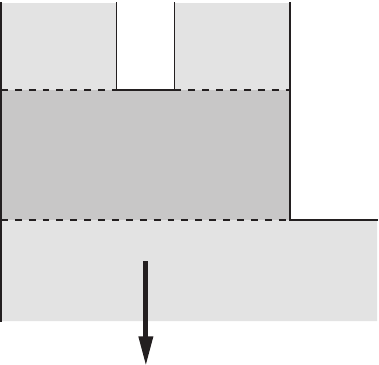}
\large
\put(-149,138){$\mathcal B$}
\put(-110,89){$\mathcal B'$}
\put(-71,138){$\mathcal B''$}
\put(-96,0){$\widetilde{\mathcal B}$}
\caption{Sketch of a brick in a castle.}
\label{fig4:stackcastle}
\end{figure}

Since $\mathcal S'$ is connected and contains only a finite number of bricks, for large-enough $h$ we have $\mathcal S'=\mathcal S'_h$, implying that $\mathcal S'$ is a castle.
\end{itemize}
\end{proof}

\begin{remark}
In the next theorem, we give an upper bound on the number of guards in terms of $r$, $g$ and $b$, where $b$ is the number of collars. The presence of $b$ may look awkward, but we will actually have to carry this parameter along to the next section, in order to prove Theorem~\ref{t4:2orthoe}. Regardless, the casual reader may happily ignore this term (as it just contributes to lowering the bound), and read (\ref{eq4:main}) simply as
$$\left\lfloor \frac{r-g}{2} \right\rfloor +1.$$
\end{remark}

\begin{theorem}\label{t4:2orthor}
Any open (resp.\ closed) 2-reflex orthogonal polyhedron with $r>0$ reflex edges, $b$ collars and genus $g$ is guardable by at most
\begin{equation}
\left\lfloor \frac{r-g}{2} \right\rfloor -b +1
\label{eq4:main}
\end{equation}
open (resp.\ closed) reflex edge guards.
\end{theorem}
\begin{proof}
Preliminarily, we eliminate all contact rectangles of types~(k) to~(t) (refer to Figure~\ref{fig4:1}), by transforming them into other types of contacts, without altering the genus of the polyhedron or the amount of collars and reflex edges.

Let $\mathcal P^\circ$ be the interior of a given 2-reflex polyhedron $\mathcal P$ (if $\mathcal P$ is open, $\mathcal P^\circ = \mathcal P$), and let $\alpha$ be a horizontal plane containing at least a contact rectangle of $\mathcal P$. We translate upward by $d>0$ every point of $\mathcal P^\circ$ that lies strictly above $\alpha$. As a result, the faces lying on $\alpha$ and pointing downward are translated upward, while those pointing upward do not move. Then we extrude $\mathcal P^\circ \cap \alpha$ upward by $d$, hence closing the ``gaps'', and obtaining a new open polyhedron. If $\mathcal P$ is open, then let $\mathcal P'$ be this new polyhedron. If $\mathcal P$ is closed, then let $\mathcal P'$ be the closure of the new polyhedron.

Figure~\ref{fig4:stretch} shows how this transformation acts on a contact rectangle of type~(k). Observe that the number of reflex edges is preserved.

\begin{figure}[h]
\centering
\includegraphics[width=0.95\linewidth]{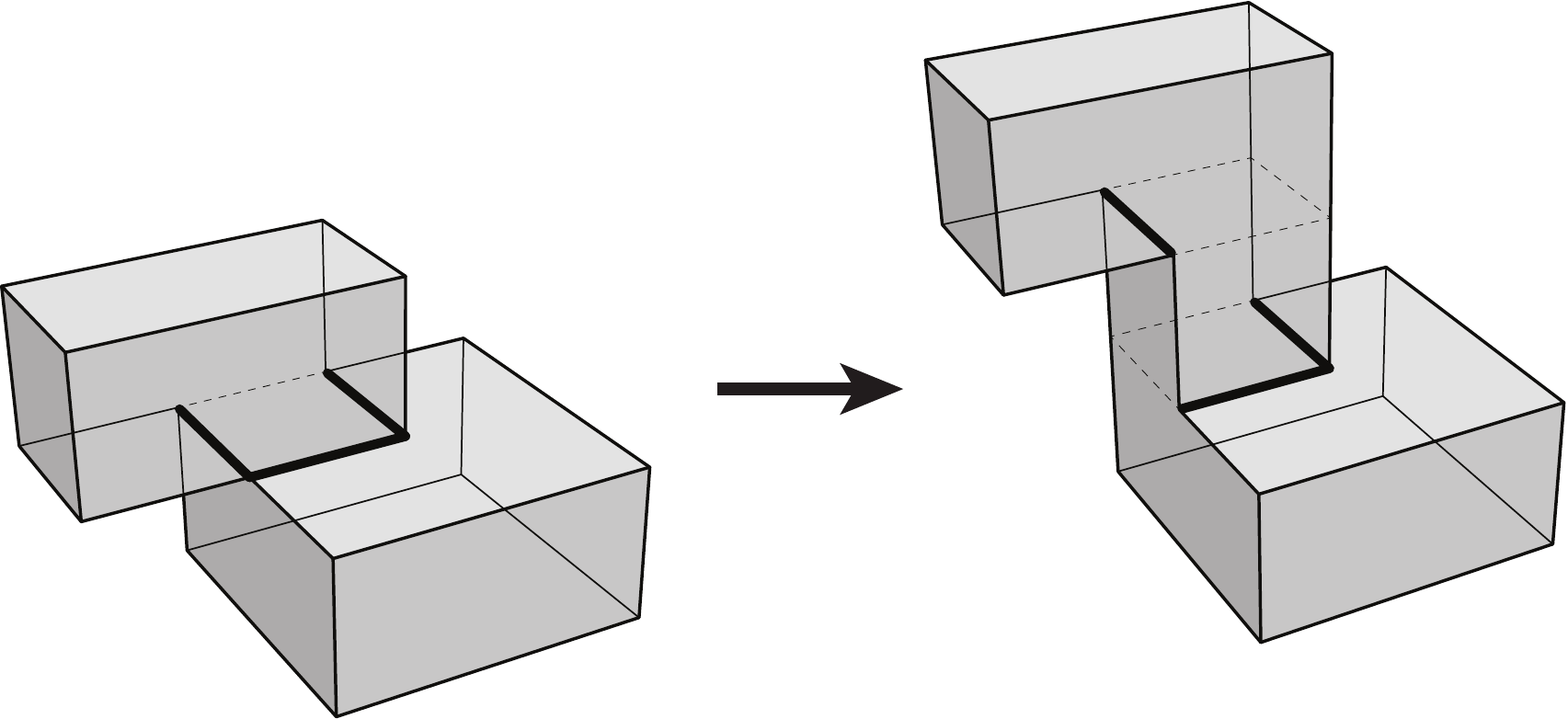}
\caption{Transforming a contact rectangle of type~(k) into one of type~(i) (above) and one of type~(c) (below), while preserving the number of reflex edges (thick lines).}
\label{fig4:stretch}
\end{figure}

In general, contact rectangles of type~(a) to~(e) that lie on $\alpha$ are not altered by our transformation, whereas those of type~(f) to~(j) are merely translated upward by $d$, but do not change their type. If follows that no collar is lost in the process (recall that collars are contacts of type (a) and (f)).

On the other hand, contacts of type~(k) to~(t) lying on $\alpha$ are eliminated, as detailed below.

\begin{itemize}
\item Each contact of type (k) is transformed into two contacts of type (i) and (c).
\item Each contact of type (l) is transformed into two contacts of type (i) and (b).
\item Each contact of type (m) is transformed into two contacts of type (i) and (e).
\item Each contact of type (n) is transformed into two contacts of type (h) and (c).
\item Each contact of type (o) is transformed into two contacts of type (j) and (e).
\item Each contact of type (p) is transformed into two contacts of type (h) and (d).
\item Each contact of type (q) is transformed into two contacts of type (g) and (d).
\item Each contact of type (r) is transformed into two contacts of type (j) and (d).
\item Each contact of type (s) is transformed into two contacts of type (i) and (d).
\item Each contact of type (t) is transformed into two contacts of type (i) and (d).
\end{itemize}

Observe that, if $\mathcal P'$ is guarded by reflex edge guards, then applying the ``reverse transformation'' to such guards yields a guarding set for $\mathcal P$ of equal size. In particular, referring to Figure~\ref{fig4:stretch}, if the two contact rectangles on the right-hand side are guarded, it implies that also the contact rectangle on the left-hand side must be guarded, no matter if $\mathcal P$ and $\mathcal P'$ are open or closed.

We can repeat the same procedure for different $\alpha$'s, until no contacts of type~(k) to~(t) are left. It will then suffice to prove our theorem for polyhedra featuring only contacts of type~(a) to~(j).

We proceed by induction on the number of non-primitive contact rectangles. The base case is given by non-convex stacks, for which (\ref{eq4:main}) holds due to Lemma~\ref{l4:stack} and the fact that $b=0$.

For the inductive step, let $\mathcal P$ be a 2-reflex orthogonal polyhedron with $r>0$ reflex edges, $b$ collars, genus $g$, and a non-primitive contact rectangle $R$. We cut $\mathcal P$ along $R$, thus resolving one non-primitive contact rectangle, and we distinguish two cases.
\begin{itemize}
\item If the cut does not disconnect $\mathcal P$, then it lowers its genus by 1. Let $\mathcal P'$ be the resulting polyhedron (which is degenerate, see Remark~\ref{r4:degen}), and let $g'=g-1$ be its genus. By inductive hypothesis, $\mathcal P'$ is guardable by
\begin{equation}
\left\lfloor \frac{r'-g'}{2} \right\rfloor -b'+1
\label{eq4:main2}
\end{equation}
guards, where $r'$ and $b'$ are, respectively, the number of reflex edges and collars of $\mathcal P'$.  Two sub-cases arise.
\begin{itemize}
\item If $R$ is a collar, then $r'=r-4$ and $b'=b-1$. By plugging these values into (\ref{eq4:main2}), we obtain that $\mathcal P\setminus R$ is guardable by
$$\left\lfloor \frac{r-4-g+1}{2} \right\rfloor -b+1+1 \leqslant \left\lfloor \frac{r-g}{2} \right\rfloor -b+1$$
reflex edge guards.

\item If $R$ is not a collar, then $b'=b$. Because $R$ is not primitive, $r'\leqslant r-2$ (refer to Figure~\ref{fig4:1}). Hence, $\mathcal P\setminus R$ is guardable by at most
$$\left\lfloor \frac{r-2-g+1}{2} \right\rfloor -b+1 \leqslant \left\lfloor \frac{r-g}{2} \right\rfloor -b+1$$
reflex edge guards.
\end{itemize}

\item If the cut disconnects $\mathcal P$ into $\mathcal P_1$ and $\mathcal P_2$, then $g_1+g_2=g$, where $g_1$ (resp.\ $g_2$) is the genus of $\mathcal P_1$ (resp.\ $\mathcal P_2$).  Let $r_1$ and $b_1$ (resp.\ $r_2$ and $b_2$) be, respectively,  the number of reflex edges and collars of $\mathcal P_1$ (resp.\ $\mathcal P_2$). Two sub-cases arise.

\begin{itemize}
\item If $R$ is a collar, then $r_1+r_2=r-4$ and $b_1+b_2=b-1$. If $r_1>0$, then we can apply the inductive hypothesis on $\mathcal P_1$ and guard it with at most
\begin{equation}
\left\lfloor \frac{r_1-g_1}{2} \right\rfloor -b_1+1
\label{eq4:main3}
\end{equation}
guards. Otherwise, $\mathcal P_1$ is a cuboid, and we can guard it by assigning a guard to any edge of $R$ (they are all reflex). In this case, $r_1=g_1=b_1=0$, and (\ref{eq4:main3}) is still satisfied. We proceed similarly with $\mathcal P_2$, and thus we have assigned a combined number of guards that is at most
$$\left\lfloor \frac{r_1-g_1}{2} \right\rfloor + \left\lfloor \frac{r_2-g_2}{2} \right\rfloor -b_1-b_2+2 \leqslant \left\lfloor \frac{r-4-g}{2} \right\rfloor -b+3 = \left\lfloor \frac{r-g}{2} \right\rfloor -b+1.$$
This many guards are then sufficient to guard $\mathcal P\setminus R$.

\item If $R$ is not a collar, then $b_1+b_2=b$. Because $R$ is not primitive, $r_1+r_2\leqslant r-2$ (refer to Figure~\ref{fig4:1}). Once again, if $r_1>0$, then we can apply the inductive hypothesis on $\mathcal P_1$ and guard it with a number of guards that is bounded by (\ref{eq4:main3}). Otherwise, $\mathcal P_1$ is a cuboid, and we can guard it by assigning a guard to any reflex edge of $R$ (there is at least one). In this case, $r_1=g_1=b_1=0$, and (\ref{eq4:main3}) is still satisfied. Again, we do the same with $\mathcal P_2$, and thus we have assigned a combined number of guards that is at most
$$\left\lfloor \frac{r_1-g_1}{2} \right\rfloor + \left\lfloor \frac{r_2-g_2}{2} \right\rfloor -b_1-b_2+2 \leqslant \left\lfloor \frac{r-2-g}{2} \right\rfloor -b+2 = \left\lfloor \frac{r-g}{2} \right\rfloor -b+1,$$
and $\mathcal P\setminus R$ is guarded.
\end{itemize}
\end{itemize}
It remains to show that, in all the cases above, also $R$ is guarded. For closed polyhedra, this is trivial, because $\mathcal P'$, $\mathcal P_1$ and $\mathcal P_2$ (as defined above) all contain $R$.

For open polyhedra, the claim follows from the fact that we eliminated all contact rectangles of type~(k) to~(t) beforehand. Indeed, if $R$ is of type~(a) to~(j), then it separates two bricks $\mathcal B_1$ and $\mathcal B_2$, such that the vertical projection of $\mathcal B_1$ is entirely contained in the vertical projection of $\mathcal B_2$ (refer to Figure~\ref{fig4:1}). Therefore, the guards that guard the interior of $\mathcal B_1$ in $\mathcal P'$, $\mathcal P_1$, or $\mathcal P_2$ either lie on the boundary of $R$ (if $\mathcal P_1$ or $\mathcal P_2$ is a cuboid), or do not lie on the same plane or $R$. In either case, $R$ is guarded in $\mathcal P$, even if $\mathcal P$ is open.
\end{proof}

\begin{observation}\label{o4:affine}
Theorem~\ref{t4:2orthor} holds more generally for 2-reflex \emph{3-oriented} polyhedra. Indeed, each 3-oriented polyhedron is mapped in an orthogonal polyhedron by a suitable affine transformation, which preserves visibility.
\end{observation}

Note that we cannot say ``orthogonally guardable'' in the statement of Theorem~\ref{t4:2orthor}, due to the following.

\begin{observation}\label{obs4:noreflex}
Not every 2-reflex orthogonal polyhedron can be orthogonally guarded by reflex edge guards, as Figure~\ref{fig4:noreflex} shows. Indeed, even if guards are placed on both reflex edges, a large cuboid remains orthogonally unguarded.
\end{observation}

\begin{figure}[h]
\centering
\includegraphics[width=0.34\linewidth]{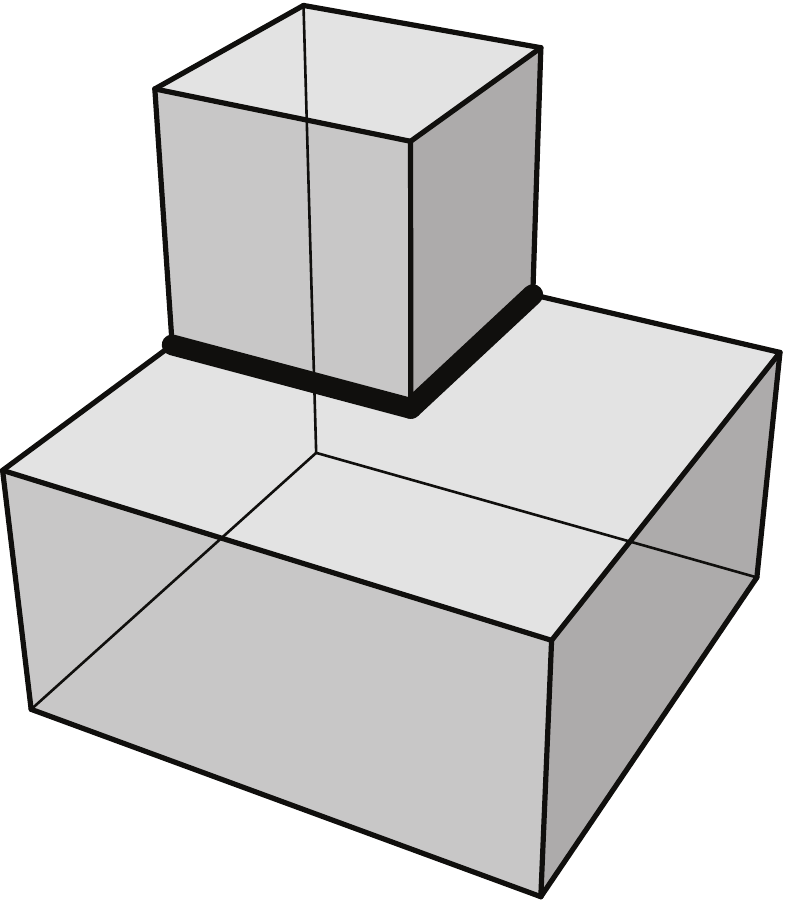}
\caption{2-reflex orthogonal polyhedron that is not orthogonally guardable by reflex edge guards.}
\label{fig4:noreflex}
\end{figure}

For $g=0$, the upper bound given in Theorem~\ref{t4:2orthor} is tight, as Figure~\ref{fig4:3} implies (cf.\ Observation~\ref{o3:rourke}). In contrast, for arbitrary genus, the situation is still unclear, and even the analogous 2-dimensional problem is open (i.e., optimally vertex-guarding orthogonal polygons with holes, see~\cite{art,urrutia2000}).

\begin{figure}[h]
\centering
\includegraphics[width=0.85\linewidth]{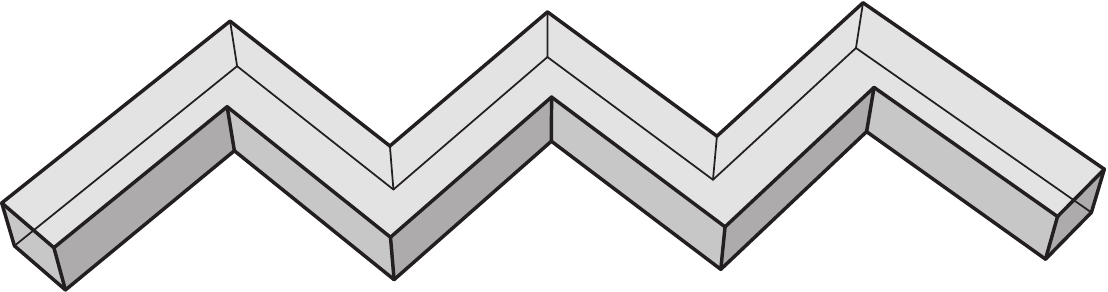}
\caption{2-reflex orthogonal polyhedron needing $\lfloor r/2 \rfloor +1$ reflex edge guards.}
\label{fig4:3}
\end{figure}

We believe that the bound given in Theorem~\ref{t4:2orthor} holds also for general (i.e., 3-reflex) orthogonal polyhedra.

\begin{conjecture}\label{con4:1}
Any open (resp.\ closed) orthogonal polyhedron with $r>0$ reflex edges and genus $g$ is guardable by at most
$$\left\lfloor \frac {r-g} 2\right\rfloor +1$$
open (resp.\ closed) reflex edge guards.
\end{conjecture}

\section{Bounds in terms of $e$}
\markthissection{4.3. \ BOUNDS IN TERMS OF $e$}

In this section we bound the number of reflex edge guards required to guard a 2-reflex orthogonal polyhedron in terms of $e$, as opposed to $r$. Rather than providing a radically new upper bound construction, we bound $e$ in terms of $r$, and then apply Theorem~\ref{t4:2orthor}.

Note that a straightforward application of this method would not improve on Urrutia's upper bound of $\lfloor e/6\rfloor$ given in Theorem~\ref{t:urrutia}. Indeed, the sharpest possible inequality between $e$ and $r$ (involving also the genus $g$) is
\begin{equation}
e\geqslant 3r-12g+12,
\label{eq4:weakineq}
\end{equation}
which still yields an upper bound of
$$\frac{e}{6}+O(g)$$
guards when applied to Theorem~\ref{t4:2orthor}.

To get around this, we will refine (\ref{eq4:weakineq}) by introducing the number of collars as an additional parameter.

\begin{lemma}\label{t4:ineq}
In every 2-reflex orthogonal polyhedron, the number of edges $e$, the number of reflex edges $r$, the number of collars $b$ and the genus $g$ satisfy the inequality
$$e\geqslant 4r-12g-4b+12.$$
\end{lemma}
\begin{proof}
We will prove that, for any collection of $k$ (mutually internally disjoint) 2-reflex orthogonal polyhedra,
\begin{equation}
e\geqslant 4r-12g-4b+12k
\label{eq4:ineq}
\end{equation}
holds. Here, $e$ (resp.\ $r$, $g$, $b$) is the sum of the edges (resp.\ reflex edges, genera, collars) of the $k$ polyhedra. Then, by plugging $k=1$, we will obtain our claim.

Our proof proceeds by induction on $r$. For $r=0$ we have a collection of $k$ cuboids, each of which has 12 edges, so $e=12k$, $g=b=0$, and (\ref{eq4:ineq}) holds as desired.

If $r>0$, there is at least one (horizontal) reflex edge, which is a side of the contact rectangle of two adjacent bricks, both belonging to the same polyhedron $\mathcal P$ of the collection. We can resolve this reflex edge (and up to three others) by separating the two bricks with a horizontal cut along the contact rectangle. As a consequence, either $\mathcal P$ gets partitioned in two polyhedra (in which case the new number of polyhedra is $k'=k+1$), or the genus of $\mathcal P$ decreases by 1 (in which case the new total genus is $g'=g-1$). Either way,
\begin{equation}
k'-g' = k-g+1.
\label{eq4:ineq4}
\end{equation}
By inductive hypothesis,
\begin{equation}
e'\geqslant 4r'-12g'-4b'+12k',
\label{eq4:ineq3}
\end{equation}
where $e'$ (resp.\ $r'$, $b'$) is the new number of edges (resp.\ reflex edges, collars), after the cut.

Two cases arise. If the cut is along a collar, then $b'=b-1$, $e'=e$, and $r'=r-4$ (see Figure~\ref{fig4:1}). Plugging these into (\ref{eq4:ineq3}) and combining the result with (\ref{eq4:ineq4}) immediately yields (\ref{eq4:ineq}), as claimed.

Otherwise (i.e., the cut is not along a collar), by inspection of Figure~\ref{fig4:1}, it is clear that
\begin{equation}
e-e'\geqslant 4(r-r')-12.
\label{eq4:ineq5}
\end{equation}
(Recall that cases (a) and (f) must be ignored, because they correspond to a collar.) By combining (\ref{eq4:ineq4}), (\ref{eq4:ineq3}), (\ref{eq4:ineq5}) and plugging $b'=b$, we obtain again (\ref{eq4:ineq}), concluding the proof.
\end{proof}

\begin{theorem}\label{t4:2orthoe}
Any open (resp.\ closed) non-convex 2-reflex orthogonal polyhedron with $e$ edges and genus $g$ is guardable by at most
\begin{equation}
\left\lfloor \frac{e-4}{8} \right\rfloor +g
\label{e4:bound2}
\end{equation}
open (resp.\ closed) reflex edge guards.
\end{theorem}
\begin{proof}
Let $r>0$ be the number of reflex edges in the polyhedron. By Lemma~\ref{t4:ineq},
\begin{equation}
r\leqslant \frac e4+3g+b-3,
\label{eq4:ineq2}
\end{equation}
where $b$ is the number of collars. Applying (\ref{eq4:ineq2}) to Theorem~\ref{t4:2orthor}, we obtain that the number of guards is bounded by
$$\left\lfloor \frac {r-g}2 \right\rfloor -b+1 \leqslant \left\lfloor \frac e8 +\frac 32 g+\frac b2-\frac 32 - \frac g2 \right\rfloor -b+1 = \left\lfloor \frac{e-4}{8}-\frac b2 \right\rfloor +g\leqslant \left\lfloor \frac{e-4}{8} \right\rfloor +g.$$
\end{proof}

\begin{observation}
Once again, Theorem~\ref{t4:2orthor} holds more generally for 2-reflex \emph{3-oriented} polyhedra (cf.\ Observation~\ref{o4:affine}).
\end{observation}

We remark that our Theorem~\ref{t4:2orthoe} is an improvement upon the previous state of the art, in that it lowers the upper bound provided by Urrutia's Theorem~\ref{t:urrutia} by roughly 25\% (for $g=0$), and also shows how guards can be chosen to lie on reflex edges, rather than on arbitrary edges.

We are unable to raise the lower bound of $\lfloor e/12\rfloor+1$ guards given in Observation~\ref{obs3:urrutia}, even if just reflex edge guards are to be employed. On the other hand, our upper bound in terms of $e$ seems far from optimal: indeed, we could treat separately other contact configurations between bricks, rather than just collars, and this could be enough to improve Lemma~\ref{t4:ineq} and lower the bound in (\ref{e4:bound2}). Therefore, we formulate a stronger version of Urrutia's Conjecture~\ref{con:urrutia} (cf.~Conjecture~\ref{con4:1}, Observation~\ref{obs4:noreflex}).

\begin{conjecture}\label{con4:2}
Any non-convex orthogonal polyhedron with $e$ edges is guardable by at most
$$\frac e {12}+O(1)$$
\emph{reflex} edge guards.
\end{conjecture}

We already know this claim holds for simply connected prisms, and Lemma~\ref{l4:stack} implies that it holds more generally for simply connected stacks. Specifically, we have the following.

\begin{theorem}\label{t4:stackse}
Any open (resp.\ closed) non-convex stack with $e$ edges and genus $g$ is guardable by at most
\begin{equation}
\left\lfloor \frac{e}{12}+\frac{g}{2} \right\rfloor
\label{eq4:stackse}
\end{equation}
open (resp.\ closed) reflex edge guards.
\end{theorem}
\begin{proof}
A straightforward induction on the number of reflex edges $r$, based on Figure~\ref{fig4:1}, reveals that
$$e = 6r-12g+12.$$
(Recall that stacks only have type-(d) and type-(i) contact rectangles.) Solving for $r$ and substituting in (\ref{eq4:ddd}) immediately yields (\ref{eq4:stackse}).
\end{proof}

\section{Time complexity}

To conclude the chapter, we show how to efficiently compute guard positions matching the upper bounds given in Theorems~\ref{t4:2orthor} and~\ref{t4:2orthoe}. Notice that both bounds refer to the very same construction, described in Section~\ref{sec4:2}. In the present section, we will merely translate such construction into an algorithm that runs in $O(n \log n)$ time (see the discussion on representing polyhedra in Chapter~\ref{chapter2}).

O'Rourke's algorithm for \emph{simple} orthogonal polygons, detailed in~\cite[Section~2.6]{art}, also runs in $O(n \log n)$ time. As observed by Urrutia in~\cite{urrutia2000}, it could be optimized to run in $O(n)$ time, if Chazelle's linear time triangulation algorithm were used (see~\cite{chazelletriang}). Unfortunately, Urrutia's speedup is only applicable to orthogonal polygons without holes.

In principle, we could rephrase O'Rourke's algorithm in terms of 2-reflex orthogonal polyhedra, and obtain a new $O(n \log n)$ algorithm. However, four issues arise that require additional care:
\begin{enumerate}
\item\label{i4:1} O'Rourke's algorithm works on simply connected polygons, while our algorithm should be applied to polyhedra of any genus.

\item\label{i4:2} O'Rourke's algorithm may assign guards to convex vertices, whereas we insist on having guards only on reflex edges.

\item\label{i4:3} O'Rourke's method to find horizontal cuts in polygons does not trivially extend to polyhedra.

\item\label{i4:4} O'Rourke's algorithm relies on guarding double histograms, whereas we need guard double castles.
\end{enumerate}

We will now give a very crude sketch of our modified algorithm, showing that each of the above issues has a relatively simple solution.

Our algorithm takes as input a 2-reflex orthogonal polyhedron $\mathcal P$, and outputs the set of reflex edges to which guards are assigned.

\subsection*{Preprocessing: $O(n \log n)$}

First of all, regardless of the data structures we use to represent orthogonal polyhedra, we do some preprocessing on $\mathcal P$ to construct \emph{adjacency tables} of faces, edges and vertices. We also store every face's boundary as a sorted list of its vertices and edges. (This typically takes $O(n)$ time, but can take $\Theta(n \log n)$ time if our initial data are severely unstructured.) These tables allow us to efficiently navigate the polyhedron's boundary. We mark each edge as reflex or convex and, if needed, we turn $\mathcal P$ by $90^\circ$ so that it contains no vertical reflex edges.

\subsection*{Finding contact lines: $O(n \log n)$}

We first compute a structure that is very similar to the \emph{horizontal visibility map} (also known as \emph{trapezoidalization}) of each vertical face of $\mathcal P$. This is a well-studied 2-dimensional problem that consists in partitioning a polygon into trapezoids by drawing horizontal lines at vertices. In our case, faces are orthogonal polygons (perhaps with holes), trapezoids are actually rectangles, and cut lines are drawn at reflex vertices only.

Let $F_i$ be a vertical face of $\mathcal P$ with $n_i$ vertices. We sort all the vertical edges of $F_i$ by the $z$ coordinate of their upper vertex (in $O(n_i \log n_i)$ time), and we ``scan'' $F_i$ from top to bottom with a sweep line. We maintain a horizontally sorted list of all the vertical edges of $F_i$ pierced by the sweep line, in which insertion and deletion take $O(\log n_i)$ time. Every time our sweep line hits a new reflex vertex $v$ belonging to a reflex edge of $\mathcal P$, we draw a \emph{contact line} from $v$ to the next (or previous) edge in the list, we add a \emph{fake vertex} there (if needed), and we proceed with our sweep.

This process takes $O(n_i \log n_i)$ time and, letting $n=\sum_i n_i$, finding contact lines on every vertical face of $\mathcal P$ takes $O(n \log n)$ in total, because
$$\sum_i n_i \log n_i \leqslant \sum_i n_i \log n = n \log n.$$

Every time we find a new contact line, we also update the face, edge, and vertex data we computed in the preprocessing step. That is, as soon as a new contact line is found, the corresponding face gets a new edge and is perhaps partitioned in two coplanar faces (this step takes constant time). Moreover, as soon as a new fake vertex $v$ is found, it is added to the other face sharing it (say, $F_j$). If we still have to process $F_j$, we consider $v$ as a reflex vertex and draw a contact line at $v$ in $F_j$ when we process it. Otherwise, $F_j$ is now a rectangle, and we just draw an additional contact line in it at $v$, in constant time.

\subsection*{Finding bricks and contact rectangles: $O(n)$}

Notice that contact lines are exactly the boundaries of the contact rectangles of $\mathcal P$. Indeed, in the previous step, we did not draw contact lines only at the reflex vertices of each face, but also at convex vertices that lie on reflex edges of $\mathcal P$.

It is easy now to identify all the contact rectangles and all the bricks, navigating the boundary of $\mathcal P$ using our precomputed data structures. While we do it, we also build a \emph{brick graph} $G$, having a node for each brick and an arc connecting each pair of bricks sharing a contact rectangle.

Observe that issue~(\ref{i4:3}) above is now solved.

\subsection*{Resolving non-primitive contact rectangles: $O(n)$}

Non-primitive contact rectangles are those that are surrounded by more than one reflex edge of $\mathcal P$ and, as such, are easy to find. As proven in Theorem~\ref{t4:2orthor}, it is safe to cut $\mathcal P$ at a non-primitive contact rectangle and place guards in the resulting polyhedra (or polyhedron).

Instead of actually cutting $\mathcal P$ and updating all the data structures, we merely delete the arcs of $G$ corresponding to the non-primitive contact rectangles.

\subsection*{Forcing simple connectedness: $O(n)$}

By this point, $\mathcal P$ has been partitioned into several, possibly not simply connected, stacks. As proven in Theorem~\ref{l4:stack}, it is safe to further cut the stacks until they all become simply connected. To do so, we again process only $G$, turning it into a forest. Such a task is accomplished by a straightforward traversal, starting at each connected component and deleting arcs leading to already visited nodes. Recall that bricks in stacks have at most four neighbors, hence the time complexity of this traversal is indeed linear.

Observe that this step also solves issue~(\ref{i4:1}) above.

\subsection*{Adjusting brick parity: $O(n)$}

For reasons that will be clear shortly, we insist on having only stacks with an even number of bricks. O'Rourke solves this by adding an extra ``chip'' to the polygon, in case he wants to change the parity of its reflex vertices. Then he applies his algorithm to the new polygon, and later removes the chip. If the chip happens to hold a guard, then that guard is reassigned to the nearest convex vertex, after the chip is removed. Notice how this choice causes issue~(\ref{i4:2}) above.

In order to avoid placing guards on convex edges, we proceed as follows. We compute the size of each connected component of $G$ by a simple traversal. Then, in each component with an odd number of nodes, we find one leaf (recall that $G$ is a forest) and delete the arc attached to it (if one exists). Finally, we collect each isolated node, remove it from $G$, find its corresponding brick $\mathcal B$ in $\mathcal P$, find a contact rectangle bordering $\mathcal B$ (one must exist), find one reflex edge $e$ on it and assign it a guard. Referring to Figure~\ref{fig4:primitive}, it is obvious that $e$ orthogonally sees all of $\mathcal B$.

The correctness of this step follows from the remarks contained in the proof of Lemma~\ref{l4:stack}, that every contact rectangle in a stack with an odd number of bricks yields an odd cut, and that it is always safe to make odd cuts.

\subsection*{Identifying odd cuts: $O(n)$}

We are left with stacks having an even number of bricks, and we want to further partition them with odd cuts. In order to identify odd cuts, we pick each connected component of $G$ and we do a depth-first traversal, rooted anywhere. During the traversal, we compute, at each arc, the parity of the nodes in the dangling subtree. Such parity is even if and only if that arc corresponds to an odd cut.

Now, it is straightforward to notice that cutting a stack having an even number of bricks at an odd cut yields two stacks that again have an even number of bricks. Additionally, this operation does not change the parity of the cuts in the two resulting stacks. Hence, there is no need to re-identify odd cuts after a cut is made. In contrast, stacks with an odd number of bricks do not have such property, and this motivates our previous step.

It follows that we may remove from $G$ all the arcs corresponding to odd cuts, without worrying about side effects.

\subsection*{Guarding double castles: $O(n)$}

At this point, only non-convex stacks without odd cuts are left. As a consequence of the observations in Lemma~\ref{l4:stack}, these are all double castles, which we now have to guard in linear time. (In contrast, O'Rourke's algorithm was left at this point with double histograms, giving rise to issue~(\ref{i4:4}) above.)

Our algorithm is loosely based on the proofs of Lemmas~\ref{l4:mono},~\ref{l4:castle}, and~\ref{l4:dcastle}, which naturally yield a procedure that cuts along certain contact rectangles and selects guards in monotone orthogonal polyhedra.

The only non-trivial aspect is that, occasionally in the procedure, we need know if some castle (or upside-down castle) is a prism, and what the orientation of its reflex edges is. To efficiently answer this question, we precompute this information for every ``sub-castle'' of each double castle that we have. We identify the two castles constituting each double castle (in linear time), then we do a depth-first traversal of the subgraphs of $G$ corresponding to those castles, starting from their base bricks. When we reach an internal node $v$, we recursively check if the subtrees dangling from its two children correspond to prisms, and if their reflex edges are oriented in the same direction. Then, after inspecting also the brick corresponding to $v$, we know if its dangling subtree corresponds to a prism and, if so, the direction of its reflex edges. Leaves are trivial to handle, in that they always correspond to prisms with no reflex edges.
\\

Summarizing, and recalling the upper bound given in Theorem~\ref{t4:2orthor}, we have the following.

\begin{theorem}
Given a 2-reflex orthogonal polyhedron with $r>0$ reflex edges and genus $g$, a guarding set of at most 
$$\left\lfloor \frac{r-g}{2} \right\rfloor +1$$
reflex edge guards can be computed in $O(n \log n)$ time.\hfill \qed
\end{theorem}

Of course, the same algorithm also achieves the upper bound in terms of $e$ given in Theorem~\ref{t4:2orthoe}:
$$\left\lfloor \frac{e-4}{8} \right\rfloor +g.$$

Observe that the only superlinear step is the vertical sweep that finds the contact lines in every vertical face of the polyhedron, plus perhaps the preprocessing step. Whether a more efficient algorithm exists remains an open problem.
\chapter{Mutually parallel edge guards in orthogonal polyhedra}\label{chapter5}
\markthischapter{CHAPTER 5. \ MUTUALLY PARALLEL EDGE GUARDS}

\begin{chapterabstract}
We consider the problem of edge-guarding orthogonal polyhedra, with the additional constraint that our guards must be all parallel.

We show that any orthogonal polyhedron with $e$ edges, of which $r$ are reflex, is guardable by at most
$$\left\lfloor\frac{e+r}{12}\right\rfloor$$
mutually parallel open edge guards, regardless of its genus. If the guards are closed, then the polyhedron is orthogonally guardable by the same number of parallel guards.

On the other hand, asymptotically $(e+r)/14$ parallel guards may be necessary, as $r$ tends to infinity. Further lower bounds are $e/12$ and $r/2$.

We also establish tight inequalities relating $e$ with $r$, by virtue of which we obtain other upper bounds on the number of parallel edge guards, in terms of $e$ and $r$ only (and the genus $g$):
$$\left\lfloor\frac{11e}{72} - \frac{g}{6}\right\rfloor - 1\qquad \mbox{and}\qquad \left\lfloor\frac{7r}{12}\right\rfloor - g + 1.$$

In particular, we slightly improve upon the previous best upper bound of $\left\lfloor e/6\right\rfloor$ edge guards, due to Urrutia.
\end{chapterabstract}

\section{Orthogonal polyhedra}

We start by giving tight inequalities relating the number of edges in an orthogonal polyhedron with the number of reflex edges. In the next section, we will use these inequalities to obtain upper bounds on parallel edge guard numbers.

Recall from Chapter~\ref{chapter2} that orthogonal polyhedra have six types of vertices, shown in Figure~\ref{fig:vertextypes}.

\begin{figure}[h]
\centering{\includegraphics[width=0.68\linewidth]{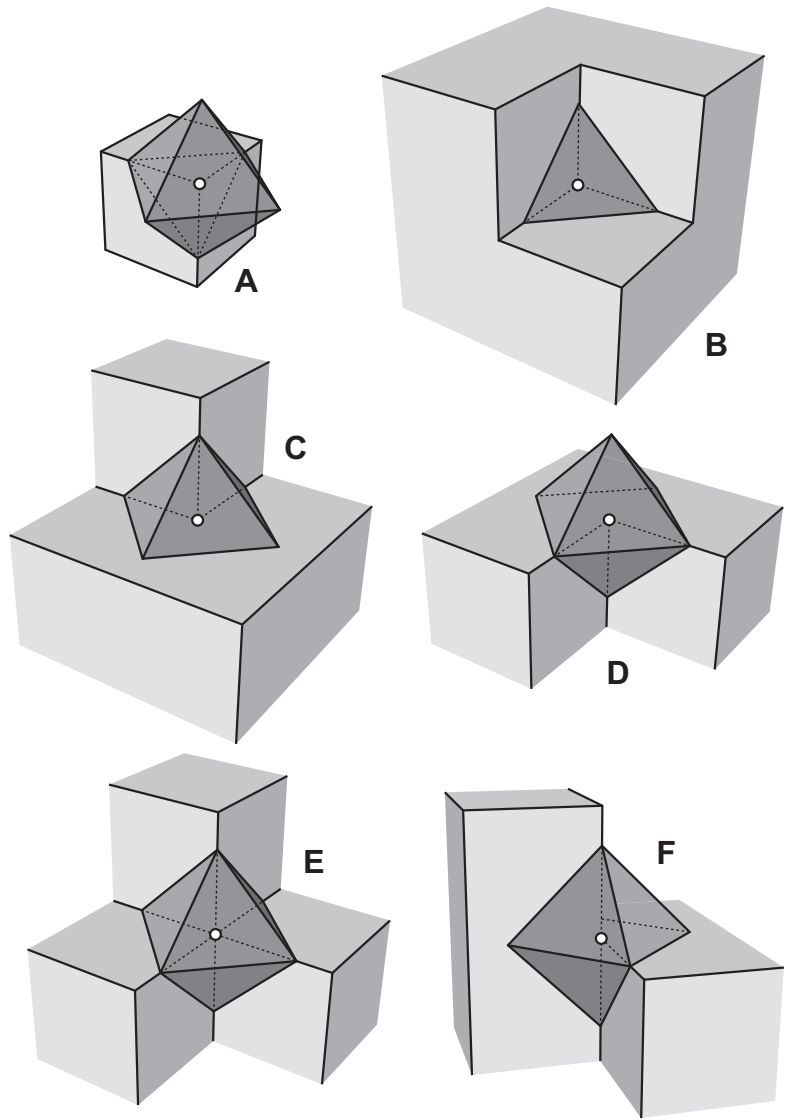}}
\caption{Vertex types.}
\label{fig:vertextypes}
\end{figure}

Let us denote by $A$ the number of A-vertices in a given orthogonal polyhedron, and so on, for each vertex class.
 
\begin{lemma}
\label{lem:2}
In any orthogonal polyhedron with $r>0$ reflex edges,
$3A + D \geqslant 28$.
\end{lemma}
\begin{proof}
Consider the bounding cuboid of the polyhedron and the set of orthogonal polygons (perhaps with
holes, without degeneracies), formed by intersection of the faces of the cuboid and the polyhedron. The vertices of those polygons are either A-vertices or D-vertices of the polyhedron (convex vertices yield A-vertices, and reflex vertices yield D-vertices).  Our strategy is to only look at the vertices belonging to the bounding faces and ensure that there is a sufficient number of them.  Namely, we only need show that there are at least
 
\begin{enumerate}
\item[(a)]ten A-vertices, or
\item[(b)]nine A-vertices and one D-vertex, or
\item[(c)]eight A-vertices and four D-vertices.
\end{enumerate}

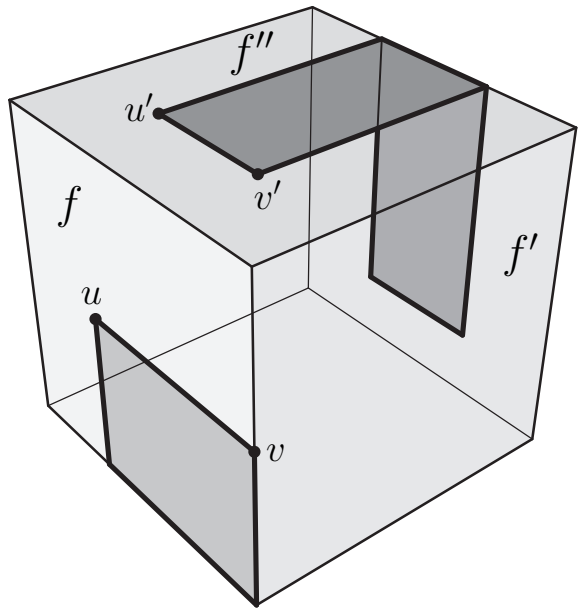
\begin{figure}[h]
\centering{\includegraphics[width=.5\linewidth]{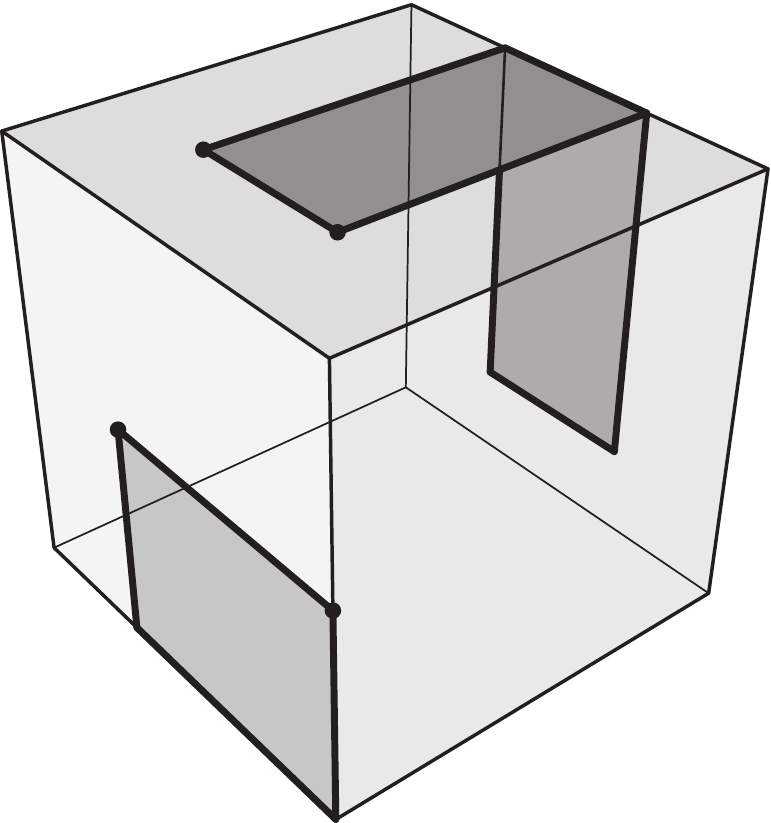}}
\large
\put(-121.46,150.42){$v'$}
\put(-170,183.42){$u'$}
\put(-116.87,55.42){$v$}
\put(-186.25,114.58){$u$}
\Large
\put(-195.62,146.87){$f$}
\put(-28.54,128.33){$f'$}
\put(-130.83,204.37){$f''$}
\caption{Illustration of the proof of Lemma~\ref{lem:2}.}
\label{fig:boundingbox}
\end{figure}

Suppose each face of the bounding cuboid contains exactly one rectangle. If all the vertices of these rectangles coincide with the corners of the bounding cuboid, then the polyhedron is convex, contradicting the assumptions.  Hence, there is a vertex $v$ that is not a corner of the bounding cuboid. Let $f$ be a face of the bounding cuboid containing $v$. At least one of the vertices, denoted by $u$, adjacent to $v$ in the rectangle contained in $f$ is such that the relative interior of $uv$ does not lie on an edge of the bounding cuboid. Let $f'$ be the bounding face opposite to $f$, and $f''$ be the bounding face chosen as shown in Figure~\ref{fig:boundingbox}: out of the four faces surrounding $f$, $f''$ is the one that lies on the \textquotedblleft side\textquotedblright\ of $uv$. $f$ and $f'$ contain two disjoint rectangles, and thus exactly eight distinct A-vertices. Additionally, $f''$ has two extra A-vertices, lying on an edge $u'v'$ parallel to $uv$ (refer to Figure~\ref{fig:boundingbox}). Collectively, $f$, $f'$ and $f''$ contain at least ten A-vertices, so (a) holds.

On the other hand, if there exists a bounding face $f$ whose intersection with the polyhedron is not a single rectangle, then we need analyze the following three cases.  Let once again $f'$ be the bounding face opposite to $f$.  

\begin{itemize}
\item If $f$ contains at least two polygons (those polygons' boundaries must be disjoint because $f$ is a bounding face), then collectively $f$ and $f'$ contain at least 12 distinct A-vertices, so (a) holds.  Indeed, every orthogonal polygon has at least four convex vertices.

\item If $f$ contains a polygon with at least one hole, then the polygon's external boundary contains at least four convex vertices (i.e., A-vertices), and the hole has at least four reflex vertices (D-vertices). $f'$ also contains at least four convex vertices (A-vertices). Together $f$ and $f'$ contain at least eight A-vertices and four D-vertices, so (c) holds.

\item If $f$ contains just one polygon, which is not convex, then such a polygon has at least five convex vertices and one reflex vertex. Together with $f'$, there are at least nine A-vertices and one D-vertex, so (b) holds.
\end{itemize} 
\end{proof}

\begin{theorem}
\label{th:2}For every orthogonal polyhedron with $e$ total edges, $r>0$ reflex edges and genus $g\geqslant 0$,
$$\frac{e}{6} + 2g - 2\ \leqslant\ r\ \leqslant\ \frac{5e}{6} - 2g - 12$$
holds. Both inequalities are tight for every $g$.
\end{theorem}
\begin{proof}
Let $c=e-r$ be the number of convex edges.
Let $A$ be the number of A-vertices, etc.  Double counting the pairs
(edge, endpoint) yields (refer to Figure~\ref{fig:vertextypes})
\begin{equation}
2c = 3A + C + 2D + 3E + 2F,
\label{eq:1}
\end{equation}
\begin{equation}
2r = 3B + 2C + D + 3E + 2F.
\label{eq:2}
\end{equation}
The curvature
of A- and B-vertices is $\pi/2$, the curvature of C- and
D-vertices is $-\pi/2$, the curvature of E- and F-vertices is $-\pi$.
Hence, by the polyhedral Gauss--Bonnet theorem (Theorem~\ref{t2:gauss}),
\begin{equation}
A + B - C - D - 2E - 2F = 8 - 8g.
\label{eq:3}
\end{equation}
Finally, since all the variables involved are non-negative,
\begin{equation}
9B + 3C + 3E + F \geqslant 0.
\label{eq:4}
\end{equation}
Subtracting 3 times (\ref{eq:3}) from 2 times (\ref{eq:4}) yields
\[-3A +15B+9C+3D+12E+8F \geqslant 24g-24.\]
Further subtracting (\ref{eq:1}) and adding 5 times (\ref{eq:2}) to the last inequality yields
\[2c - 10 r +24 g - 24 \leqslant 0,\]
which is equivalent to
$$\frac{e}{6} + 2g - 2 \leqslant r.$$

To see that the left-hand side inequality is tight for every $r$ and $g$, consider the staircase-like polyhedron with holes depicted in Figure~\ref{fig:staircaseholes}. If the staircase has $k$ \textquotedblleft segments\textquotedblright\ and $g$ holes, then it has a total of $6k+12g+6$ edges and  $k+4g-1$ reflex edges.

\begin{figure}[h]
\centering{\includegraphics[width=.5\linewidth]{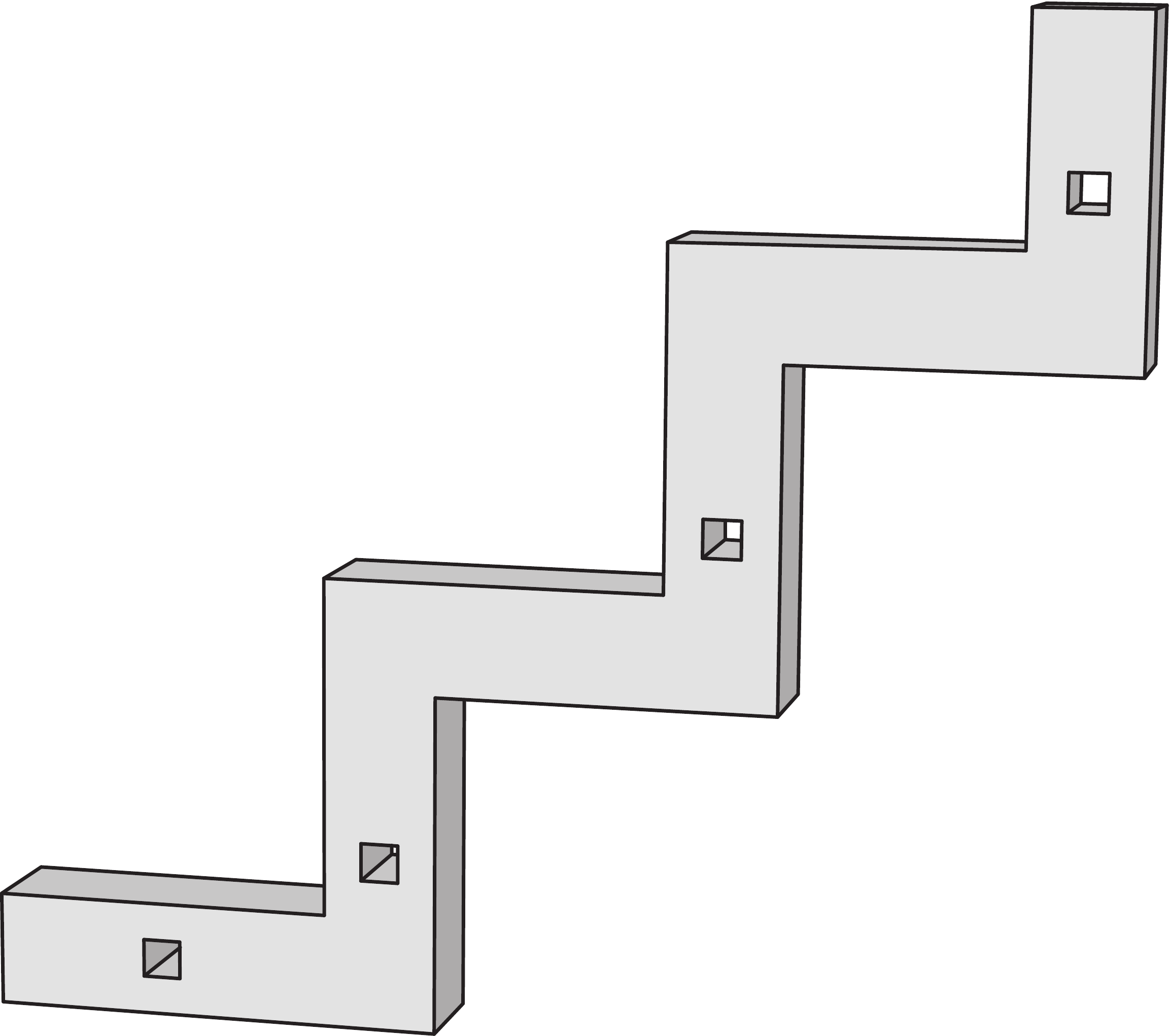}}
\caption{Polyhedron that achieves the tight left-hand side bound in Theorem~\ref{th:2}.}
\label{fig:staircaseholes}
\end{figure}

According to Lemma~\ref{lem:2}, $3A + D \geqslant 28$, unless the polyhedron is a cuboid.
Then
\begin{equation}
9A + 3D + 3E + F \geqslant 84.
\label{eq:5}
\end{equation}
 Subtract 3 times (\ref{eq:3}) from 2 times (\ref{eq:5}):
\[15A - 3B + 3C + 9D + 12E + 8F \geqslant 24g + 144.\]
Subtract (\ref{eq:2}) and add 5 times (\ref{eq:1}):
\[2r - 10c + 24g + 144 \geqslant 0,\]
which is equivalent to
$$r \leqslant \frac{5e}{6} - 2g - 12.$$

To see that the right-hand side inequality is also tight, consider the polyhedron represented in Figure~\ref{fig:reflexbound}: a cuboid with a staircase-like well carved in it, and a number of cuboidal \textquotedblleft poles\textquotedblright\ carved out from the surface of the well (i.e., the ``negative'' version of Figure~\ref{fig:staircaseholes}). If the staircase has $k$ \textquotedblleft segments\textquotedblright\ and $g$ poles, then the polyhedron has a total of $6k+12g+18$ edges and  $5k+8g+3$ reflex edges.
\end{proof}

\begin{figure}[h]
\centering{\includegraphics[width=.5\linewidth]{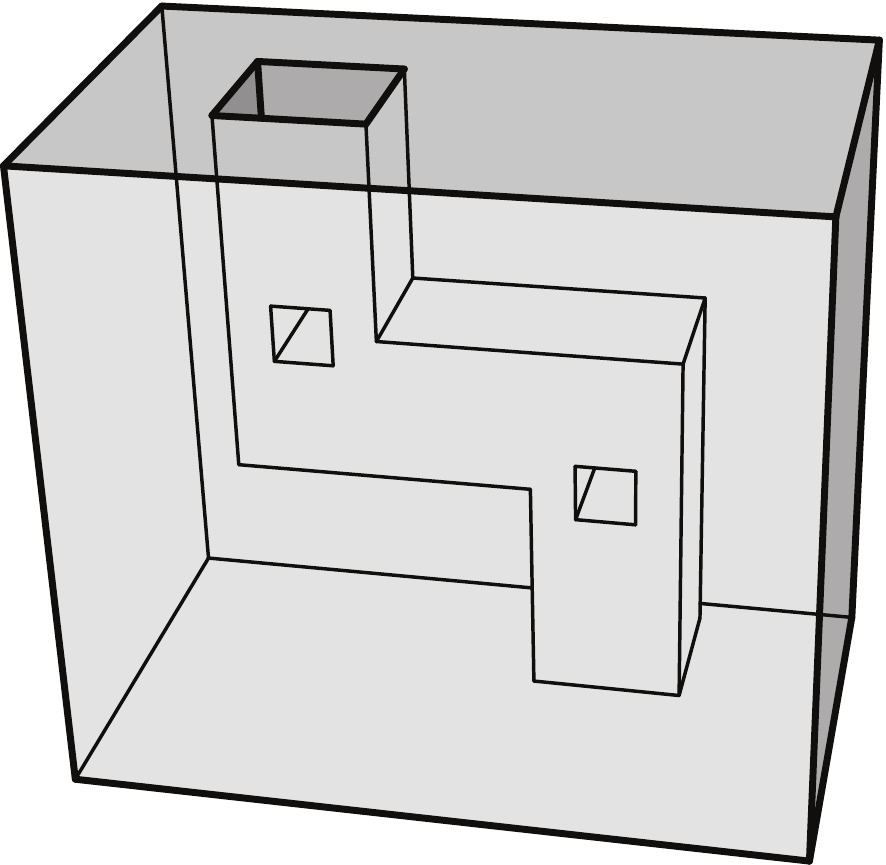}}
\caption{Polyhedron that achieves the tight right-hand side bound in Theorem~\ref{th:2}.}
\label{fig:reflexbound}
\end{figure}

\begin{observation}\label{o5:nonono}
The statement of Theorem~\ref{th:2} does not hold if we change the definition of orthogonal polyhedron by dropping the condition of connectedness of the boundary. Indeed, consider a cube and remove several smaller disjoint cubic regions from its interior. The resulting shape has unboundedly many reflex edges and just 12 convex edges.
\end{observation}

\section{Parallel edge guarding}

\subsection*{Motivations}

Now we consider the problem of guarding a given orthogonal polyhedron with mutually parallel edge guards. Recall that Urrutia gave an upper bound on the number of edge guards of $\lfloor e/6\rfloor$, where $e$ is the total number of edges (Theorem~\ref{t:urrutia}). Here we improve that upper bound to show that, asymptotically, $11e/72$ (open or closed) edge guards suffice, or, in terms of $r$, that $7r/12$ guards suffice.

One of the most obvious ways to treat guarding problems in polyhedra is to consider cross sections and try to apply known techniques and theorems to the resulting planar polygons. In the case of edge-guarding orthogonal polyhedra, one could for instance consider all the horizontal cross sections, and solve the \ART with vertex guards in each section. Such vertex guards become edge guards when extended to the whole polyhedron. Of course, the challenge is to ``consistently'' select vertex guards in neighboring cross sections, so that the overall amount of edge guards is minimized. This suggests how the problem of orthogonally guarding with mutually parallel edge guards naturally emerges.

Nonetheless, this guarding mode has notable applications in point location, tracking, and navigation. Imagine that a polyhedron is orthogonally guarded by vertical edge guards. Then, if each guard represents an array of sensors, each of which scans a horizontal area, we can immediately determine the altitude of any object in the polyhedron, based on which sensor detects the object. We can also monitor the object's movements, and guide it along a path, without ever ``losing track'' of it. In some applications, employing parallel edge guards may even be a necessity derived by environmental constraints.

\subsection*{Guarding strategy}

\begin{theorem}
\label{th:3}
Any open (resp.\ closed) orthogonal polyhedron with $e$ edges, of which $r$ are reflex, is guardable (resp.\ orthogonally guardable) by at most
$$\left\lfloor \frac{e+r}{12} \right\rfloor$$ 
mutually parallel open (resp.\ closed) edge guards. Guard locations can be computed in linear time.
\end{theorem}
\begin{proof}
Let $e_x$ and $r_x$ be the number of $x$-parallel edges and $x$-parallel reflex
edges, respectively; $e_y$, $e_z$, $r_y$, $r_z$ are similarly
defined.
Without loss of generality,
assume $x$ is the direction that minimizes the sum $e_x + r_x$, so that
$$e_x + r_x \leqslant \frac{e+r}{3}.$$
Of course, a guard on every $x$-parallel edge suffices to cover all of
$\mathcal P$, but we can do much better with a selected subset of these edges.
We argue below that selecting the three types of $x$-parallel edges
circled in Figure~\ref{fig:guardplacement} suffices
(as do three other symmetric configurations).
Let the number of  $x$-edges of each of the eight types be $\alpha, \cdots, \delta'$ as labeled
in Figure~\ref{fig:guardplacement}.

\begin{figure}[h]
\centering{\includegraphics[width=.75\linewidth]{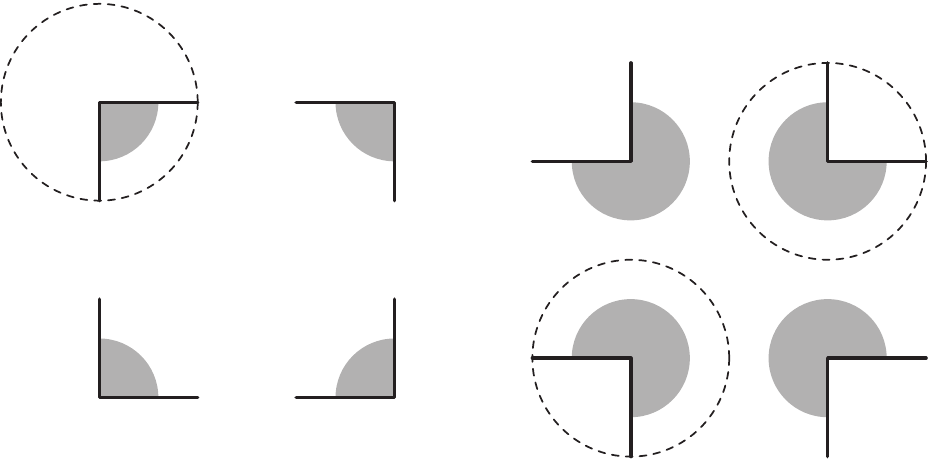}}
\put(-299.01,129.28){$\alpha$}
\put(-177.43,129.28){$\beta$}
\put(-177.24,7.3){$\gamma$}
\put(-299.01,7.3){$\delta$}
\put(-121.18,110.64){$\alpha'$}
\put(-26.45,110.64){$\beta'$}
\put(-26.84,19.14){$\gamma'$}
\put(-121.18,19.14){$\delta'$}
\caption{Possible configurations of $x$-edges. The $x$ axis is directed
  toward the reader. The circled configurations are those selected in
  the proof of Theorem~\ref{th:3}.}
\label{fig:guardplacement}
\end{figure}
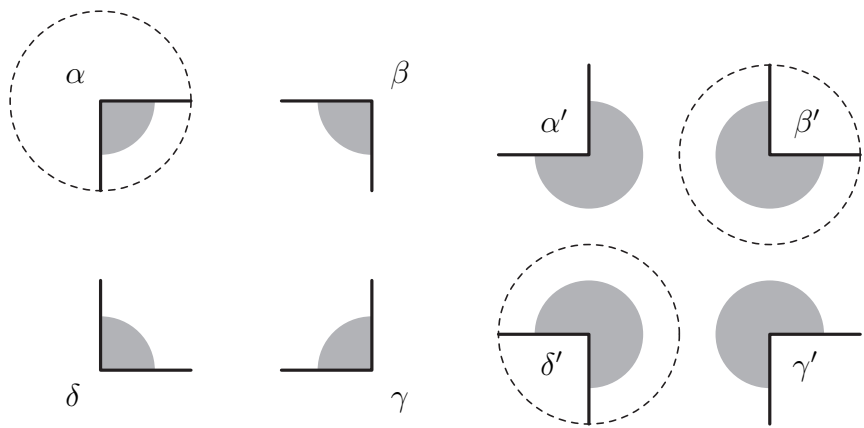

Hence we could place $\alpha + \beta' + \delta'$ guards, or (symmetrically) $\gamma +\beta'+\delta'$ guards, or $\beta+\alpha'+\gamma'$ guards, or $ \delta+\alpha'+\gamma'$ guards.

By choosing the minimum of these four sums, we place at most 
$$
\left(\alpha+\beta+\gamma+\delta+2\alpha'+2\beta'+2\gamma'+2\delta'\right)/4
$$
$$
\; = \; \frac{e_x+r_x}{4} \; \leqslant \; \frac{e+r}{12}
$$
guards.

Next we prove that our guard placement works.

We consider any point $p$ in $\mathcal{P}$ and show that $p$ is
guarded by the edges selected 
in Figure~\ref{fig:guardplacement}. 
Let $\omega$ be the $x$-orthogonal plane containing $p$ and let $Q$ be
the intersection of the (open) polyhedron $\mathcal{P}$ with
$\omega$. To prove that $p$ is guarded, we first cast an axis-parallel ray
from $p$. For our choice of guarding edges, the ray is directed upward.
Let $q$ be the intersection point of the ray and the boundary of $Q$ that is nearest
to $p$.
Next, grow leftwards a rectangle whose right side is $pq$ until it
hits a vertex $v$ of $Q$.  If it hits several vertices simultaneously, let $v$ be the
topmost.
There are three possible configurations for $v$, shown in  Figure~\ref{fig:guarding},
and each corresponds to a selected configuration in our placement of
guards (Figure~\ref{fig:guardplacement}).

If guards are closed, then $p$ is orthogonally guarded. Otherwise, if guards are open and $v$ lies in the interior of the guarding edge, then $p$ is (orthogonally)
guarded. If guards are open and $v$ is an endpoint of the guarding edge, then we show that $p$
is guarded by
a sufficiently small neighborhood of $v$ that belongs to that edge. Every face of $\mathcal{P}$ that does not intersect $\omega$ has a positive distance from $\omega$. Let $d$ be the smallest such distance. Then, the points of the guarding edge at distance strictly less than $d$ from $v$ see $p$. Hence $p$ is guarded.

If a different triplet of guarding edges is chosen, the above construction is suitably rotated by a multiple of $90^\circ$.

Computing guard locations merely involves counting edges of each type, which can be done in linear time.
\end{proof}

\begin{figure}[h]
\centering{\includegraphics[width=.85\linewidth]{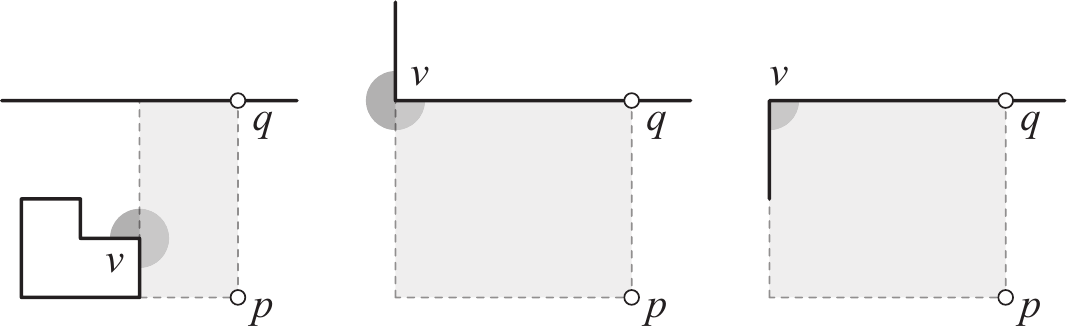}}
\caption{Illustration of the proof of Theorem~\ref{th:3}.}
\label{fig:guarding}
\end{figure}

We remark that our placement of guards in the single slices resembles a construction given in~\cite{aesu-iopof-98}, in a slightly different model.

\begin{observation}
Due to Observation~\ref{o5:nonono}, our methods do not improve on Urrutia's $\left\lfloor e/6\right\rfloor$ upper bound when applied to orthogonal shapes with disconnected boundary. Indeed, in this case the $\left\lfloor(e+r)/12\right\rfloor$ bound given by Theorem~\ref{th:3} still holds, but the $r$ to $e$ ratio can be arbitrarily close to 1.
\end{observation}

By combining the results of Theorem~\ref{th:3} with those of Theorem~\ref{th:2}, we immediately obtain two corollaries.

\begin{corollary}
\label{cor:1}
Any open (resp.\ closed) orthogonal polyhedron with $e$ edges and genus $g$ is guardable (resp.\ orthogonally guardable) by at most
$$\left\lfloor\frac{11e}{72} - \frac{g}{6}\right\rfloor - 1$$
mutually parallel open (resp.\ closed) edge guards.\hfill\qed
\end{corollary}
 
\begin{corollary}
\label{cor:2}
Any open (resp.\ closed) orthogonal polyhedron with $r$ reflex edges and genus $g$ is guardable (resp.\ orthogonally guardable) by at most
$$\left\lfloor\frac{7r}{12}\right\rfloor - g + 1$$
mutually parallel open (resp.\ closed) edge guards.\hfill\qed
\end{corollary}
 
Observe that the lower bounds on the number of edge guards given in Chapter~\ref{chapter3} in terms of $e$ and $r$ also hold for mutually parallel edge guards. Moreover, in terms of both $e$ and $r$, we have a combined lower bound of $(e+r)/14$ (refer to Figure~\ref{fig3:tight}). We conjecture that these bounds are also tight.

\begin{conjecture}\label{con5:1}
Any open (resp.\ closed) orthogonal polyhedron with $e$ edges, of which $r$ are reflex, is guardable (resp.\ orthogonally guardable) by at most
$$\frac{e+r}{14}+O(1),\qquad \frac e {12} +O(1),\qquad \frac r 2+O(1)$$
mutually parallel open (resp.\ closed) edge guards.
\end{conjecture}

\begin{observation}
Due to Theorem~\ref{th:2}, guardability by $e/{12} +O(1)$ guards implies guardability by $r/2+O(1)$ guards.
\end{observation}
\chapter{Edge guards in 4-oriented polyhedra}\label{chapter6}
\begin{chapterabstract}
We study the problem of edge-guarding polyhedra whose faces are oriented in four possible directions, any three of which are linearly independent (orthogonal polyhedra come as a subclass).

We prove that any such 4-oriented polyhedron with $e$ edges, of which $r$ are reflex, can be guarded by at most
$$\left\lfloor \frac{e+r}{6}\right\rfloor$$
edge guards, regardless of its genus. This bound is obtained in two different ways, via independent constructions that can be computed in linear time.

On the other hand, for unboundedly large values of $e$ and $r$, we exhibit 4-oriented polyhedra that require
$$\left\lfloor \frac{e}{6}\right\rfloor-1 = \left\lfloor \frac{e}{12}+\frac{r}{6}\right\rfloor$$
(reflex) edge guards.

All our results hold both in the open and in the closed guard model.
\end{chapterabstract}

\section{4-oriented polyhedra}

\subsection*{Motivations}
Recall that a polyhedron is 4-oriented if each of its faces is orthogonal to one of four possible vectors (denoted from now on by $x$, $y$, $z$, $w$), no three of which are linearly dependent. Orthogonal polyhedra, being 3-oriented, come as a proper subset of this class. In the present chapter we aim at guarding 4-oriented polyhedra with (open and closed) edge guards, and we give bounds on the number of guards that are required in each polyhedron, in terms of $e$, the total number of edges, and $r$, the number of reflex edges.

The problem of guarding $c$-oriented polyhedra naturally arises when trying to generalize the guarding technique employed in Chapter~\ref{chapter5} for orthogonal polyhedra. We still want to reduce our 3-dimensional problem to a set of planar problems, by cutting our polyhedra with parallel planes, and placing point guards in each plane. Of course, for every constant $c$, the number of different configurations we have to consider is finite, and the technique can conceivably be generalized.

However, as $c$ grows, the complexity of the possible configurations blows up, and the results become weaker. As $c$ tends to infinity, our technique will place guards on almost every edge, making our approach pointless. Nonetheless, for $c=4$, the complexity is still manageable, and the results obtained are still remarkable.

Once again, we have the same application that we mentioned for parallel edge guards in orthogonal polyhedra: instead of having, say, vertical edge guards, now we have non-horizontal edge guards. These can be interpreted again as arrays of sensors that scan the area horizontally, detecting any object in the environment, providing its altitude, and having the ability to monitor its movements and navigate it through space.

\subsection*{Structure}
An example of a 4-oriented polyhedron is given in Figure~\ref{fig6:1}, made by gluing together two regular tetrahedra.

\begin{figure}[h]
\centering
\includegraphics[width=0.35\linewidth]{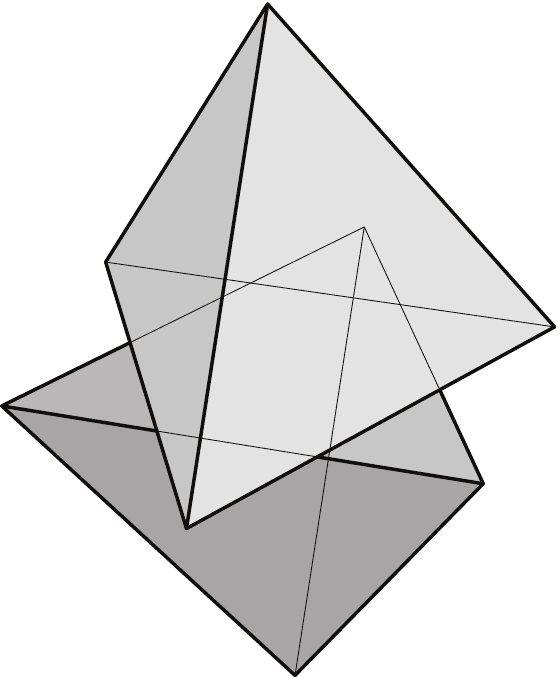}
\caption{Example of a 4-oriented polyhedron.}
\label{fig6:1}
\end{figure}

Note that, in general, each pair of \emph{orientation classes} for faces gives rise to a differently oriented edge. For instance, two adjacent faces orthogonal to $x$ and $y$, respectively, meet at an edge parallel to the cross product $x \times y$. In total, we have ${4 \choose 2} = 6$ possible orientation classes for edges. Furthermore, each such orientation class includes edges of eight different types, depending on the dihedral angle (internal with respect to the polyhedron) formed by the two adjacent faces. Figure~\ref{fig6:2} shows the eight possible ways two faces of a 4-oriented polyedron may meet to form an edge of a given orientation class.

\begin{figure}[h]
\centering
\subfigure[]{\includegraphics[scale=0.5]
{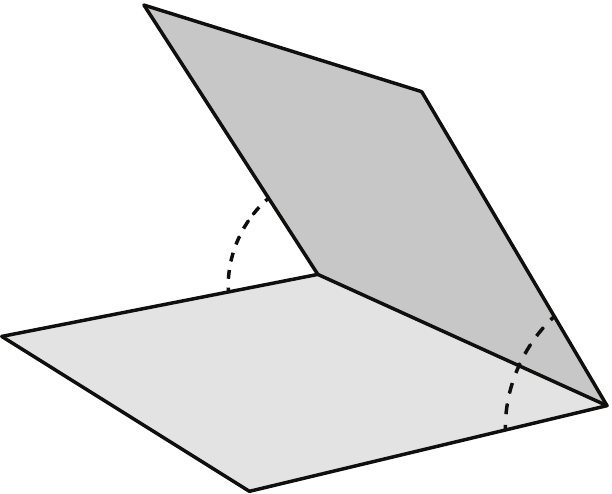}}\ 
\subfigure[]{\includegraphics[scale=0.5]{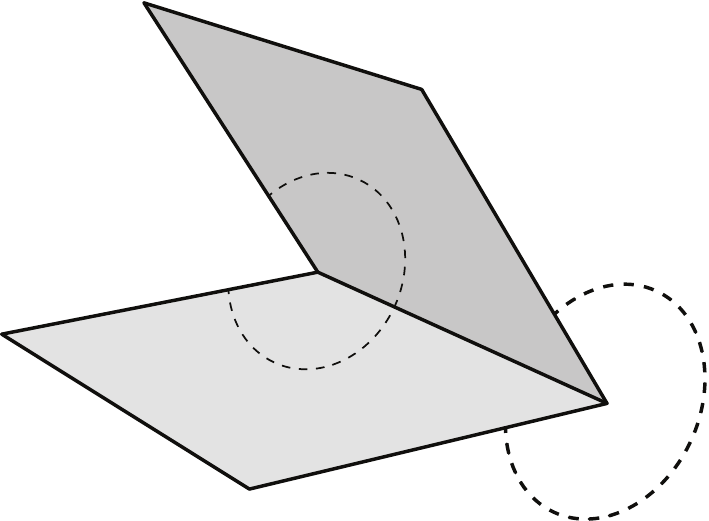}}\ 
\subfigure[]{\includegraphics[scale=0.5]{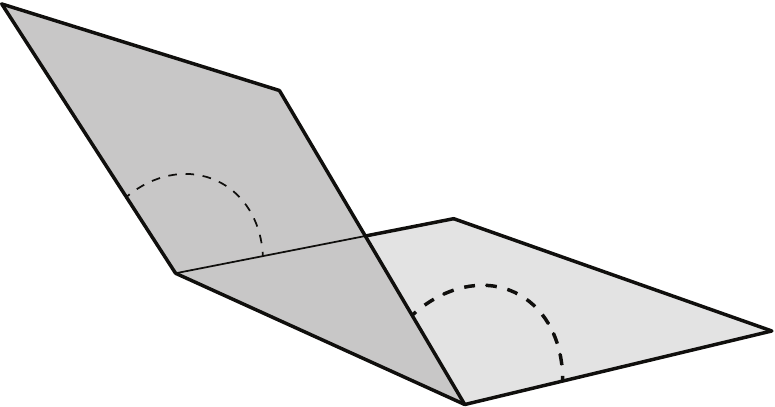}}\ 
\subfigure[]{\includegraphics[scale=0.5]{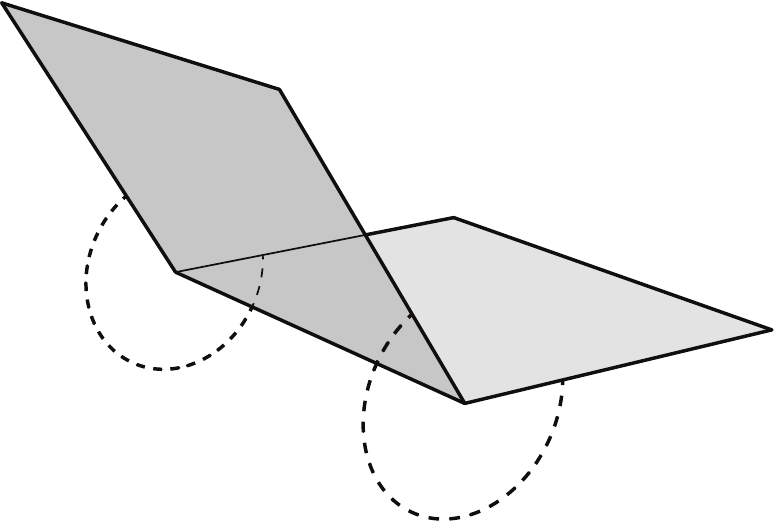}}\\
\subfigure[]{\includegraphics[scale=0.5]
{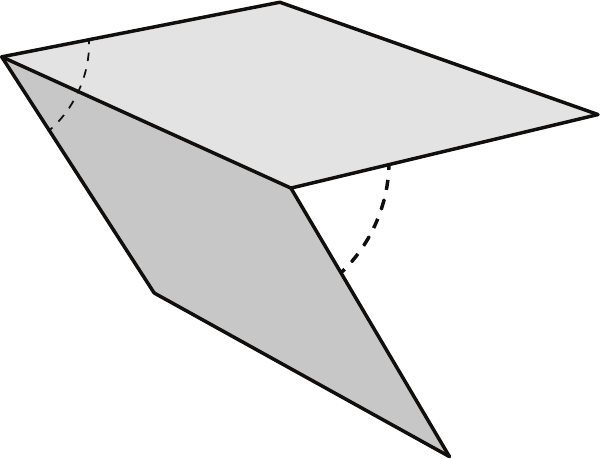}}\ 
\subfigure[]{\includegraphics[scale=0.5]{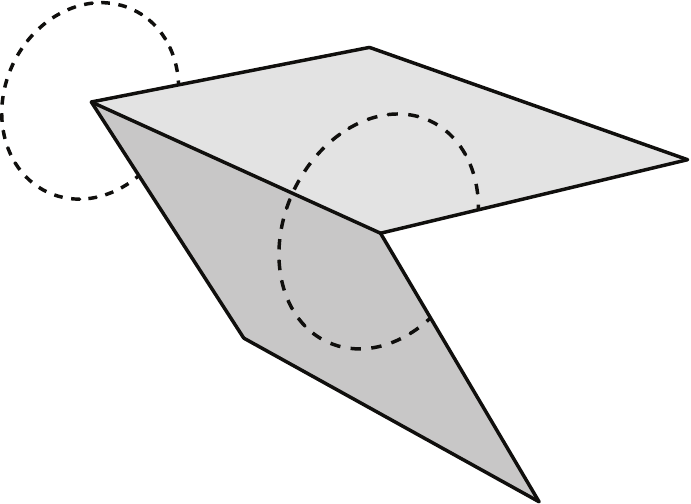}}\ 
\subfigure[]{\includegraphics[scale=0.5]{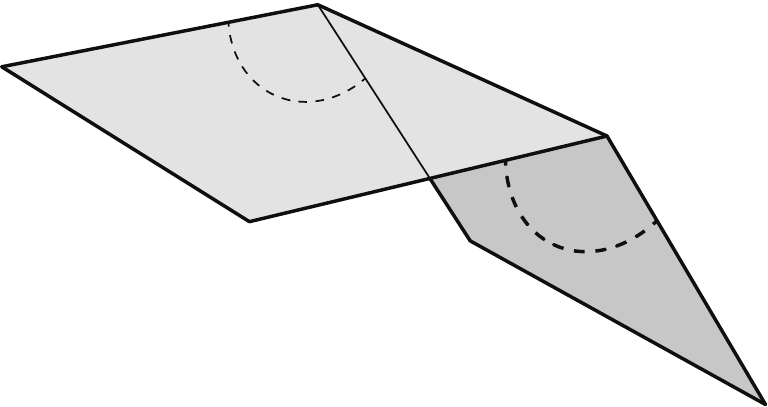}}\ 
\subfigure[]{\includegraphics[scale=0.5]{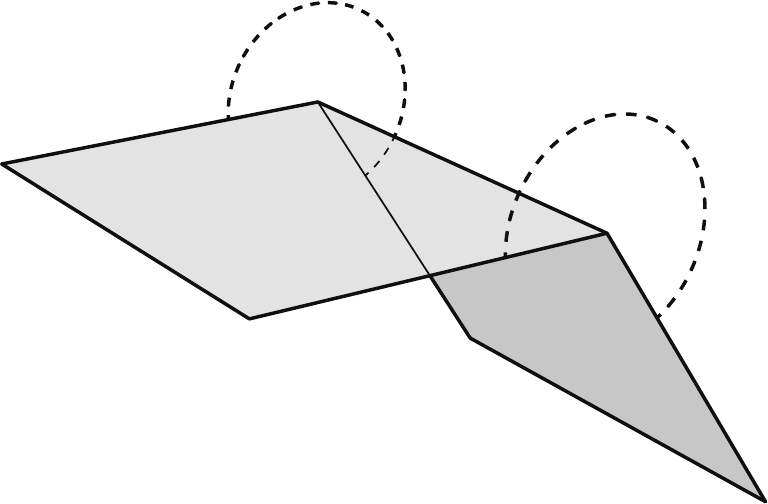}}\\
\caption{The same edge may be formed in eight possible ways.}
\label{fig6:2}
\end{figure}

Each plane parallel to a face of a 4-oriented polyhedron is also parallel to three edge classes. Therefore, if such plane is devoid of vertices of the polyehdron, it may intersect only edges from the three remaining classes. For instance, a plane orthogonal to $x$ is also parallel to $x\times y$, $x\times z$ and $x\times w$, and is not parallel to $y\times z$, $y\times w$ and $z\times w$.

In fact, the intersection of any such plane with a 4-oriented polyhedron is a collection of 3-oriented polygons, possibly with holes, whose vertices may be of 24 species, because there are eight different edge types in each edge class. The 24 species are listed in Figure~\ref{fig6:3}, each paired with its complement.

\begin{figure}[h]
\centering
\includegraphics[width=\linewidth]{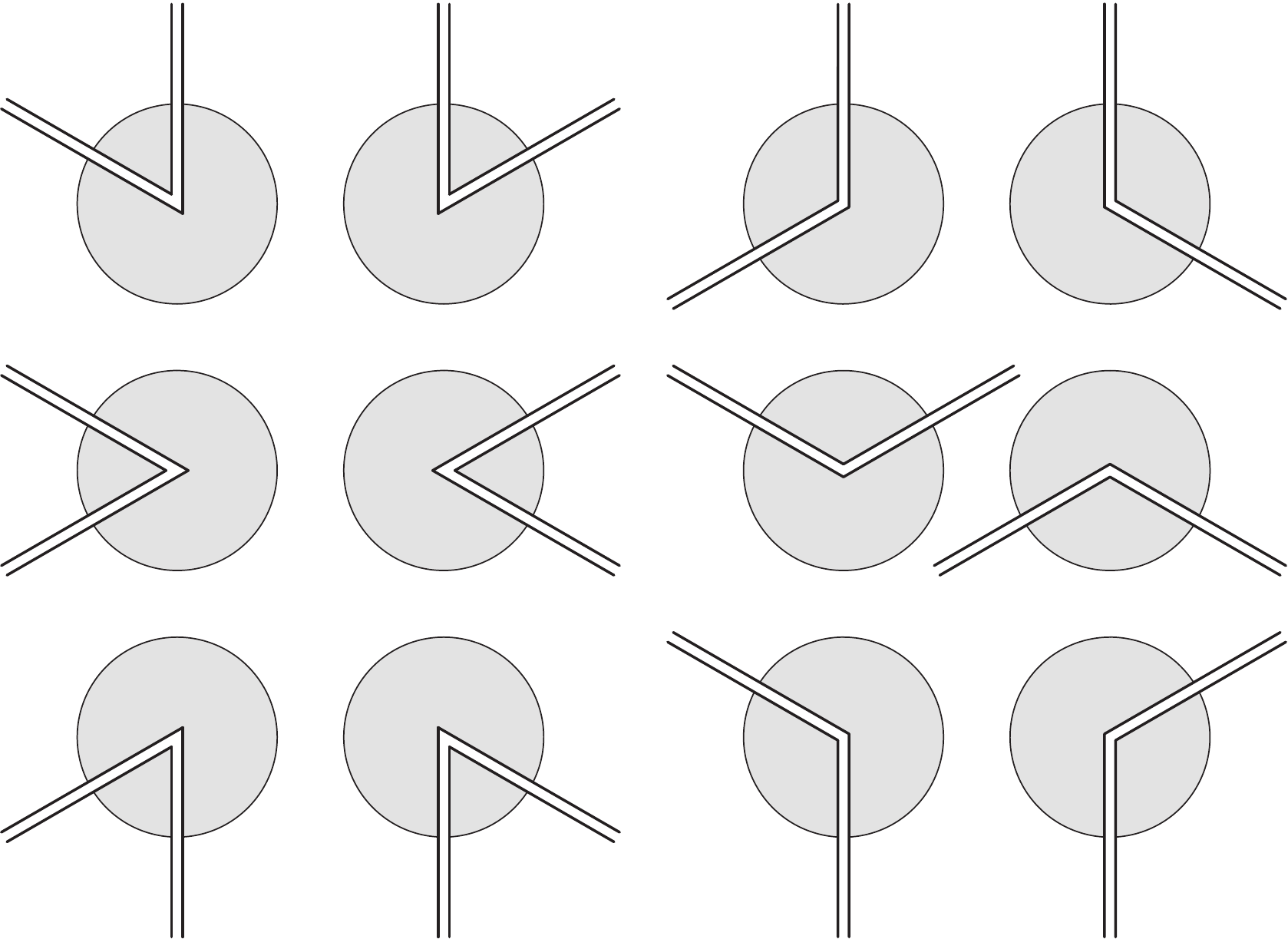}
\caption{Different types of vertices found in a section.}
\label{fig6:3}
\end{figure}

Observe that there are six \emph{smaller} and six \emph{bigger} differently oriented convex angles, and as many reflex angles, which are their complements.

Note that the terms ``smaller'' and ``bigger'' are not necessarily rigorous, depending on the orientations of $y$, $z$ and $w$ with respect to $x$: for example, vectors $y\times z$, $y\times w$ and $z\times w$, when projected onto a plane orthogonal to $x$, may induce on it a partition into six angles measuring $120^\circ$, $30^\circ$, $30^\circ$, $120^\circ$, $30^\circ$, and $30^\circ$, in this order. These would be our ``smaller'' angles, while the ``bigger'' angles would be those formed by pairs of adjacent smaller ones, and would measure $150^\circ$, $60^\circ$, $150^\circ$, $150^\circ$, $60^\circ$, and $150^\circ$. Hence, despite the names, some smaller angles would be twice as large as some bigger angles.

\section{Guarding strategy}

Much like we did in Chapter~\ref{chapter5} with orthogonal polyhedra, we cut 4-oriented polyhedra with parallel planes, and make sure that each planar section is guarded.

Let a 4-oriented polyhedron $\mathcal P$ be given, and let us pick planes orthogonal to $x$, not passing through any vertex of the polyhedron. As previously noted, each such plane intersects only edges of three different types (i.e., orientations), giving rise to the 24 species of vertices illustrated in Figure~\ref{fig6:3}. Our strategy consists in placing edge guards just on the species listed in Figure~\ref{fig6:5}

\begin{figure}[h]
\centering
\includegraphics[scale=0.7]{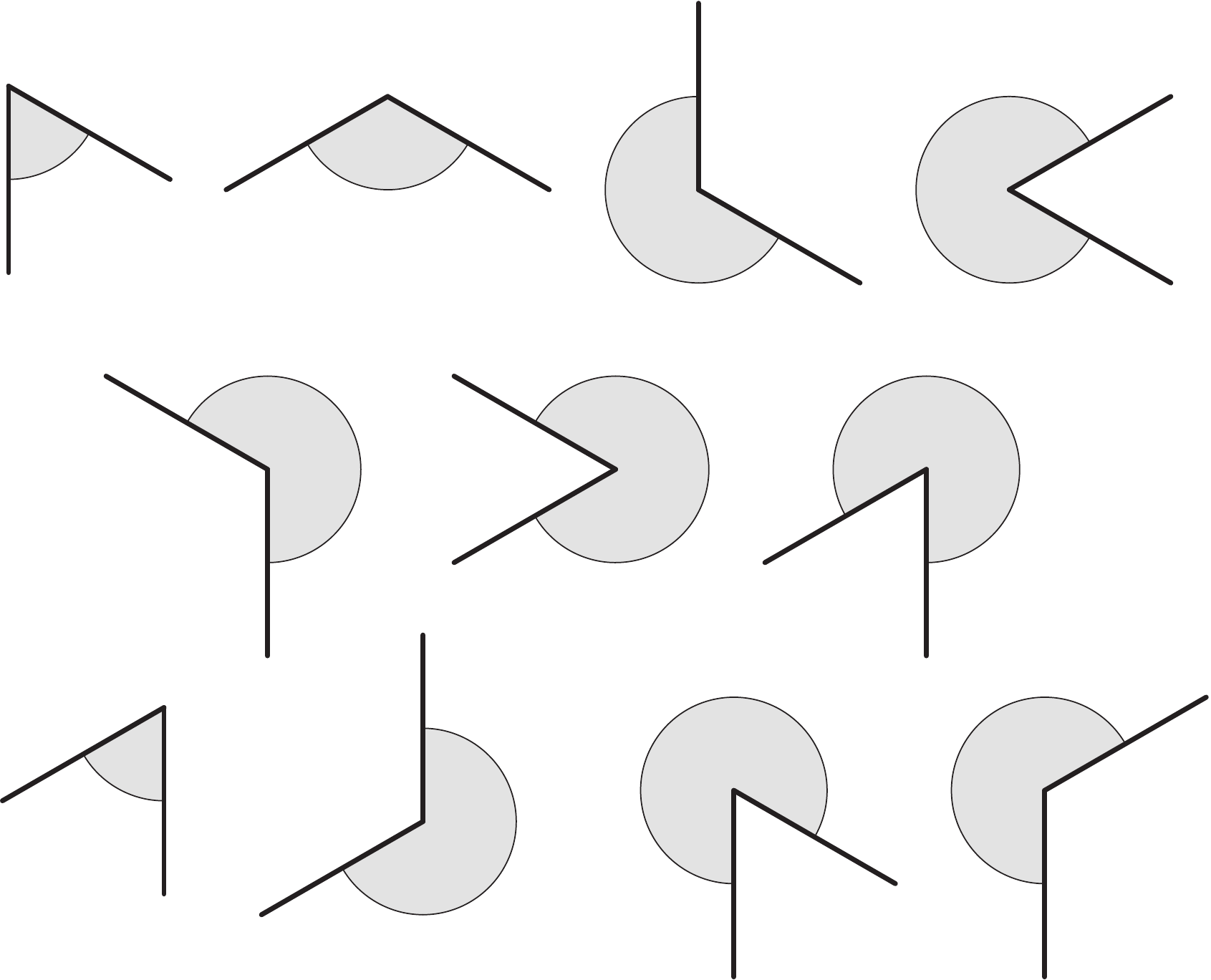}
\caption{Edge types that are assigned a guard.}
\label{fig6:5}
\end{figure}

Now, let a point $p\in \mathcal P$ be given, lying in an $x$-orthogonal section $S$ devoid of vertices of $\mathcal P$, and let us show that $p$ is guarded.

Cast a ray from $p$, parallel to $x\times y$, and extend it until it first encounters the boundary of $S$, in $q$. Because the ray is parallel to $x\times y$, $q$ must lie on an edge of $S$ that is parallel to either $x\times z$ or $x\times w$.

Suppose $q$ lies on an edge parallel to $x\times z$. We consider the segment $pq$, open at $p$ and closed at $q$, and we continuously translate it parallel to $x\times z$. Of the two possible directions of translation (leftward and rightward), we pick the one that forms the smaller angle with the ray from $p$ to $q$. We keep translating $pq$ until it first hits a vertex $v$ of $S$. If several vertices are encountered simultaneously, then $v$ is the topmost one.

Figure~\ref{fig6:6} shows the possible configurations for vertex $v$. If $v$ is one endpoint of the edge containing $q$, it belongs to one of the first four types, from (a) to (d). On the other hand, if $v$ is encountered before such endpoint is reached, then it must belong to one of the last three types, from (e) to (g). Observe that all seven vertex types have been listed in Figure~\ref{fig6:5}, and so the corresponding edges of $\mathcal P$ contain a guard. It follows that $p$ is guarded.

\begin{figure}[h]
\centering
\subfigure[]{\includegraphics[scale=1]{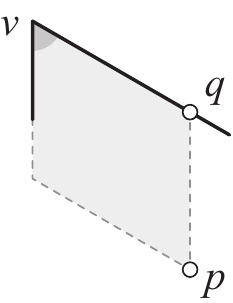}}\qquad\quad
\subfigure[]{\includegraphics[scale=1]{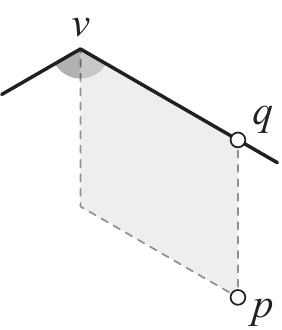}}\qquad\quad
\subfigure[]{\includegraphics[scale=1]{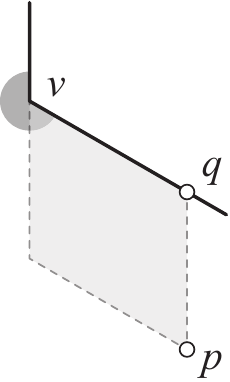}}\qquad\quad
\subfigure[]{\includegraphics[scale=1]{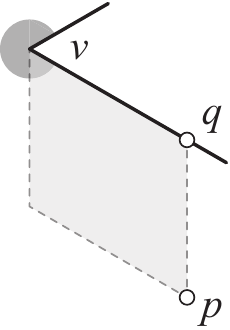}}\\ \vspace{0.2cm}
\subfigure[]{\includegraphics[scale=1]{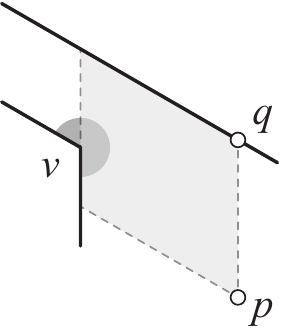}}\qquad\qquad
\subfigure[]{\includegraphics[scale=1]{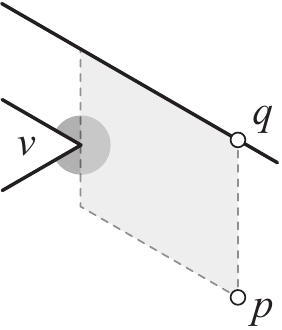}}\qquad\qquad
\subfigure[]{\includegraphics[scale=1]{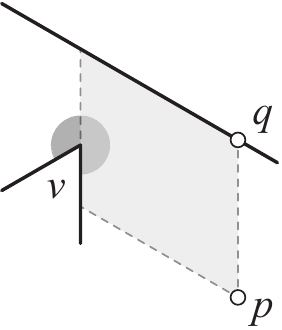}}
\caption{$p$ is necessarily guarded.}
\label{fig6:6}
\end{figure}

If, instead, $q$ lies on an edge parallel to $x\times w$, we repeat the same construction on the other side of $pq$, in a vertically symmetric fashion. It is straightforward to observe that the vertically symmetric counterparts of the seven vertex types shown in Figure~\ref{fig6:6} are also listed in Figure~\ref{fig6:5}. Hence, $p$ is guarded in this case, as well.

Based on these preliminary observations, we are now ready to prove the main result of this chapter.

\begin{theorem}\label{t6:1}
Any open (resp.\ closed) 4-oriented polyhedron with $e$ edges, of which $r$ are reflex, is guardable by at most
$$\left\lfloor \frac{e+r}{6}\right\rfloor$$
open (resp.\ closed) edge guards. Guard locations can be computed in linear time.
\end{theorem}
\begin{proof}
By placing guards as described above, we can guard the entire polyhedron $\mathcal P$, except perhaps its $x$-orthogonal sections containing vertices.

Now, if $\mathcal P$ and our guards are closed, we can guard the missing sections in exactly the same way as the other sections. Notice that our construction works even if $p$ is on the boundary of $S$.

If $\mathcal P$ and the guards are open, let $S$ be one such section, repeat our construction above, and let $p$ be guarded by a vertex $v$ of $S$ that is also a vertex of $\mathcal P$. Then there must exist an edge $\ell$ of $\mathcal P$ hosting a guard, such that $v$ is an endpoint of $\ell$. By the same argument used in Theorem~\ref{th:3}, there exists a point $t$ in the relative interior of $\ell$, close enough to $v$, that sees $p$. Indeed, because the faces of $\mathcal P$ are finitely many, there exists a minimum positive distance $d$ between the plane containing $S$ and a face of $\mathcal P$ not touching (the boundary of) $S$. If the distance between $v$ and $t$ is smaller than $d$, then $t$ sees $p$. Hence, even in the open guard model, $\mathcal P$ is completely guarded.

Recall that we placed guards only on edges of $\mathcal P$ that are not orthogonal to $x$. In particular, referring to Figure~\ref{fig6:6}, we picked:
\begin{itemize}
\item two smaller convex angles,
\item one bigger convex angle,
\item four smaller reflex angles, and
\item four bigger reflex angles.
\end{itemize}

We can repeat the same construction, rotated by different angles about $x$. Instead of casting a ray from $p$ parallel to $x\times y$, obtaining a point $q=p+\lambda(x\times y)$ with $\lambda\geqslant 0$, we could pick a negative lambda and rotate the entire construction by $180^\circ$, thus putting our guards on different edges. We could also pick $x\times z$ or $x\times w$ instead of $x\times y$, giving rise to four more different sets of guarding edges.

Because these six constructions correspond to all the six possible orientations of our species of angles, if we add up all the edges that we picked in every construction, we obtain:
\begin{itemize}
\item all smaller convex angles, twice;
\item all bigger convex angles, once;
\item all smaller reflex angles, four times;
\item all bigger reflex angles, four times.
\end{itemize}

Let $e_x$ and $r_x$ be the number of edges and reflex edges of $\mathcal P$, respectively, that are not parallel to $x$ (and similarly we define $e_y$, $r_y$, etc.). Then, the total number of edges listed above is not greater than  $2e_x+2r_x$. Indeed, this is the sum we get if we further add one copy of each bigger convex edge.

If we pick the least numerous of all the six choices, we place not more than
$$\left\lfloor\frac{2e_x+2r_x}{6}\right\rfloor = \left\lfloor\frac{e_x+r_x}{3}\right\rfloor$$
edge guards.

By the same reasoning, if we section $\mathcal P$ with planes orthogonal to $y$ instead of $x$, we obtain the upper bound
$$\left\lfloor\frac{e_y+r_y}{3}\right\rfloor,$$
and similarly for $z$ and $w$.

Among these four constructions, each edge type was considered exactly twice. For instance, the edges parallel to $x\times y$ can be assigned guards only when we section $\mathcal P$ with planes orthogonal to $z$ or to $w$. As a consequence, the smallest guarding set among these four contains at most
$$\left\lfloor\frac{e_x+e_y+e_z+e_w+r_x+r_y+r_z+r_w}{12}\right\rfloor = \left\lfloor\frac{2e+2r}{12}\right\rfloor = \left\lfloor\frac{e+r}{6}\right\rfloor$$
edge guards.

In order to place guards as described, it is sufficient to count the amount of edges of each of the 48 species, and then do a constant amount of sums and comparisons. This can be easily carried our in linear time.
\end{proof}

\begin{corollary}
Any open (resp.\ closed) 4-oriented polyhedron with $e$ edges is guardable by at most
$$\left\lfloor \frac{e}{3}\right\rfloor$$
open (resp.\ closed) edge guards. Guard locations can be computed in linear time.
\end{corollary}
\begin{proof}
Immediate by Theorem~\ref{t6:1} and the fact that $r\leqslant e$.
\end{proof}

\section{Alternative construction}

In this section we prove Theorem~\ref{t6:1} again, using a slightly different construction. Each of the two constructions performs better than the other on certain polyhedra, although they perform equally well in the respective worst case scenarios, in terms of $e$ and $r$ combined.

We start as in the previous section but, instead of choosing the edges corresponding to the 11 angles in Figure~\ref{fig6:5}, we choose those listed in Figure~\ref{fig6:7}.

\begin{figure}[h]
\centering
\includegraphics[scale=0.7]{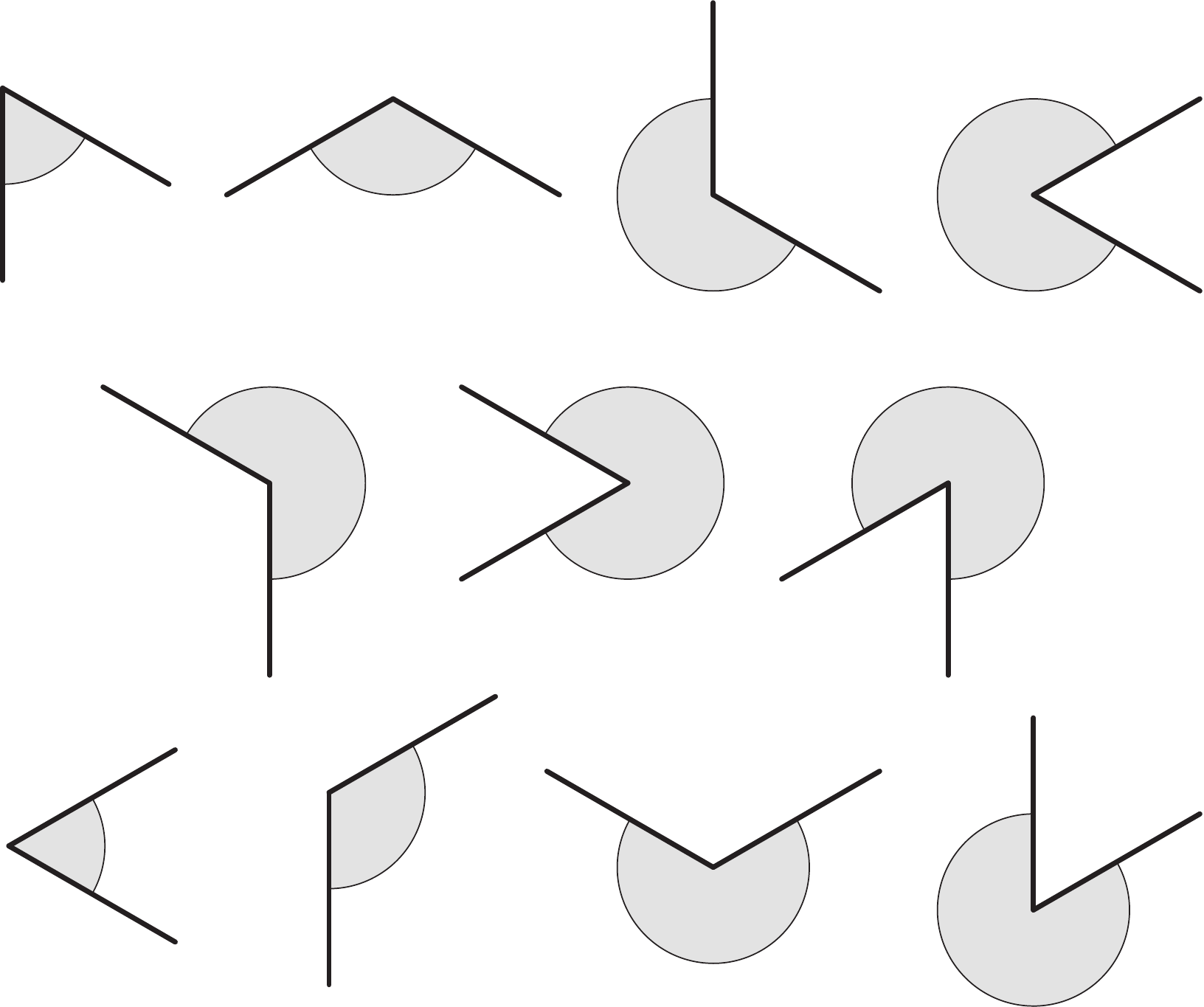}
\caption{Alternative choice of guards.}
\label{fig6:7}
\end{figure}

Now, using the same notation of the previous section, when we cast a ray from $p$ parallel to $x\times y$, two cases arise. If we find $q$ on an edge of $S$ parallel to $x\times z$, we argue as before. If, instead, the edge containing $q$ is parallel to $x\times w$, we do not repeat our argument on the other side of $pq$, but we stay on the same side. Referring to Figure~\ref{fig6:8}, we translate $pq$ (open at $p$ and closed at $q$) until we hit a vertex $v$. If $v$ is found on the edge of $S$ containing $q$, then we are in one of the first four cases, from (a) to (d). Otherwise, $v$ is hit by the interior of $pq$ as it translates, and the possible cases are those from (e) to (g). All the edges corresponding to these configurations have been selected to host a guard (see Figure~\ref{fig6:7}), hence $p$ is guarded.

\begin{figure}[h]
\centering
\subfigure[]{\includegraphics[scale=1]{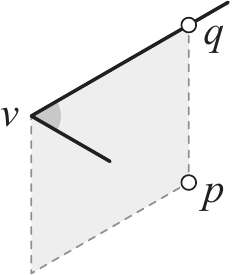}}\qquad\quad
\subfigure[]{\includegraphics[scale=1]{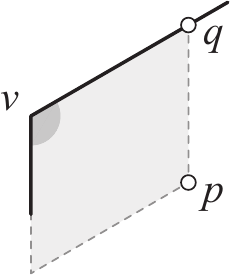}}\qquad\quad
\subfigure[]{\includegraphics[scale=1]{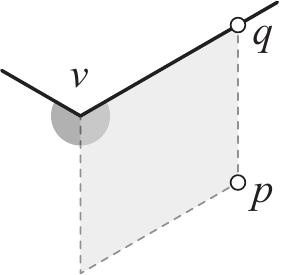}}\qquad\quad
\subfigure[]{\includegraphics[scale=1]{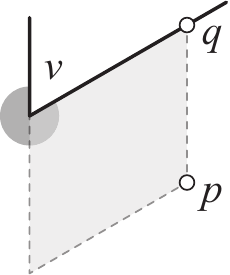}}\\ \vspace{0.2cm}
\subfigure[]{\includegraphics[scale=1]{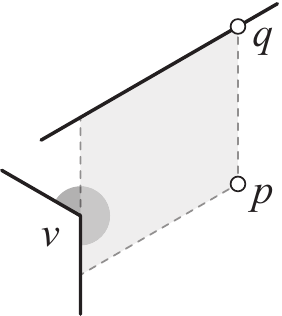}}\qquad\qquad
\subfigure[]{\includegraphics[scale=1]{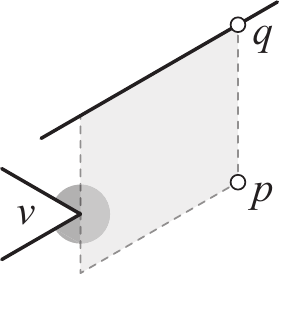}}\qquad\qquad
\subfigure[]{\includegraphics[scale=1]{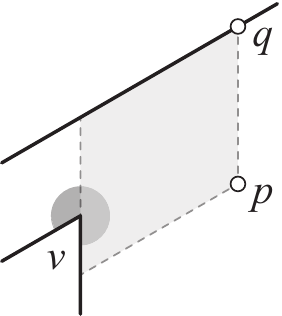}}
\caption{$p$ is necessarily guarded.}
\label{fig6:8}
\end{figure}

From this point on, we proceed as in the previous section. Note that this time we chose:
\begin{itemize}
\item two smaller convex angles,
\item two bigger convex angles,
\item four smaller reflex angles, and
\item three bigger reflex angles.
\end{itemize}

By adding up the contributions of the six possible orientations of our construction around $x$, we obtain:
\begin{itemize}
\item all smaller convex angles, twice;
\item all bigger convex angles, twice;
\item all smaller reflex angles, four times;
\item all bigger reflex angles, three times.
\end{itemize}

Now, if we add one more copy of each bigger reflex angle to this sum, we get $2e_x+2r_x$, as before. Hence, averaging this sum over the 24 possible orientations of our construction, we reach once again the bound given in Theorem~\ref{t6:1}, of
$$\left\lfloor \frac{e+r}{6}\right\rfloor$$
edge guards.

We remark that, in our first construction, we had to add some extra bigger convex angles to our sums, in order to obtain a ``clean'' upper bound formula. In our second construction, we added extra bigger reflex angles, instead. This leads to the following:

\begin{observation}
Despite having the same worst-case upper bound in terms of $e$ and $r$, our first construction is expected to outperform the second one if and only if the amount of ``bigger'' convex edges exceeds the amount of ``bigger'' reflex edges.
\end{observation}

\section{Lower bounds}

Now we give some lower bounds on the number of edge guards required to guard a 4-oriented polyhedron, in terms of $e$ and $r$.

Of course, our lower bound of
$$\left\lfloor \frac r 2\right\rfloor +1$$
guards, which holds for orthogonal polyhedra due to Figure~\ref{fig4:3}, also holds for 4-oriented polyhedra.

But in this more general scenario, we can provide better lower bounds in terms of $e$, and also in terms of $e$ and $r$ combined, as Figure~\ref{fig6:4} implies (cf.\ Figure~\ref{fig3:lower1}).

\begin{observation}
There are 4-oriented polyhedra with $e$ edges, of which $r$ are reflex, that require at least
$$\left\lfloor \frac{e}{6}\right\rfloor-1 = \left\lfloor \frac{e}{12}+\frac{r}{6}\right\rfloor$$
(reflex) edge guards, for unboundedly large values of $e$ and $r$.
\end{observation}

\begin{figure}[h]
\centering
\includegraphics[width=0.6\linewidth]{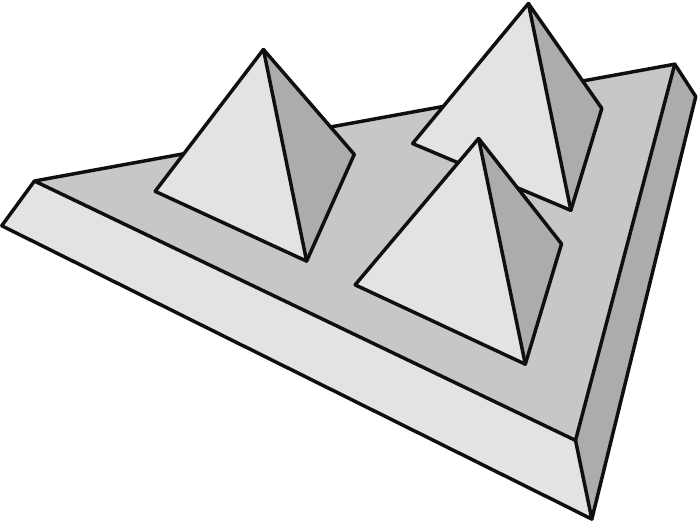}
\caption{4-oriented polyhedron with $e=6k+9$ and $r=3k$ that requires $k$ (reflex) edge guards.}
\label{fig6:4}
\end{figure}

We believe that such lower bounds are asymptotically tight, and that guards may always be chosen to lie on reflex edges.

\begin{conjecture}
Any 4-oriented polyhedron with $e$ edges can be guarded by at most
$$\frac{e}{6}+O(1)$$
reflex edge guards. If $r$ is the number of reflex edges, the same polyhedron can be also guarded by at most
$$\frac{e}{12}+\frac{r}{6}+O(1)$$
reflex edge guards.
\end{conjecture}

\part{Searching polyhedra}\label{part3}
\chapter{Model definition}\label{chapter7}
\begin{chapterabstract}
We give a thorough definition of the \TSSPext. In our model, guards are 1-dimensional segments that rotate a half-plane with one degree of freedom. We motivate our model choice with observations and examples, also discussing some alternative models.

Notably, we introduce the notion of filling guard, which will be central in the next chapter.
\end{chapterabstract}

\section{Guards and searchplanes}

We start by defining our notion of guard for the \TSSPext. Later, we will motivate our choice and discuss some alternative models. In the terminology of Chapter~\ref{chapter3}, we will employ open segment guards, not necessarily lying on an edge.

\begin{definition}[guard]\label{guard}A \emph{guard} in a closed polyhedron $\mathcal P$ is a positive-length straight line segment minus its endpoints, lying in $\mathcal P$. A \emph{boundary guard} is a guard that lies entirely on $\mathcal P$'s boundary. An \emph{edge guard} is a guard that coincides with the relative interior of an edge of $\mathcal P$.\end{definition}

\begin{remark}
For the \TSSPext, we consider only \emph{closed} polyhedra and \emph{open} guards. If we take \emph{open} polyhedra instead, then some of our results do not hold anymore, notably the one-way sweep strategy of Chapter~\ref{chapter8}. In order to obtain similar results with open polyhedra, we would have to model intruders as balls with positive radius, as opposed to dimensionless points.
\end{remark}

Throughout Part~\ref{part3}, we will be concerned mainly with boundary guards. Unlike the situation in the planar \SSPext, boundary guards already yield a rich and diverse theory.

In Definition~\ref{guard}, we exclude endpoints because we do not want guards to see beyond reflex edges or non-convex vertices, as the next definitions will clarify.

Recall that $\mathcal V(\ell)$,the visibility region of a guard $\ell$, is the set of points that are visible to $\ell$ (see Definition~\ref{d3:visregion}).

\begin{definition}[searchplane]\label{def:searchplane}A \emph{searchplane} of a guard $\ell$ is the intersection between $\mathcal V(\ell)$ and any (topologically closed) half-plane whose bounding line contains~$\ell$.\end{definition}

Consequently, every guard has a searchplane for every possible direction of the half-plane generating it, and the union of a guard's searchplanes coincides with its visibility region.

\begin{figure}[h]
\centering
\includegraphics[scale=0.8]{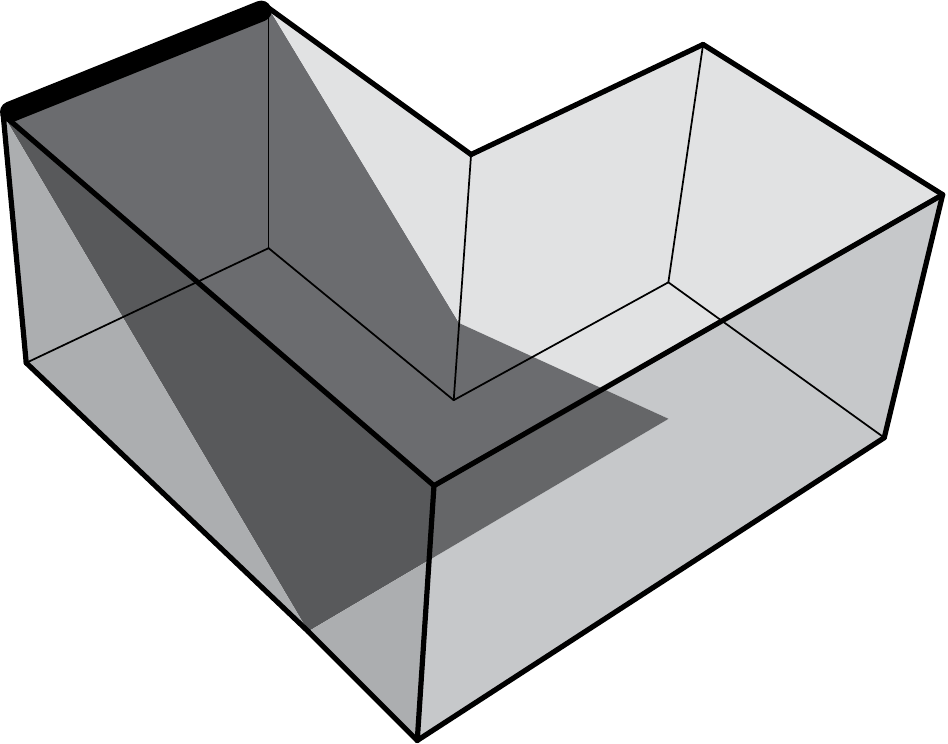}
\caption{Edge guard with one of its searchplanes, depicted as a thick line and a dark surface, respectively.}
\label{fig:3}
\end{figure}

\begin{remark}
Throughout Part~\ref{part3}, we will always use the term ``searchplane'' with the meaning of Definition~\ref{def:searchplane}, i.e., a static planar area visible to some guard. Nonetheless, we will keep using the term ``searchlight'' informally, referring to the tool that guards use to scan the environment. Thus, guards can turn searchlights around, but cannot turn searchplanes. While a searchplane is a well-defined mathematical object, a searchlight is not, in our terminology (in contrast with the standard terminology of 2-dimensional \SSP).
\end{remark}

\paragraph{Blind directions.}

If a searchplane is just a line segment, it is said to be \emph{trivial}, and the corresponding direction is said to be \emph{blind} for its guard. A guard has blind directions if and only if it is a boundary guard. We arbitrarily define a \emph{left} and \emph{right} side for each boundary guard, and we call \emph{leftmost position} the leftmost non-blind direction, for each boundary guard. Similarly, we define the \emph{rightmost position} of every boundary guard. Observe that the leftmost and rightmost positions are well-defined, because the polyhedron is a closed set, and every direction aiming straight at its exterior is blind for a boundary guard, even if the endpoints of (the topological closure of) the guard lie on reflex edges or vertices. This is because we did not include endpoints in Definition~\ref{guard}, a choice motivated also by Theorem~\ref{one}. Conversely, every other direction is not blind, because the corresponding searchplanes must contain a plane neighborhood of some point.

\section{Search schedules}

\begin{definition}[schedule]A \emph{schedule} for a guard $\ell$ is a continuous function $f_\ell: [0,T] \to S^1$, where $T\in \R^+$ and $S^1$ is the unit circle.\end{definition}

Intuitively, $f_\ell(t)$ expresses the orientation of the guard at time $t\in [0,T]$, which is the angle at which $\ell$ is aiming its searchlight. In other words, $\ell$ is able to emit a half-plane of light in any desired direction, and to rotate it continuously about the axis defined by $\ell$ itself. We will say that, at time $t$, $\ell$ is \emph{aiming its searchlight} at point $x$ if the orientation expressed by $f_\ell(t)$ corresponds to a searchplane of $\ell$ containing $x$ (assuming that one exists).

For the following definitions, we stipulate that a polyhedron $\mathcal P$ is given, along with a finite \emph{multiset} of guards, each of which is provided with a schedule. 

\begin{remark}
The reason why we use multisets rather than sets is that some guards may be coincident, yet distinct. This is the same as providing guards with unique identifiers, and also naturally models an analogous of the $k$-searcher of~\cite{search2}, i.e., a guard carrying $k$ independent searchlights. However, we will need multisets of guards only to prove Theorem~\ref{speed1} and Corollary~\ref{speed2}, which are of independent interest and may be safely disregarded without invalidating any other result in Part~\ref{part3}.
\end{remark}

\begin{definition}[illuminated point]A point is \emph{illuminated} at a given time if some guard is aiming its searchlight at it.\end{definition}

\begin{definition}[contaminated point, clear point]A point $x$ is \emph{contaminated} at time $t$, with respect to a given schedule, if there exists a continuous function $h: [0, t] \to \mathcal P$ such that $h(t)=x$ and there is no time $t'\in [0, t]$ at which $h(t')$ is illuminated. A point that is not contaminated is said to be \emph{clear}.\end{definition}

It follows that, at any time in a schedule, a maximal connected region of $\mathcal P$ without illuminated points is either all clear or all contaminated.

\begin{definition}[search schedule]A set of schedules of the form $f_\ell: [0,T] \to S^1$, where $\ell$ ranges over a finite guard multiset in a polyhedron $\mathcal P$, is a \emph{search schedule} if every point in $\mathcal P$ is clear at time $T$.\end{definition}

\begin{sloppypar}Next we define the \TSSPext\ (\TSSP).\end{sloppypar}

\begin{definition}[\TSSP]\emph{\TSSP}\ is the problem of deciding if a given multiset of guards in a given polyhedron has a search schedule.\end{definition}

An instance of \TSSP\ is said to be \emph{searchable} or \emph{unsearchable}, depending on the existence of a search schedule for its guards. Since an instance is trivially unsearchable if its guards cannot see the whole polyhedron, it is worth singling out those instances.

\begin{definition}[guard-visible instance]\label{viable}An instance of \TSSP\ is \emph{guard-visible} if every point of the polyhedron belongs to the visibility region of at least one guard.\end{definition}

\section{Filling guards}

Finally, a relevant role is played by a special type of (boundary) guard.

\begin{definition}[filling searchplane]A searchplane of a guard in a polyhedron $\mathcal P$ is \emph{filling} if it is a closed set whose relative boundary lies entirely on $\mathcal P$'s boundary.\end{definition}

Consider a guard $\ell$ with a non-trivial searchplane $S$, and let $\alpha$ be the plane containing $S$. Then $S$ is filling if and only if it coincides with the connected component of $\alpha \cap \mathcal P$ containing $\ell$.

Notice that the searchplane depicted in Figure~\ref{fig:3} is not filling, because one of its edges lies in the interior of the polyhedron (let us call this edge $e$). In addition, this searchplane is not even a closed set: recall that guards have no endpoints, so edge $e$ is not actually part of the searchplane (except for one endpoint). 

\begin{definition}[filling guard]\label{def:exhaustive}A guard is \emph{filling} if all its searchplanes are filling.\end{definition}

Because searchplanes are induced by half-planes, it follows that only boundary guards can be filling. Intuitively, a filling guard is similar to a traditional boundary guard from \SSP\ in simple polygons, in that its searchlight provides at any time an effective barrier which cannot be crossed by the intruder just by walking past its borders. The importance of such guards in developing search algorithms will be elucidated in Chapter~\ref{chapter8}.

\section{Motivations and alternative models}

Now that we have given a complete definition of our 3-dimensional searchlight model, we can discuss some of its basic features, and compare it with other plausible models.

First of all, it is obvious that no finite set of 1-dimensional searchlights can ever capture a dimensionless intruder in a polyhedron. Hence, our searchlights must be at least 2-dimensional to make any sense theoretically. Moreover, our model should achieve a good tradeoff among realism, theoretical appeal, and ease of manipulation.

We can imagine each of our guards (as previously defined in this chapter) to be an array of sensors arranged on a support that is able to rotate, say, horizontally. Each sensor is a laser beam that casts a ray and continuously turns it up and down at high speed. As a result, the support can reorient itself to aim at any desired searchplane, and the sensors can detect any intruder lying in that searchplane, provided that they turn fast enough. Implementing such a structure is in general quite demanding. However, note that it is also very scalable: depending on the application and the size of the intruder, fewer sensors can be placed on a single rotating support.

On the theoretical side, we want our model to be as close as possible to the classic 2-dimensional one, so that it can hopefully provide some insights also on planar \SSP. At the very least, we want \SSP to be a proper subproblem of \TSSP, and trivially so. This is indeed the case:

\begin{proposition}$\SSP\preceq_{\LOGSPACE} \TSSP$.\label{p7:redu}\end{proposition}
\begin{proof}Any polygon can be extruded to a prism, while each point guard can be transformed into a segment guard by stretching it parallel to the prism's sides. The resulting construction has a search schedule if and only if the original one does.\end{proof}

The above proof is enabled by the fact that our searchlights rotate with only one degree of freedom, which, incidentally, also makes for a simple description and analysis of their schedules.

Because of this, any hardness result for \SSP also holds for \TSSP, and any algorithm that works for \TSSP also works for \SSP. On the other hand, a hardness result for \TSSP can perhaps serve as an intermediate step for an analogous hardness result for \SSP. Indeed, the expressiveness of polyhedral structures may give us enough freedom to let us focus on the ``big picture'' of the reduction, which can later be refined and perhaps embedded in the plane. An example is found in Chapter~\ref{chapter11}, where we will deal with the complexity of the \PSSPext, first proving that the 3-dimensional version is \PSPACE-hard, and then modifying our construction to work also for the 2-dimensional version of the problem.

Likewise, some algorithms or techniques that work for \SSP may have very natural generalizations in \TSSP, and once again the simple reduction given in Proposition~\ref{p7:redu} may make their discovery easier. In the next chapter we will give some concrete examples of this general principle.

We point out that not for every reasonable model of searchlights the reduction from \SSP to \TSSP is immediate. As an example, consider a point guard carrying a flashlight that can be rotated with two degrees of freedom. The flashlight emits a cone of light of constant width. This is a very realistic model, whose 2-dimensional version has been proposed and studied by Obermayer et al. in~\cite{bullo}. However, while the planar model still resembles the original one, and actually generalizes it (1-dimensional searchlights are flashlights whose cone of light has null width), this is not the case for the 3-dimensional model. In fact, the relationship between searching polyhedra with 3-dimensional flashlights and searching polygons with 1-dimensional searchlights is not clear at all, and the first problem may even be computationally easier than the second one.

Furthermore, one could consider a ``hybrid'' model, in which guards rotate dihedral angles about an axis, as opposed to half-planes of light. In other terms, we could add another parameter to our model, which is the width of the beam of light emanating from each searchlight. To our understanding, such a modification shifts the problem from the ``topological side'' to the ``metric side''. Recall that the main results for classic \SSP were all expressed in terms of visibility graphs among guards, and the presence of bounday guards. In a sense, these are all topological aspects of the problem. In contrast, as hinted in~\cite{bullo} for the planar case, the model with positive-width searchlights becomes very sensitive to the actual distances between elements, such as the width of certain ``bottlenecks'' in the polygon, etc. We think this addition, despite making the model more general, deprives it of much of its appeal from the theoretical point of view. For our research, we prefer to keep our model simple and our results more dependent on the topology of the problem.

\section{Counterexamples}\label{basic}

The most noteworthy aspect of our guard model is that searchplanes of boundary guards, as opposed to searchlight rays emanating from boundary guards in \SSP, may fail to disconnect a polyhedron when aimed at its interior, regardless of its genus. As it turns out, this is the main reason why \TSSP\ seems harder than \SSP, in that exploiting such a property will enable the relatively simple \NP-hardness proofs in Chapter~\ref{chapter10}, as well as the construction of several counterexamples to positive statements about \SSP.

\begin{figure}[h]
\centering
\includegraphics[width=0.5\linewidth]{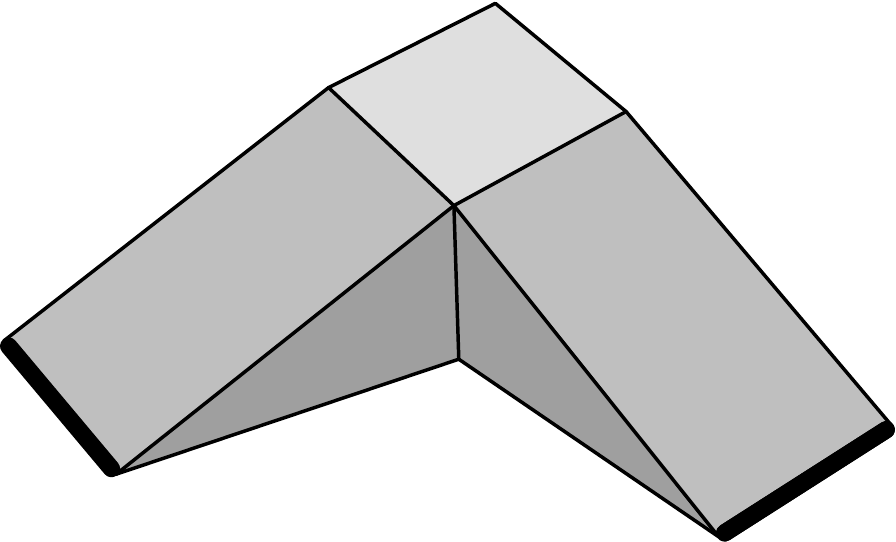}
\caption{Searchable instance of \TSSP\ with no sequential search schedule. Thick lines mark guards.}
\label{fig:4}
\end{figure}

For example, the reduction of the search space to \emph{sequential schedules} (i.e., schedules in which the guards sweep in turns) given by Obermeyer et al.\ in~\cite{bullo} is no longer possible. Figure~\ref{fig:4} shows an instance of \TSSP\ whose two guards are forced to turn their searchlights simultaneously, or else they would create gaps in the illuminated surface which would result in the recontamination of the whole polyhedron.

Moreover, in spite of the searchability of all simply connected \SSP\ instances whose guards lie on the boundary and collectively see the whole polygon (see~\cite{search1}), it is easy to construct simply connected guard-visible but unsearchable instances of \TSSP\ with only boundary guards, such as those in Figure~\ref{fig:5}.

\begin{figure}[h]
\centering
\includegraphics[scale=0.75]{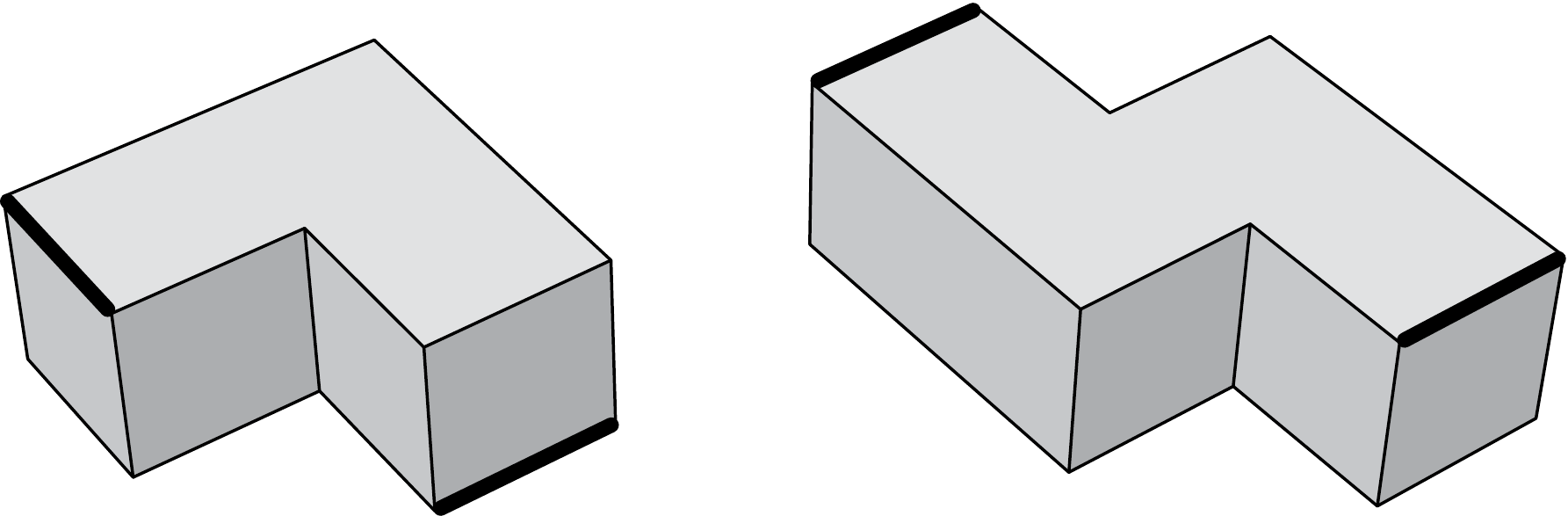}
\caption{Two unsearchable guard-visible simply connected instances of \TSSP\ whose guards lie on the boundary.}
\label{fig:5}
\end{figure}

Indeed, whenever the two guards attempt to clear the center (in either of these two instances), they fail to disconnect the polyhedron, since their searchplanes are not coplanar, which results in the recontamination of the entire instance. In Chapter~\ref{chapter10} we will provide more sophisticated unsearchable but guard-visible instances of \TSSP\ with only boundary guards.
\chapter{Searching with filling guards}\label{chapter8}
\begin{chapterabstract}
We investigate the properties of filling guards, as introduced in the previous chapter.

We point out how filling guards act as boundary guards in the planar \SSPext, in that both the one-way sweep strategy and the sequentiability of search schedules generalize to 3-dimensional problem instances with only this type of guards.

Moreover, we characterize the searchable polyhedra containing just one guard, showing that such guard must be filling.

Finally, we sketch a polynomial time algorithm to determine whether a guard is filling, implying that the hypotheses of our previous theorems are practically checkable.
\end{chapterabstract}

\section{Filling guards}

Notice that all the counterexamples described in Section~\ref{basic} employ boundary guards that are not filling. Conversely, it comes as no surprise that employing only filling guards yields positive results. To see why, let us give a characterization of the different shapes a searchplane can take with respect to the surrounding polyhedral environment. The topological closure of a searchplane is always a polygon, perhaps with holes, perhaps with some additional segments sticking out radially, and the whole searchplane is visible to some line segment lying on its external boundary, which would be the guard generating it (refer to Figure~\ref{fig:6}). There may be intersections between a searchplane's relative interior and the polyhedral boundary, which could be collections of polygons, straight line segments, and isolated points (recall that polyhedra in our model are closed sets). But what is central for our purposes is the searchplane's relative boundary, which may entirely lie on the polyhedron's boundary, or may not. If it does, and the searchplane is a closed set, then the only way an intruder could travel from one side of the searchplane to the other, without crossing the light and being caught, would be to take a detour through a suitable handle of the polyhedron. In particular, in genus-zero polyhedra, that would be impossible. In other words, any searchplane emanating from a guard aiming its searchlight at the interior of a genus-zero polyhedron disconnects the polyhedron if and only if the searchplane is filling. Thus, any filling guard aiming its searchlight at the interior of a genus-zero polyhedron always disconnects it.

\begin{figure}[h]
\centering
\subfigure[]{\label{fig:6a}\includegraphics[width=0.4\linewidth]{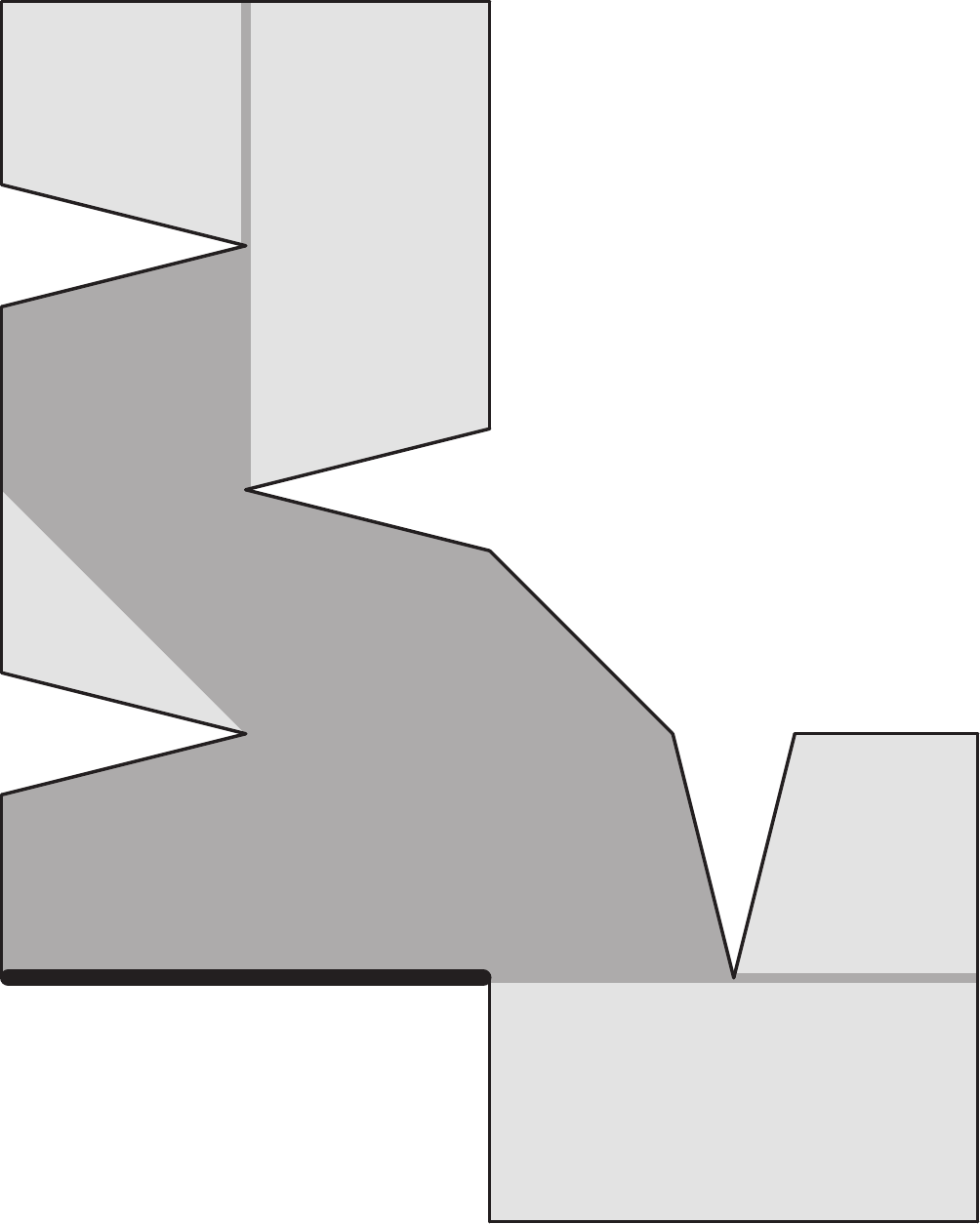}}\qquad \qquad
\subfigure[]{\label{fig:6b}\includegraphics[width=0.4\linewidth]{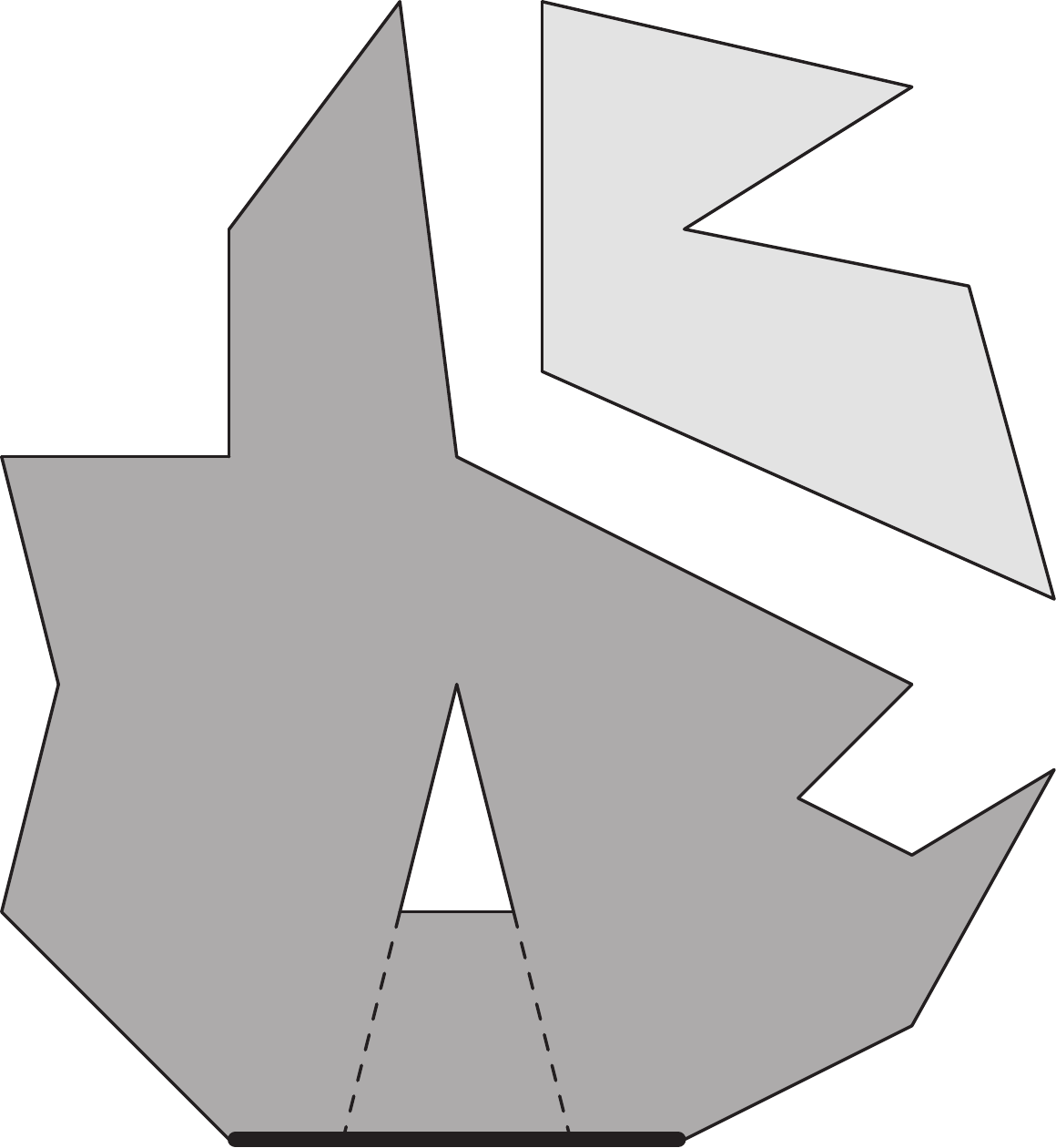}}
\caption{Two sections of polyhedra, with searchplanes represented as dark regions. The searchplane in \subref{fig:6a} has two dangling segments, while the searchplane in \subref{fig:6b} is filling but not simply connected.}
\label{fig:6}
\end{figure}

On the other hand, if a guard $\ell\subset\mathcal P$ is filling, the topological closure of $\mathcal V(\ell)$ is always a polyhedron, perhaps with some dangling polygons (which may originate from dangling segments of single searchplanes, such as those in Figure~\ref{fig:6a}). Of course, the boundary of $\mathcal V(\ell)$ may contain polygons that do not lie on $\mathcal P$'s boundary. But, because $\ell$ is filling, every such polygon is entirely contained in some searchplane of $\ell$.

\begin{definition}[critical searchplane, critical position]\label{def:critical}A searchplane $S$ of a filling guard $\ell$ is \emph{critical} if it shares a polygon with the boundary of $\mathcal V(\ell)$ that does not lie on $\mathcal P$'s boundary. Whenever $\ell$ is aimed at $S$, it is in a \emph{critical position}.\end{definition}

Every filling guard $\ell$ has only a finite number of critical searchplanes, because $\mathcal V(\ell)$'s surface is made of finitely many polygons. Equivalently, $\ell$ has a finite number of critical positions.

Now, as a filling guard $\ell$ in a genus-zero polyhedron starts turning from its leftmost position toward the right, every point that it illuminates will remain clear forever, unless the illuminated searchplane becomes tangent to some region of the polyhedron that is not in $\mathcal V(\ell)$ and that would be responsible for recontamination, once the tangency is crossed by the searchlight. This happens exactly when $\ell$ reaches a critical position.

\section{One-way sweeping}

We now have the tools required to generalize the one-way sweep strategy for guards in simple polygons (outlined by Sugihara et al.\ in~\cite{search1}, and in Chapter~\ref{chapter1}) to work with filling guards in simply connected polyhedra.

\begin{theorem}\label{exhaustive}Every guard-visible genus-zero instance of the \TSSPext\ is searchable if its guards are filling.\end{theorem}

\begin{proof}We first sketch a search schedule before detailing it further. Select any guard $\ell$ and start turning it rightward from its leftmost position. As soon as it reaches a critical position, it means that some \emph{subpolyhedron} $\mathcal R \subset \mathcal P$ has been encountered that is invisible to $\ell$. So stop turning $\ell$ and select another guard to continue the job. Proceed recursively until $\mathcal R$ is clear, then turn $\ell$ rightward again, stopping at every critical position, until the entire polyhedron is clear.

At any time, there is a clear region of $\mathcal P$ that is steadily growing, and a \emph{semiconvex subpolyhedron} $\mathcal R$ \emph{supported} by a set of guards $L$ that is being cleared by some guard not in $L$, while the guards in $L$ hold their searchlights fixed. Intuitively, the term \textquotedblleft semiconvex\textquotedblright\ is used because the only points of non-convexity of such a polyhedron lie on $\mathcal P$'s boundary. For a formal definition of \emph{semiconvex subpolygon} and \emph{support}, refer to~\cite{bullo2} or~\cite{search1}. Extending these definitions to polyhedra is straightforward, since we are considering only filling guards in genus-zero polyhedra. Thus, part of the boundary of $\mathcal R$ coincides with $\mathcal P$'s boundary, while the remaining part is determined by the searchplanes of the guards in $L$. Moreover, all the guards in $L$ are in a critical position, waiting for $\mathcal R$ (or some larger semiconvex subpolyhedron of $\mathcal P$, supported by a subset of $L$) to be cleared. It follows that none of the guards in $L$ can see any point in the interior of $\mathcal R$.

Hence, there must be some guards not in $L$ that can see an internal portion of $\mathcal R$, otherwise the problem instance would not be guard-visible. We select one of them, say $\ell'$, and start sweeping $\mathcal R$ from left to right (according to $\ell'$'s notion of left and right). Notice that a searchplane bounding $\mathcal R$ could have holes, and thus $\mathcal R$ itself could have strictly positive genus. But that does not affect our invariants, because every searchplane of $\ell'$ passing through $\mathcal R$'s interior still disconnects it, or it would not even disconnect $\mathcal P$. As a consequence, the points in $\mathcal R$ that are illuminated by $\ell'$ never get recontaminated as $\ell'$ continues its sweep.

Again, whenever $\ell'$ reaches a critical position, it stops there until the semiconvex subpolyhedron $\mathcal R' \subset \mathcal R$ supported by $L\cup \{\ell'\}$ has been cleared by some other guard, and so on recursively. Since every guard has only a finite number of critical positions, eventually $\mathcal P$ gets cleared.\end{proof}

Remarkably enough, the core argument supporting the planar one-way sweep strategy of Sugihara et al.\ (see~\cite{search1}) applies also to our polyhedral model, where there is no well-defined global concept of clockwise rotation.

\section{Sequentiality}

In addition to this, a version of the main result of Obermeyer et al.\ (see~\cite{bullo}) also extends to instances of \TSSP\ with filling guards. Assuming that all the guards are filling, we say that a search schedule is \emph{critical} and \emph{sequential} if at most one guard is turning at any given time, and guards stop or change direction only at critical positions, or at their leftmost or rightmost positions.

\begin{corollary}\label{sequentiality}Every searchable genus-zero instance of the \TSSPext\ whose guards are filling has a critical and sequential search schedule.\end{corollary}

\begin{proof}Since the instance is searchable, it must be guard-visible. Thus, the search schedule detailed in the proof of Theorem~\ref{exhaustive} applies, which is indeed critical and sequential.\end{proof}

\section{Searching with one guard}

As a further application of the concept of filling guard, we characterize the searchable instances of \TSSP\ having just one guard. We include polyhedra of any genus, not just zero.

\begin{theorem}\label{one}An instance of the \TSSPext\ (of any genus) with just one guard is searchable if and only if it is guard-visible and the only guard is a boundary guard.\end{theorem}

\begin{proof}Both conditions are clearly necessary, for the same reasons why they are necessary in 2-dimensional \SSP (see~\cite{search1} and Chapter~\ref{chapter1}). Specifically, if the guard has a point that does not lie on the polyhedron's boundary, then no small-enough ball centered at that point can ever be cleared.

Conversely, since the instance is guard-visible, the visibility region of its only guard $\ell$ coincides with the whole polyhedron $\mathcal P$. Let $\ell'$ be the maximal straight line segment contained in $\mathcal P$ and containing $\ell$. Then $\ell'$ entirely belongs to the boundary of $\mathcal P$. Indeed, if a point $x\in \ell'$ lay strictly in the interior of $\mathcal P$, then also some neighborhood of $x$ would. Recall that $\ell$ has a range of blind directions past its leftmost and rightmost position, where all its searchplanes degenerate to the single line $\ell'$. In all those directions, part of the neighborhood of $x$ would lie outside $\mathcal V(\ell)$, contradicting the guard-visibility of the instance.

\begin{figure}[h]
\centering
\includegraphics[scale=0.5]{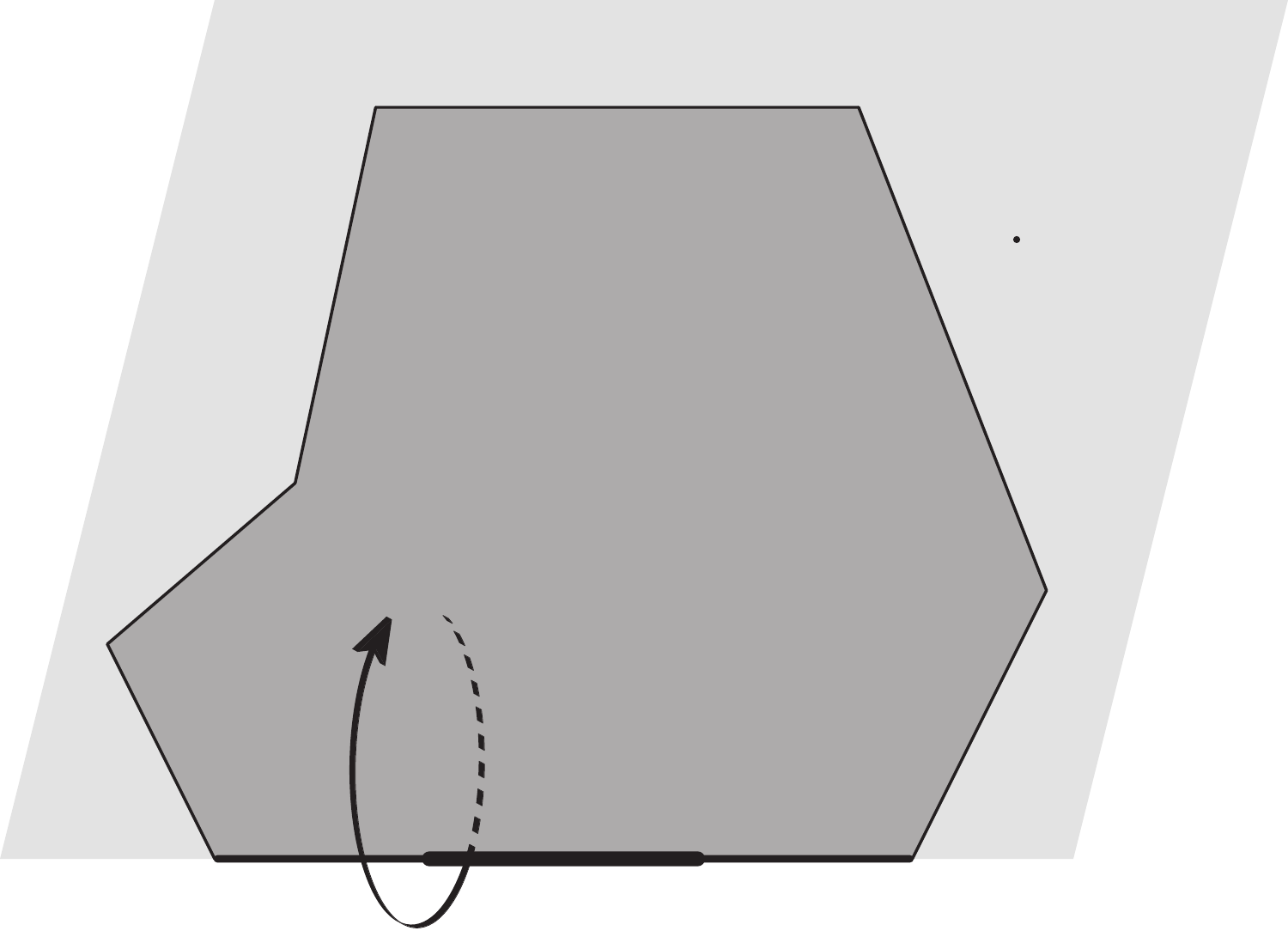}
\put(-174,146){$\alpha$}
\put(-42,119){$x$}
\put(-122,1){$\ell$}
\put(-82.5,1){$\ell'$}
\large
\put(-114.5,72){$S$}
\caption{Illustration of the proof of Theorem~\ref{one}.}
\label{fig:n1}
\end{figure}
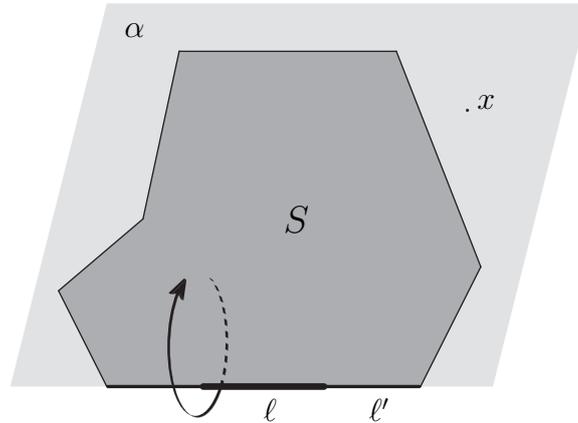

Consider a searchplane $S$, lying on a half-plane $\alpha$ whose bounding line contains $\ell$ (refer to Figure~\ref{fig:n1}). If a point $x \in \alpha \cap \mathcal P$ lay outside $S$, then $x$ would necessarily be covered by another searchplane, because $\mathcal V(\ell)=\mathcal P$. But the only points in $\alpha$ that could lie on a searchplane different from $S$ would be those in $\ell'$, because any two searchplanes of $\ell$ intersect just on $\ell'$. Nonetheless, $\ell'$ belongs to $S$ too, which yields a contradiction. Hence $S=\alpha \cap \mathcal P$ and also $\ell'$ lies on the boundary of $\mathcal P$, implying that $\ell$ is a filling guard with no critical positions. Moreover, $S$ disconnects $\mathcal P$ if it intersects its interior. Suppose by contradiction that an intruder could walk from one side of $S$ to the other side. Since $S=\alpha \cap \mathcal P$, the intruder would necessarily have to take the long route around $\ell'$ (i.e., opposite to $\alpha$, see Figure~\ref{fig:n1}), and by doing so it would cross all the half-planes in the blind directions of $\ell$. But the only points of $\mathcal P$ that belong to such half-planes are those in $\ell'$, so the intruder would be caught in any case.

It follows that turning $\ell$ from its leftmost position to its rightmost position produces a search schedule.\end{proof}

Notice that, had we included endpoints in Definition~\ref{guard}, the statement of Theorem~\ref{one} would have been false. Indeed, the guard-visibility assumption would have been satisfied by more \TSSP\ instances, including unsearchable ones. For example, consider a cube with a \textquotedblleft pyramidal well\textquotedblright\ pointing inside, and an edge guard on a reflex edge. The entire polyhedron is visible to the pyramid's vertex, which is an endpoint of the guard. However, the guard cannot search the polyhedron, because no searchplane disconnects it, in either guard model---with or without endpoints.

\section{Checking for fillingness}

We conclude this chapter by sketching an argument supporting the claim that the conditions of Theorem~\ref{exhaustive} are practically checkable.

\begin{proposition}\label{efficient}Whether a guard $\ell\subset\mathcal P$ is filling can be decided in time polynomial in the size of $\mathcal P$.\end{proposition}

\begin{proof}First of all, if $\ell$ is not a boundary guard, then it cannot be filling and may be discarded.

Then, observe that the fillingness of a searchplane $S$ generated by a half-plane $\alpha$ can be efficiently checked by computing the polygonal section $P=\alpha \cap \mathcal P$, and then computing the boundary of $S$ by drawing lines through every pair of vertices of $P$. Once the boundary of $S$ is known, its fillingness can be checked easily.

Now, determining if $\ell$ is filling can be carried out by inspecting just a polynomial number of its searchplanes. For every face $F\subset \mathcal P$ and every edge $e\subset \mathcal P$ that is not parallel to $F$, we call the point of intersection between the plane containing $F$ and the line containing $e$ a \emph{critical point} (not to be confused with ``critical position'', see Definition~\ref{def:critical}). In particular, every vertex of $\mathcal P$ is a critical point. Imagine turning the searchlight of $\ell$ from its leftmost position to the rightmost: it is straightforward to see that the fillingness of the illuminated searchplane can change only when the searchlight crosses a critical point. Hence, it suffices to check the fillingness of every searchplane corresponding to a critical point, plus one searchplane for each interval between two consecutive critical points. The number of searchplanes to check is thus polynomial.\end{proof}

Making the above proposition into an optimal algorithm goes beyond our scopes. However, the interested reader can readily see that even a naive implementation yields a time complexity of $O(n^4)$, regardless of the actual data structures used to represent polyhedra.
\chapter{Guard placing strategies}\label{chapter9}
\begin{chapterabstract}
We tackle the problem of efficiently placing boundary guards in a given polyhedron, in order to make it searchable.

Our strategy for general polyhedra with $r>0$ reflex edges employs $r^2$ boundary guards, which we believe to be asymptotically suboptimal. However, we prove that $r$ reflex edge guards are sufficient for orthogonal polyhedra.

A related goal is to minimize the search time, assuming that each guard's angular speed is bounded by some constant. We show that, in our previous constructions, it is possible to trade guards for search time, to some extent.
\end{chapterabstract}

\section{Preliminaries}

In this chapter we attempt to place boundary guards in a given polyhedron in order to make it searchable. We aim at proving the following.

\begin{conjecture}\label{con9:1}
Any non-convex polyhedron with a guard on each reflex edge is searchable.
\end{conjecture}

Although we only manage to prove it for orthogonal polyhedra, we also obtain a quadratic upper bound on the number of guards for general polyhedra.

A complementary goal is to minimize the \emph{search time}, which is the total time of a search schedule, assuming that the maximum angular speed of every guard is $2\pi~\mathrm{rad} / \mathrm{sec}$. We will show that our constructions allow to trade guards for search time, to some extent.

\begin{remark}
Formally, when considering search time, the set of a guard's legal schedules is restricted to Lipschitz continuous functions, whose Lipschitz constant corresponds to an angular speed of $2\pi~\mathrm{rad} / \mathrm{sec}$.
\end{remark}

The problem of minimizing the number of (boundary) guards required to search a given polyhedron is strongly \NP-hard, even for genus-zero orthogonal polyhedra. The problem is strongly \NP-hard also restricted to edge guards, or to edge guards lying on reflex edges. Indeed, our construction in Theorem~\ref{t3:hard} not only holds for the \ART, but also for \TSSP, due to the reduction in Proposition~\ref{p7:redu}, and the fact that point guards can search a given simple polygon if and only if they solve the \ART (see Chapter~\ref{chapter1} or~\cite{search1}).

\section{Searching general polyhedra}

\subsection*{Basic search strategy}

To begin with, we argue that $r^2+r$ boundary guards are sufficient to search any given polyhedron with $r>0$ reflex edges. Subsequently, we show how to lower this bound to $r^2$, at the expense of increasing search time.

\begin{theorem}\label{speed1}Any polyhedron with $r>0$ reflex edges can be searched in less than 1 second by at most $r^2+r$ suitably chosen boundary guards.\end{theorem}

\begin{proof}Let $\mathcal P$ be a non-convex polyhedron. We first partition it into convex regions $\mathcal C_i$, showing how this partition is induced by at most $r^2$ boundary guards. Next, we obtain a superpartition $\mathcal D_j$ with some better properties that we exploit to obtain a search schedule, using $r$ new boundary guards (a very similar partition was described by Chazelle in~\cite{polypart1}, under the name of \textquotedblleft revised naive decomposition\textquotedblright, see Chapter~\ref{chapter2}).

\paragraph{First partition.}

Let $e$ be a reflex edge of $\mathcal P$, and let $\alpha_e$ be a plane through $e$, close enough to its angle bisector, but not containing any vertex of $\mathcal P$ other than $e$'s endpoints. Let $\mathcal Q_e$ be the connected component of $\mathcal \alpha_e \cap \mathcal P$ containing $e$. We claim that $\mathcal Q_e$ is a polygon with at most $r-1$ reflex vertices, possibly with holes. Indeed, $e$ is an edge of $\mathcal Q_e$, and each reflex vertex of $\mathcal Q_e$ lies on a reflex edge of $\mathcal P$. Moreover, if an endpoint of $e$ is a reflex vertex of $\mathcal Q_e$, then it belongs at least to another reflex edge of $\mathcal P$, different from $e$. But $\alpha_e$ intersects every edge of $\mathcal P$ other than $e$ in at most one point (otherwise it would contain its endpoints, as well), hence there are at most $r-1$ reflex vertices of $\mathcal Q_e$, i.e., one for every reflex edge of $\mathcal P$ other than $e$.

By Lemma~\ref{edges}, $\mathcal Q_e$ can be guarded by at most $r$ open edge guards, one of which lies on $e$. Equivalently, $\mathcal Q_e$ can be completely illuminated by at most $r$ suitably placed boundary guards (with searchlights), lying on $\alpha_e$ and aimed parallel to it, one of which is an edge guard lying on $e$. By repeating the same construction with every other reflex edge, at most $r^2$ boundary guards are placed, and $\mathcal P$ is partitioned by illuminated searchplanes into convex polyhedra $\mathcal C_i$. We remark that, during our construction, previously placed searchplanes are not considered as part of $\mathcal P$'s boundary. Hence, every $\mathcal Q_{e_k}$ is indeed bounded by $\mathcal P$. This is because we need place guards on the boundary of $\mathcal P$, and not in its interior.

\paragraph{A coarser partition.}

Consider now a slightly different partition: proceed as above by drawing angle bisectors through reflex edges, but this time every previously drawn splitting polygon acts as a boundary. In other words, as soon as $\mathcal P$ splits into several subpolyhedra, we partition them recursively one by one, confining each split to just one subpolyhedron. So, if we consider some intermediate subpolyhedron $\mathcal P' \subseteq \mathcal P$ and select a reflex edge $e'\subset \mathcal P'$, we look at the reflex edge $e$ of $\mathcal P$ that contains $e'$, and let $\mathcal R_{e'}$ be the connected component of $\alpha_e \cap \mathcal P'$ (as opposed to $\alpha_e \cap \mathcal P$) containing $e'$. As a result, $\mathcal P$ is again partitioned into convex polyhedra $\mathcal D_j$, in such a way that every $\mathcal C_i$ is contained in some $\mathcal D_j$ (i.e., $\{\mathcal C_i\}$ is a \emph{refinement} of $\{\mathcal D_j\}$). Notice that, even though two different \emph{splitting polygons} $\mathcal R_{e'}$ and $\mathcal R_{e''}$ may correspond to the same reflex edge $e$ of $\mathcal P$ (because $e$ itself has been split by a previous cut), they are nonetheless coplanar, as they both belong to the same plane $\alpha_e$.

\begin{figure}[h]
\centering
\subfigure[$\{\mathcal C_i\}$]{\label{fig:n2a}\includegraphics[scale=0.75]{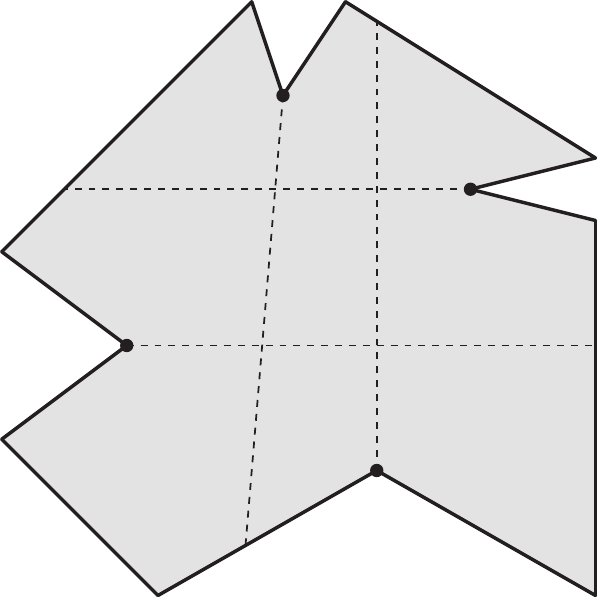}}\qquad \qquad \quad
\subfigure[$\{\mathcal D_j\}$]{\label{fig:n2b}\includegraphics[scale=0.75]{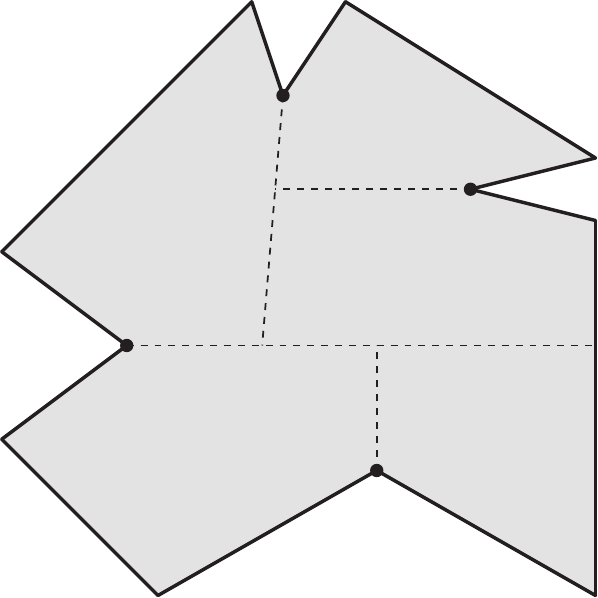}}
\caption{Comparison between partitions $\{\mathcal C_i\}$ and $\{\mathcal D_j\}$ in a section of a polyhedron. For simplicity, only the splitting planes corresponding to visible reflex edges are shown.}
\label{fig:n2}
\end{figure}

\paragraph{Search schedule.}

The point of having this coarser partition is that every $\mathcal D_j$ contains a subsegment of a reflex edge of $\mathcal P$ (see Figure~\ref{fig:n2}). This property can be checked by a straightforward induction on the construction steps.

Now we add a new edge guard on each reflex edge, and we turn these additional guards simultaneously, from their leftmost position to the rightmost. While this happens, the $r^2$ previously placed guards are never moved, so that the partition $\{\mathcal D_j\}$ is always preserved. Because each $\mathcal D_j$ is completely visible to some reflex edge, at the end of the schedule every $\mathcal D_j$ is successfully cleared.

Since every guard turns by less than $2\pi~\mathrm{rad}$, the search time is less than one second.\end{proof}

\subsection*{Improved search strategy}

We can slightly lower the number of guards used in the previous theorem, at the cost of increasing the search time. Next we show how to use only $r^2$ boundary guards, although the search time of our schedule could be quadratic in $r$, as well.

\begin{theorem}\label{heur}Any polyhedron with $r>0$ reflex edges can be searched by at most $r^2$ suitably chosen boundary guards.\end{theorem}

\begin{proof}
As in the proof of Theorem~\ref{speed1}, we partition the polyhedron $\mathcal P$ into convex regions $\mathcal D_j$ by placing at most $r^2$ boundary guards. Next we show that some of these guards can be turned in a certain order to clear every piece of the partition.

In the following, we use the same notation as in the proof of Theorem~\ref{speed1}.

\paragraph{Partition tree.}

During the splitting process of $\mathcal P$ that culminates in the partition $\{\mathcal D_j\}$, we build a tree whose nodes represent the intermediate subpolyhedra, and whose arcs are marked by the splitting polygons $\mathcal R_{e}$ (see Figure~\ref{fig:n3}). Thus, the root of the tree is $\mathcal P$ and its leaves are the $\mathcal D_j$'s. Every time we draw a splitting polygon $\mathcal R_{e}$ for a subpolyhedron $\mathcal P'$ corresponding to some node $v$ of the tree, we could either decrease the genus of $\mathcal P'$, or we could partition it into two subpolyhedra $\mathcal P'_1$ and $\mathcal P'_2$, one for each side of $\mathcal R_{e}$. In the first case we just attach to $v$ an arc labeled $\mathcal R_{e}$ with a child node labeled as the new polyhedron, say $\mathcal P''$. In the second case we attach to $v$ two arcs labeled $\mathcal R_{e}$, with two children nodes labeled $\mathcal P'_1$ and $\mathcal P'_2$. This structure is somewhat reminiscent of a Binary Space Partitioning tree. Again, as in the proof of Lemma~\ref{edges}, we have to slightly extend the notion of polyhedron, to include those with internal polygons acting as faces, resulting from non-disconnecting splits.

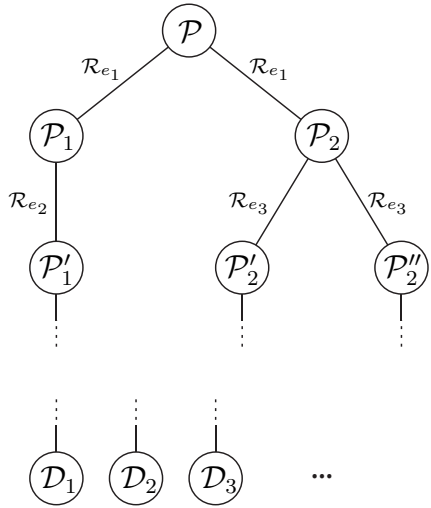
\begin{figure}[h]
\centering
\includegraphics[scale=1.1]{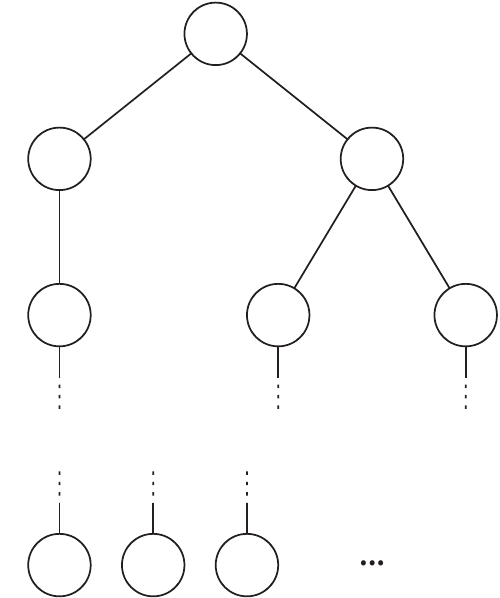}
\put(-94.5,175){$\mathcal P$}
\put(-146,136){$\mathcal P_1$}
\put(-46.5,136){$\mathcal P_2$}
\put(-146.5,86){$\mathcal P_1'$}
\put(-77,86){$\mathcal P_2'$}
\put(-17.5,86){$\mathcal P_2''$}
\put(-146,6.5){$\mathcal D_1$}
\put(-116,6.5){$\mathcal D_2$}
\put(-86,6.5){$\mathcal D_3$}
\put(-130.75,165){${}_{\mathcal R_{e_1}}$}
\put(-67.75,165){${}_{\mathcal R_{e_1}}$}
\put(-158,114.5){${}_{\mathcal R_{e_2}}$}
\put(-75.5,114.5){${}_{\mathcal R_{e_3}}$}
\put(-23.5,114.5){${}_{\mathcal R_{e_3}}$}
\caption{Sketch of a partition tree.}
\label{fig:n3}
\end{figure}

\paragraph{Search schedule.}

Next we describe a search schedule for the $r^2$ boundary guards we placed. Recall that each $\mathcal D_j$ contains a subsegment of a reflex edge of $\mathcal P$, and that on each such reflex edge lies one edge guard. We turn only \emph{some} of these $r$ guards, one by one, while all the other guards stay still. The order of activation of the guards is determined by the above partition tree, and the same guard could be activated more than once. The subpolyhedra of $\mathcal P$ are cleared recursively, following a depth-first walk of the partition tree, starting from the root. Every time a leaf is reached through an arc labeled $\mathcal R_{e}$, its corresponding $\mathcal D_j$ is swept by the edge guard lying on $e$. This is feasible because $e$ lies on the boundary of $\mathcal D_j$, which is convex. When $\mathcal D_j$ is clear, the guard moves back to its initial position and the depth-first walk proceeds.

\paragraph{Correctness.}

It remains to show that no \emph{significant} recontamination can occur among the $\mathcal D_j$'s while their bounding searchlights are rotated. Suppose the depth-first walk of the partition tree reaches a leaf labeled $\mathcal D_j$, and let $\ell_e$ be the edge guard whose duty is to clear $\mathcal D_j$. Perhaps the corresponding edge $e$ of $\mathcal P$ is divided several times by the partition, so let $E$ be the set of subsegments of $e$ whose corresponding splitting planes actually appear as labels of the edges of the partition tree. Let $e'\in E$ be the subsegment of $e$ that is also an edge of $\mathcal D_j$. On the other side of $\mathcal R_{e'}$ lies a subpolyhedron $\mathcal P'$, such that $\mathcal D_j$ and $\mathcal P'$ are represented by sibling nodes in the partition tree. In order to clear $\mathcal D_j$, $\ell_e$ turns its searchlight from $\mathcal R_{e'}$ to sweep over $\mathcal D_j$, and then back to $\mathcal R_{e'}$. Since the \emph{restriction} of $\ell_e$ to $\mathcal D_j$ is filling, no recontamination occurs between $\mathcal D_j$ and $\mathcal P'$ during this back-and-forth sweep  (see Figure~\ref{fig:8}).

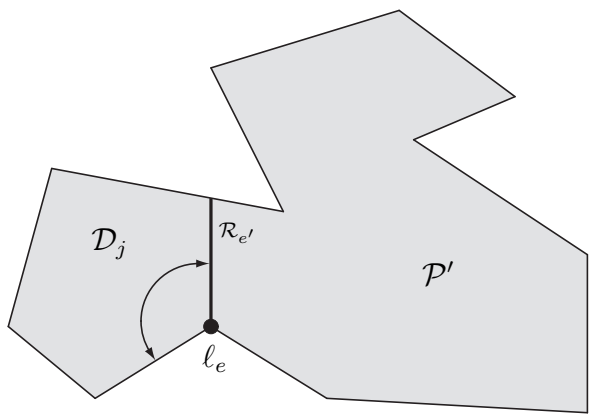
\begin{figure}[h]
\centering
\includegraphics[scale=0.6]{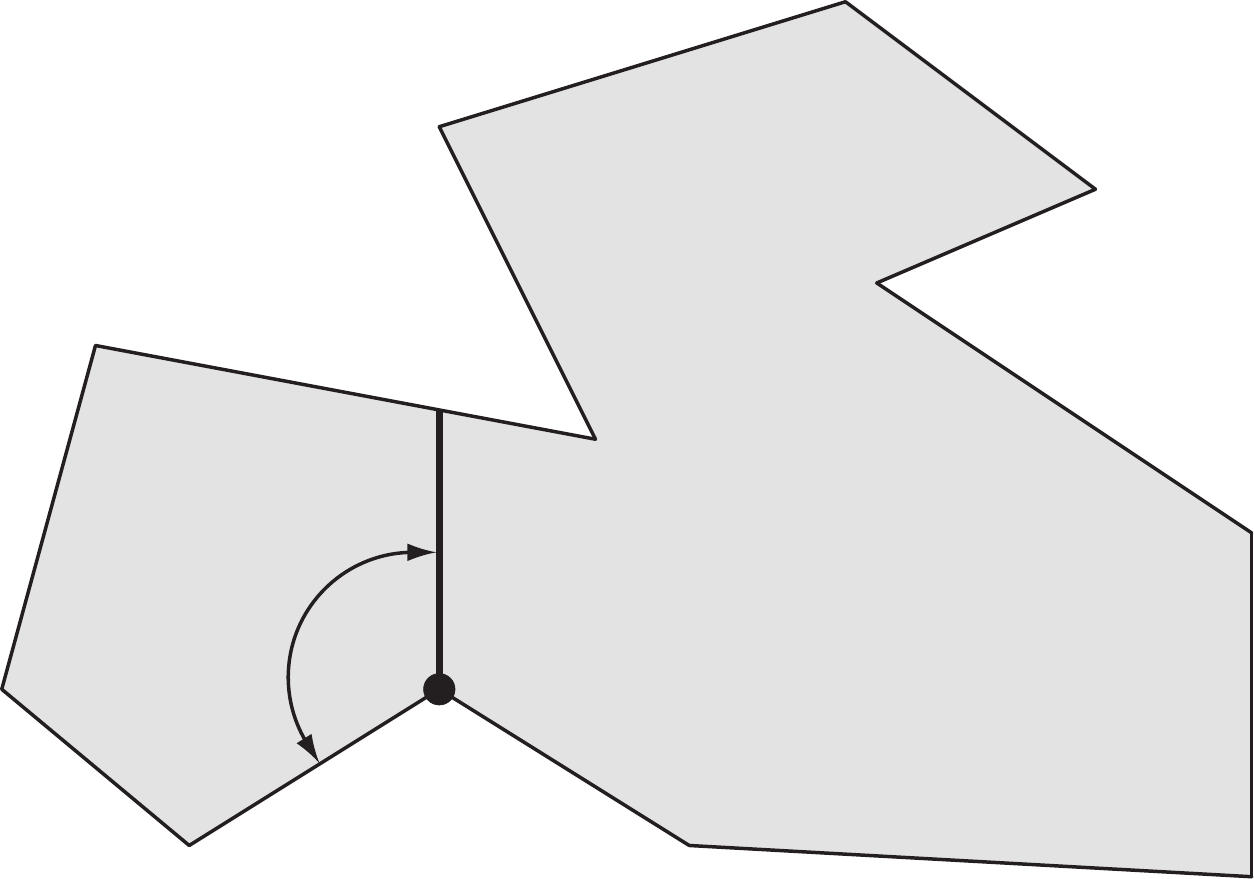}
\put(-144.5,17.5){$\ell_e$}
\put(-138.25,68.5){${}_{\mathcal R_{e'}}$}
\put(-186,61.75){$\mathcal D_j$}
\put(-62.1,48.65){$\mathcal P'$}
\caption{$\mathcal P'$ is not recontaminated while $\ell_e$ sweeps $\mathcal D_j$.}
\label{fig:8}
\end{figure}

Nonetheless, recontamination could still occur between other subpolyhedra bounded by $\alpha_e$. Let $\{e_1, e_2\}\subseteq E$, and let the subpolyhedra $\mathcal P_1$ and $\mathcal P_2$ be partitioned by $\mathcal R_{e_1}$, and $\mathcal R_{e_2}$, respectively (observe that the relative interiors of $e_1$ and $e_2$ must be disjoint, by construction). Obviously, no recontamination between $\mathcal P_1$ and $\mathcal P_2$ is possible while $\ell_e$ sweeps, because $\alpha_e$ is not a common boundary (refer to Figure~\ref{fig:n4}). However, recontamination could occur within $\mathcal P_1$ (or within $\mathcal P_2$), provided that $e_1\neq e'$. As it turns out, this type of recontamination is irrelevant, because it would imply that the node labeled $\mathcal P_1$ (call it $p$) has not been reached yet by the depth-first walk in the partition tree. We set up an induction argument on the partition tree, assuming that all the subpolyhedra corresponding to previously visited leaves are still clear before $\ell_e$ sweeps. By construction, $p$ cannot be an ancestor of the node labeled $\mathcal D_j$, otherwise $\mathcal D_j$ would be a subpolyhedron of $\mathcal P_1$, while in fact their interiors are disjoint. So, had the depth-first walk reached $p$, it would have also visited its entire dangling subtree and, by inductive assumption, $\mathcal P_1$ would still be all clear. Moreover, since $\mathcal P_1$ is not bounded by $\alpha_e$, it cannot get recontaminated while $\ell_e$ sweeps. On the other hand, if the walk has not reached $p$ yet, then perhaps some portions of $\mathcal P_1$ have been accidentally cleared, but can safely be recontaminated, because the systematic clearing process of $\mathcal P_1$ has yet to start.\end{proof}

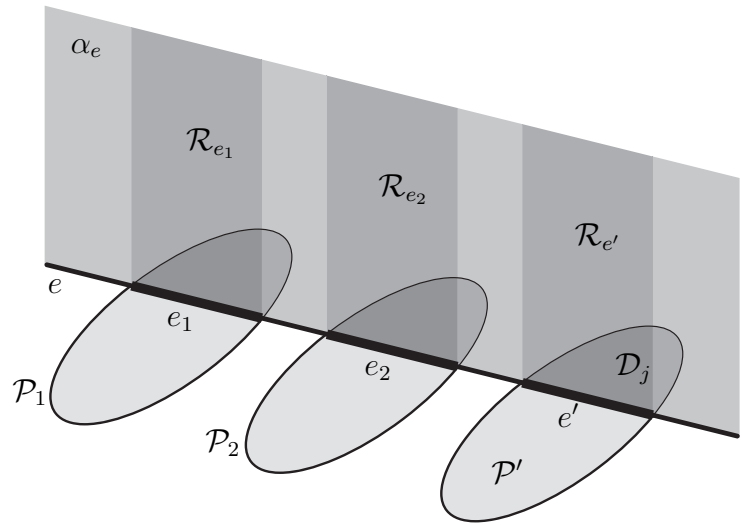
\begin{figure}[h]
\centering
\includegraphics[scale=0.9]{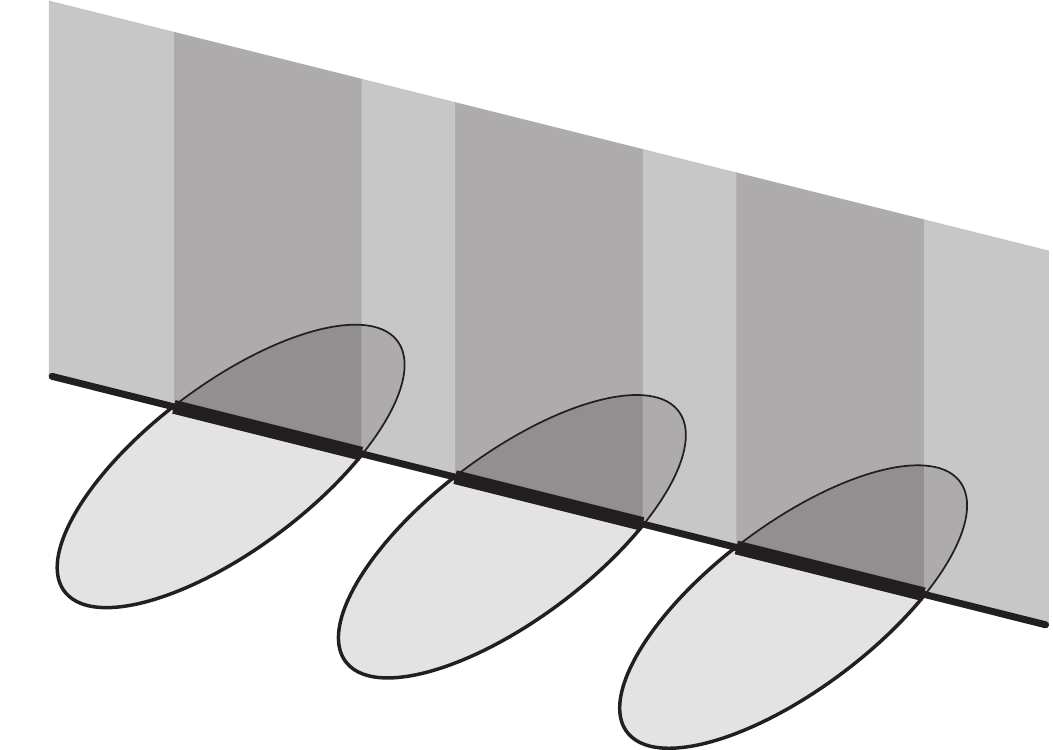}
\put(-251.25,177.25){$\alpha_e$}
\put(-208,140.5){$\mathcal R_{e_1}$}
\put(-136,123){$\mathcal R_{e_2}$}
\put(-62.5,104.5){$\mathcal R_{e'}$}
\put(-259.75,86.5){$e$}
\put(-215.2,74){$e_1$}
\put(-141.25,55.5){$e_2$}
\put(-69,36){$e'$}
\put(-273,46){$\mathcal P_1$}
\put(-200.75,27){$\mathcal P_2$}
\put(-93.75,15){$\mathcal P'$}
\put(-47.25,55.5){$\mathcal D_j$}
\caption{Sketch of a splitting plane.}
\label{fig:n4}
\end{figure}

\begin{remark}
A bonus feature of this improved construction, apart from using fewer guards, is that no two guards are coincident: indeed, if two guards end up coinciding, we can slightly perturb their respective planes $\alpha_e$. Observe that practically implementing two coincident guards (equivalently, one guard with two independent searchlights) may conceivably be unfeasible in some applications.
\end{remark}

\section{Searching orthogonal polyhedra}

The previous results can be greatly improved if the polyhedron is orthogonal.

\subsection*{Erecting fences}

We first obtain a partition of a given orthogonal polyhedron $\mathcal P$ into cuboids by constructing vertical \emph{fences} in a 3-step process (recall the terminology for orthogonal polyhedra introduced in Chapter~\ref{chapter2}). Then we will argue that these fences may be regarded as (parts of) searchplanes of suitably oriented edge guards lying on the reflex edges of $\mathcal P$. Finally, after pointing out some basic properties of our partition, we will show that rotating all the searchlights one by one clears every cuboid.

The term \textquotedblleft fence\textquotedblright\ is borrowed from~\cite{polypart2}, where Chazelle and Palios describe a partition of general polyhedra into prisms that somewhat resembles that of our Step~\ref{step1} (see also Chapter~\ref{chapter2}).
\begin{enumerate}
\item\label{step1} Each horizontal reflex edge $r$ of $\mathcal P$ has a horizontal adjacent face and a vertical adjacent face $F$, going upwards or downwards. From every point on $r$, cast a vertical ray in the direction opposite to $F$, until it again reaches the boundary of $\mathcal P$. Repeat the process with every horizontal reflex edge, so that each generates a vertical \emph{fence}, either upwards or downwards. It is easy to see that $\mathcal P$ is partitioned by fences into orthogonal prisms (i.e., extruded polygons) with horizontal bases.

\item\label{step2} Consider any vertical reflex edge $r$ of $\mathcal P$ with an internal point lying on a lateral face (i.e., an $x$-orthogonal face) of a prism $\mathcal Q$ generated in Step~\ref{step1}. By construction, $r$ must lie entirely on the boundary of $\mathcal Q$. Extend $r$ to a maximal segment $r'$ contained in the boundary of $\mathcal Q$. If possible, cast a horizontal ray from every point of $r'$, going through the interior of $\mathcal Q$ in the left-right direction, until it reaches $\mathcal Q$'s boundary, as shown in Figure~\ref{fig:n5}. The rectangle so formed is again called a fence. Repeating the same construction with every other vertical reflex edge of $\mathcal P$ lying on a lateral face of some prism, we obtain a refinement of the initial partition by these new fences. Clearly, to every such reflex edge corresponds just one fence, which in turn goes through the interior of just one prism.

\begin{figure}[h]
\centering
\includegraphics[scale=0.7]{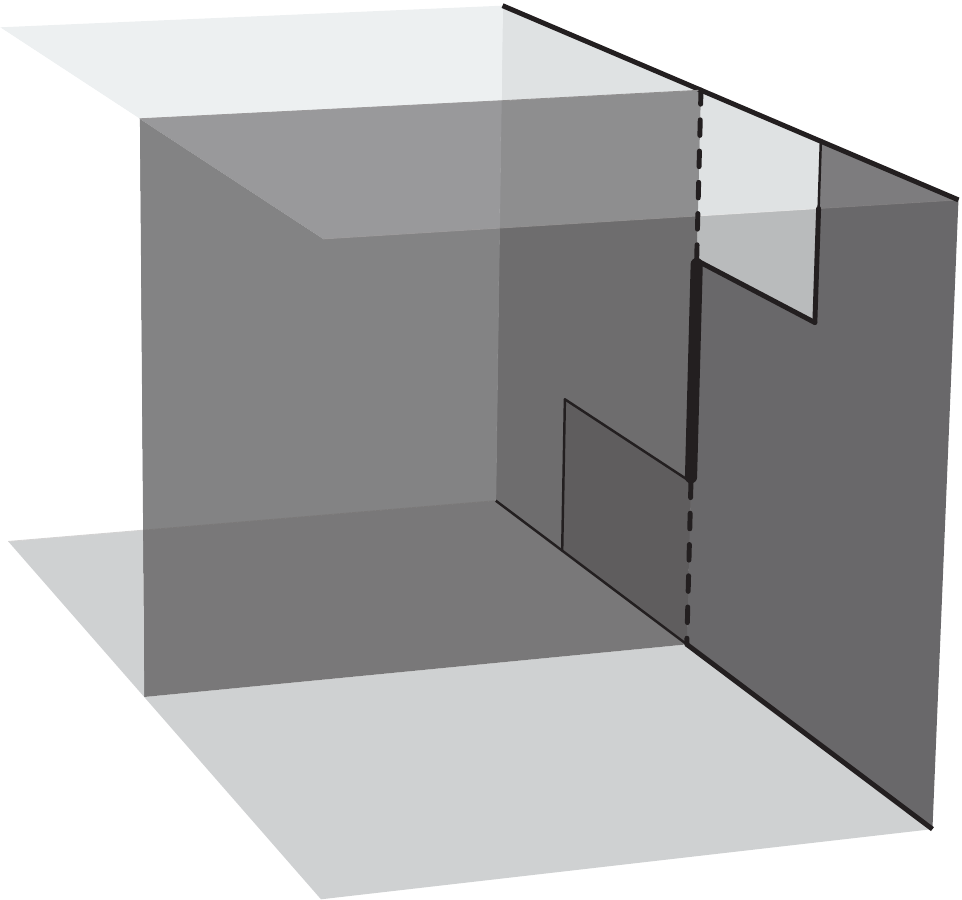}
\put(-49.5,102){$r$}
\put(-51.5,62.5){$r'$}
\large
\put(-101,19.65){$\mathcal Q$}
\caption{Fence generated by $r$ during Step~\ref{step2}. Fences are shown as darker regions.}
\label{fig:n5}
\end{figure}

\item\label{step3} The pieces resulting from Step~\ref{step2} are once again prisms, perhaps containing some additional vertical fences that are not faces. Indeed, some fences constructed in Step~\ref{step2} split a prism in two subprisms, while others just lower a prism's genus and act as degenerate faces (which will become legitimate faces of cuboids shortly). By construction, each reflex edge of such a prism is vertical, with no edges of $\mathcal P$ lying on it. Repeat the procedure in Step~\ref{step2} also with these reflex edges, casting horizontal rays, and thus building vertical fences that extend laterally the front faces of the prisms, as shown in Figure~\ref{fig:9}. As a result, $\mathcal P$ gets partitioned into cuboids.
\end{enumerate}

\begin{figure}[h]
\centering
\includegraphics[scale=1]{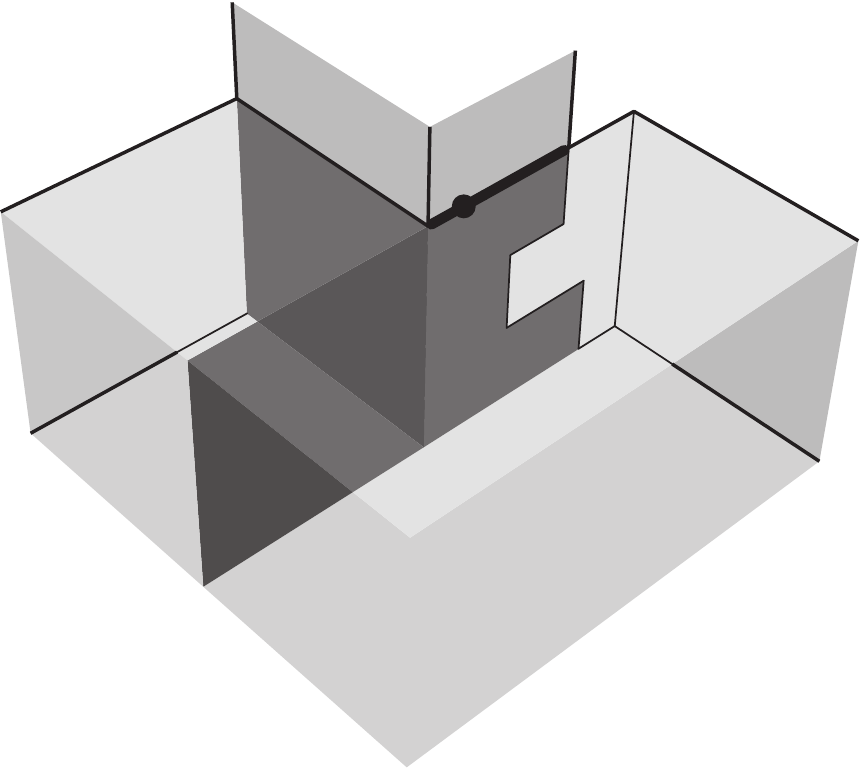}
\put(-99.75,176.5){$\ell$}
\put(-95.5,139.75){$F$}
\put(-163.25,151.15){$F'$}
\put(-122.4,126){$r$}
\put(-118.3,168.5){$x$}
\large
\put(-105.5,48.5){$\mathcal C$}
\put(-212.75,90.75){$\mathcal C'$}
\caption{Fence generated by $r$ during Step~\ref{step3}. Fences are shown as darker regions. The labels comes from Lemmas~\ref{fences} and~\ref{cuboids}.}
\label{fig:9}
\end{figure}
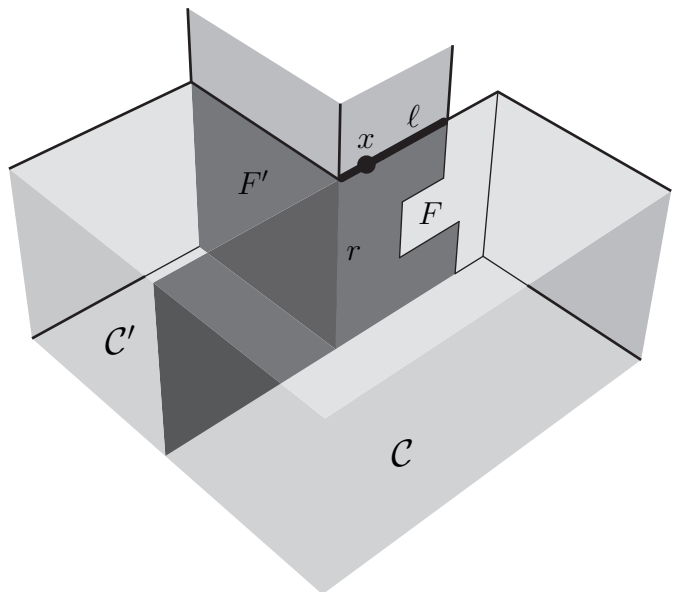

Observe that all the fences generated in Steps~\ref{step2} and~\ref{step3} are $y$-orthogonal rectangles. The reason why we distinguish three steps is that the respective fences need be treated differently in the proofs to come.

\subsection*{Placing guards}

Now place an edge guard on each reflex edge of $\mathcal P$. Aim horizontal guards vertically and aim vertical guards in the left-right direction, in such a way that every guard aims at the interior of $\mathcal P$. We first want to prove that the fences in our 3-step construction are \emph{coherent} with searchlights.

\begin{lemma}\label{fences}Every fence is contained in some illuminated searchplane.\end{lemma}

\begin{proof}The claim is obvious for fences constructed in Step~\ref{step1} and Step~\ref{step2}. As for fences constructed in Step~\ref{step3}, consider an edge $r$ generating one of them (see Figure~\ref{fig:9}). Recall that $r$ is a reflex edge of a prism, let it be $\mathcal Q$, in the partition obtained in Step~\ref{step1}. By the construction in Step~\ref{step2}, $r$ contains no reflex edge of $\mathcal P$. Hence its adjacent front face $F\subset \mathcal Q$ has no intersection with the boundary of $\mathcal P$, in a thin-enough neighborhood of $r$. But both bases of $\mathcal Q$ belong to the boundary of $\mathcal P$, because no horizontal fences were constructed (recall that the top face and the bottom face of a prism are both called ``bases''). Then, at least a subsegment of a horizontal edge of $F$, sharing an endpoint with $r$, belongs to a reflex edge of $\mathcal P$, and hence belongs to a guard $\ell$. As a consequence, $\ell$ illuminates the whole fence generated by $r$ in Step~\ref{step3}.\end{proof}

Thus, every fence \emph{belongs} to one guard. We could also incorporate each fence built in Step~\ref{step3} into its adjacent vertical fence built in Step~\ref{step1} that belongs to the same guard (e.g., in Figure~\ref{fig:9}, the two fences that are coplanar with face $F$ are ``merged'' and belong to guard $\ell$). As a result, every guard \emph{generates} at most one fence.

\begin{lemma}\label{cuboids}Every cuboid in the partition belongs entirely to the visibility region of any guard whose fence bounds it.\end{lemma}

\begin{proof}If a fence built in Step~\ref{step1} bounds a cuboid, then it belongs to a guard $\ell$ located, at least partially, on its border. Indeed, fences are vertical and the upper and lower faces of a cuboid belong to faces of $\mathcal P$, and therefore $\ell$ cannot lie entirely outside the cuboid. It follows that $\ell$ sees the whole cuboid.

Fences built in Step~\ref{step2} bound exactly two cuboids each, because the fences built in Step~\ref{step3} are all parallel to them. Again, any guard generating one such fence belongs to a common edge of the cuboids that it bounds, which of course are entirely visible to the guard. 

Also fences built in Step~\ref{step3} bound exactly two cuboids each. With the same notation as in the proof of Lemma~\ref{fences}, the fence generated by $r$ bounds a cuboid $\mathcal C$ that is also bounded by $F$, and therefore contains, at least partially, guard $\ell$ on one of the horizontal edges of $\mathcal C$ (again, $\ell$ is defined as in the proof of Lemma~\ref{fences}). $\ell$ also shares one endpoint with $r$. Moreover, the other cuboid $\mathcal C'$ is bounded also by a lateral face $F'\subset \mathcal Q$ adjacent to $r$. The interior of $F'$ has no intersection with the boundary of $\mathcal P$. Indeed, if the two had any intersection, then there would be also a vertical reflex edge of $\mathcal P$ lying in the interior of $F'$, or on $r$. But this disagrees with the construction in Step~\ref{step2}, which eliminates all such reflex edges. Let $x$ be a point in $\ell$ that is close enough to $r$ (e.g., whose distance to $r$ is lower than the minimum positive difference between any two coordinates of vertices of $\mathcal P$ (such minimum exists due to the finiteness of $\mathcal P$'s vertices). Of course, $x$ completely sees $\mathcal C$, because it lies on its boundary. But $x$ also sees every point in $\mathcal C'$, through the fence $F'$ (see Figure~\ref{fig:9}). Indeed, by our choice of $x$, the pyramid determined by $F'$ and $x$ is completely contained in $\mathcal P$.\end{proof}

\subsection*{Search strategy}

\begin{theorem}\label{orth}Any non-convex orthogonal polyhedron with an edge guard on each reflex edge is searchable.\end{theorem}

\begin{proof}Aim the guards as described above, in such a way that every fence is illuminated by some guard, by Lemma~\ref{fences}. Clearly, illuminated searchplanes induce the same partition of $\mathcal P$ into cuboids (perhaps even a finer partition). Now pick a guard $\ell$ generating a fence $F$, and let $\mathcal Q$ be the union of the cuboids bounded by $F$. Since $F$ is connected and has cuboids on both sides, it follows that $\mathcal Q$ is connected as well, and therefore it is a polyhedron. Moreover, by Lemma~\ref{cuboids}, $\mathcal Q$ entirely belongs to $\mathcal V(\ell)$. Hence, by Theorem~\ref{one}, $\mathcal Q$ can be cleared by $\ell$ while all the other guards keep their searchlights fixed. Turn $\ell$ to clear $\mathcal Q$, put $\ell$ back in its original position, and repeat the procedure for all the other guards, one at a time. Notice that every turning guard clears all the cuboids that it bounds, while the other cuboids cannot be recontaminated, because their boundaries remain fixed. Since $\mathcal P$ is both connected and non-convex, every cuboid is bounded by at least one fence, which in turn is generated by some guard. Thus, after the last guard has finished sweeping, $\mathcal P$ is completely clear.\end{proof}

\subsection*{Improving search time}

The search time of the schedule given by Theorem~\ref{orth} is linear in the number of reflex edges of $\mathcal P$, but once again we can achieve constant search time by doubling the number of guards.

\begin{corollary}\label{speed2}Any non-convex orthogonal polyhedron with two (coincident) edge guards on each reflex edge is searchable in $3/4$ seconds.\end{corollary}

\begin{proof}Half of the guards are positioned as in the proof of Theorem~\ref{orth} and never move, thus preserving the partition into cuboids. The other guards simultaneously sweep their visibility region by turning from their leftmost position to the rightmost. Since every guard lies on a reflex edge, they all have to turn by an angle of $3/2~\pi~\mathrm{rad}$, which can be done in $3/4$ seconds.\end{proof}
\chapter{Complexity of full searching}\label{chapter10}
\begin{chapterabstract}
We consider the computational complexity of deciding if a given instance of the $\TSSPext$ is searchable.

First we show that the problem is strongly \NP-hard for general polyhedra and boundary guards. Then we prove that, even if we are given a searchable orthogonal polyhedron, approximating the minimum search time within a certain constant factor is \NP-hard. As a consequence, no \PTAS exists for this problem, unless $\P=\NP$.

No analogous result is known for the planar \SSPext.
\end{chapterabstract}

\section{\NP-hardness of searchability}\label{complexity}

Recall from Chapter~\ref{chapter1} that the computational complexity of 2-dimensional \SSP is still unknown, and left as an open question by Obermeyer et al. in~\cite{bullo}.

Next we prove that \TSSP\ is strongly \NP-hard by a reduction from \TSAT, thus converting a formula $\varphi$ in 3-conjunctive normal form into an instance of \TSSP\ that is searchable if and only if $\varphi$ is satisfiable.

\subsection*{Building blocks}

A \emph{variable gadget} is a cuboid with a guard lying on one edge. The two faces adjacent to the guard are called \emph{A-side} and \emph{B-side}, respectively, as shown in Figure~\ref{fig:10a}. The guard itself is called \emph{variable guard}.

\begin{figure}[h]
\centering
\subfigure[variable gadget]{\label{fig:10a}\includegraphics[scale=0.8]{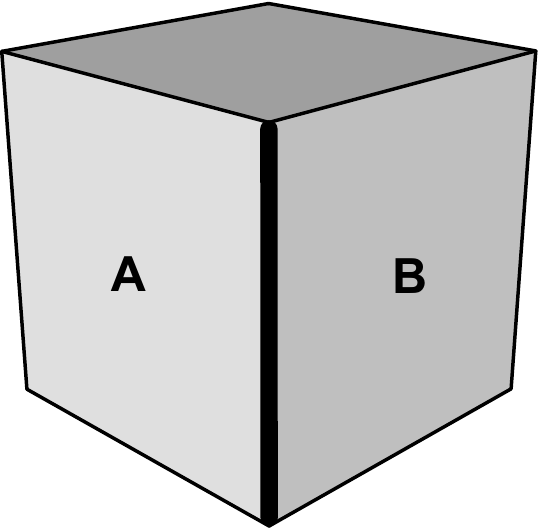}}\qquad \qquad
\subfigure[link]{\label{fig:10b}\includegraphics[scale=0.8]{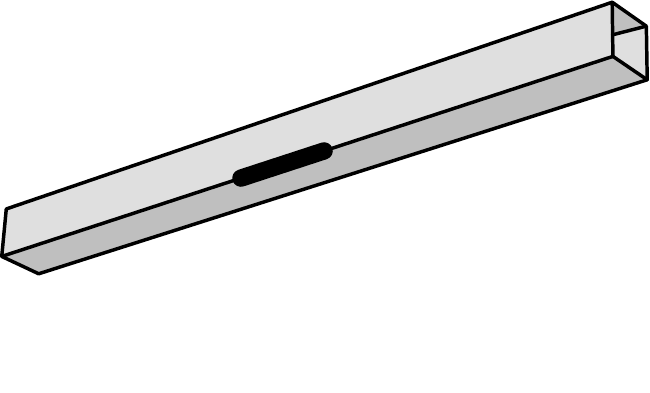}}
\caption{Two building blocks of the reduction.}
\label{fig:10}
\end{figure}

A \emph{clause gadget} is made of several parts, shown in Figure~\ref{fig:11}. The first part is a prism, shaped like a wide cuboid with three \emph{cavities} on one side. On its top there is a \emph{nook} shaped as a triangular prism lying on a side face, with a long guard on the upper edge. The guard is called \emph{separator}, because its searchplanes partition the clause gadget in two regions. One of the two regions contains none of the three cavities, and its top face is called \emph{A-side}. On the other hand, the back face of each cavity is a \emph{B-side}, and a \emph{literal guard} lies on the top edge of each B-side. When any literal guard is aiming at the A-side of its clause gadget, it also completely closes the nook containing the separator. We define the leftmost position of the separator to be the one that is closer to the A-side. All three cavities are then pairwise connected by V-shaped prisms with vertical bases. When two literal guards from the same clause gadget are both aiming at their B-sides, their searchplanes intersect in an area around the bottom of the V: such area is called \emph{C-side}. Thus, every clause gadget has three B-sides and three C-sides, all coplanar.

The A-sides, B-sides and C-sides of all the gadgets are collectively referred as \emph{distinguished sides}.

\begin{figure}[h]
\centering
\includegraphics[scale=1]{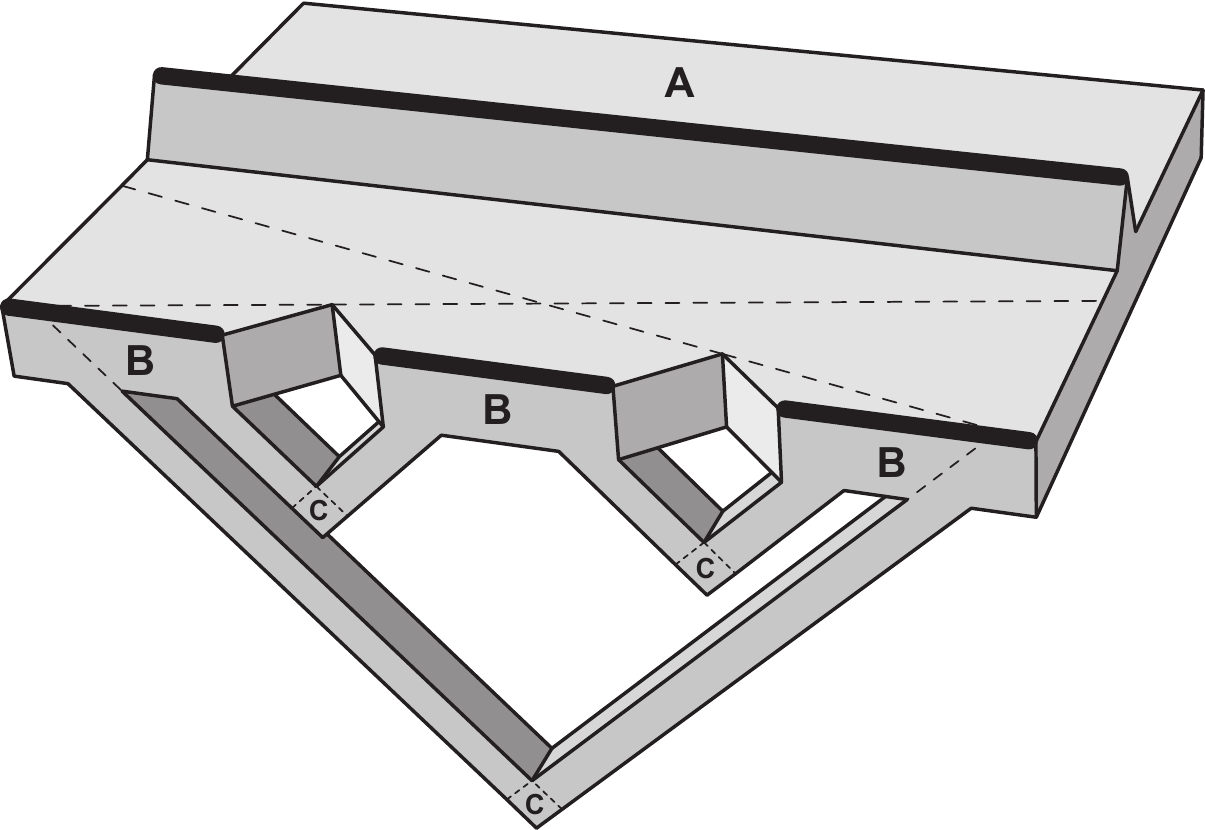}
\caption{Clause gadget.}
\label{fig:11}
\end{figure}

To connect together all the different gadgets we use structures called \emph{links}. A link is a very thin prism with its two bases removed, and with a short \emph{link guard} lying in the middle of an edge (see Figure~\ref{fig:10b}). In the following, whenever we wish to connect two gadgets, we will cut a hole in their surfaces (where indicated from time to time), and place a link stretching from one hole to the other. The angle of incidence does not matter, as long as no guard, other than its link guard, can see inside a link. Conversely, we will arrange the links in such a way that every link guard's searchlight will not interfere with the gadget guards. So, in every gadget there will be a thin illuminated polygon jutting from each of its links, which will be easily avoidable by the intruder.  As a consequence, a link can be cleared by its guard only while both its bases are \emph{capped} by some guards lying in the adjacent gadgets.

Given a Boolean formula $\varphi$ in 3-conjunctive normal form, we construct a row of variable gadgets, one for every variable of $\varphi$, and a row of clause gadgets, one for every clause of $\varphi$. We arrange the variable gadgets so that all the A-sides are coplanar, and all the B-sides are coplanar. We arrange the clause gadgets similarly, and we place the two rows of gadgets in such a way that every distinguished side of every variable gadget can see every distinguished side of every clause gadget. We also associate the $i$-th B-side of a clause gadget to the $i$-th literal in the corresponding clause of $\varphi$, for $1 \leqslant i \leqslant 3$.

Finally we add a \emph{bridge}, whose purpose is to prevent variable gadgets from being cleared if the corresponding literal guards do not behave ``coherently'', as explained later in the proof of Theorem~\ref{hard}. The bridge is constructed like a variable gadget, but it is not associated to any variable of $\varphi$. It is shaped as a long, thin pole, whose guard lies on one of the long edges, and it is arranged in such a way that its distinguished sides can see the B-side of every clause gadget.

\subsection*{Connections}

Then we connect the distinguished sides of our gadgets by placing links. The exact points of connection and angles of incidence do not matter, as long as guards are non-interfering, as specified earlier, when introducing links. Mutual intersections of links will be resolved later.

Referring to Figures~\ref{fig:13} and~\ref{fig:12}, we make the following connections.
\begin{itemize}
\item Connect the A-side of every clause gadget to both the A-side and the B-side of every variable gadget.
\item Connect all the B-sides of every clause gadget to both the A-side and the B-side of the bridge.
\item Connect all the C-sides of every clause gadget to both the A-side and the B-side of every variable gadget.
\item Connect each B-side of each clause gadget to the A-side (resp.\ B-side) of the variable gadget corresponding to its associated literal, if that literal is negative (resp.\ positive) in $\varphi$.
\end{itemize}

\begin{figure}[h]
\centering
\includegraphics[scale=0.6]{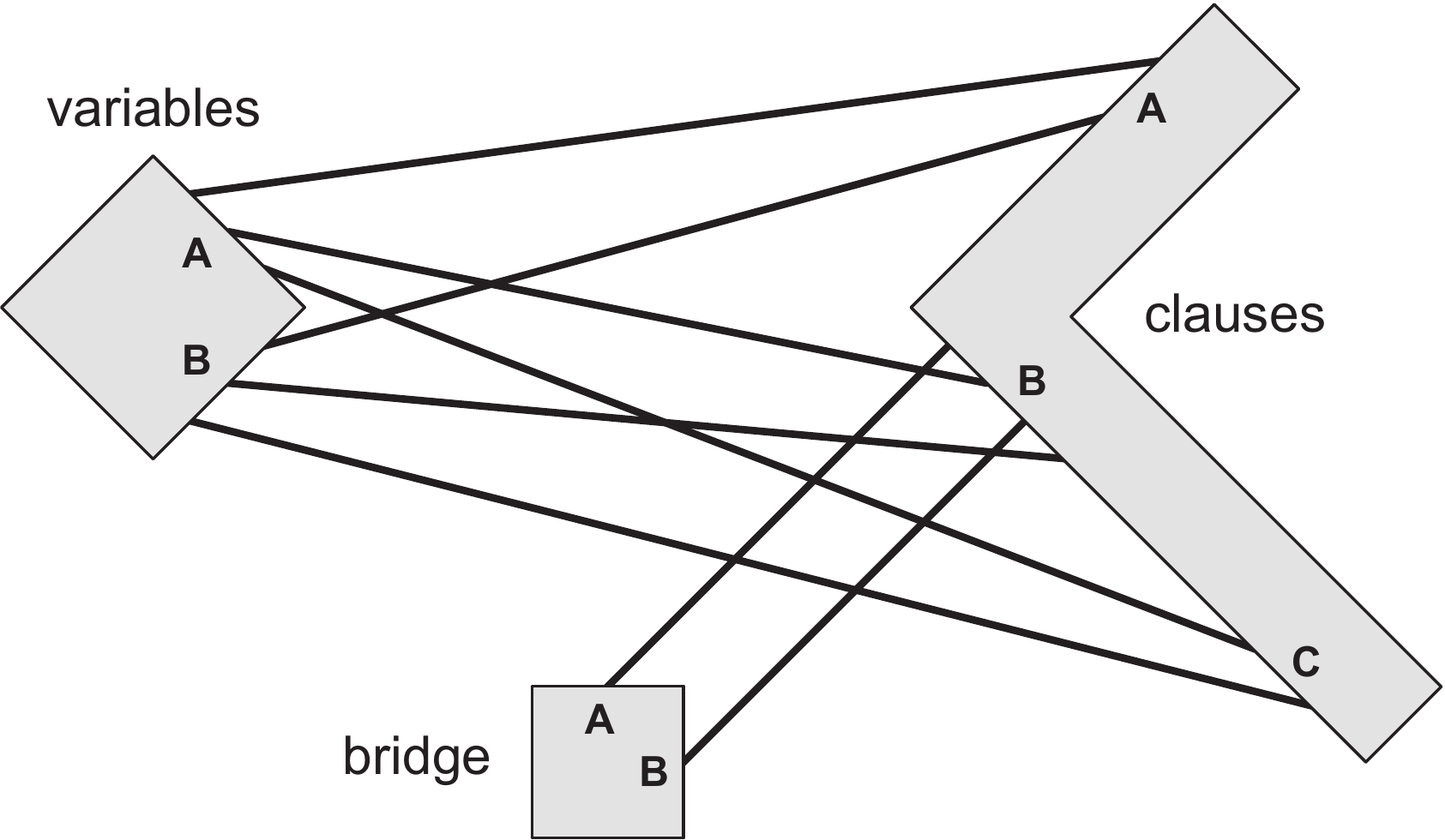}
\caption{Relative positions of the gadgets and the bridge, with their links.}
\label{fig:13}
\end{figure}

We can easily position the bridge so that it is not accidentally hit by any link running between a variable gadget and a clause gadget, such as in Figure~\ref{fig:13}. We also want links to be pairwise disjoint. To achieve this, we consider any pair of intersecting links, and shrink them while translating them slightly, until their intersection vanishes. This can be accomplished without creating new intersections with other links, for example by making sure that the \textquotedblleft new version\textquotedblright\ of each link is always strictly contained in its \textquotedblleft previous version\textquotedblright.

\begin{figure}[h]
\centering
\subfigure[clause gadget, top view]{\label{fig:12a}\includegraphics[scale=0.8]{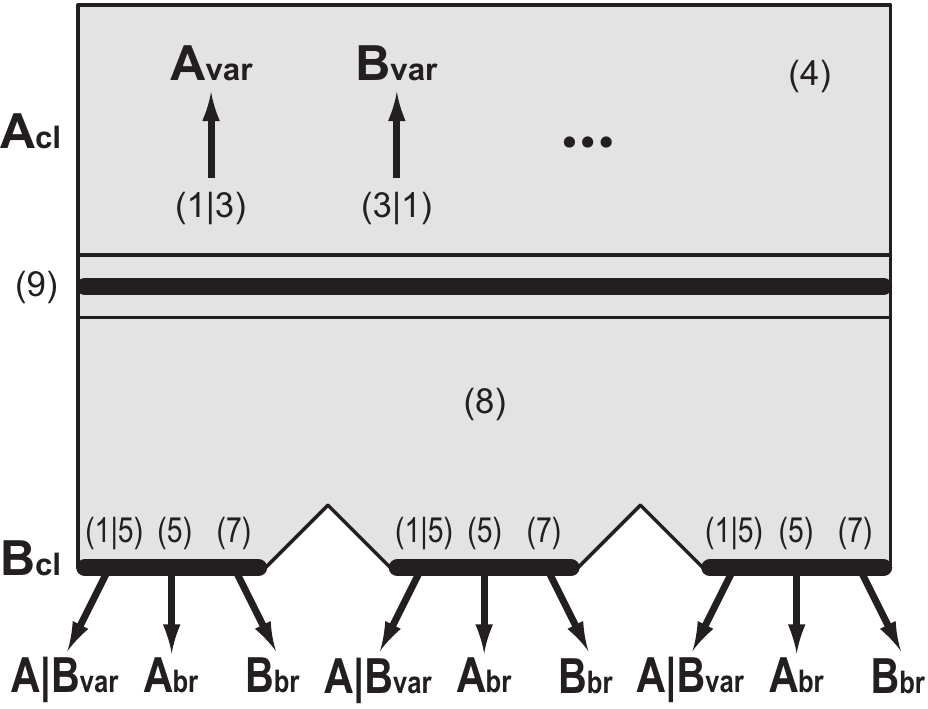}}\qquad
\subfigure[clause gadget, C-side]{\label{fig:12b}\includegraphics[scale=0.8]{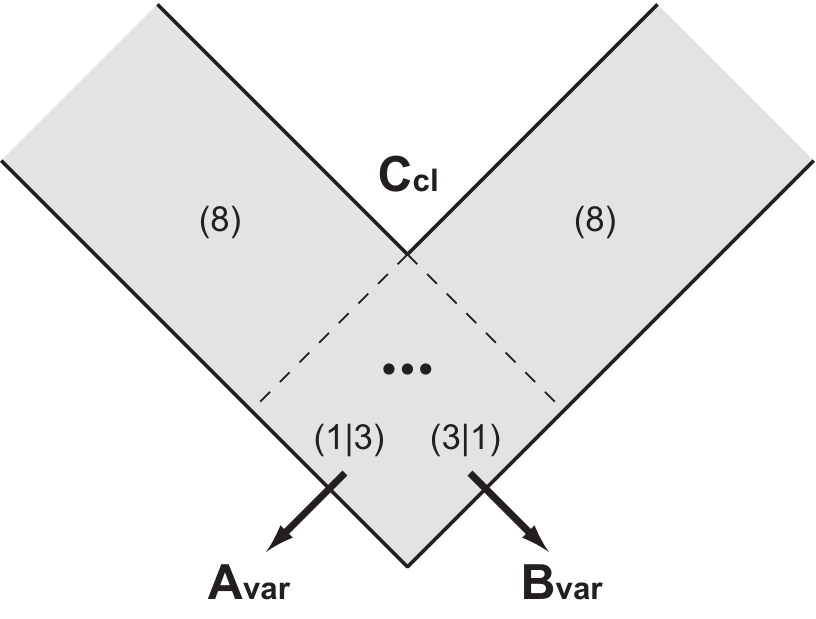}}\vspace{0.5cm}
\subfigure[variable gadget]{\label{fig:12c}\includegraphics[scale=0.8]{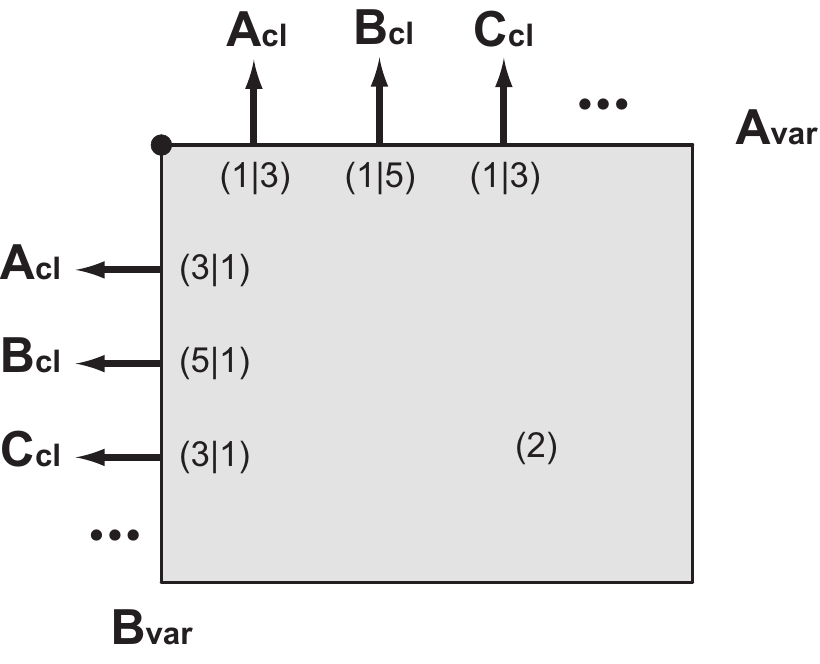}}\qquad \qquad
\subfigure[bridge]{\label{fig:12d}\includegraphics[scale=0.8]{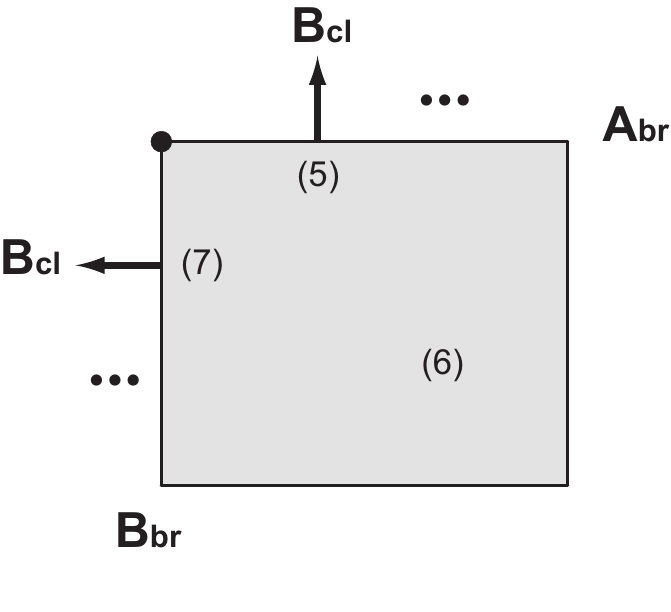}}
\caption{Connections among gadgets, and their clearing order given in Theorem~\ref{hard}. The abbreviation \textquotedblleft $\mathrm{A_{cl}}$\textquotedblright\ stands for \textquotedblleft A-side of a clause gadget\textquotedblright, etc.}
\label{fig:12}
\end{figure}

\subsection*{Reduction}

\begin{theorem}\label{hard}The \TSSPext\ is strongly \NP-hard.\end{theorem}

\begin{proof}Given an instance $\varphi$ of \TSAT, we construct the instance of \TSSP\ described above. It is indeed a polyhedron, because the bridge and all the variable gadgets are connected to all the clause gadgets. Moreover, the number of links is quadratic in the size of $\varphi$, and we may also assume that the coordinates of every vertex are rationals, with a number of digits that is polynomial in the size of $\varphi$. Removing the intersections between links takes polynomial time as well, hence the whole construction is computable in \P.

\paragraph{Positive instances.}

If $\varphi$ is satisfiable, we choose a satisfying assignment for its variables and give a search schedule that clears our polyhedron. Initially we aim each variable guard at its A-side (resp.~B-side) if the corresponding variable is true (resp.~false) in the chosen assignment. By assumption there is at least a true literal in every clause of $\varphi$. For each clause, we pick exactly one true literal, and aim the corresponding literal guard at the A-side of its clause gadget. We aim all the other literal guards at their respective B-sides. As a result, the A-side and the three C-sides of every clause gadget have the corresponding links capped by the literal guards. Finally, we aim the bridge guard at its A-side and put every separator in its leftmost position.

From this starting configuration, we specify a search schedule in nine steps (refer to Figure~\ref{fig:12}).
\begin{enumerate}
\item\label{stepp1} Clear all the links that are capped by some variable guard. This is possible because, by construction, the other end of every such link is capped by some literal guard as well.
\item\label{stepp2} Clear every variable gadget by turning its guard. While this happens, the literal guards retain caps on their own side of the links cleared during Step~\ref{stepp1}, thus preventing recontamination.
\item Clear the remaining links connected to the A-side or to a C-side of a clause gadget. This is now possible because all the variable guards switched side in Step~\ref{stepp2}.
\item Aim at its B-side every literal guard that is currently aiming at its A-side. One half of each clause gadget gets cleared as a result, while the separator prevents the still uncapped links on the B-side from recontaminating the clear links on the A-side.
\item Clear the remaining links connected to a variable gadget, and clear the links capped by the bridge guard.
\item Clear the bridge by turning its guard to the B-side.
\item Clear the remaining links that connect the bridge with the clause gadgets.
\item Turn all the literal guards simultaneously, thus clearing the last half of each clause gadget, and capping the upper nooks. Since the three literal guards of a clause gadget are collinear, when they move together they act as a single filling guard.
\item Clear every nook by turning the separators.
\end{enumerate}
When this is done, the whole polyhedron is clear, which proves that the instance of \TSSP\ is searchable.

\paragraph{Negative instances.}

Conversely, assuming that $\varphi$ is not satisfiable, we claim that the variable gadgets can never be all simultaneously clear, no matter what the guards do.

Recall that every A-side of every clause gadget is linked to both sides of every variable gadget. Hence, as soon as the A-side of any clause gadget is not covered by at least one literal guard, all the variable gadgets get immediately recontaminated, unless they were all clear in the first place. For the same reason, no variable gadget can ever be cleared while the A-side of some clause gadget is uncovered. Similarly, if a C-side of any clause gadget is not covered by at least one literal guard, all variable gadgets get recontaminated, and none of them can be cleared.

It follows that, in order for a schedule to start clearing any variable gadget, it must ensure that each clause gadget has exactly one literal guard covering the A-side and exactly two literal guards covering the C-sides. Moreover, the literal guards that cover the A-sides must be chosen once and for all. Indeed, whenever a schedule attempts to \textquotedblleft switch gears\textquotedblright\ in some clause gadget and cover the A-side with a different literal guard, all the variable gadgets become immediately contaminated, and the search must start over.

Suppose that a schedule selects exactly one literal guard for each clause gadget, to cover its A-side. Since $\varphi$ is not satisfiable, there exist two selected literal guards $\ell_1$ and $\ell_2$ whose corresponding literals in $\varphi$ are a positive occurrence and a negative occurrence of the same variable $x$. Otherwise, if all selected literals were \emph{coherent}, setting them to true would yield a satisfying assignment for the variables of $\varphi$, which is a contradiction.

But in this case, it turns out that the variable gadget corresponding to variable $x$ is impossible to clear. Indeed, there is a non-illuminated path connecting its A-side with its B-side, passing through the B-side of $\ell_1$, the bridge, and the B-side of $\ell_2$. Since $\ell_1$ and $\ell_2$ correspond to incoherent literals, their B-sides are connected to opposite sides of the same variable gadget, by construction.

Summarizing, all the variable gadgets are initially contaminated. In order to clear some of them, a schedule must first select a literal guard from each clause gadget and put it on its A-side. While that position is maintained, there is at least one variable gadget that is impossible to clear. As soon as one literal guard is moved, all the variable gadgets get recontaminated again. It follows that the variable gadgets can never be all clear at the same time, and in particular the polyhedron is unsearchable.

We gave a polynomial time reduction from \TSAT\ to \TSSP\ whose generated numerical coordinates are polynomially bounded in size, hence \TSSP\ is strongly \NP-hard.\end{proof}

\subsection*{Optimization problems}

Obviously, the previous theorem implies that the problems of minimizing search time and minimizing \emph{total angular movement} are both \NP-hard to approximate (where ``unsearchable'' translates to ``infinite search time and angular movement'').

As we will show in the next section, the problem of minimizing search time stays \NP-hard even when restricted to searchable instances. On the other hand, similar results can be obtained also for the problem of minimizing total angular movement. There are several ways to rearrange the links in the construction employed in Theorem~\ref{hard}, so that the unsatisfiable Boolean formulas are mapped into polyhedra that are indeed searchable, but only with very demanding schedules.

\section{Inapproximability of search time}

Let the \TTSSPext (\TTSSP) be the optimization problem of minimizing search time in a given instance of \TSSP.

Next we show a reduction from \TSAT to \TTSSP. Our construction will transform a formula $\varphi$ in 3-conjunctive normal form into an orthogonal polyhedron $\mathcal P$ and a set of guards lying on $\mathcal P$'s edges, each with maximum angular speed of $90^\circ / \sec$, such that $\mathcal P$ can be searched within three seconds if and only if $\varphi$ is satisfiable. Moreover, even if $\varphi$ is not satisfiable, $\mathcal P$ is still searchable. As a consequence, \TTSSP stays hard even under the promise that the instance is searchable.

\begin{remark}
In contrast with Chapter~\ref{chapter9}, where guards turned their searchlights at $2\pi~\mathrm{rad} / \mathrm{sec}$, here our guards are slower. This has the only purpose of yielding more ``pleasant'' numeric values, but has no influence on the substance of our results.
\end{remark}

\subsection*{Building blocks}

We use the same type of link depicted in Figure~\ref{fig:10b}, i.e., a thin ``uncapped'' cuboid, with a short guard in the middle. Clearing a link requires one second, and it is possible if and only if both ends remain capped by some external searchplanes. A \emph{room} is a cube with a guard lying on some edge, much like a variable gadget in our previous construction. Rooms and links can be attached together to form a \emph{chain}, in such a way that the guard in each room has links on both its adjacent faces, as depicted in Figure~\ref{fig:link}.
\begin{figure}[h]
  \centering
  \includegraphics[width=.6\linewidth]{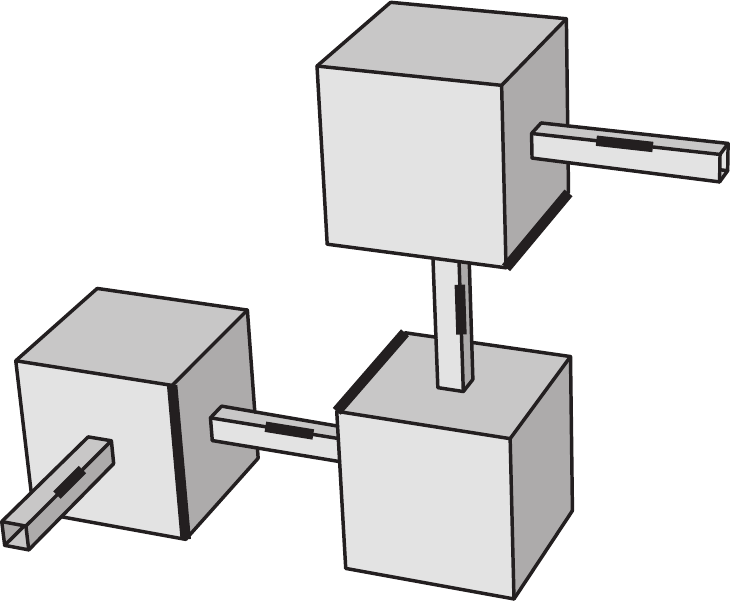}
  \caption{Odd chain made of rooms and links.}
  \label{fig:link}
\end{figure}

If a schedule is going to clear a chain in three seconds, it has to turn all the guards' searchlights in every room during second~2, while during seconds~1 and~3 they have to remain on their sides, in order to cap the surrounding links. As a consequence, the links in a chain must be cleared at second~1 and second~3, alternately. A chain is called \emph{even chain} or \emph{odd chain}, depending on the number of its rooms. Chains will connect \emph{variable gadgets} with the corresponding \emph{clause gadgets}, described next. (Note that we are referring to our ``new'' variable and clause gadget constructions, not the ones used in the previous section.)

To represent a Boolean variable $x$, we use a pair of rooms $\mathcal C_x$ and $\mathcal C'_x$, each with a \emph{fake link} on its left side, as Figure~\ref{fig:variable} suggests. The far end of each fake link is capped, as they are required for synchronization purposes only. For each occurrence of $x$ in $\varphi$, we attach an even chain to the right side of $\mathcal C_x$, and an even chain (resp.\ odd chain) to the right side of $\mathcal C'_x$ if the occurrence of $x$ is positive (resp.\ negative). Both chains will be connected to the same \emph{valve} of the proper clause gadget, as described later. We will say that variable $x$ is assigned the value \emph{true} if and only if both its guards turn their searchlights in the same direction (clockwise or counterclockwise) during second~2. Notice that the truth value of a variable gadget is well-defined in every schedule that clears it in three seconds.

\begin{figure}[h]
  \centering
  \includegraphics[width=.75\linewidth]{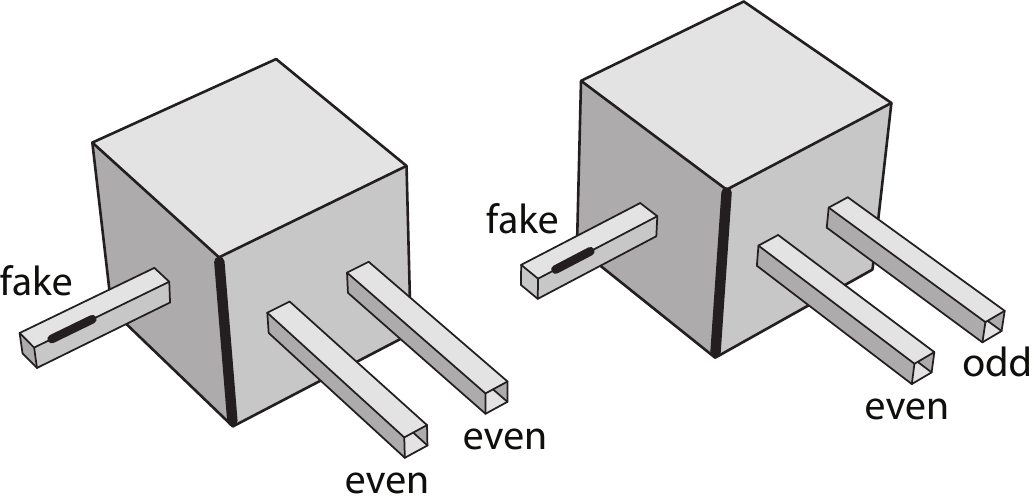}
  \put(-239.21,67.11){$x$}
  \put(-218.68,79.34){$\neg x$}
  \put(-85.66,88.03){$x$}
  \put(-67.89,99.87){$\neg x$}
\large
  \put(-249.32,105){$\mathcal C_x$}
  \put(-97.53,123.95){$\mathcal C'_x$}
  \caption{Boolean variable with two occurrences: one positive and one negative.}
  \label{fig:variable}
\end{figure}
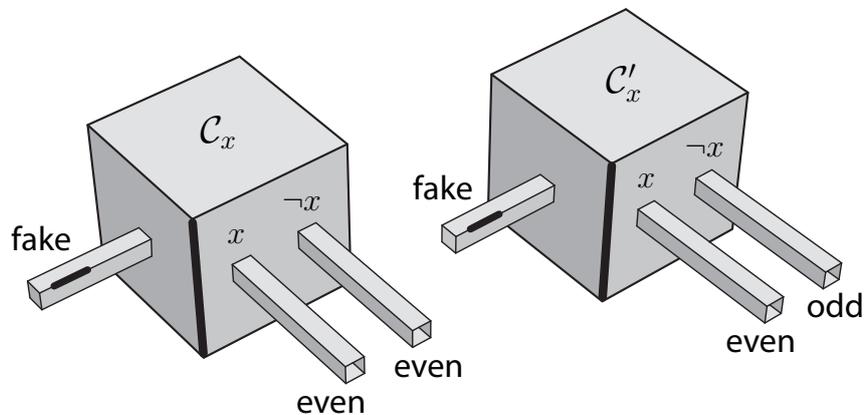

For each clause in $\varphi$, we construct a \emph{clause gadget}, consisting of a \emph{3-input OR gate} and three \emph{valves} attached to it. Each valve corresponds to an occurrence of a variable, and is then attached to the proper chain pair coming from that variable.

An \emph{OR gate} is a cube with three small holes and two guards, such that the guards can simultaneously close at most two holes with their searchlights (see Figure~\ref{fig:or}). Its clearing time is one second, provided that recontamination through holes is avoided, via external searchplanes.
\begin{figure}[h]
  \centering
  \includegraphics[width=.5\linewidth]{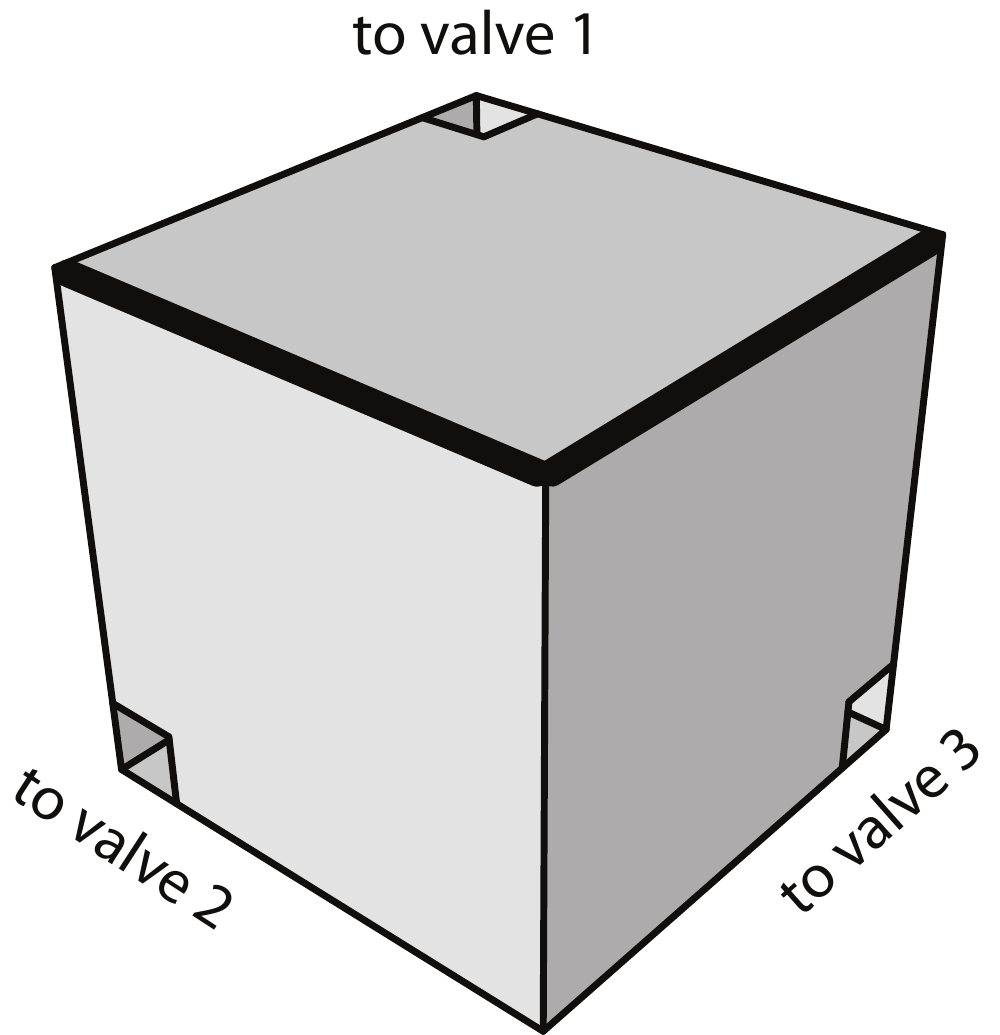}
  \caption{3-input OR gate.}
  \label{fig:or}
\end{figure}

Each hole in an OR gate directly connects to a small \emph{valve}, shown in Figure~\ref{fig:valve}. Guard $\ell_1$ is able to cap the chain pair and close the hole leading to the OR gate but, in order to do both, it has to spend one second switching position. Guard $\ell_2$ must cap the auxiliary fake links, one at a time, during seconds $1$ and $3$. Hence it must switch position during second $2$, incidentally sweeping the whole valve (and possibly clearing it).
\begin{figure}[h]
  \centering
  \includegraphics[width=.7\linewidth]{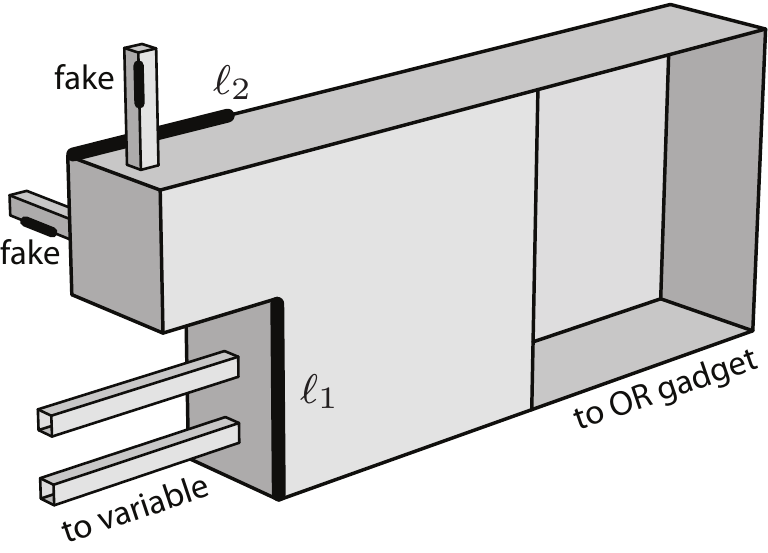}
  \caption{Valve.}
  \label{fig:valve}
\end{figure}

\subsection*{Reduction}

\begin{theorem}\label{hardopt}
It is strongly \NP-hard to compute the optimum search time in the \TTSSPext, even restricted to searchable orthogonal instances.
\end{theorem}
\begin{proof}
We use the construction described above, and prove that it yields the expected reduction.

Notice that building a chain with the correct parity and direction is never an issue, because three rooms and three links can be added to extend a chain without changing the direction of its end links. Moreover, if the chains are thin enough, the number of their bends can be bounded by a constant, while keeping them pairwise disjoint and with rational vertices. Hence the size of the whole construction is indeed polynomial in the size of~$\varphi$.

\paragraph{Positive instances.}
Suppose now that $\varphi$ is satisfiable, and let us show that our construction can be cleared in three seconds. For each variable $x$, the guard in $\mathcal C_x$ starts with its searchlight aiming at its right, while the guard in $\mathcal C'_x$ starts with its searchlight aiming at its left (resp.\ right) if $x$ is false (resp.\ true) in the chosen satisfying assignment. Consider now an occurrence of $x$, and its corresponding valve. If the occurrence is satisfied, both links attached to the valve can be cleared at second~$1$. So guard $\ell_1$ caps the links during second~$1$, waits for $\ell_2$ to move during second~$2$ (thus protecting the chains from the contaminated fake link still uncapped by $\ell_2$), and finally moves to close the valve during second~$3$. On the other hand, if the occurrence of $x$ is not satisfied, $\ell_1$ keeps capping the links and never moves (since they have to be cleared during seconds~$1$ and~$3$, respectively), while $\ell_2$ moves during second~$2$. Thus, by assumption, at the end of second~$3$, each OR gate will have at least one hole closed by the corresponding valve. The guards in the OR gate keep the remaining two holes closed throughout seconds~$1$ and~$2$, and finally the proper guard turns to clear the OR gate during second~$3$. If a valve's hole is initially closed by the OR gate, then the valve itself is successfully cleared by guard $\ell_2$ during second~$2$, hence the OR gate cannot be contaminated by the valve during second~$3$. If a valve is initially open, then a small portion of it (the \textquotedblleft niche\textquotedblright\ with the fake links) is cleared by $\ell_2$ at second~$2$, while the rest is cleared by $\ell_1$ at second~$3$. Notice also that no recontamination can ever occur between an OR gate and some auxiliary fake link in one of its valves, because guard $\ell_1$'s searchlight always separates the two regions.
 
 \paragraph{Negative instances.}
Conversely, let us assume by contradiction that $\varphi$ is not satisfiable, and some search schedule clears our construction within three seconds. As already noted, every variable gadget must be cleared at second~$2$, while each link must be cleared at second~$1$ or~$3$. By assumption, in at least one clause gadget every valve is attached to a link that is cleared during second~$1$ and to a link that is cleared during second~$3$. While links are being cleared, guard $\ell_1$ has to cap them, and it is too slow to approach the valve's hole and come back in place within second~$2$. Therefore, the region around the hole must be cleared by guard $\ell_2$, which is allowed to move only during second~$2$, because of the fake links it has to cap at seconds~$1$ and~$3$. To avoid recontamination while $\ell_2$ sweeps the hole, some guard in the OR gate must keep it closed for $\varepsilon>0$ seconds. But since there are only two such guards for three holes, one guard has to close two holes in strictly less than one second, which is impossible.
\end{proof}

By further inspecting the construction given above, it is clear that \TTSSP does not even have a \PTAS, unless $\mbox{\P} =\mbox{\NP}$. Indeed, satisfiable Boolean formulas are transformed into polyhedra searchable in three seconds, while the unsatisfiable ones are transformed into polyhedra that are unsearchable in less than $3+\varepsilon$ seconds, for a suitable small-enough $\varepsilon > 0$.

\begin{corollary}
The \TTSSPext does not have a \PTAS (unless $\P=\NP$), even when restricted to searchable orthogonal instances.\hfill\qed
\end{corollary}

Determining approximating algorithms for \TTSSP is left as an open problem (as it is even for 2-dimensional \SSP).
\chapter{Complexity of partial searching}\label{chapter11}
\begin{chapterabstract}
We introduce the \PSSPext, in which just a given subregion of the environment has to be cleared, either by catching the intruder or by confining it outside. Note that the traditional \SSPext comes as a subproblem, in which the entire environment must be cleared.

In the 3-dimensional case, we show that the decision version of the \PSSPext is strongly \PSPACE-hard, even restricted to orthogonal polyhedra and boundary guards.

Then we improve our construction, to show that the problem is strongly \PSPACE-complete also in the 2-dimensional case, restricted to orthogonal polygons and a rectangular region to be searched. Our last result stands as the first characterization of the computational complexity of a \SSPext.
\end{chapterabstract}

\section{Partial searching}

Here we introduce a slightly generalized problem: suppose the guards have to clear only a given subregion of the environment, while the rest may remain contaminated. In particular, for 3-dimensional polyhedra, we stipulate that the \emph{target area} that need be cleared is expressed as a ball, whose center and radius are given as input along with the polyhedron and the multiset of guards. We call the resulting problem \textsc{3-dimensional Partial Searchlight Scheduling Problem} (\TRSSP).

\begin{definition}[\TRSSP]\emph{\TRSSP}\ is the problem of deciding if the guards of a given instance of \TSSP\ have a schedule that clears a ball with given center and radius.\end{definition}

The terminology defined in Chapter~\ref{chapter7} for \TSSP\ extends straightforwardly to \TRSSP.

In the 2-dimensional case, we are less restrictive, and we allow any subpolygon to be the target area.

\begin{definition}[\PSSP]\emph{\PSSP}\ is the problem of deciding if the guards of a given instance of \SSP\ have a schedule that clears a given subpolygon.\end{definition}

\section{\PSPACE-hardness for polyhedra}

Next we are going to prove that \TRSSP\ is strongly \PSPACE-hard, even for orthogonal polyhedra. To do so, we give a reduction from the \emph{edge-to-edge} problem for \emph{AND/OR networks} in the \emph{nondeterministic constraint model} (see~\cite{ncl}).

\subsection*{Nondeterministic constraint logic machines}

Consider an undirected 3-connected 3-regular planar graph, whose vertices can be of two types: \emph{AND vertices} and \emph{OR vertices}. Of the three edges incident to an AND vertex, one is called its \emph{output edge}, and the other two are its \emph{input edges}. Such a graph is (a special case of) a \emph{nondeterministic constraint logic machine (NCL machine)}. A \emph{legal configuration} of an NCL machine is an orientation (direction) of its edges, such that:
\begin{itemize}
\item for each AND vertex, either its output edge is directed inward, or both its input edges are directed inward;
\item for each OR vertex, at least one of its three incident edges is directed inward.
\end{itemize}
A \emph{legal move} from a legal configuration to another configuration is the reversal of a single edge, in such a way that the above constraints remain satisfied (i.e., such that the resulting configuration is again legal).

Given an NCL machine with two \emph{distinguished edges} $e_a$ and $e_b$, and a \emph{target orientation} for each, we consider the problem of deciding if there are legal configurations $A$ and $B$ such that $e_a$ has its target orientation in $A$, $e_b$ has its target orientation in $B$, and there is a sequence of legal moves from $A$ to $B$. In a sequence of moves, the same edge may be reversed arbitrarily many times. We call this problem \textsc{Edge-to-Edge for Nondeterministic Constraint Logic machines} (\EENCL).

A proof that \EENCL\ is \PSPACE-complete is given in~\cite{ncl}, by a reduction from \computproblem{True Quantified Boolean Formula}. Based on that reduction, we may further restrict the set of \EENCL\ instances on which we will be working. Namely, we may assume that $e_a\neq e_b$, and that in no legal configuration both $e_a$ and $e_b$ have their target orientation.

\subsection*{Asynchrony}

For our main reduction, it is more convenient to employ an \emph{asynchronous} version of \EENCL. Intuitively, instead of \textquotedblleft instantaneously\textquotedblright\ reversing one edge at a time, we allow any edge to start reversing at any given time, and the \emph{reversal phase} of an edge is not \textquotedblleft atomic\textquotedblright\ and instantaneous, but may take any strictly positive amount of time. It is understood that several edges may be in a reversal phase simultaneously. While an edge is reversing, its orientation is undefined, hence it is not directed toward any vertex. During the whole process, at any time, both the above constraints on AND and OR vertices must be satisfied. We also stipulate that no edge is reversed infinitely many times in a bounded timespan, or else its orientation will not be well-defined in the end. With these extended notions of configuration and move, and with the introduction of \textquotedblleft continuous time\textquotedblright, \EENCL\ is now called \textsc{Edge-to-Edge for Asynchronous Nondeterministic Constraint Logic machines} (\EEANCL).

Despite its asynchrony, such new model of NCL machine has precisely the same power of its traditional synchronous counterpart.

\begin{proposition}\label{asynch}$\EENCL = \EEANCL$.\end{proposition}

\begin{proof}
Obviously $\EENCL \subseteq \EEANCL$, because any sequence of moves in the synchronous model trivially translates into an equivalent sequence for the asynchronous model.

For the opposite inclusion, we show how to \textquotedblleft serialize\textquotedblright\ a legal sequence of moves for an asynchronous NCL machine going from a legal configuration $A$ to configuration $B$ in a bounded timespan, in order to make it suitable for the synchronous model. An asynchronous sequence is represented by a set
$$S=\{(e_m, s_m, t_m) \mid m\in M\},$$ where $M$ is a set of \textquotedblleft edge reversal events\textquotedblright, $e_m$ is an edge with a reversal phase starting at time $s_m$ and terminating at time $t_m > s_m$. For consistency, no two reversal phases of the same edge may overlap.

Because no edge can be reversed infinitely many times, $S$ must be finite. Hence we may assume that $M=\{1, \cdots, n\}$, and that the moves are sorted according to the (weakly increasing) values of $s_m$, i.e., $1\leqslant m < m' \leqslant n \implies s_m \leqslant s_{m'}$. Then we consider the serialized sequence
$$S'=\{(e_m, m, m) \mid m\in M\},$$
and we claim that it is valid for the synchronous model, and that it is equivalent to $S$.

Indeed, each move of $S'$ is instantaneous and atomic, no two edges reverse simultaneously, and every edge is reversed as many times as in $S$, hence the final configuration is again $B$ (provided that the starting configuration is $A$). We still have to show that every move in $S'$ is legal. Let us do the first $m$ edge reversals in $S'$, for some $m\in M$, starting from configuration $A$, and reaching configuration $C$. To prove that $C$ is also legal, consider the configuration $C'$ reached in the asynchronous model at time $s_m$, according to $S$, right when $e_m$ starts its reversal phase (possibly simultaneously with other edges). By construction of $S'$, every edge whose direction is well-defined in $C'$ (i.e., every edge that is not in a reversal phase) has the same orientation as in $C$. It follows that, for each vertex, its inward edges in $C$ are a superset of its inward edges in $C'$. By assumption on $S$, $C'$ satisfies all the vertex constraints, then so does $C$, {\it a fortiori}.
\end{proof}

\begin{corollary}\label{eeancl}\EEANCL\ is \PSPACE-complete.\end{corollary}

\begin{proof}
Recall from~\cite{ncl} that \EENCL\ is \PSPACE-complete, and that $\EENCL = \EEANCL$ by Proposition~\ref{asynch}.
\end{proof}

\subsection*{Building blocks}

We realize a given NCL machine in terms of an orthogonal polyhedron consisting of three levels, called \emph{basement}, \emph{floor} and \emph{attic}, respectively. The floor contains the actual AND/OR network, arranged as an orthogonal plane graph, and is completely unsearchable. The attic is reachable from the floor through two \emph{stairs}, and contains the target ball that has to be cleared. The basement level is just a network of \emph{pipes} connecting different parts of the floor to the stairs, whose purpose is to recontaminate the stairs and the attic unless the floor guards actually \textquotedblleft simulate\textquotedblright\ a legal sequence of moves of the edges of an asynchronous NCL machine.

It is well-known that any 3-regular planar graph can be embedded in the plane as an \emph{orthogonal drawing}. For instance, we can employ the algorithm given in~\cite{planar}, which works with 3-connected 3-regular planar graphs, such as the networks of our NCL machines. The resulting drawing is orthogonal, in the sense that every edge is a sequence of horizontal and vertical line segments, which will be called \emph{subedges}. For our construction, we turn every subedge into a thin-enough cuboid, then we place a \emph{subedge guard} in the middle of each cuboid, on the bottom face, as depicted in Figure~\ref{figp:bend}. The dashed squares between consecutive subedges are called \emph{trapdoors}, and denote areas that will be attached to pipes and connected to other regions, as described later.

\begin{figure}[h]
\centering
\includegraphics[scale=1.25]{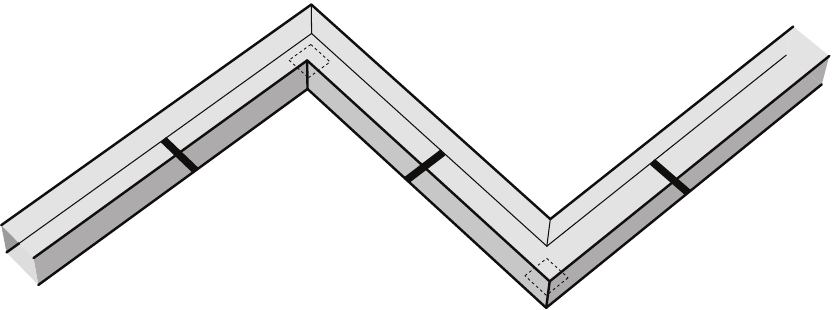}
\caption{Edge made of three orthogonal subedges.}
\label{figp:bend}
\end{figure}

Next we model OR vertices like in Figure~\ref{figp:or}. The three incoming cuboids carrying guards are subedges constructed previously. Again, the dashed square in the middle is a trapdoor that will be attached to pipes. Notice that the trapdoor completely belongs to the visibility region of each of the three guards, as the dashed lines suggest. Moreover, the two subedges coming from opposite directions are displaced, so they do not interfere with each other, in the sense that none of their two guards can see the opposite subedge through the end.

\begin{figure}[h]
\centering
\includegraphics[scale=1.15]{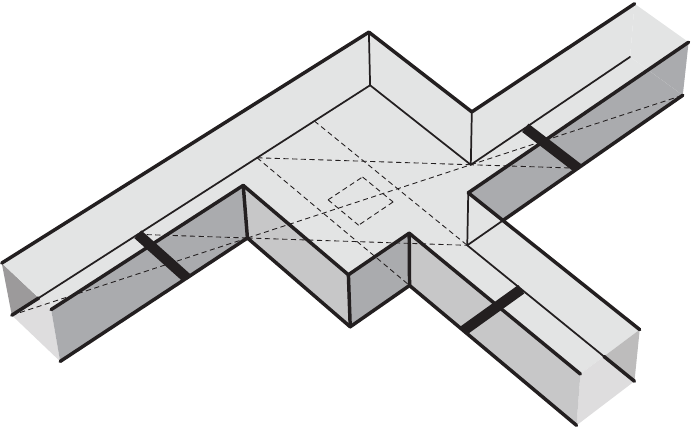}
\caption{OR vertex.}
\label{figp:or}
\end{figure}

We model the AND vertices as shown in Figure~\ref{figp:and}. The output edge is the one whose guard sees both trapdoors, while the guards in the two input edges can see only one trapdoor each. We can always arrange the drawing of our graph by further bending its edges, in such a way that this construction is feasible (i.e., the output edge is located \textquotedblleft between\textquotedblright\ the input edges), as suggested in Figure~\ref{figp:graph}.

\begin{figure}[h]
\centering
\includegraphics[scale=1.15]{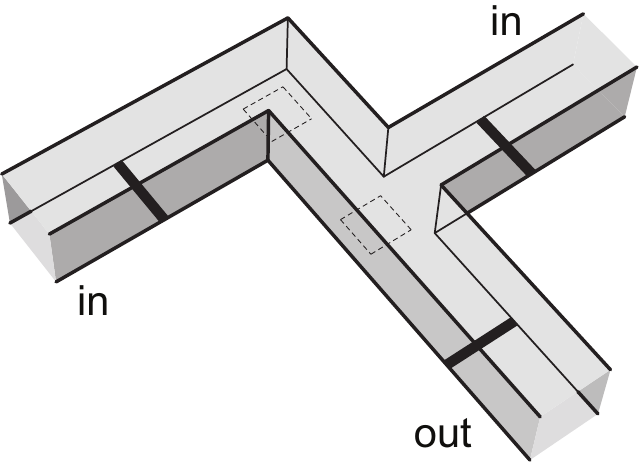}
\caption{AND vertex.}
\label{figp:and}
\end{figure}

\begin{figure}[h]
\centering
\subfigure[]{\includegraphics[scale=1.2]{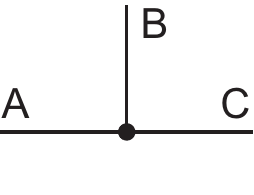}}\qquad \qquad
\subfigure[]{\includegraphics[scale=1.2]{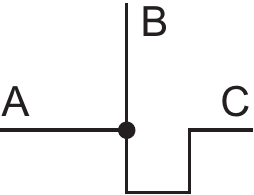}}
\caption{Rearranging the edges of an orthogonal drawing.}
\label{figp:graph}
\end{figure}

Recall that an instance of \EEANCL\ comes with two distinguished edges $e_a$ and $e_b$, each of which is embedded in our construction as a sequence of cuboidal subedges. We select an \emph{internal} subedge $s_a$ of $e_a$ (i.e., not the first, nor the last subedge of $e_a$) and an internal subedge $s_b$ of $e_b$, such that $s_a$ and $s_b$ run in two orthogonal directions. If no such subedges exist in our construction, we can further subdivide $e_a$ and $e_b$ into more small-enough \textquotedblleft redundant\textquotedblright\ subedges, in order to obtain internal ones running in the desired directions. Recall also that both $e_a$ and $e_b$ come in \EEANCL\ with a target orientation, which we want them to reach in order to solve the problem instance. Such target orientation is therefore naturally \textquotedblleft inherited\textquotedblright\ by $s_a$ and $s_b$, too.

\begin{figure}[h]
\centering
\includegraphics[scale=1.75]{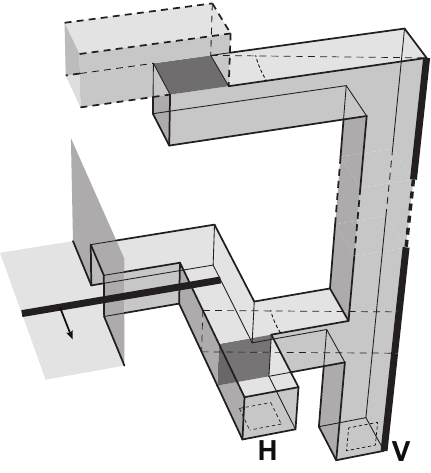}
\put(-194,89){$\ell_a$}
\put(-80,205.5){${}_\alpha$}
\put(-74,76.75){${}_\alpha$}
\caption{Stair.}
\label{figp:stair}
\end{figure}

Let $\ell_a$ and $\ell_b$ be the two subedge guards lying in $s_a$ and $s_b$, respectively. We add two stairs going up to the attic, one for $\ell_a$ and one for $\ell_b$. Figure~\ref{figp:stair} shows a close-up of one end of $\ell_a$, where we have attached a stair (the two \textquotedblleft incomplete\textquotedblright\ rectangles represent faces of $s_a$). Observe that, besides adding the polyhedral model of the stair to $s_a$, we also extend $\ell_a$ to the stair itself. The arrow attached to $\ell_a$ indicates the target orientation of $s_a$, inherited from the instance of \EEANCL. There are two trapdoors, depicted as horizontal and vertical dashed squares. The horizontal trapdoor (labeled H in the picture) lies on the bottom face of an enclosed cuboidal region that we call \emph{alcove}, whose opening is indicated by a darker vertical square. The alcove completely belongs to $\mathcal V(\ell_a)$, and it can also be capped by the long vertical guard showed in the picture, called \emph{stair guard}. While the stair guard caps the alcove, it also covers the vertical trapdoor (labeled V). On the other hand, $\ell_a$ is able to cover the horizontal trapdoor. The top cuboid with dashed edges belongs to the attic, and is not considered part of the stair. The darker horizontal square that separates the attic from the stair is called \emph{attic entrance}.

A similar construction is then repeated for $\ell_b$, and analogous remarks hold.

\begin{figure}[h]
\centering
\includegraphics[scale=1.3]{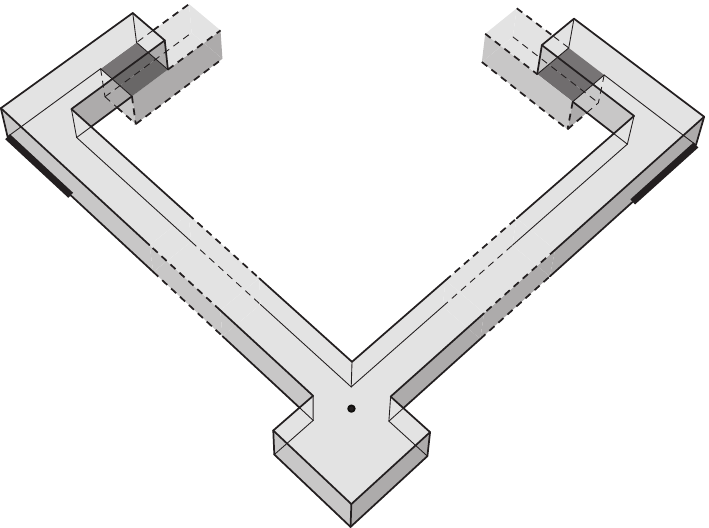}
\caption{Attic.}
\label{figp:attic}
\end{figure}

The whole attic is illustrated in Figure~\ref{figp:attic}, where the darker squares denote the two attic entrances mentioned above, and the underlying cuboids with dashed edges belong to the stairs. A long L-shaped \emph{corridor}, made of two orthogonal \emph{branches}, connects the two entrances, and each entrance can be covered by an \emph{attic guard} located at one end of the corridor. The small dot in the picture, in the middle of the corridor, denotes the target ball that has to be cleared by the guards. There is also a \emph{widening} on one side of the corridor, which is not fully visible to the guards, and hence is a perpetual source of recontamination.

To make sure that the floor is really unsearchable, we add a \emph{groove} all around an end of each subedge guard, including $\ell_a$ and $\ell_b$, like in Figure~\ref{figp:groove}. Every groove is always a source of recontamination, because some parts of it can never be seen by any guard, and cannot be isolated from the rest of the floor, either.

\begin{figure}[h]
\centering
\includegraphics[scale=1]{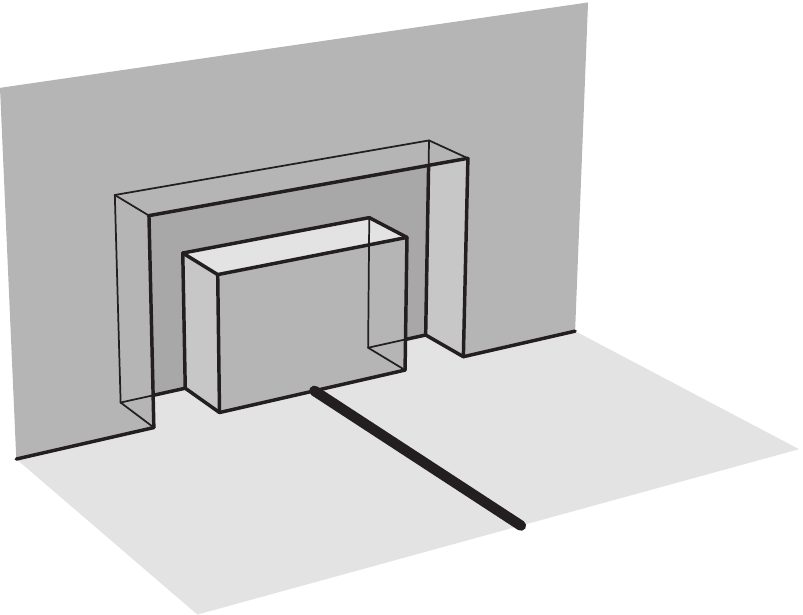}
\caption{Close-up of a subedge guard, with its groove.}
\label{figp:groove}
\end{figure}

Finally, at the basement level, we add pipes, i.e., twisted chains of very thin cuboids, to connect pairs of trapdoors. We connect each trapdoor in each stair to every trapdoor in the floor (i.e., the trapdoors in the AND/OR vertices and the trapdoors between subedges). Since there are two different types of trapdoors in the stairs (horizontal and vertical), while the trapdoors in the floor are all horizontal, we need two types of pipes, as Figure~\ref{figp:pipes} suggests. The end of a pipe that is labeled A goes into a stair trapdoor, whereas the end labeled B goes into a floor trapdoor. Observe that a pipe shaped like in Figure~\ref{figp:pipesb} can always connect a vertical stair trapdoor with any floor trapdoor, except when the latter lies exactly \textquotedblleft behind\textquotedblright\ the former, and the lower part of the stair gets in the way. This can be easily prevented, for example by constructing a \textquotedblleft thinner\textquotedblright\ version of the stair itself, such that no floor trapdoor lies completely behind the lower part of the stair (i.e., the part with the vertical trapdoor). In general, all pipes lie below the floor level, and their mutual intersections can be resolved by shrinking them, as already discussed in Chapter~\ref{chapter10} for links.

\begin{figure}[h]
\centering
\subfigure[]{\label{figp:pipesa}\includegraphics[scale=0.85]{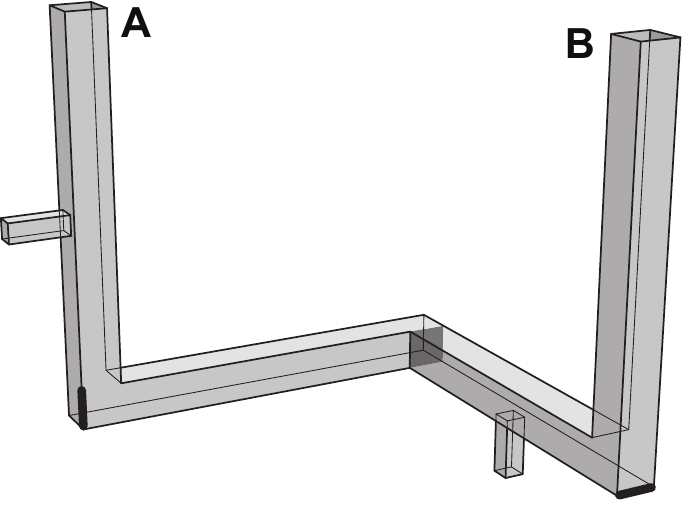}}\qquad
\subfigure[]{\label{figp:pipesb}\includegraphics[scale=0.85]{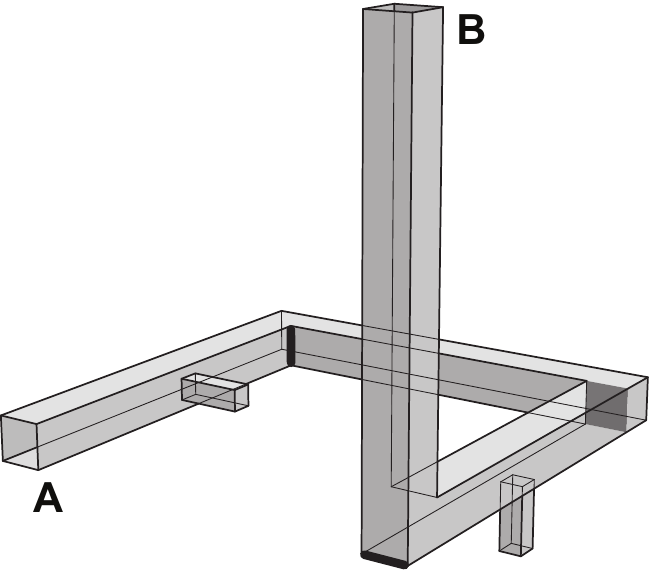}}
\put(-334.75,9.75){$\ell_1$}
\put(-188,2.75){$\ell_2$}
\put(-90.65,39.3){$\ell_1$}
\put(-81.75,3.45){$\ell_2$}
\caption{Two pipes. The pipe in \subref{figp:pipesa} connects two horizontal trapdoors, the pipe in \subref{figp:pipesb} connects a vertical and a horizontal trapdoor.}
\label{figp:pipes}
\end{figure}

Notice that two very small extra cuboids are attached to each pipe: these are called \emph{pits}. Pits are unsearchable regions, but each pit can indeed be capped by a \emph{pipe guard}. Such guards are also responsible for clearing the rest of the pipe, when both its ends A and B are capped by external searchplanes (belonging to a stair guard and to a subedge guard, respectively). Thus, by \emph{clearing a pipe} we will mean clearing its chain of four bigger cuboids connecting A with B, disregarding the two pits.

\begin{lemma}\label{pipe1}While both its A and B ends are completely illuminated by external guards, a pipe can be cleared.\end{lemma}

\begin{proof}
A similar schedule works for both types of pipe. Referring to Figure~\ref{figp:pipes}, guard $\ell_1$ covers the darker square, allowing $\ell_2$ to clear two of the four cuboids. Then $\ell_2$ keeps capping its pit, while $\ell_1$ sweeps the remaining two cuboids, and caps the other pit.
\end{proof}

On the other hand, an uncapped pipe acts as a \textquotedblleft one-way recontaminator\textquotedblright\ from B to A.

\begin{lemma}\label{pipe2}If neither A nor B is illuminated by external guards, then A is (partly) contaminated.\end{lemma}

\begin{proof}
If $\ell_1$ does not cap its pit, then the pit contaminates A. If $\ell_1$ caps its pit but $\ell_2$ does not, then the second pit contaminates A. If both guards cap their respective pits, then B and A are connected by a non-illuminated path. Because of the unsearchable grooves around all subedge guards, B must be (partly) contaminated, and so is A.
\end{proof}

\subsection*{Reduction}

\begin{theorem}\label{trssp}\TRSSP\ is strongly \PSPACE-hard, even for orthogonal polyhedra.\end{theorem}

\begin{proof}
We give a reduction from \EEANCL\ to \TRSSP, by proving that the target ball in the above construction is clearable if and only if the two distinguished edges $e_a$ and $e_b$ can be oriented in their target directions one after the other, by a legal sequence of asynchronous moves. Observe that our construction is obviously a polyhedron (it is indeed connected) whose vertices' numerical coordinates may be chosen to be polynomially bounded in size, with respect to the size of the NCL network. Moreover, the orthogonal drawing of the network can be obtained in linear time, for instance as shown in~\cite{planar}.

\paragraph{Positive instances.}

Suppose that the given instance of \EEANCL\ is solvable. Then there exists a legal sequence $S$ of asynchronous moves that, starting from a configuration $A$ in which $e_a$ is in its target direction, ends in a configuration $B$ in which $e_b$ is in its target direction. By the assumptions we made on NCL networks, $e_a\neq e_b$ and both $e_a$ and $e_b$ are reversed by $S$ at least once.

We start by \textquotedblleft replicating\textquotedblright\ configuration $A$ on the subedge guards in our construction: if an edge has an orientation in $A$, then all its subedges in the drawing inherit its orientation, and all the corresponding subedge guards are oriented accordingly. As a result, every trapdoor in the floor (not the four trapdoors in the stairs) is covered by a guard, since $A$ is a legal configuration. Indeed, the structure of the AND/OR vertices that we built implies that the NCL constraints on a vertex are satisfied if and only if all the trapdoors in its polyhedral model are covered. Also the trapdoors between subedges happen to be covered, because all the guards in a subedge chain are oriented in the same \textquotedblleft direction\textquotedblright. Incidentally, $\ell_a$ covers the horizontal trapdoor in its corresponding stair, as well. Next we cover the vertical trapdoors in both stairs with the stair guards, thus incidentally also capping both alcoves. Finally, we cover both attic entrances with the two attic guards.

In order to clear the target ball in the attic, our schedule proceeds as follows.
\begin{enumerate}
\item Clear the pipes that are attached to the stair trapdoors corresponding to $s_a$. This is feasible by Lemma~\ref{pipe1}, because both ends of such pipes are capped.
\item\label{st2} Replicate all the moves of $S$, with the correct timing. If a move reverses edge $\{u,v\}$ by turning it from vertex $u$ toward vertex $v$, we first reverse the guard corresponding to the subedge incident to $u$, then we reverse all the other subedge guards one by one in order, until we reverse the last guard in the chain, whose subedge is incident to $v$. By doing so, no trapdoor is ever uncovered, hence no pipe is recontaminated. Moreover, both $\ell_a$ and $\ell_b$ reverse during the process, incidentally clearing both alcoves, which are still capped by the stair guards.
\item Clear the pipes attached to the stair trapdoors corresponding to $s_b$, again by Lemma~\ref{pipe1}. Indeed, after Step~\ref{st2}, these pipes are all capped.
\item\label{st3} Clear the regions underlying the two attic entrances, by turning both stair guards by $\alpha$ (refer to Figure~\ref{figp:stair}). Angle $\alpha$ is such that, in the end, both attic entrances are separated from the contaminated floor by the searchlights of the stair guards. No contamination is possible through the stair trapdoors either, because all the pipes are still clear (their B ends are still capped).
\item Turn the two attic guards in concert, until they clear the target ball. No recontamination may occur through the attic entrances, whose underlying regions have effectively been cleared in Step~\ref{st3}, whereas the contaminated widening of the corridor is always separated from the portion of corridor that has been swept.
\end{enumerate}

\paragraph{Negative instances.}

Conversely, suppose that no legal sequence of asynchronous moves solves the given instance of \EEANCL, and let us prove that the target ball in the attic is unclearable.

In order to clear the target ball, both attic guards have to turn in concert, away from the attic entrances. Indeed, just one guard is insufficient to clear anything. On the other hand, the attic guards cannot sweep toward the entrances, because of the unavoidable recontaminations from the widening in the corridor.

Therefore, while the attic guards operate, recontamination has to be avoided from the attic entrances. It follows that both stair guards must keep the attic separated from the unsearchable floor. Referring to Figure~\ref{figp:stair}, each stair guard's angle has to be at least $\alpha$.

When in that position, the stair guards cannot cover the vertical trapdoors, nor cap the alcoves. This means that all the pipes attached to a vertical trapdoor and both alcoves have to be simultaneously clear at some time $t$. According to Lemma~\ref{pipe2}, since the stair guards are not capping the A ends of those pipes, then their B ends have to be capped. In other words, all the trapdoors in the floor (not in the stairs) have to be covered. Equivalently, the subedge guards' orientations must correspond to a legal configuration of an asynchronous NCL machine operating in the given network. If the orientations of the subedge guards in the chain corresponding to a same edge do not agree with each other, then that edge is considered in a reversal phase.

Recall that, in any legal configuration of the NCL network, at least one of the two distinguished edges must be oriented strictly opposite to its target direction (i.e., it is not even in a reversal phase), hence at least one horizontal trapdoor in a stair (say, the stair attached to $s_a$, without loss of generality) must be uncovered at time $t$. However, its alcove and attached pipes must indeed be clear, which means that $\ell_a$ has capped the pipes for some time, allowing them to clear themselves, and then has left the alcove for the last time at $t'$, with $t'<t$. Between time $t'$ and time $t$, the orientation of the subedge guards must always correspond to a legal configuration of the NCL network, otherwise the alcove just cleared by $\ell_a$ would be recontaminated by its pipes, again by Lemma~\ref{pipe2}.

Observe that, since at time $t'$ guard $\ell_a$ is not oriented strictly opposite to its (inherited) target direction, then $\ell_b$ must be. But, at time $t$, also the alcove corresponding to $s_b$ must be clear, hence there must be some time $t''$, with $t''<t$, when $\ell_b$ last touched that alcove. Once again, between time $t''$ and time $t$, all the subedge guards must be oriented according to some legal configuration of the contraint graph.

Between time $t'$ and time $t''$ (no matter which comes first), the search schedule must simulate an asynchronous legal sequence of moves between a configuration in which $e_a$ has its target orientation (more appropriately, $e_a$ is in a reversal phase), and one in which $e_b$ does. By hypothesis on the given NCL network, such sequence of moves does not exist: notice in fact that a sequence of moves from $A$ to $B$ is legal if and only if its reverse sequence from $B$ to $A$ is legal.

This concludes the reduction. As a result, because \EEANCL\ is \PSPACE-complete by Corollary~\ref{eeancl}, \TRSSP\ is strongly \PSPACE-hard.
\end{proof}

\section{\PSPACE-completeness for polygons}

Here we consider the 2-dimensional \PSSPext (\PSSP), proving that it is \PSPACE-complete, even for orthogonal polygons. Our reduction is not exactly an improvement on the one given in the previous section, in that our target area is not an arbitrarily small ball, but it is a polygon that could be fairly large. However, we will also show how this target polygon can be assumed to be rectangular.

To prove that \PSSP belongs to \PSPACE, we use the discretization technique of~\cite{bullo} (see also Chapter~\ref{chapter1}), and to prove that \PSSP is \PSPACE-hard we give a reduction from \EEANCL.

\subsection*{\PSSP is in \PSPACE}

\begin{lemma} \label{lemma1}
\PSSP $\in$ \PSPACE.
\end{lemma}
\begin{proof}
As explained in~\cite{bullo} and Chapter~\ref{chapter1}, a technique known as \emph{exact cell decomposition} allows to reduce the space of all possible schedules to a finite graph $G$. Each searchlight has a linear number of \emph{critical angles}, which yield an overall partition of the polygon into a polynomial number of \emph{cells}. Searchlights take turns moving, and can stop or change direction only at critical angles. Thus, a vertex of $G$ encodes the \emph{status} of each cell (either \emph{contaminated} or \emph{clear}) and the critical angle at which each searchlight is oriented.

As a consequence, $G$ can be navigated nondeterministically by just storing one vertex at a time, which requires polynomial space. Notice that deciding if two vertices of $G$ are adjacent can be done in polynomial time: an edge in $G$ represents a move of a single searchlight between two consecutive critical angles, and the updated status of each cell can be easily evaluated. Indeed, cells' vertices are intersections of lines through input points, hence their coordinates can also be efficiently stored and handled as rational expressions involving the input coordinates.

Now, in order to verify that a path in $G$ is a witness for SSP, one checks if the last vertex encodes a status in which every cell is clear. But the very same cell decomposition works also for \PSSP: the analysis in~\cite{bullo} applies even if just a subregion of the polygon has to be searched, and a path in $G$ is a witness for \PSSP if and only if its last vertex encodes a status in which every cell that has a non-empty intersection with the target subregion is clear.

Since $\PSPACE = \NPSPACE$, due to Savitch's theorem (see~\cite{papadimitriou}), our claim follows.
\end{proof}

Recalling that \SSP is a special case of \PSSP, we immediately have the following:

\begin{corollary}
\SSP $\in$ \PSPACE.\hfill\qed
\end{corollary}

\subsection*{\PSSP is \PSPACE-hard}

For the \PSPACE-hardness part, we first give a reduction in which the target region to be cleared is an orthogonal hexagon. Later, we will explain how we would have to modify our construction, should we insist on having a rectangular (hence convex) target region.

\begin{lemma} \label{lemma2}
\EEANCL $\preceq_{\P}$ \PSSP restricted to orthogonal polygons.
\end{lemma}
\begin{proof}
We show how to transform a given asynchronous NCL machine $G$ with two distinguished edges $e_a$ and $e_b$ into an instance of \PSSP.

\paragraph{Construction.}
A rough sketch of our construction is presented in Figure~\ref{fig1}. All the vertices of $G$ are placed in a row~(a), and are connected together by a network of thin \emph{corridors}~(b), turning at right angles, representing edges of $G$. Each \emph{subsegment} of a corridor is a thin rectangle, containing a \emph{subsegment guard} in the middle (not shown in Figure~\ref{fig1}). Two subsegments from different corridors may indeed cross each other like in~(c), but in such a way that the crossing point is far enough from the ends of the two subsegments and from the two subsegment guards (so that no subsegment guard can see all the way through another subsegment). All the vertices of $G$ and all the \emph{joints} between consecutive subsegments (i.e., the turning points of each corridor) are connected via extremely thin \emph{pipes}~(d) to the upper area~(e), which contains the target region (shaded in Figure~\ref{fig1}).

\begin{figure}[ht]
\centering
\includegraphics[width=\linewidth]{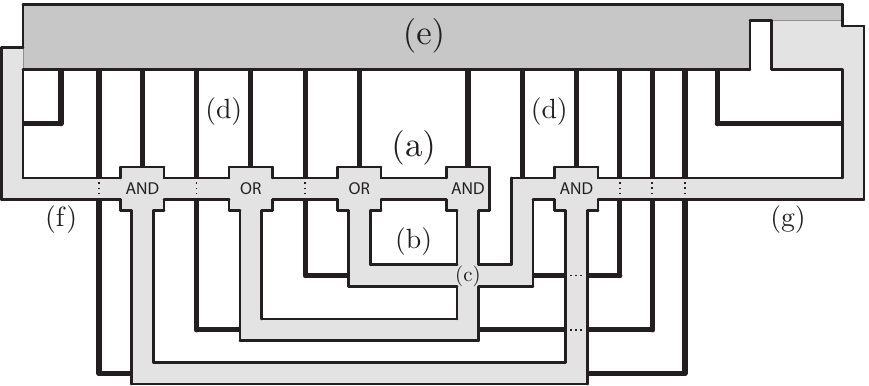}
\caption{Construction overview.}
\label{fig1}
\end{figure}

Two corridors~(f) and~(g) also reach the upper area, and they correspond to the distinguished edges of $G$, $e_a$ and $e_b$, respectively. That is, if $e_a=\{u,v\}$, and the target orientation of $e_a$ is toward $v$, then the corridor corresponding to $e_a$ connects vertex $u$ in our construction to the upper area~(e), rather than to $v$. The same holds for $e_b$. Indeed, observe that we may assume that $e_a$ and $e_b$ are reversed only once (on the first and last move, respectively) in a sequence of moves that solves \EEANCL on $G$. As a consequence, contributions to vertex constraints given by distinguished edges oriented in their target direction may be ignored.

Each pipe bends at most once, and contains one \emph{pipe guard} in the middle, lying on the boundary. Notice that straight pipes never intersect corridors, but some turning pipes do. Figure~\ref{fig2} shows a turning pipe, with its pipe guard~(a) and an intersection with a corridor~(b) (proportions are inaccurate). The \emph{intersection guards}~(c) separate the pipe from the corridor with their searchlights (dashed lines in Figure~\ref{fig2}), without \qq{disconnecting} the pipe itself. Although a pipe narrows every time it crosses a corridor, its pipe guard can always see all the way through it, because it is located in the middle. The small \emph{nook}~(d) is unclearable because no guard can see its bottom, hence it is a constant source of recontamination for the target region~(e), unless the pipe guard is covering it with its laser. (Each straight pipe also has a similar nook.)

\begin{figure}[h]
\centering
\includegraphics[width=0.75\linewidth]{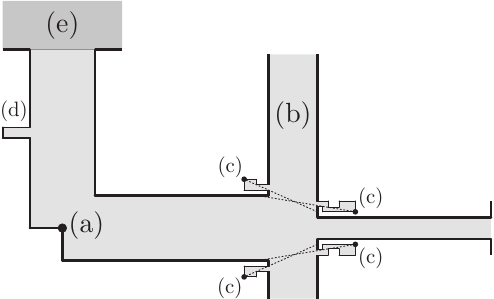}
\caption{Intersection between a pipe and a corridor.}
\label{fig2}
\end{figure}

Once again, corridor guards implement edge orientations in $G$: whenever all the subsegment guards in a corridor connecting vertices $u$ and $v$ have their searchlights oriented in the same \qq{direction} from vertex $u$ to vertex $v$, it means that the corresponding edge $\{u,v\}$ in $G$ is oriented toward $v$.

Figure~\ref{fig3} shows an OR vertex. The three subsegment guards from incoming corridors~(a) can all \qq{cap} pipe~(b) with their searchlights, and nook~(c) guarantees that the pipe is recontaminated whenever all three guards turn their searchlights away.

\begin{figure}[h]
\centering
\includegraphics[width=0.7\linewidth]{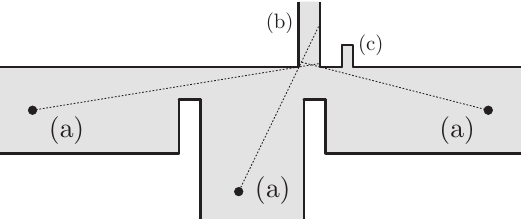}
\caption{OR vertex.}
\label{fig3}
\end{figure}

AND vertices are implemented as in Figure~\ref{fig4}. The two subsegment guards~(a) correspond to input edges, and are able to cap one pipe~(e) each, whereas guard~(c) can cover them both simultaneously. But that leaves pipe~(d) uncovered, unless it is capped by guard~(b), which belongs to the corridor corresponding to the output edge. Again, uncovered pipes are recontaminated by the unclearable nooks~(f).

\begin{figure}[h]
\centering
\includegraphics[width=0.7\linewidth]{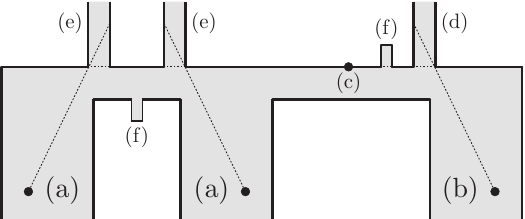}
\caption{AND vertex.}
\label{fig4}
\end{figure}

Joints between consecutive subsegments of a corridor may be viewed as OR vertices with two inputs, shaped like in Figure~\ref{fig3}, but without the corridor coming from the left.

Finally, Figure~\ref{fig5} shows the upper area of the construction, reached by the distinguished edges $e_a$ and $e_b$ ((a) and~(b), respectively), and by all the pipes~(c). The guard in~(d) can cap all the pipes, one at a time, and its purpose is to clear the left part of the target region, while the small rectangle~(e) on the right will be cleared by the guard in~(f). The two pipes~(g) implement additional OR vertices with two inputs, and prevent~(d) and~(f) from acting, unless the respective distinguished edges are in their target orientations. Nook~(h) will contaminate part of the target region, unless~(d) is aiming down. Nooks~(i) prevent area~(e) from staying clear whenever guard~(f) is not aiming up. The guard in~(j) separates the two parts of the target region with its laser, so that they can be cleared in two different moments.

\begin{figure}[h]
\centering
\includegraphics[width=\linewidth]{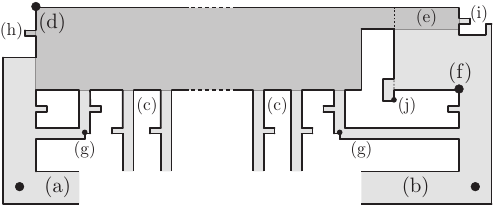}
\caption{Target region.}
\label{fig5}
\end{figure}

\paragraph{Positive instances.}
Suppose $G$ is a solvable instance of \EEANCL. Then we can \qq{mimic} the transition from configuration $A$ to configuration $B$ by turning subsegment guards. Specifically, if edge $e=\{u,v\}$ in $G$ changes its orientation from $u$ to $v$, then all the subsegment guards in the corridor corresponding to $e$ turn their lasers around, one at a time, starting from the guard closest to $u$. Before this process starts, each pipe has one end capped by some subsegment guard, and in particular pipe~(g) on the left of Figure~\ref{fig5} is capped by the guard in~(a). Hence, guard~(d) is free to turn and cap all the pipes one by one, stopping for a moment to let each pipe's internal guard clear the pipe itself (which now has both ends capped) and cover its nook (see Figure~\ref{fig2}). As a result, the left part of the target region can be cleared by rotating~(d) clockwise, from right to down. Then the subsegment guards start rotating as explained above, until configuration $B$ is reached. If done properly, this keeps all the pipes capped and clear, thus preventing the left part of the target region from being recontaminated. When $B$ is reached, guard~(f) can turn up to clear~(e) and finally solve our \PSSP instance. Notice that each subsegment guard has to turn in the correct direction (clockwise or counterclockwise) when it turns away from a pipe, to avoid the formation of \qq{contaminated bubbles} that may extend to the whole pipe.

\paragraph{Negative instances.}
Conversely, suppose that $G$ is not solvable. Observe that rectangle~(e) in Figure~\ref{fig5} has to be cleared by guard~(f) as a last thing, because it will be recontaminated by nooks~(i) as soon as~(f) turns away. On the other hand, as soon as a pipe has both ends uncapped by external guards, some portion of the target region necessarily gets recontaminated by some nook, regardless of where the pipe guard is aiming its searchlight. But guard~(d) can cap just one pipe at a time and, while it does so, nook~(h) keeps some portion of the target region contaminated. Thus, the entire process must start from a configuration $A$ in which all the pipes are simultaneously capped and guard~(d) is free to turn right (i.e., $e_a$ is in its target orientation), then proceed without ever uncapping any pipe (i.e., preserving legality), and finally reach a configuration $B$ in which guard~(f) is free to turn up (i.e., $e_b$ is in its target orientation). By assumption this is impossible, hence our \PSSP instance is unsolvable.
\end{proof}

By putting together Lemma~\ref{lemma1} and Lemma~\ref{lemma2}, we immediately obtain the following:

\begin{theorem} \label{maintheorem}
Both the \PSSPext and its restriction to orthogonal polygons are strongly \PSPACE-complete. \hfill \qed
\end{theorem}

The term \qq{strongly} is implied by the fact that all the vertex coordinates generated in the \PSPACE-hardness reduction of Lemma~\ref{lemma2} are numbers with polynomially many digits (or can be made so through negligible adjustments).

\subsection*{Convexifying the target region} \label{convexregion}

We can further improve our Theorem~\ref{maintheorem} by making the target region in Lemma~\ref{lemma2} rectangular.

Our new target region has the same width as the previous one, and the height of the small rectangle~(e) in Figure~\ref{fig5}. In order for this to work, we have to make sure that some portion of the target region is \qq{affected} by each contaminated pipe that is not capped by guard~(d), no matter where \emph{all} the pipe guards are oriented. To achieve this, we make pipes reach the upper area of our construction at increasing heights, from left to right, in a staircase-like fashion.

\begin{figure}[h]
\centering
\includegraphics[width=0.85\linewidth]{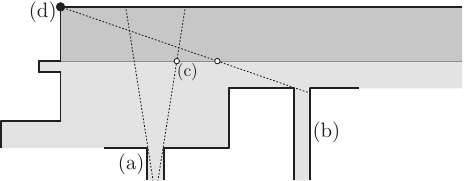}
\caption{Rectangular target region.}
\label{fig6}
\end{figure}

Assume we already placed pipe~(a) as in Figure~\ref{fig6}, and we need find the correct height at which it is safe to connect pipe~(b). First we find the rightmost intersection~(c) between a searchlight emanating from the pipe guard of~(a) and the lower border of the target region. Then we set the height of pipe~(b) so that it is capped by guard~(d) when it aims slightly to the right of~(c). This is always feasible, provided that pipes are thin enough, which is not an issue.

After we have set all pipes' heights from left to right, the construction is complete and the proof of Lemma~\ref{lemma2} can be repeated verbatim, yielding:

\begin{theorem}
Both the \PSSPext and its restriction to orthogonal polygons with rectangular target regions are strongly \PSPACE-complete. \hfill \qed
\end{theorem}

Investigating the complexity of \PSSP or \TPSSP on simply connected environments is left as an open problem.

Finally, we conjecture that all the other variants of the \SSPext are \PSPACE-complete, as well.

\begin{conjecture}
\SSP, \TSSP and \TPSSP are \PSPACE-complete.
\end{conjecture}

\begin{conclusions}\markthischapter{CONCLUSIONS}

\section*{Summary and research scenario}\markthissection{SUMMARY AND RESEARCH SCENARIO}

In this thesis we consider two guarding and searching problems in 3-dimensional polyhedra: the \ART and the \SSPext. Both are well-studied and understood in planar environments, but very little research had been done, prior to this thesis, in higher dimensional settings.

One of the reasons why research has eluded these problems so far is arguably the lack of effective ways to generalize most of the techniques and tools, like triangulations, that are often employed in solving 2-dimensional Computational Geometric problems.

Our goal is to obtain general results that extend well-known classic planar theorems as naturally as possible. Ideally, non-trivial 2-dimensional facts should come as special cases of our theorems, and increasing the dimension of our problems should allow us to look at them under a more insightful perspective.

The first thing to establish in the process of attempting such a generalization is how to meaningfully model the concept of guard, and in which fashion a guard should be able to see its surroundings. To this end, in Chapter~\ref{chapter3} we elucidate the main features of point guards, edge guards, and face guards, in relation to their ability to guard 3-dimensional polyhedra. As it turns out, point guards appear too ``weak'' to be effectively employed in polyhedra, in that a superlinear amount of them may be required to guard even orthogonal polyhedra. On the other hand, face guards fail to model any common type of agent, and appear too ``powerful'' to be realistic. Indeed, a ``unit weight'' face guard may represent a quadratic number of patroling point guards, even in orthogonal polyhedra.

In contrast, from preliminary observations, edge guards seem to act as the natural counterparts of vertex guards in polygons. Both the analogies with fluorescent lights and with patroling guards are sound, and a linear amount of edge guards is sufficient to guard any polyhedron. As edge guards are an excellent candidate to model 3-dimensional guards, the rest of the thesis is devoted to studying their capabilities.

We further introduce several different guarding modes: orthogonal guarding occurs when visibility rays are orthogonal to guards, while an edge guard can be closed or open, depending on whether it contains its endpoints. We argue that the open guard model is more realistic in some scenarios. Surprisingly, a small change such as depriving edge guards of their endpoints may increase the number of required guards by a factor of 3.

\paragraph{\ART.} For this problem, we mainly aim at finding relatively simple techniques that make up for the lack of generalized triangulations and convex quadrilateralizations, in order to obtain strong upper bounds on edge guard numbers in different settings.

We optimistically conjecture that the tight upper bound of $\lfloor r/2\rfloor+1$ vertex guards for orthogonal polygons with $r$ reflex vertices holds also for orthogonal polyhedra and edge guards. We are unable to prove this for all orthogonal polyhedra, but in Chapter~\ref{chapter4} we do it for the 2-reflex ones (i.e, those with reflex edges oriented in at most two distinct directions, as opposed to three). Further, we show that our edge guards can be open, and can even be placed on reflex edges.

This result not only incorporates the 2-dimensional one as a special case (i.e., when the polyhedron is a prism), but our algorithm partitions each 2-reflex orthogonal polyhedron into pieces that look like the polyhedron in Figure~\ref{fig4:monoguard}, and does so in $O(n \log n)$ time. When applied to orthogonal prisms, it yields the same partition into L-shaped pieces that O'Rourke's algorithm produces in the 2-dimensional case (see Chapter~\ref{chapter1} and\cite{art,urrutia2000}).

On the other hand, the class of 2-reflex orthogonal polyhedra, despite not including every orthogonal polyhedron, is already quite rich, and can model a wide range of 3-dimensional shapes. Moreover, we believe our partial result to be a necessary step toward a proof of the same upper bound for general orthogonal polyhedra.

An alternative technique that we explore is cutting a polyhedron with parallel planes, in order to reduce a 3-dimensional problem to a set of 2-dimensional ones. The main challenge here is to make sure that sub-problems belonging to ``neighboring'' planes are solved in a somewhat similar fashion, in such a way that all the 2-dimensional solutions can be subsequently merged into an efficient 3-dimensional one.

In Chapter~\ref{chapter5}, we successfully apply this concept to the problem of edge-guarding orthogonal polyhedra. As a by-product, our linear-time algorithm yields guards that are all parallel and, when applied to orthogonal polygons, it gives the same solution as the one given by Abello et al.\ in~\cite{aesu-iopof-98} for a type of guard called ``flashlight''. Applications of mutually parallel edge guards include point location, tracking, and navigation.

One of the highlights of Chapter~\ref{chapter5} are two tight inequalities relating the number of edges $e$ in an orthogonal polyhedron and the number of reflex edges $r$. We apply these inequalities to our previous algorithm to slightly lower the state-of-the-art upper bound for edge guards in terms of $e$ and to produce a new upper bound in terms of $r$.

Cutting with parallel planes proves to be a fruitful technique also for 4-oriented polyhedra, i.e., polyhedra whose faces are oriented in at most four different directions. This class greatly extends orthogonal polyhedra, and our technique yields an upper bound of $\lfloor (e+r)/6 \rfloor$ open edge guards, which can be found in linear time, as we show in Chapter~\ref{chapter6}. However, further ``pushing'' this strategy by applying it to broader classes of polyhedra produces results that are less attractive, and becomes completely pointless for general polyhedra.

Table~\ref{tc:1} summarizes some asymptotic upper and lower bounds on the number of \emph{open} edge guards required to solve the \ART in several types of genus-zero polyhedra, both in terms of the total number of edges $e$ and the number of reflex edges $r$.

\begin{table}[h]
\centering
\caption{Bounds on open edge guard numbers in several classes of polyhedra.}
\hfill \\
\begin{tabular}{c||c|c||c|c}
& \textbf{$e$-lower b.} & \textbf{$e$-upper b.} & \textbf{$r$-lower b.} & \textbf{$r$-upper b.}\\
\hline\hline
\textbf{1-reflex} & $e/12$ & $e/12$ & $r/2$ & $r/2$\\
\hline
\textbf{2-reflex} & $e/12$ & $e/8$ & $r/2$ & $r/2$\\
\hline
\textbf{3-reflex} & $e/12$ & $11e/72$ & $r/2$ & $7r/12$\\
\hline\hline
\textbf{4-oriented} & $e/6$ & $e/3$ & $r/2$ & $r$\\
\hline\hline
\textbf{General} & $e/6$ & $e$ & $r$ & $r$
\end{tabular}
\label{tc:1}
\end{table} 

The first three rows refer to orthogonal polyhedra. All lower bounds come from Chapter~\ref{chapter3}. The upper bounds in the first row follow directly from the classic theorems on planar orthogonal polygons, surveyed in Chapter~\ref{chapter1}. The second, third and fourth rows summarize some of the results of Chapters~\ref{chapter4},~
\ref{chapter5} and~\ref{chapter6}, respectively. Finally, the last row on general polyhedra follows from Chapter~\ref{chapter3} (recall that the $27e/32$ upper bound stated in Theorem~\ref{thm:cano} was established only for closed edge guards, hence it does not affect Table~\ref{tc:1}).

\paragraph{\SSPext.} Concerning this problem, our goal is once again to model searchlights as naturally as possible, yet in such a way that as many classic 2-dimensional results carry over to our 3-dimensional model. In addition, we would like our generalization to offer insights that could inspire new approaches to the 2-dimensional problem, as well.

After considering several plausible variations on the searchlight model, in Chapter~\ref{chapter7} we opt for half-planes of light rotating about segment guards. Hence, a guard can be thought of as an orientable array of sensors, each of which rapidly and continuously scans a planar region. This model is also quite scalable, in that, depending on the application, a guard can be approximated by a sparser array of sensors.

Having searchlights that can turn with just one degree of freedom allows for a trivial reduction from the 2-dimensional \SSPext to the 3-dimensional one. Moreover, we identify a broad class of guards, called filling guards, that act as a natural generalization of planar guards, and can be identified in polynomial time. Indeed, as shown in Chapter~\ref{chapter8}, the one-way sweep strategy detailed in~\cite{search1} extends to polyhedral environments with filling guards, along with all its consequences. Similarly, the sequentiability of search schedules, proved in~\cite{bullo} for polygons, extends to polyhedra with filling guards, too. Filling guards play a role also in the characterization of the searchable polyhedra that have only one guard.

Next we pose the problem of efficiently placing boundary guards in a given polyhedron in order to make it searchable. The corresponding question for polygons has never been asked in the literature, but it follows from the one-way sweep strategy that placing a point guard on each reflex vertex of a simple polygon is sufficient. We conjecture the same to be true for polyhedra (i.e., that placing a guard on each reflex edge makes any polyhedron searchable), and we manage to prove it in Chapter~\ref{chapter9} for orthogonal polyhedra. For general polyhedra, we only provide a quadratic upper bound on the number of required boundary guards.

Finally, we move on to computational complexity issues. Determining the complexity of the \SSPext (i.e., deciding if a given guard set can search a given polygon) has been an open problem for decades, and we show that our 3-dimensional extension is computationally hard. Specifically, in Chapter~\ref{chapter10} we prove that deciding if a 3-dimensional instance is searchable is strongly \NP-hard. In addition, even under the promise that the given instance is indeed searchable, the problem of minimizing search time has no \PTAS (unless $\P=\NP$).

In Chapter~\ref{chapter11}, we further generalize our problem, by asking if only a given subregion of the environment can be searched, regardless of the rest. We prove that such a \PSSPext is \PSPACE-hard, by a reduction from Nondeterministic Constraint Logic machines.

Our technique has some impact also on 2-dimensional problems: by adapting our reduction to polygons, we manage to prove that the 2-dimensional version of the \PSSPext is \PSPACE-complete, thus giving the first characterization of the computational complexity of a variation of the \SSPext.

The main computational complexity results of Chapters~\ref{chapter10} and~\ref{chapter11} are summarized in Table~\ref{tc:2}.

\begin{table}[h]
\centering
\caption{Computational complexity of several variations of \SSP.}
\hfill \\
\begin{tabular}{c|c|c}
& \textbf{Full} & \textbf{Partial}\\
\hline
\textbf{2D} & \PSPACE & \PSPACE-complete\\
\hline
\textbf{3D} & \NP-hard & \PSPACE-hard
\end{tabular}
\label{tc:2}
\end{table} 

In conclusion, we succeeded, to some extent, in elucidating the role of edge guards in polyhedra, both in relation to the \ART and the \SSPext. We showed that these guards naturally generalize point guards in polygons, and we illustrated how several planar theorems carry over to 3-dimensional settings, despite the lack of powerful general theoretical tools such as triangulations.

We also argued that reasoning about polyhedra could provide a more insightful view of planar problems, by revealing a ``bigger picture'' that is often hidden or unclear when confined to the planar case. We demonstrated that tackling 3-dimensional problems may even suggest novel approaches to their 2-dimensional versions, as in the case of the \PSSPext.

\section*{Open problems}\markthissection{OPEN PROBLEMS}

Several open problems are left for further research. We name a few:

\begin{itemize}
\item Some of the lower and upper bounds in Table~\ref{tc:1} do not match. We believe all the lower bounds to be tight, as stated in several conjectures in Chapters~\ref{chapter3},~\ref{chapter4},~\ref{chapter5}, and~\ref{chapter6}.

\item The same construction used in Theorem~\ref{th:3} could
be analyzed more closely to achieve a better upper bound. In contrast with the fact that the polyhedra with highest $r$ to $e$ ratio are responsible for the worst cases in
our analysis, such polyhedra are nonetheless
intuitively easy to guard by selecting a small fraction of
their reflex edges. Isolating these cases and analyzing them
separately may yield an improved overall bound.

\item The technique used in Chapter~\ref{chapter5} for orthogonal polyhedra and extended in Chapter~\ref{chapter6} to 4-oriented polyhedra seems to perform poorly on $c$-oriented polyhedra, with $c\geqslant 5$. Indeed, the number of configurations quickly blows up and our method becomes less effective. A better technique should be designed for these more complex polyhedra.

\item As proved in Theorem~\ref{orth}, placing a guard on each reflex edge of an orthogonal polyhedron makes it searchable. We conjecture this holds also for general polyhedra, but it turns out to be a surprisingly deep problem. One way to lower our quadratic bound on the number of guards would be to slightly modify the partition used in Theorem~\ref{speed1}. Instead of aiming the guards at the angle bisectors of the reflex edges, we could aim them in any direction, cut the polyhedron with the \textquotedblleft extended\textquotedblright\ searchplanes, and still eliminate all the reflex edges. The advantage would be that, by carefully choosing a plane for each reflex edge that minimizes the intersections with other reflex edges, the overall number of intersections could be significantly less than quadratic. Even if it is still quadratic in the worst case, it is likely much lower on average, with respect to any reasonable probability measure over polyhedra.

\item In the case of orthogonal polyhedra, the upper bound of $r$ guards given by Theorem~\ref{orth} is also not known to be optimal. In fact, we believe that it can still be lowered by a constant factor.

\item The search time of the schedules given in Theorems~\ref{heur} and~\ref{orth} could be dramatically reduced by clearing several regions in parallel.  For instance, in Theorem~\ref{orth} we could turn in concert two guards whose fences do not bound a same cuboid. Generalizing, we could construct a graph $G$ on the set of reflex edges, with an edge for every pair of reflex edges whose fences bound a same cuboid. Then the search time would be proportional to the chromatic number $\chi(G)$, which is likely sublinear, at least on average.

\item The complexity of \TSSP could be considered also for a restricted set of instances, such as genus-zero polyhedra or orthogonal polyhedra. Moreover, the technique used for Theorem~\ref{hard} could perhaps suggest a way to prove the \NP-hardness of \SSP. Adding new elements to \SSP in order to increase its expressiveness, while preserving its 2-dimensional nature, could be an intermediate step. For example, we could introduce \emph{curtains}, i.e., line segments that block visibility but do not block movement, or \emph{windows}, i.e., line segments that block movement but not visibility. Then we could attempt to replicate the construction of Theorem~\ref{hard} for this extension of \SSP.

\item We may try to modify the constructions used in Chapter~\ref{chapter11} to prove the \PSPACE-hardness of the full \TSSP, as well. We observe that such constructions are not guard-visible, as per Definition~\ref{viable}. As a matter of fact, incorporating regions that can never be cleared is an effective way to force the recontamination of other critical regions when certain conditions are met. However, if we want our constructions to be guard-visible instances of \TSSP, no such expedient is of any use, and cleverer tools have to be devised.

\item Observe that, throughout Part~\ref{part3}, we focused primarily on boundary guards (with the exception of Theorem~\ref{one}, whose statement encompasses also general guards). This is certainly a bonus for Theorems~\ref{speed1},~\ref{heur},~\ref{orth},~\ref{hard},~\ref{hardopt} and~\ref{trssp}, in which restricting to boundary guards yields stronger statements. On the other hand, extending the concept of filling guard (Definition~\ref{def:exhaustive}) to guards not lying on the boundary may lead to interesting generalizations of Theorems~\ref{exhaustive} and~\ref{sequentiality}.
\end{itemize}

\end{conclusions}

\end{document}